\documentclass[twoside,leqno,pdftex]{report}

\pdfoutput=1

\usepackage{RR}
\usepackage[a4paper]{geometry}
\usepackage{a4wide}

\usepackage[utf8]{inputenc} 
\usepackage[T1]{fontenc} 

\usepackage[english]{babel}

\usepackage[numbers]{natbib}

\usepackage{xr-hyper}
\usepackage
  [pdftex, extension=pdf, colorlinks=true,
   citecolor=blue, filecolor=blue, linkcolor=blue, urlcolor=blue]{hyperref}

\usepackage{amssymb, amsmath, mathtools}
\usepackage{amsthm, thmtools} 
\usepackage{amsfonts, bbold, dsfont}

\usepackage{caption}
\usepackage[nottoc,other]{tocbibind}

\usepackage{minitoc} 

\usepackage[normalem]{ulem}

\usepackage{todonotes}

\usepackage[boxed]{algorithm2e}

\usepackage{enumitem}
\newenvironment{CompactList}{%
  \setlength{\parskip}{0pt}
  \begin{itemize}[nosep]}{%
  \end{itemize}}

\usepackage{rotating}
\usepackage{graphicx}
\graphicspath{{./}}

\usepackage{color}

\definecolor{darkred}{rgb}{0.7,0.1,0.1}
\definecolor{darkgreen}{rgb}{0.1,0.7,0.1}
\definecolor{darkblue}{rgb}{0.1,0.1,0.7}
\definecolor{orange}{rgb}{1,0.5,0}
\definecolor{lightred}{RGB}{255,232,232}
\definecolor{lightgreen}{RGB}{232,255,232}
\definecolor{lightblue}{RGB}{232,255,255}
\definecolor{lightorange}{rgb}{1,0.7,0.3}
\definecolor{lightgray}{gray}{0.9}

\newcommand{\assume}[1]{\textcolor{darkred}{\bf\boldmath{#1}}}
\newcommand{\assumemm}[1]{\textcolor{darkred}{{#1}}}

\newcommand{\defcolor}{lightgreen}
\newcommand{\lemcolor}{lightblue}
\newcommand{\thmcolor}{lightred}
\newcommand{\rmkcolor}{lightgray}

\newcommand{\coloredbox}[2]{\fcolorbox{black}{#1}{#2}}
\newcommand{\defbox}[1]{\coloredbox{\defcolor}{#1}}
\newcommand{\lembox}[1]{\coloredbox{\lemcolor}{#1}}
\newcommand{\thmbox}[1]{\coloredbox{\thmcolor}{#1}}
\newcommand{\rmkbox}[1]{\coloredbox{\rmkcolor}{#1}}

\usepackage{epic,eepic}

\definecolor{metriccolor}{RGB}{225,255,255}
\definecolor{completecolor}{RGB}{195,255,255}
\definecolor{vectorcolor}{RGB}{255,225,225}
\definecolor{normedcolor}{RGB}{255,195,195}
\definecolor{innercolor}{RGB}{255,135,135}
\definecolor{hilbertcolor}{RGB}{212,162,212}

\definecolor{tanIII}{RGB}{205,135,63}
\definecolor{yellowII}{RGB}{238,238,0}
\definecolor{darkOliveGreenIII}{RGB}{162,205,90}
\definecolor{turquoiseII}{RGB}{0,229,238}

\newcommand{\colorA}{tanIII}
\newcommand{\colorB}{yellowII}
\newcommand{\colorC}{darkOliveGreenIII}
\newcommand{\colorD}{turquoiseII}

\usepackage{tikz}
\usetikzlibrary{shapes}
\usetikzlibrary{shadows}
\usetikzlibrary{mindmap,trees,snakes}
\usetikzlibrary{arrows}

\usetikzlibrary{tikzmark,calc}
\tikzset{math3d/.style={x= {(-0.55cm,-0.55cm)}, z={(0cm,1cm)},y={(1cm,0cm)}}}

\newcommand{\thref}[1]{\ref{#1} ({\em \nameref{#1}})}
\newcommand{\threfc}[2]{\ref{#1} ({\em \nameref{#1}}, \uline{#2})}

\newcommand{\proofpar}[1]{{\bf\boldmath{#1}.}}
\newcommand{\proofparskip}[1]{\medskip\noindent\proofpar{#1}}

\newcommand{\FEM}{Finite Element Method}

\newcommand{\fem}{finite element method}

\newcommand{\vectorspace}{vector space}

\newcommand{\nonempty}{nonempty}
\newcommand{\nontrivial}{nontrivial}
\newcommand{\nonzero}{nonzero}
\newcommand{\nonaffine}{non-affine}
\newcommand{\nondegenerate}{nondegenerate}

\newcommand{\software}[1]{{\sf #1}}
\newcommand{\coq}{\software{Coq}}
\newcommand{\isabelle}{\software{Isabelle/HOL}}

\newcommand{\st}{{\,|\,}}
\newcommand{\leftst}{\;\left|\;}
\newcommand{\rightst}{\;\right|\;}

\newcommand{\cf}{cf.}
\newcommand{\eg}{e.g.}
\newcommand{\ie}{i.e.}

\newcommand{\half}{{\frac{1}{2}}}

\newcommand{\p}{\partial}

\newcommand{\dps}{\displaystyle}

\newcommand{\fhi}{\varphi}

\newcommand{\lbrabra}{[\hspace*{-0.12em}[}
\newcommand{\rbrabra}{]\hspace*{-0.12em}]}

\newcommand{\lacac}{\{\hspace*{-0.32em}\{}
\newcommand{\racac}{\}\hspace*{-0.32em}\}}

\newcommand{\makespace}[1]{\quad#1\quad}
\newcommand{\Equiv}{\Leftrightarrow}
\newcommand{\EQUIV}{\makespace{\Longleftrightarrow}}
\newcommand{\Implies}{\Rightarrow}
\newcommand{\IMPLIES}{\makespace{\Longrightarrow}}
\newcommand{\Conj}{\land}

\newcommand{\CONJ}{\makespace{\Conj}}
\newcommand{\Disj}{\lor}

\newcommand{\AND}{\makespace{\mathrm{and}}}

\newcommand{\eqdef}{\stackrel{\mathrm{def.}}{=}}
\newcommand{\equivdef}{\stackrel{\mathrm{def.}}{\EQUIV}}
\newcommand{\nin}{\not\in}

\newcommand{\fact}[1]{{#1!}}
\newcommand{\ffact}[1]{{#1\boldsymbol{!}}}
\newcommand{\kron}[2]{\delta_{#1#2}}
\newcommand{\kkron}[2]{\ddelta_{#1,#2}}

\newcommand{\binomnp}{\binom{n}{p}}
\newcommand{\binomkpdd}{\binom{k+d}{d}}

\newcommand{\calA}{\mathcal{A}}
\newcommand{\calB}{\mathcal{B}}
\newcommand{\calC}{\mathcal{C}}

\newcommand{\calE}{\mathcal{E}}
\newcommand{\calF}{\mathcal{F}}
\newcommand{\calH}{\mathcal{H}}

\newcommand{\calL}{\mathcal{L}}
\newcommand{\calM}{\mathcal{M}}

\newcommand{\calP}{\mathcal{P}}

\newcommand{\calR}{\mathcal{R}}
\newcommand{\calS}{\mathcal{S}}

\newcommand{\matC}{\mathbb{C}}
\newcommand{\matK}{\mathbb{K}}
\newcommand{\matN}{\mathbb{N}}
\newcommand{\matP}{\mathbb{P}}
\newcommand{\matQ}{\mathbb{Q}}
\newcommand{\matR}{\mathbb{R}}
\newcommand{\matS}{\mathbb{S}}

\renewcommand{\aa}{{\bf a}}
\newcommand{\cc}{{\bf c}}
\newcommand{\cco}{{\cc_0}}
\newcommand{\dd}{{\bf d}}
\newcommand{\ee}{{\bf e}}
\renewcommand{\gg}{{\bf g}}
\newcommand{\uu}{{\bf u}}

\newcommand{\vv}{{\bf v}}
\newcommand{\xx}{{\bf x}}
\newcommand{\xxo}{{\xx_0}}
\newcommand{\yy}{{\bf y}}
\newcommand{\yyo}{\yy_0}

\newcommand{\XX}{{\bf X}}
\newcommand{\zzero}{{\bf 0}}
\newcommand{\oone}{{\bf 1}}

\newcommand{\aalpha}{{\boldsymbol{\alpha}}}
\newcommand{\bbeta}{{\boldsymbol{\beta}}}
\newcommand{\ggamma}{{\boldsymbol{\gamma}}}
\newcommand{\ddelta}{{\boldsymbol{\delta}}}
\newcommand{\llambda}{{\boldsymbol{\lambda}}}
\newcommand{\mmu}{{\boldsymbol{\mu}}}
\newcommand{\pphi}{{\boldsymbol{\phi}}}
\newcommand{\ffhi}{{\boldsymbol{\fhi}}}

\newcommand{\cfhi}{\check{\fhi}}

\newcommand{\caalpha}{\check{\aalpha}}
\newcommand{\cbbeta}{\check{\bbeta}}
\newcommand{\cggamma}{\check{\ggamma}}
\newcommand{\cddelta}{\check{\ddelta}}

\newcommand{\ha}{\hat{a}}
\newcommand{\hf}{\hat{f}}
\newcommand{\hg}{\hat{g}}
\newcommand{\hp}{\hat{p}}
\newcommand{\hq}{\hat{q}}
\newcommand{\hv}{\hat{v}}
\newcommand{\hx}{\hat{x}}

\newcommand{\hH}{\hat{H}}

\newcommand{\hK}{\hat{K}}
\newcommand{\hP}{\hat{P}}

\newcommand{\hU}{\hat{U}}
\newcommand{\halpha}{\hat{\alpha}}

\newcommand{\hsigma}{\hat{\sigma}}

\newcommand{\hSigma}{\hat{\Sigma}}

\newcommand{\haa}{\hat{\aa}}
\newcommand{\hgg}{\hat{\gg}}
\newcommand{\hvv}{\hat{\vv}}
\newcommand{\hxx}{\hat{\xx}}
\newcommand{\hyy}{\hat{\yy}}
\newcommand{\hmmu}{{\hat{\mmu}}}

\newcommand{\hcalH}{\hat{\calH}}
\newcommand{\hcalL}{\hat{\calL}}

\newcommand{\hmatP}{\hat{\matP}}

\newcommand{\bpi}{{b^\prime_i}}

\newcommand{\fp}{{f^\prime}}

\newcommand{\ip}{{i^\prime}}
\newcommand{\jp}{{j^\prime}}

\newcommand{\up}{{u^\prime}}
\newcommand{\vp}{{v^\prime}}

\newcommand{\xp}{{x^\prime}}

\newcommand{\Ep}{{E^\prime}}
\newcommand{\Fp}{{F^\prime}}

\newcommand{\Kp}{{K^\prime}}

\newcommand{\alphap}{\alpha^\prime}

\newcommand{\lambdap}{{\lambda^\prime}}

\newcommand{\phip}{{\phi^\prime}}

\newcommand{\calBp}{{\calB^\prime}}

\newcommand{\calEp}{{\calE^\prime}}
\newcommand{\calFp}{{\calF^\prime}}

\newcommand{\ccp}{{\cc^\prime}}
\newcommand{\ccpo}{{\cc^\prime_0}}
\newcommand{\uup}{{\uu^\prime}}
\newcommand{\uupo}{{\uu^\prime_0}}
\newcommand{\vvp}{{\vv^\prime}}
\newcommand{\xxp}{{\xx^\prime}}
\newcommand{\xxpo}{{\xx^\prime_0}}
\newcommand{\yyp}{{\yy^\prime}}
\newcommand{\yypo}{{\yy^\prime_0}}

\newcommand{\aalphap}{\aalpha^\prime}
\newcommand{\bbetap}{\bbeta^\prime}

\newcommand{\xpp}{{x^{\prime\prime}}}

\newcommand{\tf}{{\tilde{f}}}

\newcommand{\tp}{{\tilde{p}}}
\newcommand{\tq}{{\tilde{q}}}
\newcommand{\tr}{{\tilde{r}}}
\newcommand{\ts}{{\tilde{s}}}

\newcommand{\txx}{\tilde{\xx}}

\newcommand{\tXX}{{\widetilde{\XX}}}

\newcommand{\taalpha}{{\widetilde{\aalpha}}}
\newcommand{\tbbeta}{{\widetilde{\bbeta}}}
\newcommand{\tggamma}{{\widetilde{\ggamma}}}
\newcommand{\tddelta}{{\widetilde{\ddelta}}}
\newcommand{\tfhi}{{\widetilde{\fhi}}}

\newcommand{\uaa}{\underline{\aa}}
\newcommand{\uvv}{\underline{\vv}}

\newcommand{\calAkd}{{\calA}^d_k}
\newcommand{\calAld}{{\calA}^d_l}
\newcommand{\calAko}{{\calA}^0_k}
\newcommand{\calAki}{{\calA}^1_k}

\newcommand{\calAkdmi}{{\calA}^{d-1}_k}
\newcommand{\calAkmid}{{\calA}^d_{k-1}}
\newcommand{\calAkdi}{{\calA}^d_{k,i}}
\newcommand{\calAkdo}{{\calA}^d_{k,0}}
\newcommand{\calAkdone}{{\calA}^d_{k,1}}
\newcommand{\calAiiidiiione}{{\calA}^3_{3,1}}
\newcommand{\calAkpid}{{\calA}^d_{k+1}}
\newcommand{\calAod}{{\calA}^d_0}
\newcommand{\calAid}{{\calA}^d_1}
\newcommand{\calAiiid}{{\calA}^d_3}
\newcommand{\calAiiidii}{{\calA}^2_3}
\newcommand{\calAiidiii}{{\calA}^3_2}
\newcommand{\calAiiidiii}{{\calA}^3_3}

\newcommand{\calCd}[1]{\calC^d_{#1}}
\newcommand{\calCkd}{\calCd{k}}
\newcommand{\calCkpid}{\calCd{k+1}}
\newcommand{\calCld}{\calCd{l}}
\newcommand{\calCod}{\calCd{0}}
\newcommand{\calCid}{\calCd{1}}
\newcommand{\calCiid}{\calCd{2}}
\newcommand{\calCiiid}{\calCd{3}}
\newcommand{\calCdpi}[1]{\calC^{d+1}_{#1}}
\newcommand{\calCkdpi}{\calCdpi{k}}
\newcommand{\calCdmi}[1]{\calC^{d-1}_{#1}}

\newcommand{\calCkpidmi}{\calCdmi{k+1}}
\newcommand{\calCki}{\calC^1_k}
\newcommand{\calCkii}{\calC^2_k}
\newcommand{\calCodii}{\calC^2_0}
\newcommand{\calCidii}{\calC^2_1}
\newcommand{\calCiidii}{\calC^2_2}
\newcommand{\calCiiidii}{\calC^2_3}
\newcommand{\calCiiidiii}{\calC^3_3}
\newcommand{\Cinf}{\calC^\infty}
\newcommand{\CinfRdR}{\Cinf(\matRd,\matR)}
\newcommand{\CinfRdRd}{\Cinf(\matRd,\matRd)}

\newcommand{\calSt}{{\widetilde{\calS}}}
\newcommand{\calStkd}[1]{\calSt^d_{k,#1}}
\newcommand{\calStkdi}{\calStkd{i}}
\newcommand{\calStkdpi}[1]{\calSt^{d+1}_{k,#1}}
\newcommand{\calStkdpii}{\calStkdpi{i}}
\newcommand{\calStkpid}[1]{\calSt^d_{k+1,#1}}
\newcommand{\calStkpido}{\calStkpid{0}}
\newcommand{\calStkpidi}{\calStkpid{i}}
\newcommand{\calSc}{\check{\calS}}
\newcommand{\calSckd}[1]{\calSc^d_{k,#1}}
\newcommand{\calSckdi}{\calSckd{i}}

\newcommand{\FElagP}[2]{\prescript{L\!a\!g}{}{\matP}_{#1}^{#2}}
\newcommand{\FElagPref}[2]{\prescript{L\!a\!g}{}{\hmatP}_{#1}^{#2}}
\newcommand{\FElagQ}[2]{\prescript{L\!a\!g}{}{\matQ}_{#1}^{#2}}

\newcommand{\matNstar}{\matN^\star}

\newcommand{\matNd}{\matN^d}
\newcommand{\matPdpi}[1]{\matP^{d+1}_{#1}}

\newcommand{\matPkdpi}{\matPdpi{k}}
\newcommand{\matPkmidpi}{\matPdpi{k-1}}
\newcommand{\matPd}[1]{\matP^d_{#1}}
\newcommand{\matPkd}{\matPd{k}}
\newcommand{\matPld}{\matPd{l}}
\newcommand{\matPnd}{\matPd{n}}
\newcommand{\matPkpid}{\matPd{k+1}}
\newcommand{\matPkmid}{\matPd{k-1}}

\newcommand{\matPkpld}{\matPd{k+l}}
\newcommand{\matPlmid}{\matPd{l-1}}
\newcommand{\matPnmid}{\matPd{n-1}}
\newcommand{\matPnmiid}{\matPd{n-2}}
\newcommand{\matPod}{\matPd{0}}
\newcommand{\matPid}{\matPd{1}}
\newcommand{\matPiid}{\matPd{2}}

\newcommand{\matQd}[1]{\matQ^d_{#1}}
\newcommand{\matQkd}{\matQd{k}}
\newcommand{\matQid}{\matQd{1}}
\newcommand{\matPl}[1]{\matP^l_{#1}}
\newcommand{\matPkl}{\matPl{k}}
\newcommand{\matPil}{\matPl{1}}
\newcommand{\matPdmi}[1]{\matP^{d-1}_{#1}}
\newcommand{\matPkdmi}{\matPdmi{k}}
\newcommand{\matPkpidmi}{\matPdmi{k+1}}
\newcommand{\matPkmidmi}{\matPdmi{k-i}}
\newcommand{\matPkpimidmi}{\matPdmi{k+1-i}}
\newcommand{\matPldmi}{\matPdmi{l}}

\newcommand{\matPndmi}{\matPdmi{n}}
\newcommand{\matPidmi}{\matPdmi{1}}
\newcommand{\matPodmi}{\matPdmi{0}}
\newcommand{\matPi}[1]{\matP^1_{#1}}
\newcommand{\matPki}{\matPi{k}}
\newcommand{\matPli}{\matPi{l}}
\newcommand{\matPoi}{\matPi{0}}
\newcommand{\matPii}{\matPi{1}}
\newcommand{\matPkpli}{\matPi{k+l}}
\newcommand{\matQi}[1]{\matQ^{1}_{#1}}
\newcommand{\matQki}{\matQi{k}}
\newcommand{\matPk}[1]{\matP^{#1}_k}
\newcommand{\matPkii}{\matPk{2}}
\newcommand{\matPkiii}{\matPk{3}}

\newcommand{\RT}{\mathbb{RT}}
\newcommand{\BDM}{\mathbb{BDM}}
\newcommand{\BDFM}{\mathbb{BDFM}}
\newcommand{\Ned}{\mathbb{N}}
\newcommand{\Hdiv}{{\bf H}(\mathrm{div})}
\newcommand{\Hcurl}{{\bf H}(\mathrm{curl})}

\newcommand{\matRplus}{\matR_+}

\newcommand{\matRd}{\matR^d}
\newcommand{\matRdmi}{\matR^{d-1}}

\newcommand{\matRkpi}{\matR^{k+1}}
\newcommand{\matRl}{\matR^l}
\newcommand{\matRlmi}{\matR^{l-1}}
\newcommand{\matRlpi}{\matR^{l+1}}
\newcommand{\matRn}{\matR^n}
\newcommand{\matRnpi}{\matR^{n+1}}
\newcommand{\matRp}{\matR^p}
\newcommand{\matRq}{\matR^q}

\renewcommand{\emptyset}{\varnothing}
\renewcommand{\leq}{\leqslant}
\renewcommand{\geq}{\geqslant}

\newcommand{\Interior}[1]{\stackrel{\circ}{#1}}

\newcommand{\Ball}{\calB}
\newcommand{\OpenBall}{\Ball^o}

\newcommand{\ind}[1]{{\pi^{#1}}}
\newcommand{\indl}{\ind{l}}
\newcommand{\indd}{\ind{d}}
\newcommand{\indinv}[1]{(\ind{#1})^{-1}}
\newcommand{\indlinv}{\indinv{l}}
\newcommand{\inddinv}{\indinv{d}}
\newcommand{\permc}[2]{{c^{#1}_{#2}}}
\newcommand{\permcd}[1]{{\permc{d}{#1}}}
\newcommand{\permcdo}{{\permcd{0}}}
\newcommand{\permcinv}[2]{{(c^{#1}_{#2})^{-1}}}
\newcommand{\permcdinv}[1]{\permcinv{d}{#1}}
\newcommand{\permtrsp}[2]{{\tau^{#1}_{#2}}}

\newcommand{\permtrspd}[1]{{\permtrsp{d}{#1}}}

\newcommand{\card}{\mathrm{card}}

\newcommand{\Span}[1]{{\mathrm{span}\left(#1\right)}}

\newcommand{\Ker}[1]{{\ker\left(#1\right)}}
\newcommand{\Rg}[1]{{\mathrm{rg}\left(#1\right)}}

\newcommand{\AffSp}[2]{{\calA\!f\!\!f\!\left(#1, #2\right)}}
\newcommand{\FSp}[2]{{{#2}^{#1}}}
\newcommand{\LSp}[2]{{\calL \left(#1, #2\right)}}
\newcommand{\LcSp}[2]{{\calL_c \left(#1, #2\right)}}

\newcommand{\FEF}{\FSp{E}{F}}

\newcommand{\LEF}{\LSp{E}{F}}
\newcommand{\LEpF}{\LSp{\Ep}{F}}

\newcommand{\LFG}{\LSp{F}{G}}
\newcommand{\LEG}{\LSp{E}{G}}

\newcommand{\AffEF}{\AffSp{E}{F}}
\newcommand{\AffcalEpF}{\AffSp{\calEp}{F}}
\newcommand{\AffcalEpcalFp}{\AffSp{\calEp}{\calFp}}
\newcommand{\AffcalFpcalEp}{\AffSp{\calFp}{\calEp}}
\newcommand{\AffFG}{\AffSp{F}{G}}
\newcommand{\AffEG}{\AffSp{E}{G}}
\newcommand{\AffFE}{\AffSp{F}{E}}
\newcommand{\AffRR}{\AffSp{\matR}{\matR}}
\newcommand{\AffRdR}{\AffSp{\matRd}{\matR}}
\newcommand{\AffRdRd}{\AffSp{\matRd}{\matRd}}
\newcommand{\AffRlRd}{\AffSp{\matRl}{\matRd}}
\newcommand{\AffRdmiRd}{\AffSp{\matRdmi}{\matRd}}
\newcommand{\AffRpRq}{\AffSp{\matRp}{\matRq}}

\newcommand{\LcEF}{\LcSp{E}{F}}

\newcommand{\Arrow}[2]{#1\to#2}
\newcommand{\ArEF}{\Arrow{E}{F}}

\newcommand{\ArcalEpF}{\Arrow{\calEp}{F}}
\newcommand{\ArcalEpcalFp}{\Arrow{\calEp}{\calFp}}

\newcommand{\ArRR}{\Arrow{\matR}{\matR}}

\newcommand{\ArRdRd}{\Arrow{\matRd}{\matRd}}
\newcommand{\ArRlRd}{\Arrow{\matR^l}{\matRd}}

\newcommand{\ArRpRq}{\Arrow{\matR^p}{\matR^q}}
\newcommand{\ArRdR}{\Arrow{\matRd}{\matR}}
\newcommand{\ArRdmiR}{\Arrow{\matR^{d-1}}{\matR}}

\newcommand{\Identity}{\mathrm{Id}}
\newcommand{\identity}[1]{\Identity_{#1}}

\newcommand{\vertiii}[1]{
  {\left\vert\kern-0.25ex\left\vert\kern-0.25ex\left\vert #1
    \right\vert\kern-0.25ex\right\vert\kern-0.25ex\right\vert}}

\newcommand{\norm}[2]{\left\|#2\right\|_{#1}}
\newcommand{\tnorm}[2]{\vertiii{#2}_{#1}}

\newcommand{\lto}[1]{<^{\mathrm{#1}}}
\newcommand{\ltlex}{\lto{lex}}
\newcommand{\ltcolex}{\lto{colex}}
\newcommand{\ltrevlex}{\lto{revlex}}
\newcommand{\ltsymlex}{\lto{symlex}}
\newcommand{\ltgrlex}{\lto{grlex}}
\newcommand{\ltgrevlex}{\lto{grevlex}}
\newcommand{\ltgrcolex}{\lto{grcolex}}
\newcommand{\ltgrsymlex}{\lto{grsymlex}}

\newcommand{\len}[1]{\left|#1\right|}

\newcommand{\nE}[1]{\norm{E}{#1}}
\newcommand{\nEdot}{\nE{\cdot}}

\newcommand{\nF}[1]{\norm{F}{#1}}
\newcommand{\nFdot}{\nF{\cdot}}

\newcommand{\tnEF}[1]{\tnorm{F, E}{#1}}

\newcommand{\famvert}[2]{{\lbrabra#1\rbrabra}^{#2}}
\newcommand{\famvertd}[1]{{\famvert{#1}{d}}}
\newcommand{\famvertdmi}[1]{{\famvert{#1}{d-1}}}
\newcommand{\famverti}[1]{{\famvert{#1}{1}}}

\newcommand{\famnode}[2]{{\lacac#1\racac}_{\!#2}}
\newcommand{\famnodekd}[1]{{\famnode{#1}{\calAkd}}}
\newcommand{\famnodekdmi}[1]{{\famnode{#1}{\calAkdmi}}}
\newcommand{\famnodekmid}[1]{{\famnode{#1}{\calAkmid}}}
\newcommand{\famnodeod}[1]{{\famnode{#1}{\calAod}}}
\newcommand{\famnodeid}[1]{{\famnode{#1}{\calAid}}}

\newcommand{\famnodeki}[1]{{\famnode{#1}{\calAki}}}

\newcommand{\hKgen}[1]{\hK^{#1}}
\newcommand{\hKi}{\hKgen{1}}
\newcommand{\hKd}{\hKgen{d}}
\newcommand{\hKl}{\hKgen{l}}
\newcommand{\Kv}[1]{K^{\famvert{\vv}{#1}}}
\newcommand{\Kvi}{K^{\famverti{v}}}
\newcommand{\Kvd}{\Kv{d}}
\newcommand{\Khvd}{K^{\famvertd{\hvv}}}
\newcommand{\Hvd}{H^{\famvertd{\vv}}}
\newcommand{\calHvd}{\calH^{\famvertd{\vv}}}
\newcommand{\calFvd}[1]{\calF^{\famvertd{\vv}}_{#1}}
\newcommand{\calFvdindo}{\calFvd{\ind{0}}}
\newcommand{\calFvdindl}{\calFvd{\indl}}
\newcommand{\calFvdindd}{\calFvd{\indd}}
\newcommand{\Fvd}[1]{F^{\famvertd{\vv}}_{#1}}
\newcommand{\Fvdindl}{\Fvd{\indl}}
\newcommand{\Fvdindd}{\Fvd{\indd}}

\newcommand{\barc}[1]{\lambda^{#1}}
\newcommand{\barcvd}{\barc{\famvertd{\vv}}}
\newcommand{\bbarc}[1]{\llambda^{#1}}
\newcommand{\bbarcvd}{\bbarc{\famvertd{\vv}}}
\newcommand{\isobaryc}[1]{g^{#1}}
\newcommand{\isobarycv}{\isobaryc{\famverti{v}}}
\newcommand{\bisobaryc}[1]{\gg^{#1}}
\newcommand{\bisobarycv}{\bisobaryc{\famvertd{\vv}}}
\newcommand{\hisobaryc}[1]{\hg^{#1}}
\newcommand{\hisobaryci}{\hisobaryc{1}}
\newcommand{\hbisobaryc}[1]{\hgg^{#1}}
\newcommand{\hbisobarycd}{\hbisobaryc{d}}
\newcommand{\Zkd}[2]{\zeta_{#1}^{#2}}
\newcommand{\zkd}{\Zkd{k}{d}}

\newcommand{\phigeoi}[1]{\varphi_{\mathrm{geo}}^{\famverti{#1}}}
\newcommand{\phigeoiv}{\phigeoi{v}}
\newcommand{\invphigeoiv}{\left(\phigeoiv\right)^{-1}}
\newcommand{\phigeo}{\ffhi_{\mathrm{geo}}}
\newcommand{\phigeod}[1]{{\phigeo^{\famvertd{#1}}}}
\newcommand{\phigeodv}{\phigeod{\vv}}
\newcommand{\invphigeod}[1]{\left(\phigeod{#1}\right)^{-1}}
\newcommand{\invphigeodv}{\invphigeod{\vv}}
\newcommand{\Ageod}[1]{A_{\mathrm{geo}}^{\famvertd{#1}}}
\newcommand{\Ageodv}{\Ageod{\vv}}
\newcommand{\invAgeod}[1]{\left(\Ageod{#1}\right)^{-1}}
\newcommand{\invAgeodv}{\invAgeod{\vv}}

\newcommand{\phind}[2]{\pphi_{#1}^{\famvertd{#2}}}
\newcommand{\phindl}[1]{\phind{\indl}{#1}}
\newcommand{\phindv}[1]{\phind{#1}{\vv}}
\newcommand{\phindlv}{\phindl{\vv}}
\newcommand{\phindd}[1]{\phind{\indd}{#1}}
\newcommand{\phinddv}{\phindd{\vv}}
\newcommand{\invphindlv}{\left(\phindl{\vv}\right)^{-1}}
\newcommand{\invphinddv}{\left(\phindd{\vv}\right)^{-1}}
\newcommand{\phipermcd}[1]{\phind{\permcd{#1}}{\vv}}
\newcommand{\invphipermcd}[1]{\left(\phipermcd{#1}\right)^{-1}}
\newcommand{\phipermtrspd}[1]{\phind{\permtrspd{#1}}{\vv}}
\newcommand{\Bind}[2]{B_{\ind{#1}}^{\famvertd{#2}}}
\newcommand{\Bindlv}{\Bind{l}{\vv}}

\newcommand{\fkd}[1]{f^d_{k,#1}}
\newcommand{\fkdo}{\fkd{0}}
\newcommand{\fkdi}{\fkd{i}}
\newcommand{\tfkdo}{\tf^d_{k,0}}

\newcommand{\thetinj}[2]{\theta_{#1}^{#2}}
\newcommand{\thetinjd}[1]{\thetinj{#1}{d}}
\newcommand{\thetinjdmi}[1]{\thetinj{#1}{d-1}}

\newcommand{\pder}{\partial}
\newcommand{\Diff}{\mathrm{D}}

\newcommand{\PropPP}{P}
\newcommand{\PropPPmn}{{\PropPP(m,n)}}
\newcommand{\PropQQ}{Q}
\newcommand{\PropQQm}{{\PropQQ(m)}}

\newcommand{\nsh}{n_{\mathrm{sh}}}
\newcommand{\hSigmakd}{\hSigma^d_k}
\newcommand{\hSigmaki}{\hSigma^1_k}

\newcommand{\calLp}{\calL^p}

\newcommand{\Ltwo}{L^2}

\newcommand{\ToPPki}{\texorpdfstring{$\matPki$}{Pk1}}

\newcommand{\ToPPid}{\texorpdfstring{$\matPid$}{P1d}}

\newcommand{\ToPFEki}{\texorpdfstring{$\FElagP{k}{1}$}{Pk1}}
\newcommand{\ToPFErefki}{\texorpdfstring{$\FElagPref{k}{1}$}{Pk1}}

\newcommand{\ToPFEkd}{\texorpdfstring{$\FElagP{k}{d}$}{Pkd}}

\declaretheoremstyle[
  headfont=\normalfont\bfseries,
  notefont=\normalfont\bfseries\itshape,
  bodyfont=\normalfont]{mydef}
\declaretheoremstyle[
  headfont=\normalfont\bfseries,
  notefont=\normalfont\bfseries\itshape,
  bodyfont=\normalfont\itshape]{mythm}
\declaretheoremstyle[
  headfont=\normalfont\itshape,
  notefont=\normalfont\itshape,
  bodyfont=\normalfont]{myrmk}

\newcommand{\mydeclaretheoremshaded}[3]{
  \declaretheorem[
    #1,
    shaded={
      textwidth=0.99\textwidth,
      bgcolor=#2,
      rulecolor=black,
      rulewidth=0.5pt}]{#3}
}

\newcommand{\mydeclaretheoremstyle}[3]{
  \mydeclaretheoremshaded{numberlike=theorem, style=#1}{#2}{#3}
}

\newcommand{\thm}[1]{#1 theorem} 

\newcommand{\BL}{Beppo Levi}
\newcommand{\BLt}{\thm{{\BL} (monotone convergence)}}

\newcommand{\Cara}{Carathéodory}
\newcommand{\Carat}{\thm{{\Cara}'s extension}}

\newcommand{\EdPkd}{of Euclidean division by monomial in~$\matPkd$}
\newcommand{\EdPkdL}{Lemma {\EdPkd}}
\newcommand{\EdPkdl}{lemma {\EdPkd}}

\newcommand{\Fl}{Fatou's lemma}

\newcommand{\FL}{Fatou--Lebesgue}
\newcommand{\FLt}{\thm{\FL}}

\newcommand{\Ldcv}{Lebesgue's dominated convergence}
\newcommand{\Ldcvt}{\thm{\Ldcv}}

\newcommand{\Ledcv}{Lebesgue's extended dominated convergence}
\newcommand{\Ledcvt}{\thm{\Ledcv}}

\newcommand{\LedcvLpt}{{\Ledcvt} in~$\calLp$}

\newcommand{\LM}{Lax--Milgram}
\newcommand{\LMT}{\thm{\LM}}

\newcommand{\LMC}{Lax--Milgram--Céa}
\newcommand{\LMCT}{\thm{\LMC}}

\newcommand{\RFi}{Riesz--Fischer}
\newcommand{\RFit}{\thm{\RFi}}

\newcommand{\RFr}{Riesz--Fréchet}
\newcommand{\RFrT}{\thm{{\RFr} representation}}

\newcommand{\Tont}{\thm{Tonelli}}

\newcommand{\UPkd}{of unisolvence of~$\FElagP{k}{d}$}
\newcommand{\UPkdT}{Theorem {\UPkd}}
\newcommand{\UPkdt}{theorem {\UPkd}}

\newcommand{\MIndAkd}{of multi-indices~$\calAkd$ and~$\calCkd$}
\newcommand{\MIndAkdD}{Definition {\MIndAkd}}

\newcommand{\LagPkd}{of Lagrange nodes}
\newcommand{\LagPkdD}{Definition {\LagPkd}}

\mydeclaretheoremshaded{style=mythm}{\thmcolor}{theorem}
\mydeclaretheoremstyle{mythm}{\lemcolor}{lemma}
\mydeclaretheoremstyle{mydef}{\defcolor}{definition}
\mydeclaretheoremstyle{myrmk}{\rmkcolor}{remark}

\mydeclaretheoremshaded{numbered=no, name={\LMCT}}{\thmcolor}{lmcthm}
\mydeclaretheoremshaded{numbered=no, name={\LMT}}{\thmcolor}{lmthm}
\newcommand{\rLbbf}{Representation lemma for bounded bilinear forms}
\mydeclaretheoremshaded{numbered=no, name={\rLbbf}}{\lemcolor}{bilinlem}
\mydeclaretheoremshaded{numbered=no, name={\RFrT}}{\thmcolor}{rfrthm}
\newcommand{\opTcs}{\thm{Orthogonal projection} for a complete subspace}
\mydeclaretheoremshaded{numbered=no, name={\opTcs}}{\thmcolor}{projthm}
\newcommand{\fpT}{\thm{Fixed point}}
\mydeclaretheoremshaded{numbered=no, name={\fpT}}{\thmcolor}{fpthm}

\mydeclaretheoremshaded{numbered=no, name={\RFit}}{\thmcolor}{rfithm}
\mydeclaretheoremshaded{numbered=no, name={\LedcvLpt}}{\thmcolor}{ledcLpthm}
\mydeclaretheoremshaded{numbered=no, name={\Ledcvt}}{\thmcolor}{ledcthm}
\mydeclaretheoremshaded{numbered=no, name={\Ldcvt}}{\thmcolor}{ldcthm}
\mydeclaretheoremshaded{numbered=no, name={\FLt}}{\thmcolor}{falthm}
\mydeclaretheoremshaded{numbered=no, name={\Fl}}{\thmcolor}{flem}
\mydeclaretheoremshaded{numbered=no, name={\Tont}}{\thmcolor}{ftthm}
\newcommand{\Eutpm}{Existence and uniqueness of the tensor product measure}

\mydeclaretheoremshaded{numbered=no, name={\Eutpm}}{\thmcolor}{eutpmlem}
\mydeclaretheoremshaded{numbered=no, name={\BLt}}{\thmcolor}{blthm}
\mydeclaretheoremshaded{numbered=no, name={\Carat}}{\thmcolor}{cthm}

\newcommand{\mct}{\thm{monotone class}}
\newcommand{\Dplt}{%
  \thm{Dynkin \texorpdfstring{$\pi$}{pi}--\texorpdfstring{$\lambda$}{lambda}}}
\newcommand{\Abst}{{\Dplt} / {\mct}}
\mydeclaretheoremshaded{numbered=no, name={\Abst}}{\thmcolor}{absthm}

\mydeclaretheoremshaded{numbered=no, name={\UPkdT}}{\thmcolor}{upkdthm}
\mydeclaretheoremshaded{numbered=no, name={\EdPkdL}}{\thmcolor}{edpkdlem}
\mydeclaretheoremshaded{numbered=no, name={\MIndAkdD}}{\defcolor}{mindakddef}
\mydeclaretheoremshaded{numbered=no, name={\LagPkdD}}{\defcolor}{lagpkdnodesdef}

\newcounter{prevtheorem}
\newcounter{nexttheorem}
\stepcounter{nexttheorem}

\RRdate{September 2024}

\RRauthor{
  François Clément%
  \thanks[serena]{%
    Inria \&
    CERMICS, École des Ponts, 77455 Marne-la-Vallée Cedex 2, France.
    \texttt{Francois.Clement@inria.fr}.}
  \and
  Vincent Martin%
  \thanks[lmac]{%
    Université de technologie de Compiègne,
    LMAC, 
    60203 Compiègne, France. 
    \texttt{Vincent.Martin@utc.fr}.}
}
\authorhead{%
  F. Clément, \& V. Martin}

\RRtitle{%
  Méthode des éléments finis.\\
  Preuves détaillées en vue d'une formalisation en Coq}
\RRetitle{%
  Finite element method.\\
  Detailed proofs to be formalized in Coq}
\titlehead{%
  Detailed proofs for the finite element method}

\RRnote{%
  This work was partly supported by the European Research Council (ERC) under
  the European Union's Horizon 2020 Research and Innovation Programme – Grant
  Agreement n$^\circ$810367.}

\RRresume{%
  Pour obtenir la plus grande confiance dans l'exactitude des résultats de
  programmes de simulation numérique pour la résolution d'Équations aux
  Dérivées Partielles (EDPs), il faut formaliser les notions mathématiques
  sur lesquelles ils sont basés.
  La méthode des éléments finis est l'un des outils les plus utilisés pour la
  résolution de larges gammes d'EDPs.
  L'objectif de ce document est de fournir à la communauté des chercheurs en
  preuve formelle des preuves papiers très détaillées pour la construction des
  éléments finis de Lagrange de tout degré sur des simplexes en dimension
  quelconque.
}
\RRabstract{%
  To obtain the highest confidence on the correction of numerical simulation
  programs for the resolution of Partial Differential Equations (PDEs), one has
  to formalize the mathematical notions and results that allow to establish the
  soundness of the approach.
  The finite element method is one of the popular tools for the numerical
  resolution of a wide range of PDEs.
  The purpose of this document is to provide the formal proof community with
  very detailed pen-and-paper proofs for the construction of the Lagrange
  finite elements of any degree on simplices in positive dimension.
}

\RRmotcle{%
  équation aux dérivées partielles,
  méthode des éléments finis,
  éléments finis de Lagrange,
  simplexe,
  preuve mathématique détaillée,
  preuve formelle en analyse réelle}
\RRkeyword{%
  partial differential equation,
  finite element method,
  Lagrange finite elements,
  simplex,
  detailed mathematical proof,
  formal proof in real analysis}

\RRprojet{Serena}

\RCParis

\begin{document}

\RRNo{9557}
\makeRR

\chapter*{Foreword}

\newcommand{\HALRR}{https://hal.inria.fr/hal-04713897}
\newcommand{\Version}[1]{\href{{\HALRR}v#1}{Version~#1}}

\noindent
This document is intended to evolve over time.
Last version is release~1.0 ({\ie} version~1).\\
It is available at \url{\HALRR/}.

\bigskip\noindent
\Version{1} (release 1.0, 2024/09/30) is the first release.\\
It covers:
\begin{CompactList}
\item the general definition of finite element;
\item some results about simplicial geometry;
\item
  the construction of simplicial Lagrange finite elements on a segment,
  including results about $\matPki$~Lagrange polynomials;
\item
  the construction of simplicial Lagrange finite elements in
  dimension~$d\geq1$, including:
  \begin{CompactList}
  \item the construction of multi-indices of given maximum length;
  \item
    results about multivariate polynomials, such as
    the linear independence of monomials, and
    the Euclidean division by a monomial;
  \item
    results about~$\matPid$ Lagrange polynomials, such as
    their view as barycentric coordinates;
  \item
    results about affine geometric mappings such as
    the transformation of $l$-faces of dimension~$l\leq d$;
  \item results about Lagrange nodes and Lagrange linear forms of~$\matPkd$;
  \item the proofs of unisolvence of~$\matPkd$, and of face unisolvence.
  \end{CompactList}
\end{CompactList}

\dominitoc
\setcounter{minitocdepth}{3}

\tableofcontents

\part{Overview}
\label{p:overview}

\chapter{Introduction}
\label{c:introduction-1}

\section{Formal proof}
\label{s:formal-proof}

A formal proof is conducted in a logical framework that provides dedicated
computer programs to mechanically check the validity of the proof, the
so-called formal proof assistants.
Such formal proofs may concern known mathematical theorems, but also properties
of some piece of other computer programs, {\eg} see~\cite{har:fpt:08},
and~\cite[Glossary p. 343]{bol:tcm:14}.
This field of computer science is extremely popular as it allows to certify
with no doubt the behavior of critical programs.

Interactive theorem provers are now known to be able to tackle real analysis.
For instance, in the field of Ordinary Differential Equations (ODEs) with
{\isabelle}~\cite{ih:nao:12,imm:fvc:14,it:fod:16}, and
{\coq}~\cite{ms:pao:13}, or in the field of Partial Differential Equations
(PDEs), again with {\isabelle}~\cite{ap:ihf:16}, and
{\coq}~\cite{bol:wen:13,bol:tcm:14}.
In the latter example, the salient aspect is that the round-off error due to
the use of IEEE-754 floating-point arithmetic can also be fully taken into
account.
But the price to pay is that all the details of the proofs has to be dealt
with, and thus the availability of very detailed pen-and-paper proofs is a
major asset.

\section{Objective}
\label{s:objective}

Our current long-term purpose is to formally prove programs implementing the
{\FEM} (FEM).
The FEM is widely used to solve a broad class of PDEs, mainly because of its
sound mathematical foundation, see Section~\ref{ss:fem-ref}.
A description of the typical various components that are necessary to solve a
partial differential equation with the FEM is shown in Figure~\ref{f:FEM-map}.

Consider a physical stationary problem over a domain~$\Omega\subset\matRd$,
where~$d$ is often~1, 2 or~3.
First, a {\em continuous} variational formulation is set on some functional
spaces, typically Hilbert or Banach spaces (complete normed vector spaces,
equipped or not with a scalar product) over the domain, denoted by~$V$.
In all cases, $V$ is infinite-dimensional and thus cannot be represented on a
computer.
In favorable cases, a theorem, such as the {\LMT} ({\eg}
see~\cite{cm:lmt:16}), provides existence and uniqueness of the continuous
solution.
Here, ``continuous'' means that the solution is a general function defined
over the {\em whole}~$\Omega$, in contrast to a ``discrete'' solution that is
entirely defined by a {\em finite} number of data.

Second, a {\em discrete} variational formulation is built.
To do so, a mesh is constructed on~$\Omega$, {\ie} a finite set of polygons
that covers~$\Omega$ (with additional properties).
These polygons are called geometric elements.
A {\em finite element}, {\ie} a space of polynomials up to a certain degree
and the type of primary values that are computed (such as values at certain
points, means, fluxes), is chosen in these geometric elements.
This defines the finite dimensional approximation space~$V_h$ that
(inexactly) represents the continuous Hilbert or Banach space~$V$.
The discrete variational formulation is also chosen.
In some cases, it is simply the restriction of the continuous variational
formula to~$V_h$.
In most cases, additional terms are introduced to improve some numerical
features (such as robustness and convergence), this is often called the
numerical scheme.
A theorem should provide existence, uniqueness of the discrete solution, and
the convergence in some sense of this discrete solution towards the
continuous solution.
This discrete formulation also involves generally an approximation to compute
the integrals, which is called the quadrature.

Finally, the discrete system (linear or nonlinear) is built and solved
numerically.
The result is a vector of floating points that can be stored and used to
visualize the solution.

In addition, {\em a posteriori} computations can be devised to improve the
results, with a change of mesh for instance.
When the problem is unsteady, the discrete part of the procedure is
partly repeated: one tries to re-use as much as possible the computed data to
avoid numerous expensive computations.


\begin{figure}[t]
  \centering
      \begin{tikzpicture}
    \tikzstyle{rectr}=
      [draw, rounded corners=3pt, text width=3.1cm, align=center];
    \tikzstyle{rectd}=[draw, text width=2.2cm, align=center];
    \tikzstyle{fl}=[->, >=latex, very thick];
    \tikzstyle{flr}=[->, >=stealth', thin,dashed, red];

    \node[rectr, fill=\colorA] (CVF) at (4,5.5)
      {Continuous Variational Formulation};
    \node[rectr, fill=\colorB] (DVF) at (4,3.5)
      {Discrete Variational Formulation};
    \node[rectr, fill=\colorC] (Res) at (4,2) {Resolution};
    \node[rectr, fill=\colorD] (Vis) at (4,0.5) {Visualization};
    \node[rectr, fill=\colorD] (APos) at (7.7,2.75) {A posteriori};

    \draw[fl] (CVF) to (DVF);
    \draw[fl] (DVF) to (Res);
    \draw[fl] (Res) to (Vis);
    \draw[fl, dashed, rounded corners=3pt]
      (Res.east) -- (6,2) -- (6,3.5) -- (DVF.east);

    \node[rectd] (LM) at (0.5,5.5) {{\small {Lax--Milgram}}};
    \node[rectd] (Hilb) at (7.5,5.5) {{\small {Hilbert space}}};
    \node[rectd] (Msh) at (0.5,4.4) {{\small Mesh}};
    \node[rectd] (FE) at (0.5,3.8) {{\small {Finite Element}}};
    \node[rectd] (Quad) at (0.5,3.2) {{\small {Quadrature}}};
    \node[rectd] (Sch) at (0.5,2.6) {{\small Scheme}};
    \node[rectd] (LinS) at (0.5,1.6) {{\small Linear Solver}};
    \node[rectd] (NLS) at (0.5,0.7) {{\small Non Linear Solver}};

    \draw[flr] (LM) to (CVF);
    \draw[flr] (Hilb) to (CVF);
    \draw[flr] (LM) to (DVF);
    \draw[flr] (Msh.east) to[bend left] (DVF.170);
    \draw[flr] (FE) to (DVF);
    \draw[flr] (Quad.east) to (DVF.183);
    \draw[flr] (Sch.east) to[bend right] (DVF.191);
    \draw[flr] (LinS) to (Res);
    \draw[flr] (NLS.east) to[bend right] (Res.200);
  \end{tikzpicture}
  \caption[Components of a typical (stationary) FEM formulation.]{%
    Brief recall of the various components of a typical (stationary) FEM
    formulation.\protect\\
    The black thick arrows depict the usual construction steps when solving a
    problem via the FEM.
    The black thick dashed arrow shows the possible feedback loop that
    {\em a posteriori} analysis allows, to change the mesh, the finite
    element ({\ie} the order of polynomial approximation, the quadrature
    formula or the numerical scheme).
    The red dashed arrows depict when a theorem or a numerical tool takes
    place.}
  \label{f:FEM-map}
\end{figure}
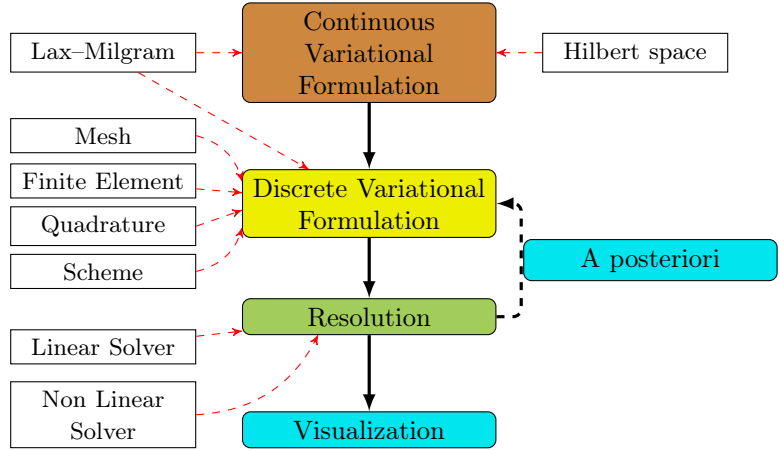

\bigskip

The {\LMT}, one of the key ingredients to establish the FEM in some cases,
was already addressed, see~\cite{cm:lmt:16} for a detailed pen-and-paper
proof, and~\cite{bol:cfp:17} for a formal proof in~{\coq}.
A part of the Lebesgue integration theory was also addressed,
see~\cite{cm:li:21} for detailed pen-and-paper proofs,
and~\cite{bol:cfl:22,bol:cfl:23} for a formalization in~{\coq}.

\bigskip

The present document is a further contribution to our ultimate goal.
It focuses only on the definition of the finite element itself as a triple,
see Chapter~\ref{c:fe}, and on the construction of Lagrange finite elements
of degree~$k$ on $d$-simplices, see Chapter~\ref{c:Pkd-lag-fe}.
These Lagrange finite elements are defined over simplices (segments, triangles
or tetrahedra when~$d=1,2,3$), with polynomials having a total degree no
greater than~$k$ (denoted $\matPkd$ in the sequel), and with nodal values that
are evenly distributed over the whole simplex, including its vertices.

\section{{The finite elements}}
\label{s:fem}

In this document, $d\in\matNstar$ denotes the space dimension (typically
$d\in\{1,2,3\}$), and $k\in\matN$ denotes the {\em degree of approximation},
{\ie} some polynomial degree.

Let us first define two polynomial spaces that are useful in the sequel.
$\matPkd$, as already said, is the {\vectorspace} of polynomials of~$d$
variables and of total degree at most~$k$,
\[
  \matPkd \eqdef \left\{
    \left( \xx \longmapsto
      \sum_{\aalpha \in \matNd, \sum_{i = 1}^d \alpha_i \leq k}
      a_\aalpha x_1^{\alpha_1} \ldots x_d^{\alpha_d} \right) : \ArRdR
    \rightst
    \left.
    \vphantom{%
      \sum_{\aalpha \in \matNd, \sum_{i = 1}^d \alpha_i \leq k}
      a_\aalpha x_1^{\alpha_1} \ldots x_d^{\alpha_d}}
    a_\aalpha \in \matR
  \right\}.
\]
$\matQkd$ is the {\vectorspace} of polynomials of~$d$ variables and of degree
at most~$k$ for each variable,
\[
  \matQkd \eqdef \left\{
    \left( \xx \longmapsto
      \sum_{\bbeta \in \matNd, \forall i \in [1..d], \beta_i \leq k}
      a_\bbeta x_1^{\beta_1} \ldots x_d^{\beta_d} \right) : \ArRdR
    \rightst
    \left.
      \vphantom{
      \sum_{\bbeta \in \matNd, \forall i \in [1..d], \beta_i \leq k}
      a_\bbeta x_1^{\beta_1} \ldots x_d^{\beta_d}}
    a_\bbeta \in \matR
  \right\}.
\]
$\matPkd$ is the polynomial approximation space for many simplicial elements,
thus it will be used in the whole document, whereas~$\matQkd$, used mainly for
hexahedra, is necessary for this introductory part only.
For all $k\in\matN$, we have obviously $\matPki=\matQki$, and, for all
$d\geq2$, $\matPkd\subsetneq\matQkd$.
Note that~$\matPid$ is the space of affine functions to~$\matR$.

\subsection{Some references}
\label{ss:fem-ref}

The mathematical literature dedicated to the FEM is very rich, FEM is a very
active field of research.
Thus, it is not intended in the present document to be comprehensive, nor to
review the various approaches to this method.
We limit ourselves to some books among many other references:
Ciarlet~\cite{cia:fem:78},
Ern~\cite{ern:ef:05},
Ern and Guermond~\cite{eg:tpf:04,eg:fe1:21,eg:fe2:21,eg:fe3:21},
Babu{\v{s}}ka and Strouboulis~\cite{bs:fem:01},
Quarteroni and Valli~\cite{qv:nap:94},
Brenner and Scott~\cite{bs:mtf:08},
and Zienkiewicz, Taylor and Zhu~\cite{ztz:fem:13}.
The first part of the Handbook of Numerical Analysis is dedicated to the finite
element methods~\cite{hna:cia:91,hna:wah:91,hna:rt:91,hna:bo:91,hna:fs:91}.

The interested reader will find many references in all these books.

\subsection{The finite element as a triple}
\label{ss:fe-triple}

To define a finite element, in addition to~$d$ and~$k$, some other notations
are necessary: $q\in\matNstar$ denotes the size of the physical unknowns
($q$~is generally~1 or~$d$), and~$\nsh\in\matNstar$ represents the number of
so-called local {\em shape functions}, {\ie} a number of linear forms on~$P$
(see below).
A ``correct'' finite element needs to satisfy the unisolvence property, that
relates~$\Sigma$ (see below) and~$P$ (and thus~$\nsh$ and~$k$).

Following~\cite{cia:fem:78} and~\cite{eg:fe1:21}, a finite element is
mathematically defined as a triple $(K,P,\Sigma)$, where:
\begin{itemize}
\item
  $K$~is a geometrical ``element'', {\ie} a closed bounded set with {\nonempty}
  interior.
  Typically, $K$~is a {\nondegenerate} polygonal set (a polyhedron), or the
  image of such a set by a regular function.

\item
  $P$~is a {\nonzero}, finite-dimensional {\vectorspace} of functions
  from~$\matRd$ to~$\matRq$.
  It is meant to be a space of polynomial functions, or the image of such a
  space via regular functions.

\item
  $\Sigma$ is a collection of linear forms on $P$,
  $\Sigma\eqdef(\sigma_i)_{i\in[1..\nsh]}$ where each $\sigma_i$ is a linear
  form  on~$P$.
  The application $\phi_\Sigma:P\to\matR^{\nsh}$ defined by
  $\phi_\Sigma(p)\eqdef(\sigma_i(p))_{i\in[1..\nsh]}$ is an isomorphism.
\end{itemize}

The bijectivity of~$\phi_\Sigma$ is called {\em unisolvence}.
The~$\sigma_i$ are the degrees of freedom, and~$\nsh$ is the number of degrees
of freedom.
It can be easily seen that the unisolvence requires that the dimension of $P$
be equal to~$\nsh$.

\bigskip

A finite element is generally defined over what is called a {\em reference}
geometrical element, a regular polyhedron~$\hK$: the reference FE is
denoted~$(\hK,\hP,\hSigma)$.
Typically, $\hK$~is the unit rectangle simplex, see below
Section~\ref{sss:geom-simplex-vs-quad} or Definition~\ref{d:ref-simplex}.
It can also be the unit $d$-cube.
This reference geometrical element~$\hK$ is transported to a so-called
{\em current} element of the mesh denoted by~$K$, by a geometric mapping that
we call~$\phigeo^K:\ArRdRd$, see a simple 2D~example in
Figure~\ref{f:geo-map-d=2}.
The finite element~$(K,P,\Sigma)$ on~$K$ is then somewhat the image
of~$(\hK,\hP,\hSigma)$ by $\phigeo^K$, see~\cite[Sec. 9.1]{eg:fe1:21}.

The 1D~Lagrange finite element is constructed in this way in
Chapter~\ref{c:Pk1-lag-fe}.
Note that in Chapter~\ref{c:Pkd-lag-fe}, the general case~$d\geq1$ is treated
the other way around: we deal directly with a general {nondegenerate}
$d$-simplex, and then introduce the reference Lagrange finite element as a
particular case.
The reason for this is explained at the end of
Section~\ref{s:sketch-of-the-proof-of-unisolvence-Pkd}.

According to the hypotheses on~$\phigeo^K$, $K$~can be a polyhedron with planar
faces, or not, and can be convex, or not.
Not all geometrical elements are possible, though.
In this document, we assume that the mesh is {\em affine}, which means
that~$\phigeo^K$ is supposed to be invertible and affine.
The {\nonaffine} meshes are more tricky to handle, see
Section~\ref{sss:geom-simplex-vs-quad} below.
Note that the geometric mapping is affine iff each of its components is
in~$\matPid$.
With a slight abuse of notation, the geometrical element is commonly
labelled~$\matPid$, $\matQid$ or~$\matPiid$\ldots when the components
of~$\phigeo^K$ are in the corresponding space.

In this document, the FE~triple is specified in Definition~\ref{d:fe-triple},
using the hypothesis of an affine mesh, thus~$K$ is assumed to be a polyhedron.

\bigskip

Note that one can also define a finite element as a quadruple, adding to the
previous triple a linear interpolation operator that maps continuously
functions in a larger space containing~$P$ (to be defined) to~$P$
itself, see~\cite[Sec. 5.3.]{eg:fe1:21}.
The construction of the interpolation operator is not covered in the present
version of the document.

\subsection{Lagrange finite elements}
\label{ss:fe-lag}

There exist various families of finite elements, that are more or less adapted
to the considered problem.
For each family, there are usually a version for the simplices, and another for
the quadrangles (when~$d=2$), or the hexahedra (when~$d=3$).
Some finite elements can also be constructed on other geometries such as
hexagons ($d=2$) or prisms ($d=3$).

First, let us discuss briefly the geometry of the FE. Then, we introduce the
most common families of~FE, before presenting the nodal finite elements, and in
particular Lagrange FE.

\subsubsection{Typology of geometries: simplex vs quadrangular}
\label{sss:geom-simplex-vs-quad}

When the geometrical element is a simplex, or when it is a $d$-cuboid (defined
as $\prod_{i=1}^d[x_i, y_i]$, for some $(x_i)_{i\in[1..d]}$ and
$(y_i)_{i\in[1..d]}$ in $\matR$, with $x_i<y_i$), the geometrical
mapping~$\phigeo^K$ is affine, thus covered by the assumption in this
document.
For cuboids, note that the~FE can be defined by a tensorization of~1D
finite elements, see~\cite[Sec. 6.4]{eg:fe1:21}.

However, dealing with {\nonaffine} geometric element may be necessary.
For instance, to better approximate the geometry of the domain (think of
airplane wings for instance), there exist {\em curved} simplicial finite
elements: for instance, $\matPiid$-simplices are simplices where the
geometrical mapping is in~$(\matPiid)^d$.
Thus, when~$d=2$, each edge of~$K$ is a part of a parabola,
see~\cite[Chap. 13]{eg:fe1:21}.
Another important {\nonaffine} mesh case in practice, is the
$\matQid$-quadrangles and hexahedra: in this case, $\phigeo^K$~is
in~$(\matQid)^d$ (thus not affine {\it a priori}).
In~3D in particular, this may give complicated geometries, as general
{\nondegenerate} hexahedra are not always convex, and their faces are not
necessarily planar.

This leads to complex numerical issues: meshing the domain~$\Omega$ can be a
very hard task (it is {\em not} simple, even with simplices).
One needs to avoid ``wrong''  geometrical elements (for instance
self-intersecting), the approximation results are also more involved, the
integration needs more care, and so on.

\bigskip

To conclude on geometrical element, we say once again that we limit ourselves
to affine, invertible geometrical mapping, and more specifically to geometrical
elements that are {\nondegenerate} standard simplices.

In this case, the reference simplex~$\hK$ is the unit rectangle simplex, see
Definition~\ref{d:ref-simplex}.
It is the convex envelop of the reference vertices
$\famvertd{\hvv}\eqdef(\hvv_0,\ldots,\hvv_n)$, where~$\hvv_0$ is the canonical
origin~$\zzero\eqdef(0,\dots,0)^T\in\matRd$, and for all $i\in[1..d]$, $\hvv_i$
corresponds to the $i$-th canonical basis vector
$\ee_i\eqdef(0,\dots,0,1,0,\dots,0)^T\in\matRd$ (the~1 is in position~$i$).

Given~$d+1$ points $\famvertd{\vv}\eqdef(\vv_0,\ldots,\vv_n)$ in~$\matRd$, the
simplex~$K^{\famvertd{\vv}}$ is the convex envelop of~$\famvertd{\vv}$, see
Definition~\ref{d:simplex}.
It is proven {\nondegenerate} iff~$\famvertd{\vv}$ is affinely independent.
The geometrical mapping that maps~$\hK$ to~$K$ is defined as, see
Defintion~\ref{d:geo-mapping},
$\phigeo^K(\hxx)\eqdef\sum_{i=0}^d\hcalL^{1,d}_i(\hxx)\,\vv_i$ where
$(\hcalL^{1,d}_i)_{i\in[0..d]}$ is the Lagrange basis of~$\matPid$, such that
$\hcalL^{1, d}_0(\hxx)\eqdef1-\sum_{j=0}^d\hx_j$, and for all $i\in[1..d]$,
$\hcalL^{1, d}_i(\hxx)\eqdef\hx_i$.
It is easy to see that~$\phigeo^K$ is affine.

\subsubsection{Some finite elements}
\label{sss:some-FE}

Various approximation spaces~$P$ and linear forms~$\Sigma$ exist in the
literature.
Just to mention some of the most popular (original references can be found
in~\cite{eg:fe1:21}):
\begin{itemize}
\item
  for the {\em nodal} FE, the linear forms are the evaluation at some
  points~$\aa_i$, for $i\in[1..\nsh]$, that are called the FE~{\em nodes}:
  $\sigma_i(p)=p(\aa_i)$.
  The Lagrange FE are the most common nodal FE, see~\cite[Sec.~6.4
  and~7.4]{eg:fe1:21}, and below.

\item
  the Hermite FE are nodal FE, but the linear forms are the evaluation at the
  nodes of both the function and its derivatives (assuming enough regularity
  for the function), see~\cite{cr:glh:72}.

\item
  for the {\em modal} FE in~1D, the linear forms are the integrals of the
  product of the function and Legendre polynomials.
  It can be tensorized for cuboid elements, see~\cite[Sec.~6.3.2
  and~6.4.2]{eg:fe1:21}.

\item
  canonical hybrid FE mix evaluations at nodal values and integrals of the
  function over edges, faces, and volumes (in~3D) of the geometrical element,
  see~\cite[Sec. 7.6]{eg:fe1:21}.

\item
  the~$\RT$, $\BDM$, $\BDFM$\ldots are {\em flux} FE,
  see~\cite[Chap. 14]{eg:fe1:21}.
  They aim at approximating~$\Hdiv$ functions (vector functions in~$(\Ltwo)^d$
  whose divergence is in~$\Ltwo$).
  Some of the linear forms are integrals of the normal of the function across
  the $(d-1)$-faces of the simplex.

\item
  the Nedelec family~$\Ned$ of~FE is meant for the approximation of~$\Hcurl$
  functions (functions in~$(\Ltwo)^d$ whose curl is in~$(\Ltwo)^d$).
  Some of the linear forms are integrals of the tangential of the function
  along the edges of the simplex, see~\cite[Chap. 15]{eg:fe1:21}.
\end{itemize}

\subsubsection{Some nodal finite elements, Lagrange finite elements}
\label{sss:nodal-fe-lag}

For the nodal finite elements, there exist various choices of nodes.
The Lagrange FE are based on Lagrange nodes: these nodes are evenly distributed
in the element, and when~$k\geq1$, the vertices are always some of the nodes,
see examples for the reference simplices in Figure~\ref{f:FE-lag-P1234}.
The precise definition using barycentric coordinates is given in
Definition~\ref{d:lag-nodes-Pkd}.
For cuboids, the Lagrange nodes are merely a tensorization of the segment
nodes.

For simplices, the approximation space is~$\matPkd$, and the Lagrange FE is
denoted~$\FElagP{k}{d}$.
In this document, we clearly separate the notations for the approximation
space and for the finite element triple.
In cuboids, the approximation space is~$\matQkd$, and the Lagrange FE is
denoted~$\FElagQ{k}{d}$.
In cuboids, it is possible to remove some internal nodes and to reduce the
space dimension of the approximation space, which reduces the computational
effort, while preserving the accuracy.
These~FE are called serendipity finite elements, and the approximation
space is denoted~$\matS_k^d(\subsetneq\matQkd)$,
see~\cite[Sec. 6.4.3]{eg:fe1:21}.

The polynomial interpolation problem consists in finding a polynomial~$p$ to
approximate a given function~$f$, with the constraint that $f(\xx_i)=p(\xx_i)$
for a finite family of chosen points $(\xx_i)_{i\in I}$.
It is well known that choosing Lagrange nodes as interpolation points gives
rise to an interpolating polynomial that can have very large oscillations.
These oscillations tend to increase when the polynomial degree~$k$ increases.
This may lead to poor approximation results and to numerical difficulty for
integration.
This is why some other nodes may be favored.
For instance, in 1D, one can take Gauss--Lobatto nodes, that are the two
vertices of the segment and the roots of some Chebyshev polynomials,
see~\cite[Sec.~6.3.5]{eg:fe1:21}.
This improves drastically the high oscillation effects.
The Gauss--Lobatto nodes can be easily tensorized for cuboid elements.

\bigskip

In the present document, we focus on Lagrange simplicial FE,
because~$\FElagP{1}{d}$ and~$\FElagP{2}{d}$ are the most popular FE, because
treating~$\FElagP{k}{d}$ for all~$k$ and~$d$ is of great interest by itself,
and because~ $\matPkd$ is an approximation space that is used for many~FE.
Thus Lagrange simplicial FE is a natural first step in our development.


\section{Our main sources}
\label{s:our-main-sources}

Our main source for the proofs in the present text is~\cite{eg:fe1:21}.
The part on affine geometry is inspired by~\cite{gos:cms4:98}.

Note that, because the support for differential calculus is still limited in
{\coq}, we choose to favor proof paths using algebraic arguments over those
using analysis results.
As a consequence, some of our proofs are distinct from those
of~\cite{eg:fe1:21}.
In particular, we do not rely on differential calculus to prove the lemmas on
geometric mappings that relate various configurations of simplices or faces.
As we assume that the meshes are affine ({\ie} in our case, made up of
simplices that are {\nondegenerate} and have planar hyperfaces), and as we do
not need to integrate to define the Lagrange finite element, the results on
affine maps and affine spaces are sufficient.
Nonetheless, some results of topology (and of analysis) are obviously required
to prove that {\nondegenerate} simplices have {\nonempty} interior.

To obtain the dimension of~$\matPkd$, defined as the linear span of monomials,
we prove the linear independence of these monomials by calculating their
partial derivatives and taking the value at~$\zzero$.
A possible alternative, that is not followed in the {\coq} formalization, is to
use the Euclidean division to prove an isomorphism between~$\matPkd$ and
$\matPkmid\times\matPkdmi$, and conclude on the dimension by induction (see
Remark~\ref{r:isom-dim-Pkd}).

\section{Contents}
\label{s:contents}

The present version of this document covers all material up to the construction
of the simplicial Lagrange finite elements in dimension~$d\geq1$.

After the general definition of the notion of finite element, and some results
about simplicial geometry, it starts with the construction of the simplicial
Lagrange finite elements on a segment including results about~$\matPki$
Lagrange polynomials; this step is mostly for pedagogical purpose.
Then, for any dimension~$d\geq1$, it covers the construction of multi-indices
of given maximum length.
This allows to write the definition of multivariate polynomial spaces as the
linear span of monomials, and some results such as
the linear independence of monomials using partial derivation and providing the
dimension of the considered space of polynomial,
the product of polynomials,
the composition of polynomials and affine mappings, and the
Euclidean division by a monomial.
Then, there are some results about~$\matPid$ Lagrange polynomials, with in
particular their view as barycentric coordinates.
This provides results about affine geometric mappings, such as the
transformation of $l$-faces of dimension~$l\leq d$.
This allows to pass from the reference simplex in dimension~$l$ to a current
element in dimension~$d$.
The case $l=d$ is the most important one and is treated first.
Next, some results about Lagrange nodes and Lagrange linear forms of~$\matPkd$
are given, before the proof of unisolvence of~$\matPkd$.
This allows to conclude with the construction of the simplicial Lagrange finite
elements in dimension~$d\geq1$, and the statement and proof of face
unisolvence.
The formalization in~{\coq} of most of these aspects is presented
in~\cite{mou:phd:24,bol:cfs:24}.

It is planned to add more results in a forthcoming version.

\section{Teaching}
\label{s:teaching}

This document is not primarily meant for teaching usage.
The objective was to be as comprehensive as possible in the proofs.
This led to very detailed demonstrations, and to a compact style of writing
that is not common, and may seem daunting to the uninformed reader.

However, the authors tried to give some insights on the FEM theory and
on the proofs and theorems in the introductory chapters.
They also strove to give some indications in the proofs when they felt it
necessary.
They believe that this document could be useful for interested teachers, and
dedicated students, in complement to the usual manuals.

\section{Disclaimer}
\label{s:disclaimer}

Note that the manuscript itself is not formally proved (and will never be).
Indeed, {\LaTeX} compilers are not formal proof tools!

Moreover, formalization is not just straightforward translation of mathematical
texts and formulas.
Some design choices have to be made and proof paths may differ, mainly to favor
usability of {\coq} theorems and ease formal developments.
Thus, there may exist differences between the mathematical setting presented
here, and the formal setting developed in {\coq}~\cite{mou:phd:24,bol:cfs:24}.

Hence, despite the care taken in its writing, this document might still be
prone to errors or holes in the demonstrations.
There could also exist simpler paths in the proofs.
Please, feel free to inform the authors of any such issue, and to share any
comments or suggestions\ldots

\section{Organization}
\label{s:organization}

Part~\ref{p:overview} of this document is organized as follows.
After the present introductory Chapter~\ref{c:introduction-1}, the notations
are collected in Chapter~\ref{c:notations}.
The chosen proof paths of the main results are then sketched in
Chapter~\ref{c:statements-and-sketches-of-the-proofs}.

Part~\ref{p:detailed-proofs} (Chapters~\ref{c:introduction-2}
to~\ref{c:Pkd-lag-fe}) is the core of this document.
In this part, the definitions are presented, and the lemmas and theorems are
stated with their detailed proofs.
Its organization is briefly described at the end of the second introductory
Chapter~\ref{c:introduction-2}.

Chapter~\ref{c:conclusions-perspectives} concludes and gives some
perspectives.

Finally, an appendix gathers the list of statements in
Chapter~\ref{c:lists-of-statements}, and explicit dependencies (both ways) in
Chapters~\ref{c:the-proof-cites-explicitly}
and~\ref{c:is-explicitly-cited-in-the-proof-of}.
The appendix is not intended for printing!

\chapter{Notations}
\label{c:notations}

In this chapter (as in most of the document), we use the following
conventions:
\begin{CompactList}
\item $d$ and~$l$ denote a dimension, a (usually nonzero) natural number;

\item
  $k$ denotes the order of a finite element, a natural number, {\eg} the
  maximum degree of polynomials in an approximation space;

\item
  the hat diacritical mark ``~$\hat{\ }~$'' denotes a reference quantity:
  a reference set, {\eg} the reference simplex~$\hKd$, or a reference element
  in a reference set, {\eg} a point~$\hxx$ in~$\hKd$, see below;

\item
  {\bf boldface} denotes finite families, typically of size~$d$, and
  more generally vectors in a {\vectorspace}, or points in an affine
  subspace:
  \begin{CompactList}
  \item
    {\bf boldface} small Greek letters~$\aalpha$ and~$\bbeta$ denote
    multi-indices in~$\matNd$;

  \item
    {\bf boldface} small letters, such as~$\xx$ and~$\hxx$, denote points or
    vectors in~$\matRd$;

  \item
     the components of a finite family~$\xx$ are usually denoted~$x_i$, and by
     default the indices of components start from 1;
  \end{CompactList}

\item
  the tilde (resp. check) diacritical mark ``~$\tilde{\ }~$'' (resp.
  ``~$\check{\ }~$'') denotes a family of size~$d-1$ that corresponds to the
  $d-1$ first (resp. last) components of the regular family of size~$d$, or a
  function of such family.
  For instance, for $\aalpha\in\matNd$, the multi-indices~$\taalpha$
  (resp. $\caalpha$) satisfy $\aalpha=(\alpha_1,\caalpha)=(\taalpha,\alpha_d)$,
  see below.
\end{CompactList}

\medskip\noindent
The following notations and conventions are used throughout this document.
\begin{itemize}[itemsep=0pt,topsep=0pt]
\item Logic:
  \begin{CompactList}
  \item
    Using a compound (tuple of elements of~$X$, or subset of~$X$) in an
    expression at a location where only a single element makes sense, is a
    shorthand for the same expression expanded for all elements of the
    compound;
    for instance, ``$\forall x,\xp\in X$'' means
    ``$\forall x\in X,\forall\xp\in X$'',
    ``$\forall(x_i)_{i\in I}\in X$'' means ``$\forall i\in I$, $x_i\in X$'',
    and ``$x,\xp\calR\xpp$'' means ``$x\calR\xpp$ and $\xp\calR\xpp$'';

  \item ``iff'' is a shorthand for ``if and only if''.
  \end{CompactList}

\item Naive set theory:
  \begin{CompactList}
  \item
    $\card(E)$ denotes the cardinal of some set~$E$, especially when it is
    finite, see Lemma~\ref{l:inj-implies-unisolvence};

  \item
    the cartesian products~$E^{d+1}$, $E^d\times E$ and $E\times E^d$ are
    assimilated,
    see Lemma~\ref{l:Pkd-nondecr-d};

  \item
    the set of functions from~$E$ to~$F$ is either denoted~$\FEF$, or through
    the type annotation ``$\ArEF$''.
    Both compact expressions ``let~$f\in\FEF$'' and ``let~$f:\ArEF$'' mean
    ``let~$f$ be a function from~$E$ to~$F$'';

  \item
    $\Rg{f}$ denotes the range of function~$f:E\to F$, {\ie} the image $f(E)$
    of its domain,
    see Lemma~\ref{l:surj-aff-map-is-full-linear-rg}.
  \end{CompactList}

\item Numbers:
  \begin{CompactList}
  \item
    \emph{positive}/\emph{negative}, and \emph{increasing}/\emph{decreasing}
    are meant in their \emph{strict} sense, {\ie} without the possibility of
    equality.
    Otherwise, we use \emph{nonnegative}/\emph{nonpositive}, and
    \emph{nondecreasing}/\emph{nonincreasing};

  \item
    $[n..p]$ denotes the integer interval $[n,p]\cap\matN$ of values from
    integer~$n$ up to integer~$p$ (both included);

  \item
    $\binom{n}{p}$ denotes the binomial coefficient, {\ie} the number of
    subsets of~$p$ elements of a set of~$n$ elements,
    see Definition~\ref{d:binom-coef};

  \item
    $\delta$ denotes the Kronecker delta function, {\ie} $\kron{i}{i}\eqdef1$
    for all $i\in\matN$, and $\kron{i}{j}\eqdef0$ when $i\neq j$,
    see Definition~\ref{d:canon-fam};

  \item
    $\thetinjd{i}:\Arrow{[0..d]}{[0..d+1]}$ is the ``jump'' enumeration
    function that skips~$i$, see Lemma~\ref{l:jump-enum};

  \item
    $\zzero$ denotes the family of zeros $(0,\ldots,0)$, and
    $\oone$ denotes the family of ones $(1,\ldots,1)$,
    see Definition~\ref{d:canon-fam};

  \item
    $\ee_i$ denotes the $i$-th canonical ``basis'' family with only a~1 in
    $i$-th position,
    see Definition~\ref{d:canon-fam};

  \item
    $\len{\aalpha}$ denotes the length of a multi-index, {\ie} the sum of its
    components,
    see Definition~\ref{d:len-multi-ind};

  \item
    $\ffact{\aalpha}$ denotes the factorial of a multi-index, {\ie} the product
    of the factorials of its components,
    see Definition~\ref{d:fact-multi-ind};

  \item
    $\kkron{\aalpha}{\bbeta}$ denotes the Kronecker delta of two multi-indices,
    {\ie} the product of the Kronecker deltas of their components,
    see Definition~\ref{d:kron-multi-ind};

  \item
    $\calAkd$ denotes the set of multi-indices in dimension~$d$ of sum at
    most~$k$,
    see Definitions~\ref{d:multi-ind-Ak1} (for $d\eqdef1$)
    and~\ref{d:multi-ind-Akd-Ckd};

  \item
    $\calCkd$ denotes the set of multi-indices in dimension~$d$ of sum equal
    to~$k$,
    see Definition~\ref{d:multi-ind-Akd-Ckd};

  \item
    $\calAkdi$ denotes the subset of~$\calAkd$ of multi-indices with zero
    $i$-th component,
    see Definition~\ref{d:multi-ind-Akd-Ckd};

  \item
    $\taalpha$ denotes the first $d-1$ components of a multi-index~$\aalpha$
    in~$\matNd$,
    see Remark~\ref{r:tilde-and-check-notations};

  \item
    $\caalpha$ denotes the last $d-1$ components of a multi-index~$\aalpha$
    in~$\matNd$,
    see Remark~\ref{r:tilde-and-check-notations};

  \item
    $\calSckdi$ (resp.~$\calStkdi$) denotes the $i$-th ``vertical''
    (resp. ``horizontal'') slice of~$\calCkd$,
    see Definition~\ref{d:slices-Sckdi-and-Stkdi};

  \item
    $\fkdo$ (resp. $\tfkdo$) denotes the function $\Arrow{\calAkdmi}{\calCkd}$
    on multi-indices that puts $k$ minus the length at the beginning (resp. at
    the end),
    see Lemma~\ref{l:card-Ckd-Akdm1};

  \item
    $\fkdi:\Arrow{\calAkdmi}{\calAkdi}$ denotes the function on multi-indices
    that inserts~0 in $i$-th position,
    see Lemma~\ref{l:card-Akdi-Akdm1}.
  \end{CompactList}

\item General topology:
  \begin{CompactList}
  \item
    $\Interior{K}$ denotes the interior of some subset~$K$ of some
    topological space,
    see Definition~\ref{d:fe-triple}.
  \end{CompactList}

\item Calculus:
  \begin{CompactList}
  \item
    $\pder^\bbeta$ denotes the partial derivative of order~$\bbeta$, a
    multi-index,
    see Lemma~\ref{l:pder-monom};

    \item
    $\Diff f$ denotes the differential of the function~$f$ of appropriate type,
    see Lemma~\ref{l:diff-lag-pol-ref}.
  \end{CompactList}

\item Linear and affine algebra:
  \begin{CompactList}
  \item
    $\Span{u_1,\ldots,u_n}$ denotes the linear span of the family of
    vectors $(u_i)_{i\in[1..n]}$, with the convention that
    $\Span{\emptyset}\eqdef\{\zzero\}$,
    see Definitions~\ref{LM-d:linear-span} and~\ref{d:pol-space-Pk1}, and
    Lemma~\ref{l:baryc-closure-is-aff-sub-sp};

  \item
    $\Ker{f}$ denotes the kernel of morphism~$f$ (usually a linear map from a
    {\vectorspace} to another),
    see Definition~\ref{LM-d:kernel}, and Lemma~\ref{l:im-ker-incl-ker};

  \item
    $\dim(E)$ denotes the finite dimension of a {\vectorspace}~$E$, {\ie} the
    common cardinality of its bases,
    see Lemma~\ref{l:inj-or-surj-and-dim-implies-bij};

  \item
    $\calM_{n,p}(\matR)$, or simply $\calM_{n,p}$, denotes the space of
    matrices with $n$~lines and $m$~columns,
    see Lemma~\ref{l:equiv-def-aff-map-finite-dim};

  \item
    for a matrix $A\in\calM_{n,p}(\matR)$, where $n,p>0$, the line
    $i\in[1..n]$ of $A$ is denoted by $\underline{A}_i$, and its column
    $j\in[1..p]$ is denoted by $A_j$,
    see Lemma~\ref{l:diff-lag-pol-P1d};

  \item
    $\LEF$ denotes the {\vectorspace} of linear maps from {\vectorspace}~$E$ to
    {\vectorspace}~$F$,
    see Definition~\ref{LM-d:set-of-linear-maps}, and
    Lemma~\ref{l:inj-or-surj-and-dim-implies-bij};

  \item
    $\calEp\eqdef\xxo+\Ep$ denotes the affine subspace of direction~$\Ep$, a
    vector subspace of~$E$, and origin~$\xxo\in E$,
    see Definition~\ref{d:aff-sub-sp};

  \item
    $\AffEF$ denotes the {\vectorspace} of affine maps from {\vectorspace}~$E$
    to {\vectorspace}~$F$,
    see Lemma~\ref{l:space-aff-maps}.
  \end{CompactList}

\item Geometry:
  \begin{CompactList}
  \item
    $\famvert{\vv}{n,d}\eqdef(\vv_0,\dots,\vv_n)$ denotes a family of $n+1$
    points in~$\matR^d$, usually indexed by~$[0..n]$,
    $\famvert{\vv}{d}$ is a shorthand for $\famvert{\vv}{d,d}$,
    see Definition~\ref{d:canon-fam};
    it is the notation used to represent the vertices of the geometry of finite
    elements,
    see Remark~\ref{r:fam-nodes-fam-verts};

  \item
    $\famvert{\vv_\indl}{l,d}$ is the sub-family
    $(\vv_{\indl(i)})_{i\in[0..l]}$ where $\indl:[0..l]\to[0..n]$,
    see Lemma~\ref{l:aff-indep-closed-by-sub-family};

  \item
    $\famvertd{\hvv}\eqdef(\zzero,\ee_1,\ldots,\ee_d)$ denotes the family of
    $d+1$ reference points in~$\matRd$,
    see Definition~\ref{d:fam-ref-aff-pts};

  \item
    $\bisobarycv$ denotes the isobarycenter of the family of
    points~$\famvertd{\vv}$;
    see Definition~\ref{d:isobaryc};

  \item
    $\hbisobarycd$ denotes the isobarycenter of the family of reference
    points~$\famvertd{\hvv}$;
    see Lemma~\ref{l:ref-isobaryc};

  \item
    $\Kvd$ denotes the simplex of vertices~$\famvertd{\vv}$ in~$\matRd$,
    see Definition~\ref{d:simplex};

  \item
    $\hKd$ denotes the reference simplex in~$\matRd$, its vertices are the
    reference points~$\famvertd{\hvv}$,
    see Definition~\ref{d:ref-simplex};

  \item
    $\famnodeki{a}$ denotes~$k+1$ nodes in~$\matR$,
    see Definitions~\ref{d:fam-aff-pts} and~\ref{d:lag-pol-Pk1};

  \item
    $\famnodeki{\ha}$ denotes the reference Lagrange nodes of~$\matPki$ that
    are equally distributed over the reference simplex $[0,1]$ in~$\matR$,
    see Definition~\ref{d:lag-nodes-Pk1-ref};

  \item
    $\famnode{\aa}{\calB_d}$ denotes a family of~$\card(\calB_d)$ points
    in~$\matR^d$,
    see Definition~\ref{d:fam-aff-pts},
    it is the notation used to represent the nodes of the geometry of nodal
    finite elements,
    see Remark~\ref{r:fam-nodes-fam-verts};

  \item
    $\famnodekd{\aa}$ denotes the Lagrange nodes of~$\matPkd$ that are equally
    distributed over the simplex~$\Kvd$ in~$\matRd$,
    see Definition~\ref{d:lag-nodes-Pkd};

  \item
    $\famnodekd{\haa}$ denotes the reference Lagrange nodes of~$\matPkd$ that
    are equally distributed over the reference simplex~$\hKd$ in~$\matRd$,
    see Lemma~\ref{l:lag-nodes-Pkd-ref};

  \item
    $\famvertd{\uvv}$ denotes the sub-vertices of the Lagrange nodes
    of~$\matPkd$ with respect to~$\vv_0$,
    see Definition~\ref{d:sub-vert-lag-nodes-Pkd};

  \item
    $\famnodekmid{\uaa}$ denotes the Lagrange sub-nodes of~$\matPkmid$
    associated with the sub-vertices~$\famvertd{\uvv}$ with respect to~$\vv_0$,
    see Lemma~\ref{l:Pkm1d-sub-nodes-sub-vert-are-some-nodes-Pkd};

  \item
    $\phigeoiv$ denotes the geometric mapping associated with~$\famverti{v}$,
    see Definition~\ref{d:geo-mapping-1d};

  \item
    $\phigeodv$ denotes the geometric mapping associated with~$\famvertd{\vv}$,
    see Definition~\ref{d:geo-mapping};

  \item
    $\Ageodv$ denotes the square matrix
    $\begin{pmatrix}\vv_1-\vv_0&\vv_2-\vv_0&\dots&\vv_d-\vv_0\end{pmatrix}$,
    see Lemma~\ref{l:prop-geo-mapping};

  \item
    $\barcvd_i(\xx)$ denotes the $i$-th barycentric coordinate of~$\xx$ with
    respect to~$\famvert{\vv}{d}$,
    see Lem\-ma~\ref{l:baryc-coor};

  \item
    $\calHvd_i$ denotes the $i$-th face hyperplane with respect to
    $\Kvd$, opposite the vertex~$\vv_i$, see Definition~\ref{d:face-hyperpl};

  \item
    $\hcalH^d_i\eqdef\calH^{\famvertd{\hvv}}_i$ denotes the $i$-th reference
    face hyperplane with respect to $\hKd$, opposite the reference
    vertex~$\hvv_i$,
    see Definition~\ref{d:face-hyperpl};

  \item
    $\Hvd_i$ denotes the $i$-th hyperface of simplex $\Kvd$, opposite the
    vertex~$\vv_i$,
    see Definition~\ref{d:hyperface};

  \item
    $\calFvdindl$ denotes the $l$-face affine space having the
    vertices~$\famvert{\vv_\indl}{l,d}$,
    see Definition~\ref{d:l-face-aff-space};

  \item
    $\Fvdindl$ denotes the $l$-face having the
    vertices~$\famvert{\vv_\indl}{l,d}$,
    see Definition~\ref{d:l-face};

  \item
    $\phindlv$ is the geometric $l$-face mapping associated with the
    $l$-face~$\Fvdindl$,
    see Definition~\ref{d:geo-l-face-mapping}.
  \end{CompactList}

\item Polynomials:
  \begin{CompactList}
  \item
    $X_k\eqdef(x\in\matR\mapsto x^k)$ denotes the monomial of a single variable
    of degree~$k$ in~$\matR$ (it is simply denoted~1 when $k\eqdef0$),
    see Definition~\ref{d:pol-space-Pk1};

  \item
    $\XX^\aalpha\eqdef(\xx\in\matRd\mapsto\prod_{i=1}^dx_i^{\alpha_i})$ denotes
    the monomial of~$d$ variables of degree~$\sum_{i=1}^d\alpha_i$,
    1~is a shortcut for~$\XX^\zzero$ (the constant function of value~1),
    and~$X_i$ is a shortcut for~$\XX^{\ee_i}$,
    see Definition~\ref{d:monom-kd};

  \item
    $\calP(\matRd)$ denotes the algebra of polynomials of~$d$ variables,
    see Lemma~\ref{l:deg-Pi-leq-k};

  \item
    $\deg(p)$ denotes the degree of polynomial~$p$, {\ie} the highest degree of
    its monomials with {\nonzero} coefficients, where the degree of a monomial
    is the sum of the degrees for each variable,
    see Definition~\ref{d:deg-pol};

  \item
    $\matPkd$ denotes the {\vectorspace} of polynomials of~$d$ variables and of
    degree at most~$k$,
    see Definition~\ref{d:pol-space-Pkd};

  \item
    $\xi^d$ denotes the isomorphism between~$\matPod$ and~$\matPodmi$ that
    keeps the value of (constant) polynomials,
    see Lemma~\ref{l:isom-P0d-P0dm1};

  \item
    $\zkd$ denotes the isomorphism between~$\matPkd$
    and~$\matPkdmi\times\matPkmid$ that realizes the Euclidean division
    by~$X_d$,
    see Lemma~\ref{l:isom-Pkd-Pkdm1xPkm1d};

  \item
    $\calL^{\famnodeki{a}}_i$ for $i\in[0..k]$ denotes the family of Lagrange
    polynomials in~$\matPki$ associated with nodes $\famnodeki{a}\in\matR$,
    see Definition~\ref{d:lag-pol-Pk1};

  \item
    $\hcalL^{k,1}_i\eqdef\calL^{\famnodeki{\ha}}_i$ for $i\in[0..k]$ denotes
    the family of reference Lagrange polynomials in~$\matPki$ associated with
    the reference Lagrange nodes $\famnodeki{\ha}\in\matR$,
    see Lemma~\ref{l:lag-basis-Pk1-ref};

  \item
    $\hcalL^{1,d}_i$ for $i\in[0..d]$ denotes the family of reference
    Lagrange polynomials of~$\matPid$,
    see Definition~\ref{d:lag-pol-P1d-ref};

  \item
    $\calL^{\famvertd{\vv}}_i\eqdef\hcalL^{1,d}_i\circ\invphigeodv$ for
    $i\in[0..d]$ denotes the family of Lagrange polynomials of~$\matPid$
    associated with~$\famvertd{\vv}$,
    see Lemma~\ref{l:lag-pol-P1d}.
  \end{CompactList}

\item Finite element:
  \begin{CompactList}
  \item
    $(K,P,\Sigma)$ denotes a finite element where~$K$ represents the geometry,
    $P$ the approximation space, and~$\Sigma$ the degrees of freedom,
    see Definition~\ref{d:fe-triple};

  \item
    $\phi_\Sigma$ denotes the collection of degrees of freedom applications,
    see Definition~\ref{d:fe-triple};

  \item
    $\Sigma^{\famnodekd{\aa}}=(\sigma_\aalpha)_{\aalpha\in\calAkd}$ denote the
    Lagrange linear forms associated with the Lagrange nodes of~$\matPkd$,
    see Definition~\ref{d:lag-lin-forms-Pkd};

  \item
    $\hSigmakd=(\hsigma_\aalpha)_{\aalpha\in\calAkd}$ denote the reference
    Lagrange linear forms associated with the reference Lagrange nodes
    of~$\matPkd$,
    see Definition~\ref{d:lag-lin-forms-Pkd-ref};

  \item
    $\FElagP{k}{d}\eqdef\left(\Kvd,\matPkd,\Sigma^{\famnodekd{\aa}}\right)$
    denotes the Lagrange finite element of degree~$k$ associated with
    vertices~$\famvertd{\vv}$,
    see Theorems~\ref{t:Pkd-lag-fe}, and~\ref{t:Pk1-lag-fe} for $d\eqdef1$;

  \item
    $\FElagPref{k}{d}\eqdef\left(\hKd,\matPkd,\hSigmakd\right)$ denotes the
    reference Lagrange finite element of degree~$k$ in dimension~$d$,
    see Theorems~\ref{t:Pkd-lag-fe-ref}, and~\ref{t:Pk1-lag-fe-ref} for
    $d\eqdef1$.
  \end{CompactList}
\end{itemize}

\bigskip

Note that the vector space $\matPkd$ of polynomials of total degree at
most~$k$ does not have the same notation as the Lagrange finite
element~$\FElagP{k}{d}$.
We have tried to visually separate the two concepts, although specialists
often use the same notation.

\chapter{Statements and sketches of proofs}
\label{c:statements-and-sketches-of-the-proofs}

This chapter gathers the sketches of the proofs of the main results that are
detailed in Part~\ref{p:detailed-proofs}.
Namely: the {\UPkdt}, the {\EdPkdl},
and a more technical discussion about ordering the multi-indices.

\section{Sketch of the proof {\UPkd}}
\label{s:sketch-of-the-proof-of-unisolvence-Pkd}

\begin{upkdthm}
  \mbox{}\\
  Let~$d\geq1$.
  Let~$k\geq1$.
  Let~$\famvertd{\vv}$ be $d+1$ affinely independent points in~$\matRd$.\\
  Then, $\left(\Kvd,\matPkd,\Sigma^{\famnodekd{\aa}}\right)$ satisfies the
  unisolvence property.
\end{upkdthm}

See Theorem~\ref{t:unisolvence-Pkd}, and the proof of
Lemma~\ref{l:lag-lin-forms-Pkd-inj}.
The definition of unisolvence is explained in Section~\ref{ss:fe-triple}.

The proof of the {\UPkdt} uses a double induction scheme, and the Euclidean
division of polynomials (see
Section~\ref{s:sketch-of-the-proof-of-Euclid-div-Pkd}).
The proof of the {\UPkdt} goes as follows:
\begin{itemize}
\item
  prove that~$\phi_{\Sigma^{\famnodekd{\aa}}}$ is injective by a double
  induction on~$d,k\geq1$ (see Lemma~\ref{l:lag-lin-forms-Pkd-inj}).\\
  This amounts to prove that if~$\famvertd{\vv}$ are $d+1$ affinely independent
  points in~$\matRd$, and if the polynomial~$p\in\matPkd$ is zero on all
  Lagrange nodes~$\famnodekd{\aa}$ ({\ie} for all $\aalpha\in\calAkd$,
  $p(\aa_\aalpha)=0$), then~$p$ is zero.
  We recall that the Lagrange nodes~$\famnodekd{\aa}$ are evenly distributed
  over the simplex whose vertices are $\famvertd{\vv}$.
  \begin{itemize}
  \item
    {\bf\boldmath{for $d=1$,}} the components of any polynomial in~$\matPki$ on
    the Lagrange basis are the values at the Lagrange nodes (see
    Lemma~\ref{l:decomp-Pk1-pol-in-lag-basis}), thus the cancellation of the
    latter implies the cancellation of the polynomial itself (see
    Lemma~\ref{l:lag-lin-forms-Pk1-inj});
  \item
    {\bf\boldmath{for $k=1$,}} the components of any polynomial in~$\matPid$ on
    the Lagrange basis are the values at the vertices (which are the Lagrange
    nodes in this case, see Lemmas~\ref{l:lag-nodes-Pid-are-vert}
    and~\ref{l:decomp-P1d-pol-in-lag-basis}), thus the cancellation of the
    latter implies the cancellation of the polynomial itself (see
    Lemma~\ref{l:lag-lin-forms-P1d-inj});
  \item
    {\bf\boldmath{for $d,k\geq2$,}} assume that the injectivity result holds
    for $(d-1,k)$ and $(d,k-1)$,\\
    and let~$p\in\matPkd$ vanishing on~$\famnodekd{\aa}$, the Lagrange nodes
    of~$\matPkd$ defined for the vertices~$\famvertd{\vv}$, then:
    \begin{itemize}
    \item
      {\bf step 1:} factorization using the injectivity result for $(d-1,k)$.
      This is done by using the fact that~$p$ is zero on the nodes of a
      hyperface of the current $d$-simplex, and passing this information on the
      reference $(d-1)$-simplex via an affine bijective mapping
      between~$\matRdmi$ and this hyperface:
      \begin{itemize}
      \item
        show that~$p_0\eqdef p\circ\phindv{\thetinjdmi{0}}$ vanishes
        on~$\famnodekdmi{\haa}$, the reference Lagrange nodes of~$\matPkdmi$,
        where~$\phindv{\thetinjdmi{0}}$ is the bijective geometric hyperface
        mapping (see Lem\-ma~\ref{l:geo-hyperface-mapping}) from~$\matRdmi$
        onto~$\calHvd_0$, the face hyperplane opposite the vertex~$\vv_0$;
      \item
        then, apply the injectivity result for~$(d-1,k)$ to obtain the
        cancellation of~$p_0$ on~$\matRdmi$, and thus the cancellation
        of~$p$ on the whole hyperplane~$\calHvd_0$ (and not only the nodes on
        this hyperplane);
      \item
        finally, use the Euclidean division on~$\matPkd$ to show the existence
        of~$q\in\matPkmid$ such that~$p=\barcvd_0\times q$, where~$\barcvd_0$
        is the barycentric coordinate associated with the vertex~$\vv_0$, (the
        hyperplane~$\calHvd_0$ is the place where~$\barcvd_0$ vanishes, see
        Lemmas~\ref{l:equiv-def-face-hyperpl}
        and~\ref{l:factor-zero-pol-hyperpl-Pkd});
      \end{itemize}
    \item
      {\bf step 2:} cancellation using the injectivity result for $(d,k-1)$.
      The principle here is to observe that~$q\in\matPkmid$ is zero on all the
      Lagrange nodes except the ones on~$\calHvd_0$, and that these cancelling
      nodes are precisely the Lagrange nodes of~$\matPkmid$ that are defined
      from sub-vertices (see Figure~\ref{f:sub-vert-sub-nodes_d3-k3}):
      \begin{itemize}
      \item
        show that~$q$ vanishes on the first Lagrange nodes of~$\matPkd$,
        indexed by~$\calAkmid$;
      \item
        then, show that those nodes are equal to~$\famnodekmid{\uaa}$, the
        Lagrange nodes of~$\matPkmid$ associated with the sub-vertices with
        respect to~$\vv_0$ (see
        Lemma~\ref{l:Pkm1d-sub-nodes-sub-vert-are-some-nodes-Pkd});
      \item
        finally, apply the injectivity result for~$(d,k-1)$ to obtain
        cancellation of~$q$, then of~$p$, and thus the injectivity result
        for~$(d,k)$;
      \end{itemize}
    \end{itemize}
  \end{itemize}
\item
  use equivalence between injectivity and bijectivity of linear maps when
  the dimensions of input and output spaces coincide.
\end{itemize}

The {\UPkdt} is used in the present document to build the~$\FElagP{d}{k}$
Lagrange finite element of degree~$k$ associated with simplicial
vertices~$\famvertd{\vv}$ (see Theorem~\ref{t:Pkd-lag-fe}).

\bigskip

Note that in this proof, one needs in step~2 to take sub-nodes in order to use
the induction hypothesis for $(d,k-1)$.
This requires to pass from the initial simplex to a different simplex.
This new simplex is the first one, minus a slice that is comprised between the
hyperplane~$\calHvd_0$ opposite~$\vv_0$, and the hyperplane parallel
to~$\calHvd_0$ that passes through the nodes closest to the hyperplane, see
Figure~\ref{f:sub-vert-sub-nodes_d3-k3}.
From this remark, note that if the initial simplex is the reference one, the
new simplex cannot be the reference simplex.
Thus, this proof requires to treat general {\nondegenerate} simplices and
not only the {\em reference} simplex.

This is why the unisolvence is proved directly in all generality for any
{\nondegenerate} simplices.
The unisolvence for the reference simplex is then simply a particular case.

\clearpage
\section{Sketch of the proof {\EdPkd}}
\label{s:sketch-of-the-proof-of-Euclid-div-Pkd}

\begin{edpkdlem}
  \mbox{}\hfill
  Let~$d\geq2$.
  Let~$k\geq1$.
  Let~$p\in\matPkd$.\\
  Then, there exist unique $\tp_0\in\matPkdmi$ and $p_1\in\matPkmid$, such that
  $p=\tp_0+X_d\,p_1$, which also writes
  for all $(x_1,\ldots,x_d)\in\matRd$,
  $p(x_1,\ldots,x_d)=\tp_0(x_1,\ldots,x_{d-1})+x_d\,p_1(x_1,\ldots,x_d)$.
\end{edpkdlem}

See Lemma~\ref{l:decomp-Pkd}.
The proof of the {\EdPkdl} uses the~``$\tilde{\ }$'' notation for functions for
which the last variable is dropped (see~$p$ and ~$\tp_0$ above, and
Remark~\ref{r:tilde-and-check-notations}).
The proof of the {\EdPkdl} goes as follows:
\begin{itemize}
\item
  {\bf existence:} by induction on~$k\geq1$,
  \begin{itemize}
  \item
    {\bf\boldmath{for $k=1$,}} the result is straightforward as
    $\matPid\eqdef\Span{1,X_1,X_2,\ldots,X_d}$ (the constant on~$X_d$
    is~$p_1$);
  \item
    {\bf\boldmath{for $k\geq1$,}} assume that the result holds for~$k$,\\
    and let~$p=\sum_{\aalpha\in\calAkpid}a_\aalpha\XX^\aalpha\in\matPkpid$,
    then:
    \begin{itemize}
    \item
      from the partition of the multi-index
      sets~$\calAkpid=\calAkd\uplus\calCkpid$ (see
      Definition~\ref{d:multi-ind-Akd-Ckd} and Lemma~\ref{l:Ckd-layers-Akd}),
      we have $p=q+r$ with
      $q=\sum_{\aalpha\in\calAkd}a_\aalpha\XX^\aalpha\in\matPkd$ and\\
      $r=\sum_{\aalpha\in\calCkpid}a_\aalpha\XX^\aalpha\in\matPkpid$;
    \item
      then, apply the result for~$k$ to~$q$, and obtain existence
      of~$\tq_0\in\matPkdmi\subset\matPkpidmi$
      and~$q_1\in\matPkmid\subset\matPkd$ such that $q=\tq_0+X_d\,q_1$;
    \item
      for all $i\in[0..k+1]$, for all $\taalpha_i\in\calCdmi{k+1-i}$,
      let~$b_{\taalpha_i}\eqdef a_{(\taalpha_i,i)}$, let
      \[
        \tr_0 \eqdef \sum_{\taalpha_0 \in \calCkpidmi}
        b_{\taalpha_0} \tXX^{\taalpha_0}
        \AND
        r_1 \eqdef \sum_{i = 1}^{k + 1}
        \sum_{\taalpha_i \in \calCdmi{k + 1 - i}}
        b_{\taalpha_i} \tXX^{\taalpha_i} X_d^{i - 1},
    \]
    then show that~$\tr_0\in\matPkpidmi$, $r_1\in\matPkd$ and $r=\tr_0+r_1$;
    \item
      finally show that $\tp_0\eqdef\tq_0+\tr_0\in\matPkpidmi$,
      $p_1\eqdef q_1+r_1\in\matPkd$ and $p=\tp_0+X_dp_1$;
    \end{itemize}
  \end{itemize}
\item
  {\bf uniqueness:} assume that $\tq_0+X_d\,q_1=\tr_0+X_d\,r_1$ with
  $\tq_0,\tr_0\in\matPkdmi$ and $q_1,r_1\in\matPkmid$,
  \begin{itemize}
  \item
    applying the equality to $\xx=(\txx,0)\in\matRd$ provides $\tq_0=\tr_0$;
  \item
    let $\sum_{\aalpha\in\calAkmid}a_\aalpha\XX^\aalpha$ be the decomposition
    on the monomial basis of $q_1-r_1\in\matPkmid$, then the linear
    independence of the monomial family in~$\matPkd$ and the cancellation\\
    $X_d(q_1-r_1)=0$ provide $(a_\aalpha)_{\aalpha\in\calAkmid}=\zzero$, thus
    $q_1=r_1$.
  \end{itemize}
\end{itemize}

The {\EdPkdl} is used in the present document to build an isomorphism
between~$\matPkd$ and~$\matPkdmi\times\matPkmid$ in
Lemma~\ref{l:isom-Pkd-Pkdm1xPkm1d}, to express a multivariate polynomial as a
polynomial of~$x_d$ in Lemma~\ref{l:Pkd-as-pol-xd}, to establish the degree of
the product of two polynomials in Lemma~\ref{l:prod-2-polynom-alt-proof}, and
the factorization of a polynomial vanishing in the last reference face
hyperplane~$\hcalH^d_d$ in Lemma~\ref{l:factor-zero-pol-last-ref-hyperpl} (used
for Lemma~\ref{l:factor-zero-pol-hyperpl-Pkd}, see
Section~\ref{s:sketch-of-the-proof-of-unisolvence-Pkd}).

Note that, instead of building $\matPkd$ as the linear span of monomials as
done here (see Definition~\ref{d:pol-space-Pkd}), it is possible to build it by
incrementing the dimension of the polynomial by setting the components of 1D
polynomial as polynomials of other variables.
The Lemma~\ref{l:Pkd-as-pol-xd} somewhat establishes the link between these two
views.

\clearpage
\section{Multi-index ordering}
\label{s:multi-ind-numbering}

The second part of the present document does not focus on the ordering of
multi-indices, but it may be interesting to say a few words on this topic.

\subsection{Multi-indices and the {\fem}}
\label{ss:multi-ind-fem}

For Lagrange finite elements in $d$-simplices, we deal with multi-indices
whose sum is at most~$k$.
\begin{mindakddef}
  \mbox{}\hfill
  Let~$d\geq1$.
  Let~$k\in\matN$.
  The {\em set of multi-indices of length at most~$k$ (resp. of length~$k$)}
  is denoted~$\calAkd$ (resp.~$\calCkd$), and is defined by
  \[
    \calAkd \eqdef \{ \aalpha \in \matNd \st \len{\aalpha} \leq k \}
    \quad (\mbox{resp. }
    \calCkd \eqdef \{ \aalpha \in \matNd \st \len{\aalpha} = k \}),
  \]
  where $\forall\aalpha\in\matNd$, $\len{\aalpha}\eqdef\sum_{i=1}^d\alpha_i$.

  Let~$i\in[1..d]$.
  $\calAkdi\eqdef\{\aalpha\in\calAkd\st\alpha_i=0\}$ is the
  {\em subset of~$\calAkd$ of multi-indices with zero $i$-th component}.
\end{mindakddef}
(See also Definition~\ref{d:multi-ind-Akd-Ckd}.)
These multi-indices are used for two purposes in this document.
First, they provide the multi-exponent of the multivariate monomial
$\XX^\aalpha\eqdef\prod_{i=1}^dX_i^{\alpha_i}$ (whose total degree
is~$\len{\aalpha}$, see Definition~\ref{d:monom-kd}).
Thus, a polynomial of total degree at most~$k$ is the sum of such monomials:
for $(a_\aalpha)_{\aalpha\in\calAkd}\in\matR$,
$p\eqdef\sum_{\aalpha\in\calAkd}a_\aalpha\XX^\aalpha$ (see
Definition~\ref{d:pol-space-Pkd} and the following statements).

Second, they index the Lagrange nodes in a (nondegenerate) $d$-simplex.
Let~$\Kvd$ be a simplex defined by its $d+1$ (affinely independent)
vertices  in~$\matRd$, denoted by~$\famvertd{\vv}\eqdef(\vv_0,\dots,\vv_d)$.
Then, the Lagrange nodes~$\famnodekd{\aa}$ are defined by
\begin{lagpkdnodesdef}
  \mbox{}\hfill
  Let~$d\geq1$.
  Let~$k\in\matN$.
  Let~$\famvertd{\vv}$ be $d+1$ points in~$\matRd$.\\
  The {\em Lagrange nodes of~$\matPkd$} are denoted
  $\famnodekd{\aa}=(\aa_\aalpha)_{\aalpha\in\calAkd}$, and are defined by
  \begin{align*}
    (k = 0)&&
    \aa_\zzero &\eqdef
      \bisobarycv \left( = \frac{1}{d + 1} \sum_{i = 0}^d \vv_i \right),&&\\
    (k \geq 1)&&
    \forall \aalpha \in \calAkd,\quad
    \aa_\aalpha &\eqdef
      \vv_0 + \sum_{i = 1}^d \frac{\alpha_i}{k} (\vv_i -  \vv_0).&&
  \end{align*}
\end{lagpkdnodesdef}
(See also Definition~\ref{d:lag-nodes-Pkd}.)
The case of constant polynomials ($k=0$) is special, as the node is
set at the isobarycenter of $\famvertd{\vv}$.
We focus here on the general case when $k\geq1$.
In this case, the Lagrange nodes can be written equivalently
\[
  \aa_\aalpha = \frac{k - \len{\aalpha}}{k} \vv_0
  + \sum_{i = 1}^d \frac{\alpha_i}{k}\ \vv_i,
  \quad \mbox{with }
  \frac{k - \len{\aalpha}}{k} + \sum_{i = 1}^d \frac{\alpha_i}{k} = 1,
\]
which exhibits the barycentric coordinates
$(\frac{k-\len{\aalpha}}{k},\frac{\alpha_1}{k},\dots,\frac{\alpha_d}{k})$
of~$\aa_\aalpha$ with respect to~$\famvertd{\vv}$ (see
Section~\ref{s:baryc-coord}).
For instance, in the reference simplex~$\hK^d$, whose reference vertices are
$\hvv_0\eqdef\zzero$ and for $i\in[1..d]$, $\hvv_i\eqdef\ee_i$, this expression
shows that each reference node~$\haa_\aalpha$ has coordinates
$(\haa_\aalpha)_i=\frac{\halpha_i}{k}$.
It is thus natural to represent a Lagrange node in a simplex by its multi-index
in~$\calAkd$, and to link it to the corresponding monomial, see
Figure~\ref{f:lag-k3-d2-d3}.
These correspondences are be used throughout the document.

\begin{figure}[htb]
  \centering
  \resizebox{0.75\linewidth}{!}{
    \begin{tikzpicture}[scale=4,math3d] 

  \def\kk{3}

  \def\colk{black}
  \def\colo{magenta}
  \def\coli{darkgreen}
  \def\colii{red}
  \def\coliii{blue}

  \def\opacity{0.4}
  \def\opacityi{0.7}
  \def\opacityii{0.6}

  \coordinate (AA) at (0,0,-0.25);
  \coordinate (CC) at ($(AA) + (0,1,0)$);
  \coordinate (DD) at ($(AA) + (0,0,1)$);
  \draw (AA) node[above left] {$\hvv_0$};
  \draw (CC) node[above=1.5pt] {$\hvv_1=(1,0)$};
  \draw (DD) node[right=1.5pt] {$\hvv_2=(0,1)$};
  \draw[line width=1.0pt,rounded corners=0.5pt] (AA) -- (CC) -- (DD) -- cycle;
  \node[color=\colk] (K2) at ($(AA) + (0,0.65,0.65)$) {$\hK_{2}$};

  \fill[color=\colo] (AA) circle (0.8pt);
  \node[color=\colo,below] (Nxyz) at (AA) {$\haa_{(0,0)}$};
  \newcount\y
  \foreach \x in {1,0} {
    \pgfmathsetcount{\y}{1-\x} 
    \coordinate (Axyz) at ($(AA) + 1/\kk*(0,\x,\y)$);
    \node[color=\coli,below] (Nxyz) at (Axyz) {$\haa_{(\x,\the\y)}$};
    \fill[color=\coli] (Axyz) circle (0.8pt);
  }
  \newcount\y
  \foreach \x in {2,1,...,0} {
    \pgfmathsetcount{\y}{2-\x} 
    \coordinate (Axyz) at ($(AA) + 1/\kk*(0,\x,\y)$);
    \node[color=\colii,below] (Nxyz) at (Axyz) {$\haa_{(\x,\the\y)}$};
    \fill[color=\colii] (Axyz) circle (0.8pt);
  }
  \newcount\y
  \foreach \x in {3,2,...,0} {
    \pgfmathsetcount{\y}{3-\x} 
    \coordinate (Axyz) at ($(AA) + 1/\kk*(0,\x,\y)$);
    \node[color=\coliii,below] (Nxyz) at (Axyz) {$\haa_{(\x,\the\y)}$};
    \fill[color=\coliii] (Axyz) circle (0.8pt);
  }

  \coordinate (eps) at ($1/\kk*(0,0.05,-0.05)$); 
  \coordinate (A10) at ($(AA) + 1/\kk*(0,1,0)$);
  \coordinate (A01) at ($(AA) + 1/\kk*(0,0,1)$);
  \coordinate (A20) at ($(AA) + 1/\kk*(0,2,0)$);
  \coordinate (A02) at ($(AA) + 1/\kk*(0,0,2)$);

  \node[color=\colo,below=12pt] (Nxyz) at (AA) {$\calCod$};
  \node[color=\coli,below=12pt] (Nxyz) at (A10) {$\calCid$};
  \node[color=\colii,below=12pt] (Nxyz) at (A20) {$\calCiid$};
  \node[color=\coliii,below=12pt] (Nxyz) at (CC) {$\calCiiid$};
  
  \draw[color=\coli,line width=0.7pt] (A10) -- ($(A01) + (eps)$);
  \draw[color=\colii,line width=0.7pt] (A20) -- ($(A02) + (eps)$);
  \draw[color=\coliii,dashed,line width=1.4pt] (CC) -- ($(DD) + (eps)$);


  \coordinate (A) at (0,2.0,0);
  \coordinate (B) at ($(A) + (1,0,0)$);
  \coordinate (C) at ($(A) + (0,1,0)$);
  \coordinate (D) at ($(A) + (0,0,1)$);
  \draw (A) node[left] {$\hvv_0$}; 
  \draw (B) node[below=12pt] {$\hvv_1=(1,0,0)$}; 
  \draw (C) node[above=2pt] {$\hvv_2=(0,1,0)$}; 
  \draw (D) node[right=1.5pt] {$\hvv_3=(0,0,1)$}; 
  \draw[line width=1.0pt,rounded corners=0.5pt] (A) -- (B) -- (D) -- cycle;
  \draw[line width=1.0pt,rounded corners=0.5pt] (A) -- (C) -- (B) -- cycle;
  \draw[line width=1.0pt,rounded corners=0.5pt] (B) -- (C) -- (D) -- cycle;
  \draw[line width=1.6pt,rounded corners=0.5pt] (A) -- (C) -- (D) -- cycle;

  \node[color=\colk] (K3) at ($(A) + (0,-0.3,0.7)$) {$\hK_{3}$};

  \pgfmathparse{1-1/\kk}\let\kkp\pgfmathresult
  \coordinate (B1) at ($\kkp*(A) + 1/\kk*(B)$);
  \coordinate (C1) at ($\kkp*(A) + 1/\kk*(C)$);
  \coordinate (D1) at ($\kkp*(A) + 1/\kk*(D)$);
  \pgfmathparse{1-2/\kk}\let\kkp\pgfmathresult
  \coordinate (B2) at ($\kkp*(A) + 2/\kk*(B)$);
  \coordinate (C2) at ($\kkp*(A) + 2/\kk*(C)$);
  \coordinate (D2) at ($\kkp*(A) + 2/\kk*(D)$);

  \draw[color=\coli,fill=\coli!40,fill opacity=\opacity] (B1) -- (C1) -- (D1) -- cycle;
  \draw[color=\colii,fill=\colii!20,fill opacity=\opacity] (B2) -- (C2) -- (D2) -- cycle;
  \draw[color=\coliii,fill=\coliii!20,fill opacity=\opacity] (B) -- (C) -- (D) -- cycle;

  \fill[color=\colo,fill opacity=\opacityii] (A) circle (0.8pt);
  \foreach \x in {1,0} {
    \pgfmathparse{1-\x}\let\YY\pgfmathresult
    \foreach \y in {0,...,\YY} {
      \pgfmathparse{1-\x-\y}\let\z\pgfmathresult
      \fill[color=\coli,fill opacity=\opacityii] ($(A) + 1/\kk*(\x,\y,\z)$) circle (0.8pt);
    }
  }
  \foreach \x in {2,1,...,0} {
    \pgfmathparse{2-\x}\let\YY\pgfmathresult
    \foreach \y in {0,...,\YY} {
      \pgfmathparse{2-\x-\y}\let\z\pgfmathresult
      \fill[color=\colii,fill opacity=\opacityii] ($(A) + 1/\kk*(\x,\y,\z)$) circle (0.8pt);
    }
  }
  \newcount\z
  \foreach \x in {3,2,...,0} {
    \pgfmathparse{3-\x}\let\YY\pgfmathresult
    \foreach \y in {0,...,\YY} {
      \pgfmathsetcount{\z}{3-\x-\y} 
      \coordinate (Axyz) at ($(A) + 1/\kk*(\x,\y,\z)$);
      \node[color=\coliii,below=-0.15] (Nxyz) at (Axyz) {$\haa_{(\x,\y,\the\z)}$};
      \fill[color=\coliii] (Axyz) circle (0.8pt);
    }
  }


\end{tikzpicture}
  }
  \caption[Lagrange nodes of the reference simplex]{%
    Lagrange nodes~$\famnodekd{\haa}$ of the reference simplex~$\hK_d$
    when~$d\in\{2,3\}$ and~$k=3$ (see Section~\ref{ss:multi-ind-fem}).\\
    Each node is depicted as a colored ball, and corresponds to a
    unique element of~$\calAiiid$.
    The colors correspond to degrees $l\leq3$ of polynomials, or equivalently
    to lengths of multi-indices ({\ie} in~$\calCld$ for $l\leq3$).
    In magenta, the node~$\haa_\zzero$ corresponds to constant polynomials
    (with degree~0) in~$\matPod$, and to the multi-index~$\zzero$ in the
    singleton~$\calCod$.
    In green, the nodes correspond to non-constant affine polynomials (with
    degree~1), and to the multi-indices $\ee_1,\ldots,\ee_d$ in~$\calCid$.
    In red, the nodes correspond to non-affine quadratic polynomials (with
    degree~2), and to multi-indices in~$\calCiid$.
    In blue, the nodes correspond to non-quadratic cubic polynomials (with
    degree~3), and to multi-indices in~$\calCiiid$.
    We observe in this picture that
    $\calAiiid=\calCod\uplus\calCid\uplus\calCiid\uplus\calCiiid$.}
  \label{f:lag-k3-d2-d3}
\end{figure}

\subsection{Monomial order}
\label{ss:monom-order}

Once these notations are set, we can present the various possibilities to order
multi-indices.
Such ordering should be a {\em monomial order}, {\ie} a total order that is
compatible with the monoid structure of the monomials: for all
$\aalpha,\bbeta,\ggamma\in\calAkd$, $\XX^\aalpha<\XX^\bbeta$ implies
$\XX^\aalpha\XX^\ggamma<\XX^\bbeta\XX^\ggamma$, which means that
$\aalpha<\bbeta$ implies $\aalpha+\ggamma<\bbeta+\ggamma$.
It is also generally required that for all $\aalpha\neq\zzero$,
$\XX^\zzero=1<\XX^\aalpha$.

Among various possibilities, {\eg} see~\cite[Chap. 2]{clo:iva:15}
and~\cite{wp:mo}, we present commonly used monomial orders, and variants
including one that reveals more convenient in our context (and that we call
``grsymlex'', see Section~\ref{ss:grsymlex}).

In the sequel, $\aalpha$ and~$\bbeta$ denote any multi-indices in~$\calAkd$.

Note that we use the following notations when $d\geq2$: the check
notation~$\caalpha$ denotes the {\em last} $d-1$ components of~$\aalpha$, and
the tilde notation~$\taalpha$ denotes the {\em first} ones.
Thus, when $\aalpha\eqdef(\alpha_1,\ldots,\alpha_d)$, we have
$\aalpha=(\alpha_1,\caalpha)=(\taalpha,\alpha_d)$, see also
Remark~\ref{r:tilde-and-check-notations}.

\subsection{Lexicographic order}
\label{ss:lex}

The {\em lexicographic order}, or simply {\em ``lex'' order}, can be
recursively defined as
\begin{equation}
  \label{e:lex}
  \aalpha \ltlex \bbeta \equivdef
  \left\{
    \begin{array}{l}
      \alpha_1 < \beta_1, \mbox { or}\\
      \alpha_1 = \beta_1 \CONJ d \geq 2 \CONJ \caalpha \ltlex \cbbeta.
    \end{array}
  \right.
\end{equation}
We have $\aalpha\ltlex\bbeta$ iff $\alpha_i<\beta_i$ for the {\em first}
index~$i$ for which~$\alpha_i$ and~$\beta_i$ differ.

Starting from the right, the {\em colexicographic order}, or simply
{\em ``colex'' order}, can be recursively defined as
\begin{equation}
  \label{e:colex}
  \aalpha \ltcolex \bbeta \equivdef
  \left\{
    \begin{array}{l}
      \alpha_d < \beta_d, \mbox { or}\\
      \alpha_d = \beta_d \CONJ d \geq 2 \CONJ \taalpha \ltcolex \tbbeta.
    \end{array}
  \right.
\end{equation}
We have $\aalpha\ltcolex\bbeta$ iff $\alpha_i<\beta_i$ for the {\em last}
index~$i$ for which~$\alpha_i$ and~$\beta_i$ differ.
The {\em colex} order is also called {\em inverse lexicographic order}, or
simply {\em ``invlex'' order}, {\eg} see~\cite[p.~61]{clo:iva:15}.

We may introduce the {\em symmetrical lexicographic order}, or simply
{\em ``symlex'' order}, as
\begin{equation}
  \label{e:symlex}
  \aalpha \ltsymlex \bbeta \equivdef \bbeta \ltlex \aalpha.
\end{equation}
It is the symmetrical of the lex order, we have $\aalpha\ltsymlex\bbeta$ iff
$\beta_i<\alpha_i$ for the {\em first} index~$i$ for which~$\alpha_i$
and~$\beta_i$ differ.

We may also define the {\em reverse lexicographic order}, or simply
{\em ``revlex'' order}, as
\begin{equation}
  \label{e:revlex}
  \aalpha \ltrevlex \bbeta \equivdef \bbeta \ltcolex \aalpha.
\end{equation}
It is the symmetrical of the colex order, we have $\aalpha\ltrevlex\bbeta$ iff
$\beta_i<\alpha_i$ for the {\em last} index~$i$ for which~$\alpha_i$
and~$\beta_i$ differ.
The {\em revlex} order is also called
{\em reverse inverse lexicographic order}, or simply {\em ``rinvlex'' order},
{\eg} see~\cite[p.~61]{clo:iva:15}.

The lex, colex, symlex, and revlex orders are monomial orders.
Note also that lex and colex are obviously equivalent when~$d=1$ (and so are
symlex and revlex).
Moreover, when $d=2$ and the multi-indices have the same length, lex and colex
orders are symmetrical.
Indeed, assume that $\alpha_1+\alpha_2=\beta_1+\beta_2$.
Then, we have $\alpha_1<\beta_1\Equiv\beta_2<\alpha_2$ and
$\alpha_1=\beta_1\Equiv\alpha_2=\beta_2$, {\ie}
$(\alpha_1,\alpha_2)\ltlex(\beta_1,\beta_2)\Equiv
(\beta_1,\beta_2)\ltcolex(\alpha_1,\alpha_2)$.
Thus, in that case ($\len{\aalpha}=\len{\bbeta}$ and $d=2$), lex and revlex
orders are equivalent, as well as colex and symlex.
See a 2D~example in Figure~\ref{f:tet-tria_k3_lex_colex}.

\begin{figure}[t]
  \centering
  \hfill
  \resizebox{0.375\linewidth}{!}{
    \begin{tikzpicture}[scale=4,math3d] 

  \def\kk{3}

  \def\colk{black}
  \def\colb{blue}

  \coordinate (AA) at (0,0,-0.25);
  \coordinate (CC) at ($(AA) + (0,1,0)$);
  \coordinate (DD) at ($(AA) + (0,0,1)$);
  \draw (AA) node[above left] {$\hvv_0$};
  \draw (CC) node[above=1.5pt] {$\hvv_1=(1,0)$};
  \draw (DD) node[right=1.5pt] {$\hvv_2=(0,1)$};
  \draw[line width=1.0pt,rounded corners=0.5pt] (AA) -- (CC) -- (DD) -- cycle;
  \node[color=\colk] (K2) at ($(AA) + (0,0.65,0.65)$) {$\hK_{2}$};

  \fill[color=\colk] (AA) circle (0.8pt);
  \node[color=\colk,below] (Nxyz) at (AA) {$\haa_{(0,0)}$};
  \newcount\y
  \foreach \x in {1,0} {
    \pgfmathsetcount{\y}{1-\x} 
    \coordinate (Axyz) at ($(AA) + 1/\kk*(0,\x,\y)$);
    \node[color=\colk,below] (Nxyz) at (Axyz) {$\haa_{(\x,\the\y)}$};
    \fill[color=\colk] (Axyz) circle (0.8pt);
  }
  \newcount\y
  \foreach \x in {2,1,...,0} {
    \pgfmathsetcount{\y}{2-\x} 
    \coordinate (Axyz) at ($(AA) + 1/\kk*(0,\x,\y)$);
    \node[color=\colk,below] (Nxyz) at (Axyz) {$\haa_{(\x,\the\y)}$};
    \fill[color=\colk] (Axyz) circle (0.8pt);
  }
  \newcount\y
  \foreach \x in {3,2,...,0} {
    \pgfmathsetcount{\y}{3-\x} 
    \coordinate (Axyz) at ($(AA) + 1/\kk*(0,\x,\y)$);
    \node[color=\colb,below] (Nxyz) at (Axyz) {$\haa_{(\x,\the\y)}$};
    \fill[color=\colb] (Axyz) circle (0.8pt);
  }

  \coordinate (eps0) at ($1/\kk*(0,0,-0.05)$);
  \coordinate (eps) at ($1/\kk*(0,0.05,-0.05)$);

  \coordinate (A00) at ($(AA)$);
  \coordinate (A03) at ($(AA) + 1/\kk*(0,0,3)$);
  \coordinate (A10) at ($(AA) + 1/\kk*(0,1,0)$);
  \coordinate (A12) at ($(AA) + 1/\kk*(0,1,2)$);
  \coordinate (A20) at ($(AA) + 1/\kk*(0,2,0)$);
  \coordinate (A21) at ($(AA) + 1/\kk*(0,2,1)$);
  \coordinate (A30) at ($(AA) + 1/\kk*(0,3,0)$);
  \draw[color=\colb,dashed,line width=1.4pt,->,>=latex]
    (A00) -- ($(A03) + (eps0)$);
  \draw[color=\colb,dashed,line width=1pt,->,>=latex]
    (A03) -- ($(A10) - (eps)$);
  \draw[color=\colb,dashed,line width=1pt,->,>=latex]
    (A10) -- ($(A12) + (eps0)$);
  \draw[color=\colb,dashed,line width=1pt,->,>=latex]
    (A12) -- ($(A20) - (eps)$);
  \draw[color=\colb,dashed,line width=1pt,->,>=latex]
    (A20) -- ($(A21) + (eps0)$);
  \draw[color=\colb,dashed,line width=1.4pt,->,>=latex]
    (A21) -- ($(A30) - (eps)$);

\end{tikzpicture}
  }
  \hfill
  \resizebox{0.375\linewidth}{!}{
    \begin{tikzpicture}[scale=4,math3d] 

  \def\kk{3}

  \def\colk{black}
  \def\colb{blue}

  \coordinate (AA) at (0,0,-0.25);
  \coordinate (CC) at ($(AA) + (0,1,0)$);
  \coordinate (DD) at ($(AA) + (0,0,1)$);
  \draw (AA) node[above left] {$\hvv_0$};
  \draw (CC) node[above=1.5pt] {$\hvv_1=(1,0)$};
  \draw (DD) node[right=1.5pt] {$\hvv_2=(0,1)$};
  \draw[line width=1.0pt,rounded corners=0.5pt] (AA) -- (CC) -- (DD) -- cycle;
  \node[color=\colk] (K2) at ($(AA) + (0,0.65,0.65)$) {$\hK_{2}$};

  \fill[color=\colk] (AA) circle (0.8pt);
  \node[color=\colk,below] (Nxyz) at (AA) {$\haa_{(0,0)}$};
  \newcount\y
  \foreach \x in {1,0} {
    \pgfmathsetcount{\y}{1-\x} 
    \coordinate (Axyz) at ($(AA) + 1/\kk*(0,\x,\y)$);
    \node[color=\colk,below] (Nxyz) at (Axyz) {$\haa_{(\x,\the\y)}$};
    \fill[color=\colk] (Axyz) circle (0.8pt);
  }
  \newcount\y
  \foreach \x in {2,1,...,0} {
    \pgfmathsetcount{\y}{2-\x} 
    \coordinate (Axyz) at ($(AA) + 1/\kk*(0,\x,\y)$);
    \node[color=\colk,below] (Nxyz) at (Axyz) {$\haa_{(\x,\the\y)}$};
    \fill[color=\colk] (Axyz) circle (0.8pt);
  }
  \newcount\y
  \foreach \x in {3,2,...,0} {
    \pgfmathsetcount{\y}{3-\x} 
    \coordinate (Axyz) at ($(AA) + 1/\kk*(0,\x,\y)$);
    \node[color=\colb,below] (Nxyz) at (Axyz) {$\haa_{(\x,\the\y)}$};
    \fill[color=\colb] (Axyz) circle (0.8pt);
  }

  \coordinate (eps0) at ($1/\kk*(0,-0.05,0)$);
  \coordinate (eps) at ($1/\kk*(0,-0.05,0.05)$);

  \coordinate (A00) at ($(AA)$);
  \coordinate (A30) at ($(AA) + 1/\kk*(0,3,0)$);
  \coordinate (A01) at ($(AA) + 1/\kk*(0,0,1)$);
  \coordinate (A21) at ($(AA) + 1/\kk*(0,2,1)$);
  \coordinate (A02) at ($(AA) + 1/\kk*(0,0,2)$);
  \coordinate (A12) at ($(AA) + 1/\kk*(0,1,2)$);
  \coordinate (A03) at ($(AA) + 1/\kk*(0,0,3)$);
  \draw[color=\colb,dashed,line width=1.4pt,->,>=latex]
    (A00) -- ($(A30) + (eps0)$);
  \draw[color=\colb,dashed,line width=1pt,->,>=latex]
    (A30) -- ($(A01) - (eps)$);
  \draw[color=\colb,dashed,line width=1pt,->,>=latex]
    (A01) -- ($(A21) + (eps0)$);
  \draw[color=\colb,dashed,line width=1pt,->,>=latex]
    (A21) -- ($(A02) - (eps)$);
  \draw[color=\colb,dashed,line width=1pt,->,>=latex]
    (A02) -- ($(A12) + (eps0)$);
  \draw[color=\colb,dashed,line width=1.4pt,->,>=latex]
    (A12) -- ($(A03) - (eps)$);

\end{tikzpicture}
  }
  \hfill
  \\
  \centering
  \hfill
  \resizebox{0.375\linewidth}{!}{
    \begin{tikzpicture}[scale=4,math3d] 

  \def\kk{3}

  \def\colk{black}
  \def\colb{blue}

  \coordinate (AA) at (0,0,-0.25);
  \coordinate (CC) at ($(AA) + (0,1,0)$);
  \coordinate (DD) at ($(AA) + (0,0,1)$);
  \draw (AA) node[above left] {$\hvv_0$};
  \draw (CC) node[above=1.5pt] {$\hvv_1=(1,0)$};
  \draw (DD) node[right=1.5pt] {$\hvv_2=(0,1)$};
  \draw[line width=1.0pt,rounded corners=0.5pt] (AA) -- (CC) -- (DD) -- cycle;
  \node[color=\colk] (K2) at ($(AA) + (0,0.65,0.65)$) {$\hK_{2}$};

  \fill[color=\colk] (AA) circle (0.8pt);
  \node[color=\colk,below] (Nxyz) at (AA) {$\haa_{(0,0)}$};
  \newcount\y
  \foreach \x in {1,0} {
    \pgfmathsetcount{\y}{1-\x} 
    \coordinate (Axyz) at ($(AA) + 1/\kk*(0,\x,\y)$);
    \node[color=\colk,below] (Nxyz) at (Axyz) {$\haa_{(\x,\the\y)}$};
    \fill[color=\colk] (Axyz) circle (0.8pt);
  }
  \newcount\y
  \foreach \x in {2,1,...,0} {
    \pgfmathsetcount{\y}{2-\x} 
    \coordinate (Axyz) at ($(AA) + 1/\kk*(0,\x,\y)$);
    \node[color=\colk,below] (Nxyz) at (Axyz) {$\haa_{(\x,\the\y)}$};
    \fill[color=\colk] (Axyz) circle (0.8pt);
  }
  \newcount\y
  \foreach \x in {3,2,...,0} {
    \pgfmathsetcount{\y}{3-\x} 
    \coordinate (Axyz) at ($(AA) + 1/\kk*(0,\x,\y)$);
    \node[color=\colb,below] (Nxyz) at (Axyz) {$\haa_{(\x,\the\y)}$};
    \fill[color=\colb] (Axyz) circle (0.8pt);
  }

  \coordinate (eps0) at ($1/\kk*(0,0,-0.05)$);
  \coordinate (eps1) at ($1/\kk*(0,0.05,0)$);
  \coordinate (eps2) at ($1/\kk*(0,0.05,-0.05)$);
  \coordinate (eps) at ($1/\kk*(0,-0.025,0.05)$);

  \coordinate (A00) at ($(AA)$);
  \coordinate (A03) at ($(AA) + 1/\kk*(0,0,3)$);
  \coordinate (A10) at ($(AA) + 1/\kk*(0,1,0)$);
  \coordinate (A12) at ($(AA) + 1/\kk*(0,1,2)$);
  \coordinate (A20) at ($(AA) + 1/\kk*(0,2,0)$);
  \coordinate (A21) at ($(AA) + 1/\kk*(0,2,1)$);
  \coordinate (A30) at ($(AA) + 1/\kk*(0,3,0)$);
  \draw[color=\colb,dashed,line width=1.4pt,<-,>=latex]
    ($(A00) - (eps0)$) -- (A03);
  \draw[color=\colb,dashed,line width=1pt,<-,>=latex]
    ($(A03) - (eps)$) -- (A10);
  \draw[color=\colb,dashed,line width=1pt,<-,>=latex]
    ($(A10) - (eps0)$) -- (A12);
  \draw[color=\colb,dashed,line width=1pt,<-,>=latex]
    ($(A12) - (eps)$) -- (A20);
  \draw[color=\colb,dashed,line width=1pt,<-,>=latex]
    ($(A20) - (eps0)$) -- (A21);
  \draw[color=\colb,dashed,line width=1.4pt,<-,>=latex]
    ($(A21) + (eps2)$) -- (A30);

\end{tikzpicture}
  }
  \hfill
  \resizebox{0.375\linewidth}{!}{
    \begin{tikzpicture}[scale=4,math3d] 

  \def\kk{3}

  \def\colk{black}
  \def\colb{blue}

  \coordinate (AA) at (0,0,-0.25);
  \coordinate (CC) at ($(AA) + (0,1,0)$);
  \coordinate (DD) at ($(AA) + (0,0,1)$);
  \draw (AA) node[above left] {$\hvv_0$};
  \draw (CC) node[above=1.5pt] {$\hvv_1=(1,0)$};
  \draw (DD) node[right=1.5pt] {$\hvv_2=(0,1)$};
  \draw[line width=1.0pt,rounded corners=0.5pt] (AA) -- (CC) -- (DD) -- cycle;
  \node[color=\colk] (K2) at ($(AA) + (0,0.65,0.65)$) {$\hK_{2}$};

  \fill[color=\colk] (AA) circle (0.8pt);
  \node[color=\colk,below] (Nxyz) at (AA) {$\haa_{(0,0)}$};
  \newcount\y
  \foreach \x in {1,0} {
    \pgfmathsetcount{\y}{1-\x} 
    \coordinate (Axyz) at ($(AA) + 1/\kk*(0,\x,\y)$);
    \node[color=\colk,below] (Nxyz) at (Axyz) {$\haa_{(\x,\the\y)}$};
    \fill[color=\colk] (Axyz) circle (0.8pt);
  }
  \newcount\y
  \foreach \x in {2,1,...,0} {
    \pgfmathsetcount{\y}{2-\x} 
    \coordinate (Axyz) at ($(AA) + 1/\kk*(0,\x,\y)$);
    \node[color=\colk,below] (Nxyz) at (Axyz) {$\haa_{(\x,\the\y)}$};
    \fill[color=\colk] (Axyz) circle (0.8pt);
  }
  \newcount\y
  \foreach \x in {3,2,...,0} {
    \pgfmathsetcount{\y}{3-\x} 
    \coordinate (Axyz) at ($(AA) + 1/\kk*(0,\x,\y)$);
    \node[color=\colb,below] (Nxyz) at (Axyz) {$\haa_{(\x,\the\y)}$};
    \fill[color=\colb] (Axyz) circle (0.8pt);
  }

  \coordinate (eps0) at ($1/\kk*(0,-0.05,0)$);
  \coordinate (eps) at ($1/\kk*(0,-0.05,0.05)$);

  \coordinate (A00) at ($(AA)$);
  \coordinate (A30) at ($(AA) + 1/\kk*(0,3,0)$);
  \coordinate (A01) at ($(AA) + 1/\kk*(0,0,1)$);
  \coordinate (A21) at ($(AA) + 1/\kk*(0,2,1)$);
  \coordinate (A02) at ($(AA) + 1/\kk*(0,0,2)$);
  \coordinate (A12) at ($(AA) + 1/\kk*(0,1,2)$);
  \coordinate (A03) at ($(AA) + 1/\kk*(0,0,3)$);
  \draw[color=\colb,dashed,line width=1.4pt,->,>=latex]
    (A03) -- ($(A12) + (eps)$);
  \draw[color=\colb,dashed,line width=1pt,->,>=latex]
    (A12) -- ($(A02) - (eps0)$);
  \draw[color=\colb,dashed,line width=1pt,->,>=latex]
    (A02) -- ($(A21) + (eps)$);
  \draw[color=\colb,dashed,line width=1pt,->,>=latex]
    (A21) -- ($(A01) - (eps0)$);
  \draw[color=\colb,dashed,line width=1pt,->,>=latex]
    (A01) -- ($(A30) + (eps)$);
  \draw[color=\colb,dashed,line width=1.4pt,->,>=latex]
    (A30) -- ($(A00) - (eps0)$);
\end{tikzpicture}
  }
  \hfill
  \caption[Lex, colex, symlex, and revlex orderings]{%
    Lex (top left), colex (top right), symlex (bottom left), and revlex (bottom
    right) orderings of~$\calAkd$ when $d=2$ and $k=3$ (see
    Section~\ref{ss:lex}).\\
    The increase in the order is represented by dashed arrows.
    For~$\calAiiidii$, we have
    $(0,0)\ltlex(0,1)\ltlex(0,2)\ltlex(0,3)\ltlex(1,0)\ltlex(1,1)\ltlex(1,2)\ltlex(2,0)\ltlex(2,1)\ltlex(3,0)$,
    and
    $(0,0)\ltcolex(1,0)\ltcolex(2,0)\ltcolex(3,0)\ltcolex(0,1)\ltcolex(1,1)\ltcolex(2,1)\ltcolex(0,2)\ltcolex(1,2)\ltcolex(0,3)$.\\
    The symlex order is the symmetrical of the lex order, and the
    revlex order is the symmetrical of the colex order.
    For instance, when the length is~3 (hypotenuse of the triangles, blue
    nodes), we have $(0,3)\ltlex(1,2)\ltlex(2,1)\ltlex(3,0)$, and also
    $(0,3)\ltrevlex(1,2)\ltrevlex(2,1)\ltrevlex(3,0)$.}
  \label{f:tet-tria_k3_lex_colex}
\end{figure}

The lex order and its variants are not convenient in practice here, as they do
not sort the monomials of a given polynomial according to their total degrees:
for instance, for $d=2$, let $p\eqdef X_1^0X_2^8$ and $q\eqdef X_1^1X_2^2$, we
have $p\ltlex q$ (as $0<1$), but $\deg(p)=8>3=\deg(q)$.
Thus, in the sequel, they are only be used to define other monomial orders.

\clearpage
\subsection{Graded lexicographic order}
\label{ss:grlex}

The {\em graded lexicographic order}, or simply {\em ``grlex'' order}, is
defined by
\begin{equation}
  \label{e:grlex}
  \aalpha \ltgrlex \bbeta \equivdef
  \left\{
    \begin{array}{l}
      \len{\aalpha} < \len{\bbeta}, \mbox { or}\\
      \len{\aalpha} = \len{\bbeta} \CONJ \aalpha \ltlex \bbeta.
    \end{array}
  \right.
\end{equation}
This amounts to first compare the length of multi-indices, and in case of
equality, use the standard lex order~\eqref{e:lex}.
Thus, when $\len{\aalpha}=\len{\bbeta}$, we have $\aalpha\ltgrlex\bbeta$ iff
$\alpha_i<\beta_i$ for the first index~$i$ for which~$\alpha_i$ and~$\beta_i$
differ.

We have the following equivalence, which may be seen as an alternative
recursive definition,
\begin{equation}
  \label{e:grlex-equiv}
  \aalpha \ltgrlex \bbeta \EQUIV
  \left\{
    \begin{array}{l}
      \len{\aalpha} < \len{\bbeta}, \mbox { or}\\
      \len{\aalpha} = \len{\bbeta}
      \CONJ \alpha_1 < \beta_1, \mbox{ or}\\
      \len{\aalpha} = \len{\bbeta} \CONJ \alpha_1 = \beta_1
      \CONJ d \geq 2 \CONJ \caalpha \ltgrlex \cbbeta,
    \end{array}
  \right.
\end{equation}
as when $\len{\aalpha}=\len{\bbeta}$ and $\alpha_1=\beta_1$, we have
$\len{\caalpha}=\len{\cbbeta}$, and lex and grlex are identical.
Note that the second case ($\len{\aalpha}=\len{\bbeta}$ and $\alpha_1<\beta_1$)
implies $d\geq2$.

This ordering is a monomial order.
It is also called {\em degree lexicographic order}, or simply
{\em ``deglex'' order}.
See~2D and~3D examples in Figure~\ref{f:tet-tria_k3_grlex}.

\begin{figure}[t]
  \centering
  \resizebox{0.75\linewidth}{!}{
    \begin{tikzpicture}[scale=4,math3d] 

  \def\kk{3}

  \def\colk{black}
  \def\colo{magenta}
  \def\coli{darkgreen}
  \def\colii{red}
  \def\coliii{blue}

  \def\opacity{0.4}
  \def\opacityi{0.7}
  \def\opacityii{0.6}

  \coordinate (AA) at (0,0,-0.25);
  \coordinate (CC) at ($(AA) + (0,1,0)$);
  \coordinate (DD) at ($(AA) + (0,0,1)$);
  \draw (AA) node[above left] {$\hvv_0$};
  \draw (CC) node[above=1.5pt] {$\hvv_1=(1,0)$};
  \draw (DD) node[right=1.5pt] {$\hvv_2=(0,1)$};
  \draw[line width=1.0pt,rounded corners=0.5pt] (AA) -- (CC) -- (DD) -- cycle;
  \node[color=\colk] (K2) at ($(AA) + (0,0.65,0.65)$) {$\hK_{2}$};

  \fill[color=\colo] (AA) circle (0.8pt);
  \node[color=\colo,below] (Nxyz) at (AA) {$\haa_{(0,0)}$};
  \newcount\y
  \foreach \x in {1,0} {
    \pgfmathsetcount{\y}{1-\x} 
    \coordinate (Axyz) at ($(AA) + 1/\kk*(0,\x,\y)$);
    \node[color=\coli,below] (Nxyz) at (Axyz) {$\haa_{(\x,\the\y)}$};
    \fill[color=\coli] (Axyz) circle (0.8pt);
  }
  \newcount\y
  \foreach \x in {2,1,...,0} {
    \pgfmathsetcount{\y}{2-\x} 
    \coordinate (Axyz) at ($(AA) + 1/\kk*(0,\x,\y)$);
    \node[color=\colii,below] (Nxyz) at (Axyz) {$\haa_{(\x,\the\y)}$};
    \fill[color=\colii] (Axyz) circle (0.8pt);
  }
  \newcount\y
  \foreach \x in {3,2,...,0} {
    \pgfmathsetcount{\y}{3-\x} 
    \coordinate (Axyz) at ($(AA) + 1/\kk*(0,\x,\y)$);
    \node[color=\coliii,below] (Nxyz) at (Axyz) {$\haa_{(\x,\the\y)}$};
    \fill[color=\coliii] (Axyz) circle (0.8pt);
  }

  \coordinate (eps0) at ($1/\kk*(0,0,-0.05)$); 
  \coordinate (eps) at ($1/\kk*(0,0.05,-0.05)$); 
  \node[color=\colo,below=12pt] (Nxyz) at (AA) {$\calCod$};
  \coordinate (A10) at ($(AA) + 1/\kk*(0,1,0)$);
  \coordinate (A01) at ($(AA) + 1/\kk*(0,0,1)$);
  \draw[color=\colo,dashed,line width=1.4pt,->,>=latex] (AA) -- ($(A01) + (eps0)$);
  \draw[color=\coli,dashed,line width=1pt,->,>=latex] (A01) -- ($(A10) - (eps)$);
  \node[color=\coli,below=12pt] (Nxyz) at (A10) {$\calCid$};
  \coordinate (A20) at ($(AA) + 1/\kk*(0,2,0)$);
  \coordinate (A02) at ($(AA) + 1/\kk*(0,0,2)$);
  \draw[color=\coli,dashed,line width=1pt,->,>=latex] ($(A10) + (eps)$) -- ($(A02) + (eps)$);
  \draw[color=\colii,dashed,line width=1pt,->,>=latex] (A02) -- ($(A20) - (eps)$);
  \node[color=\colii,below=12pt] (Nxyz) at (A20) {$\calCiid$};
  \coordinate (A30) at ($(AA) + 1/\kk*(0,3,0)$);
  \coordinate (CD) at ($1/2*(CC) + 1/2*(DD)$);
  \draw[color=\colii,dashed,line width=1pt,->,>=latex] ($(A20) + (eps)$) -- ($(DD) + (eps)$);
  \draw[color=\coliii,dashed,line width=1.4pt,->,>=latex]  ($(DD) + (eps)$) -- ($(A30) - (eps)$);
  \node[color=\coliii,below=12pt] (Nxyz) at (A30) {$\calCiiid$};

  \coordinate (A) at (0,2.0,0);
  \coordinate (B) at ($(A) + (1,0,0)$);
  \coordinate (C) at ($(A) + (0,1,0)$);
  \coordinate (D) at ($(A) + (0,0,1)$);
  \draw (A) node[left] {$\hvv_0$}; 
  \draw (B) node[below=12pt] {$\hvv_1=(1,0,0)$}; 
  \draw (C) node[above=2pt] {$\hvv_2=(0,1,0)$}; 
  \draw (D) node[right=1.5pt] {$\hvv_3=(0,0,1)$}; 
  \draw[line width=1.0pt,rounded corners=0.5pt] (A) -- (B) -- (D) -- cycle;
  \draw[line width=1.0pt,rounded corners=0.5pt] (A) -- (C) -- (B) -- cycle;
  \draw[line width=1.0pt,rounded corners=0.5pt] (B) -- (C) -- (D) -- cycle;
  \draw[line width=1.6pt,rounded corners=0.5pt] (A) -- (C) -- (D) -- cycle;

  \node[color=\colk] (K3) at ($(A) + (0,-0.3,0.7)$) {$\hK_{3}$};

  \pgfmathparse{1-1/\kk}\let\kkp\pgfmathresult
  \coordinate (B1) at ($\kkp*(A) + 1/\kk*(B)$);
  \coordinate (C1) at ($\kkp*(A) + 1/\kk*(C)$);
  \coordinate (D1) at ($\kkp*(A) + 1/\kk*(D)$);
  \pgfmathparse{1-2/\kk}\let\kkp\pgfmathresult
  \coordinate (B2) at ($\kkp*(A) + 2/\kk*(B)$);
  \coordinate (C2) at ($\kkp*(A) + 2/\kk*(C)$);
  \coordinate (D2) at ($\kkp*(A) + 2/\kk*(D)$);

  \draw[color=\coli,fill=\coli!40,fill opacity=\opacity] (B1) -- (C1) -- (D1) -- cycle;
  \draw[color=\colii,fill=\colii!20,fill opacity=\opacity] (B2) -- (C2) -- (D2) -- cycle;
  \draw[color=\coliii,fill=\coliii!20,fill opacity=\opacity] (B) -- (C) -- (D) -- cycle;

  \fill[color=\colo,fill opacity=\opacityii] (A) circle (0.8pt);
  \foreach \x in {1,0} {
    \pgfmathparse{1-\x}\let\YY\pgfmathresult
    \foreach \y in {0,...,\YY} {
      \pgfmathparse{1-\x-\y}\let\z\pgfmathresult
      \fill[color=\coli,fill opacity=\opacityii] ($(A) + 1/\kk*(\x,\y,\z)$) circle (0.8pt);
    }
  }
  \foreach \x in {2,1,...,0} {
    \pgfmathparse{2-\x}\let\YY\pgfmathresult
    \foreach \y in {0,...,\YY} {
      \pgfmathparse{2-\x-\y}\let\z\pgfmathresult
      \fill[color=\colii,fill opacity=\opacityii] ($(A) + 1/\kk*(\x,\y,\z)$) circle (0.8pt);
    }
  }
  \newcount\z
  \foreach \x in {3,2,...,0} {
    \pgfmathparse{3-\x}\let\YY\pgfmathresult
    \foreach \y in {0,...,\YY} {
      \pgfmathsetcount{\z}{3-\x-\y} 
      \coordinate (Axyz) at ($(A) + 1/\kk*(\x,\y,\z)$);
      \node[color=\coliii,below=-0.15] (Nxyz) at (Axyz) {$\haa_{(\x,\y,\the\z)}$};
      \fill[color=\coliii] (Axyz) circle (0.8pt);
    }
  }

  \coordinate (eps0) at ($1/\kk*(0.05,-0.05,0)$); 
  \coordinate (eps) at ($1/\kk*(0,0.05,-0.05)$); 
  \coordinate (A210) at ($(A) + 1/\kk*(2,1,0)$);
  \coordinate (A201) at ($(A) + 1/\kk*(2,0,1)$);
  \draw[color=\coliii,dashed,line width=1.4pt,<-,>=latex] (B) -- ($(A210) + (eps0)$);
  \draw[color=\coliii,dashed,line width=0.7pt,<-,>=latex] (A210) -- ($(A201) + (eps)$);
  \coordinate (A120) at ($(A) + 1/\kk*(1,2,0)$);
  \coordinate (A102) at ($(A) + 1/\kk*(1,0,2)$);
  \draw[color=\coliii,dashed,line width=0.7pt,<-,>=latex] ($(A201) + (eps)$) -- ($(A120) - (eps)$);
  \draw[color=\coliii,dashed,line width=0.7pt,<-,>=latex] (A120) -- ($(A102) + (eps)$);
  \coordinate (CD) at ($1/2*(C) + 1/2*(D)$);
  \draw[color=\coliii,dashed,line width=0.7pt,<-,>=latex] ($(A102) + (eps)$) -- ($(C) - (eps)$);
  \draw[color=\coliii,dashed,line width=1.4pt,<-,>=latex] (C) -- ($(D) + (eps)$);

\end{tikzpicture}
  }
  \caption[Grlex ordering]{%
    Grlex (or deglex) ordering of~$\calAkd$ when $d\in\{2,3\}$ and $k=3$
    (see Section~\ref{ss:grlex}).\\
    The increase in the order is represented by dashed arrows (only in the
    case $l=3$ when $d=3$, see Figure~\ref{f:lag-k3-d2-d3}).\\
    For~$\calAiiidii$ ($d=2$), we have
    $(0,0)\ltgrlex(0,1)\ltgrlex(1,0)\ltgrlex(0,2)\ltgrlex(1,1)\ltgrlex(2,0)\ltgrlex(0,3)\ltgrlex(1,2)\ltgrlex(2,1)\ltgrlex(3,0)$.\\
    For~$\calCiiidiii$ ($d=3$), we have
    $(0,0,3)\ltgrlex(0,1,2)\ltgrlex(0,2,1)\ltgrlex(0,3,0)\ltgrlex(1,0,2)\ltgrlex(1,1,1)\ltgrlex(1,2,0)\ltgrlex(2,0,1)\ltgrlex(2,1,0)\ltgrlex(3,0,0)$.}
  \label{f:tet-tria_k3_grlex}
\end{figure}

\clearpage
\subsection{Graded colexicographic order}
\label{ss:grcolex}

We design the {\em graded colexicographic order}, or simply
{\em ``grcolex'' order}, defined by
\begin{equation}
  \label{e:grcolex}
  \aalpha \ltgrcolex \bbeta \equivdef
  \left\{
    \begin{array}{l}
      \len{\aalpha} < \len{\bbeta}, \mbox { or}\\
      \len{\aalpha} = \len{\bbeta} \CONJ \aalpha \ltcolex \bbeta.
    \end{array}
  \right.
\end{equation}
This amounts to first compare the length of multi-indices, and in case of
equality, use the colex order~\eqref{e:colex}.
Thus, when $\len{\aalpha}=\len{\bbeta}$, we have $\aalpha\ltgrcolex\bbeta$ iff
$\alpha_i<\beta_i$ for the last index~$i$ where~$\alpha_i$ and~$\beta_i$
differ.

We have the following equivalence, which may be seen as an alternative
recursive definition,
\begin{equation}
  \label{e:grcolex-equiv}
  \aalpha \ltgrcolex \bbeta \EQUIV
  \left\{
    \begin{array}{l}
      \len{\aalpha} < \len{\bbeta}, \mbox { or}\\
      \len{\aalpha} = \len{\bbeta}
      \CONJ \alpha_d < \beta_d, \mbox{ or}\\
      \len{\aalpha} = \len{\bbeta} \CONJ \alpha_d = \beta_d
      \CONJ d \geq 2 \CONJ \taalpha \ltgrcolex \tbbeta.
    \end{array}
  \right.
\end{equation}
as when $\len{\aalpha}=\len{\bbeta}$ and $\alpha_d=\beta_d$, we have
$\len{\taalpha}=\len{\tbbeta}$, and colex and grcolex are identical.
Note that the second case ($\len{\aalpha}=\len{\bbeta}$ and $\alpha_d<\beta_d$)
implies $d\geq2$.

This ordering is a monomial order.
Compare on the same~2D and~3D examples the grcolex order in
Figure~\ref{f:tet-tria_k3_grcolex} and the grlex order in
Figure~\ref{f:tet-tria_k3_grlex}.

\begin{figure}[t]
  \centering
  \resizebox{0.75\linewidth}{!}{
    \begin{tikzpicture}[scale=4,math3d] 

  \def\kk{3}

  \def\colk{black}
  \def\colo{magenta}
  \def\coli{darkgreen}
  \def\colii{red}
  \def\coliii{blue}

  \def\opacity{0.4}
  \def\opacityi{0.7}
  \def\opacityii{0.6}

  \coordinate (AA) at (0,0,-0.25);
  \coordinate (CC) at ($(AA) + (0,1,0)$);
  \coordinate (DD) at ($(AA) + (0,0,1)$);
  \draw (AA) node[above left] {$\hvv_0$};
  \draw (CC) node[above=1.5pt] {$\hvv_1=(1,0)$};
  \draw (DD) node[right=1.5pt] {$\hvv_2=(0,1)$};
  \draw[line width=1.0pt,rounded corners=0.5pt] (AA) -- (CC) -- (DD) -- cycle;
  \node[color=\colk] (K2) at ($(AA) + (0,0.65,0.65)$) {$\hK_{2}$};

  \fill[color=\colo] (AA) circle (0.8pt);
  \node[color=\colo,below] (Nxyz) at (AA) {$\haa_{(0,0)}$};
  \newcount\y
  \foreach \x in {1,0} {
    \pgfmathsetcount{\y}{1-\x} 
    \coordinate (Axyz) at ($(AA) + 1/\kk*(0,\x,\y)$);
    \node[color=\coli,below] (Nxyz) at (Axyz) {$\haa_{(\x,\the\y)}$};
    \fill[color=\coli] (Axyz) circle (0.8pt);
  }
  \newcount\y
  \foreach \x in {2,1,...,0} {
    \pgfmathsetcount{\y}{2-\x} 
    \coordinate (Axyz) at ($(AA) + 1/\kk*(0,\x,\y)$);
    \node[color=\colii,below] (Nxyz) at (Axyz) {$\haa_{(\x,\the\y)}$};
    \fill[color=\colii] (Axyz) circle (0.8pt);
  }
  \newcount\y
  \foreach \x in {3,2,...,0} {
    \pgfmathsetcount{\y}{3-\x} 
    \coordinate (Axyz) at ($(AA) + 1/\kk*(0,\x,\y)$);
    \node[color=\coliii,below] (Nxyz) at (Axyz) {$\haa_{(\x,\the\y)}$};
    \fill[color=\coliii] (Axyz) circle (0.8pt);
  }

  \coordinate (eps0) at ($1/\kk*(0,-0.05,0)$); 
  \coordinate (eps) at ($1/\kk*(0,0.05,-0.05)$); 
  \node[color=\colo,below=12pt] (Nxyz) at (AA) {$\calCod$};
  \coordinate (A10) at ($(AA) + 1/\kk*(0,1,0)$);
  \coordinate (A01) at ($(AA) + 1/\kk*(0,0,1)$);
  \draw[color=\colo,dashed,line width=1.4pt,->,>=latex] (AA) -- ($(A10) + (eps0)$);
  \draw[color=\coli,dashed,line width=1pt,->,>=latex] (A10) -- ($(A01) + (eps)$);
  \node[color=\coli,below=12pt] (Nxyz) at (A10) {$\calCid$};
  \coordinate (A20) at ($(AA) + 1/\kk*(0,2,0)$);
  \coordinate (A02) at ($(AA) + 1/\kk*(0,0,2)$);
  \draw[color=\coli,dashed,line width=1pt,->,>=latex] ($(A01) + (eps)$) -- ($(A20) - (eps)$);
  \draw[color=\colii,dashed,line width=1pt,->,>=latex] (A20) -- ($(A02) + (eps)$);
  \node[color=\colii,below=12pt] (Nxyz) at (A20) {$\calCiid$};
  \coordinate (A30) at ($(AA) + 1/\kk*(0,3,0)$);
  \coordinate (CD) at ($1/2*(CC) + 1/2*(DD)$);
  \draw[color=\colii,dashed,line width=1pt,->,>=latex] ($(A02) + (eps)$) -- ($(A30) - (eps)$);
  \draw[color=\coliii,dashed,line width=1.4pt,->,>=latex] (A30) -- ($(DD) + (eps)$);
  \node[color=\coliii,below=12pt] (Nxyz) at (A30) {$\calCiiid$};

  \coordinate (A) at (0,2.0,0);
  \coordinate (B) at ($(A) + (1,0,0)$);
  \coordinate (C) at ($(A) + (0,1,0)$);
  \coordinate (D) at ($(A) + (0,0,1)$);
  \draw (A) node[left] {$\hvv_0$}; 
  \draw (B) node[below=12pt] {$\hvv_1=(1,0,0)$}; 
  \draw (C) node[above=2pt] {$\hvv_2=(0,1,0)$}; 
  \draw (D) node[right=1.5pt] {$\hvv_3=(0,0,1)$}; 
  \draw[line width=1.0pt,rounded corners=0.5pt] (A) -- (B) -- (D) -- cycle;
  \draw[line width=1.0pt,rounded corners=0.5pt] (A) -- (C) -- (B) -- cycle;
  \draw[line width=1.0pt,rounded corners=0.5pt] (B) -- (C) -- (D) -- cycle;
  \draw[line width=1.6pt,rounded corners=0.5pt] (A) -- (C) -- (D) -- cycle;

  \node[color=\colk] (K3) at ($(A) + (0,-0.3,0.7)$) {$\hK_{3}$};

  \pgfmathparse{1-1/\kk}\let\kkp\pgfmathresult
  \coordinate (B1) at ($\kkp*(A) + 1/\kk*(B)$);
  \coordinate (C1) at ($\kkp*(A) + 1/\kk*(C)$);
  \coordinate (D1) at ($\kkp*(A) + 1/\kk*(D)$);
  \pgfmathparse{1-2/\kk}\let\kkp\pgfmathresult
  \coordinate (B2) at ($\kkp*(A) + 2/\kk*(B)$);
  \coordinate (C2) at ($\kkp*(A) + 2/\kk*(C)$);
  \coordinate (D2) at ($\kkp*(A) + 2/\kk*(D)$);

  \draw[color=\coli,fill=\coli!40,fill opacity=\opacity] (B1) -- (C1) -- (D1) -- cycle;
  \draw[color=\colii,fill=\colii!20,fill opacity=\opacity] (B2) -- (C2) -- (D2) -- cycle;
  \draw[color=\coliii,fill=\coliii!20,fill opacity=\opacity] (B) -- (C) -- (D) -- cycle;

  \fill[color=\colo,fill opacity=\opacityii] (A) circle (0.8pt);
  \foreach \x in {1,0} {
    \pgfmathparse{1-\x}\let\YY\pgfmathresult
    \foreach \y in {0,...,\YY} {
      \pgfmathparse{1-\x-\y}\let\z\pgfmathresult
      \fill[color=\coli,fill opacity=\opacityii] ($(A) + 1/\kk*(\x,\y,\z)$) circle (0.8pt);
    }
  }
  \foreach \x in {2,1,...,0} {
    \pgfmathparse{2-\x}\let\YY\pgfmathresult
    \foreach \y in {0,...,\YY} {
      \pgfmathparse{2-\x-\y}\let\z\pgfmathresult
      \fill[color=\colii,fill opacity=\opacityii] ($(A) + 1/\kk*(\x,\y,\z)$) circle (0.8pt);
    }
  }
  \newcount\z
  \foreach \x in {3,2,...,0} {
    \pgfmathparse{3-\x}\let\YY\pgfmathresult
    \foreach \y in {0,...,\YY} {
      \pgfmathsetcount{\z}{3-\x-\y} 
      \coordinate (Axyz) at ($(A) + 1/\kk*(\x,\y,\z)$);
      \node[color=\coliii,below=-0.15] (Nxyz) at (Axyz) {$\haa_{(\x,\y,\the\z)}$};
      \fill[color=\coliii] (Axyz) circle (0.8pt);
    }
  }

  \coordinate (eps0) at ($1/\kk*(0,-0.05,+0.05)$); 
  \coordinate (eps) at ($1/\kk*(0,0.05,0)$); 
  \coordinate (A012) at ($(A) + 1/\kk*(0,1,2)$);
  \coordinate (A102) at ($(A) + 1/\kk*(1,0,2)$);
  \coordinate (A021) at ($(A) + 1/\kk*(0,2,1)$);
  \coordinate (A201) at ($(A) + 1/\kk*(2,0,1)$);
  \draw[color=\coliii,dashed,line width=1.5pt,<-,>=latex] (D) -- ($(A012) + (eps0)$);
  \draw[color=\coliii,dashed,line width=0.7pt,<-,>=latex] (A012) -- ($(A102) + (eps)$);
  \draw[color=\coliii,dashed,line width=0.7pt,<-,>=latex] ($(A102) + (eps)$) -- ($(A021) - (eps)$);
  \draw[color=\coliii,dashed,line width=0.7pt,<-,>=latex] (A021) -- ($(A201) + (eps)$);
  \coordinate (CD) at ($1/2*(C) + 1/2*(D)$);
  \draw[color=\coliii,dashed,line width=0.7pt,<-,>=latex] ($(A201) + (eps)$) -- ($(C) - (eps)$);
  \draw[color=\coliii,dashed,line width=1.5pt,<-,>=latex] (C) -- ($(B) + (eps)$);

\end{tikzpicture}
  }
  \caption[Grcolex ordering]{%
    Grcolex ordering of~$\calAkd$ when $d\in\{2,3\}$ and $k=3$
    (see Section~\ref{ss:grcolex}).\\
    The increase in the order is represented by dashed arrows (only in the
    case $l=3$ when $d=3$, see Figure~\ref{f:lag-k3-d2-d3}).\\
    For~$\calAiiidii$ ($d=2$), we have
    $(0,0)\ltgrcolex(1,0)\ltgrcolex(0,1)\ltgrcolex(2,0)\ltgrcolex(1,1)\ltgrcolex(0,2)\ltgrcolex(3,0)\ltgrcolex(2,1)\ltgrcolex(1,2)\ltgrcolex(0,3)$.\\
    For~$\calCiiidiii$ ($d=3$), we have
    $(3,0,0)\ltgrcolex(2,1,0)\ltgrcolex(1,2,0)\ltgrcolex(0,3,0)\ltgrcolex(2,0,1)\ltgrcolex(1,1,1)\ltgrcolex(0,2,1)\ltgrcolex(1,0,2)\ltgrcolex(0,1,2)\ltgrcolex(0,0,3)$.
  }
  \label{f:tet-tria_k3_grcolex}
\end{figure}

\clearpage
\subsection{Graded symmetric lexicographic order}
\label{ss:grsymlex}

We also design the {\em graded symmetric lexicographic order}, or simply
{\em ``grsymlex'' order}, defined by
\begin{equation}
  \label{e:grsymlex}
  \aalpha \ltgrsymlex \bbeta \equivdef
  \left\{
    \begin{array}{l}
      \len{\aalpha} < \len{\bbeta}, \mbox { or}\\
      \len{\aalpha} = \len{\bbeta} \CONJ \aalpha \ltsymlex \bbeta.
    \end{array}
  \right.
\end{equation}
This amounts to first compare the length of multi-indices, and in case of
equality, use the symlex order~\eqref{e:symlex}.
Thus, when $\len{\aalpha}=\len{\bbeta}$, we have $\aalpha\ltgrsymlex\bbeta$ iff
$\beta_i<\alpha_i$ for the first index~$i$ where~$\alpha_i$ and~$\beta_i$
differ.

We have the following equivalence, which may be seen as an alternative
recursive definition,
\begin{equation}
  \label{e:grsymlex-equiv-1}
  \aalpha \ltgrsymlex \bbeta \EQUIV
  \left\{
    \begin{array}{l}
      \len{\aalpha} < \len{\bbeta}, \mbox { or}\\
      \len{\aalpha} = \len{\bbeta}
      \CONJ \beta_1 < \alpha_1, \mbox{ or}\\
      \len{\aalpha} = \len{\bbeta} \CONJ \beta_1 = \alpha_1
      \CONJ d \geq 2 \CONJ \caalpha \ltgrsymlex \cbbeta.
    \end{array}
  \right.
\end{equation}
as when $\len{\aalpha}=\len{\bbeta}$ and $\beta_1=\alpha_1$, we have
$\len{\caalpha}=\len{\cbbeta}$, and symlex and grsymlex are identical.
Note that the second case ($\len{\aalpha}=\len{\bbeta}$ and $\beta_1<\alpha_1$)
implies $d\geq2$.

Moreover, the recursive definition may be simplified with the following
equivalence,
\begin{equation}
  \label{e:grsymlex-equiv-2}
  \aalpha \ltgrsymlex \bbeta \EQUIV
  \left\{
    \begin{array}{l}
      \len{\aalpha} < \len{\bbeta}, \mbox { or}\\
      \len{\aalpha} = \len{\bbeta}
      \CONJ d \geq 2 \CONJ \caalpha \ltgrsymlex \cbbeta.
    \end{array}
  \right.
\end{equation}
as when $\len{\aalpha}=\len{\bbeta}$ and $\beta_1<\alpha_1$, we have~$d\geq2$
and $\len{\caalpha}<\len{\cbbeta}$, {\ie} $\caalpha\ltgrsymlex\cbbeta$.
Note that this simplification is made possible by the symmetric aspect of
symlex.
Note also that it is not possible for the grlex and grcolex orders:
indeed, for instance for grlex,
$\len{\aalpha}=\len{\bbeta}\Conj\alpha_1<\beta_1$ implies
$\len{\cbbeta}<\len{\caalpha}$, {\ie} $\cbbeta\ltgrlex\caalpha$,
but not $\caalpha\ltgrlex\cbbeta$.

This ordering is a monomial order.
Note that the grcolex and grsymlex orders are identical when $d=1$ or~2, but
differ as soon as $d\geq3$.
Compare on the same~2D and~3D examples the grsymlex order in
Figure~\ref{f:tet-tria_k3_grsymlex}, and the grcolex order in
Figure~\ref{f:tet-tria_k3_grcolex}.

\begin{figure}[t]
  \centering
  \resizebox{0.75\linewidth}{!}{
    \begin{tikzpicture}[scale=4,math3d] 

  \def\kk{3}

  \def\colk{black}
  \def\colo{magenta}
  \def\coli{darkgreen}
  \def\colii{red}
  \def\coliii{blue}

  \def\opacity{0.4}
  \def\opacityi{0.7}
  \def\opacityii{0.6}

  \coordinate (AA) at (0,0,-0.25);
  \coordinate (CC) at ($(AA) + (0,1,0)$);
  \coordinate (DD) at ($(AA) + (0,0,1)$);
  \draw (AA) node[above left] {$\hvv_0$};
  \draw (CC) node[above=1.5pt] {$\hvv_1=(1,0)$};
  \draw (DD) node[right=1.5pt] {$\hvv_2=(0,1)$};
  \draw[line width=1.0pt,rounded corners=0.5pt] (AA) -- (CC) -- (DD) -- cycle;
  \node[color=\colk] (K2) at ($(AA) + (0,0.65,0.65)$) {$\hK_{2}$};

  \fill[color=\colo] (AA) circle (0.8pt);
  \node[color=\colo,below] (Nxyz) at (AA) {$\haa_{(0,0)}$};
  \newcount\y
  \foreach \x in {1,0} {
    \pgfmathsetcount{\y}{1-\x} 
    \coordinate (Axyz) at ($(AA) + 1/\kk*(0,\x,\y)$);
    \node[color=\coli,below] (Nxyz) at (Axyz) {$\haa_{(\x,\the\y)}$};
    \fill[color=\coli] (Axyz) circle (0.8pt);
  }
  \newcount\y
  \foreach \x in {2,1,...,0} {
    \pgfmathsetcount{\y}{2-\x} 
    \coordinate (Axyz) at ($(AA) + 1/\kk*(0,\x,\y)$);
    \node[color=\colii,below] (Nxyz) at (Axyz) {$\haa_{(\x,\the\y)}$};
    \fill[color=\colii] (Axyz) circle (0.8pt);
  }
  \newcount\y
  \foreach \x in {3,2,...,0} {
    \pgfmathsetcount{\y}{3-\x} 
    \coordinate (Axyz) at ($(AA) + 1/\kk*(0,\x,\y)$);
    \node[color=\coliii,below] (Nxyz) at (Axyz) {$\haa_{(\x,\the\y)}$};
    \fill[color=\coliii] (Axyz) circle (0.8pt);
  }

  \coordinate (eps0) at ($1/\kk*(0,-0.05,0)$); 
  \coordinate (eps) at ($1/\kk*(0,0.05,-0.05)$); 
  \node[color=\colo,below=12pt] (Nxyz) at (AA) {$\calCod$};
  \coordinate (A10) at ($(AA) + 1/\kk*(0,1,0)$);
  \coordinate (A01) at ($(AA) + 1/\kk*(0,0,1)$);
  \draw[color=\colo,dashed,line width=1.4pt,->,>=latex] (AA) -- ($(A10) + (eps0)$);
  \draw[color=\coli,dashed,line width=1pt,->,>=latex] (A10) -- ($(A01) + (eps)$);
  \node[color=\coli,below=12pt] (Nxyz) at (A10) {$\calCid$};
  \coordinate (A20) at ($(AA) + 1/\kk*(0,2,0)$);
  \coordinate (A02) at ($(AA) + 1/\kk*(0,0,2)$);
  \draw[color=\coli,dashed,line width=1pt,->,>=latex] ($(A01) + (eps)$) -- ($(A20) - (eps)$);
  \draw[color=\colii,dashed,line width=1pt,->,>=latex] (A20) -- ($(A02) + (eps)$);
  \node[color=\colii,below=12pt] (Nxyz) at (A20) {$\calCiid$};
  \coordinate (A30) at ($(AA) + 1/\kk*(0,3,0)$);
  \coordinate (CD) at ($1/2*(CC) + 1/2*(DD)$);
  \draw[color=\colii,dashed,line width=1pt,->,>=latex] ($(A02) + (eps)$) -- ($(A30) - (eps)$);
  \draw[color=\coliii,dashed,line width=1.4pt,->,>=latex] (A30) -- ($(DD) + (eps)$);
  \node[color=\coliii,below=12pt] (Nxyz) at (A30) {$\calCiiid$};

  \coordinate (A) at (0,2.0,0);
  \coordinate (B) at ($(A) + (1,0,0)$);
  \coordinate (C) at ($(A) + (0,1,0)$);
  \coordinate (D) at ($(A) + (0,0,1)$);
  \draw (A) node[left] {$\hvv_0$}; 
  \draw (B) node[below=12pt] {$\hvv_1=(1,0,0)$}; 
  \draw (C) node[above=2pt] {$\hvv_2=(0,1,0)$}; 
  \draw (D) node[right=1.5pt] {$\hvv_3=(0,0,1)$}; 
  \draw[line width=1.0pt,rounded corners=0.5pt] (A) -- (B) -- (D) -- cycle;
  \draw[line width=1.0pt,rounded corners=0.5pt] (A) -- (C) -- (B) -- cycle;
  \draw[line width=1.0pt,rounded corners=0.5pt] (B) -- (C) -- (D) -- cycle;
  \draw[line width=1.6pt,rounded corners=0.5pt] (A) -- (C) -- (D) -- cycle;

  \node[color=\colk] (K3) at ($(A) + (0,-0.3,0.7)$) {$\hK_{3}$};

  \pgfmathparse{1-1/\kk}\let\kkp\pgfmathresult
  \coordinate (B1) at ($\kkp*(A) + 1/\kk*(B)$);
  \coordinate (C1) at ($\kkp*(A) + 1/\kk*(C)$);
  \coordinate (D1) at ($\kkp*(A) + 1/\kk*(D)$);
  \pgfmathparse{1-2/\kk}\let\kkp\pgfmathresult
  \coordinate (B2) at ($\kkp*(A) + 2/\kk*(B)$);
  \coordinate (C2) at ($\kkp*(A) + 2/\kk*(C)$);
  \coordinate (D2) at ($\kkp*(A) + 2/\kk*(D)$);

  \draw[color=\coli,fill=\coli!40,fill opacity=\opacity] (B1) -- (C1) -- (D1) -- cycle;
  \draw[color=\colii,fill=\colii!20,fill opacity=\opacity] (B2) -- (C2) -- (D2) -- cycle;
  \draw[color=\coliii,fill=\coliii!20,fill opacity=\opacity] (B) -- (C) -- (D) -- cycle;

  \fill[color=\colo,fill opacity=\opacityii] (A) circle (0.8pt);
  \foreach \x in {1,0} {
    \pgfmathparse{1-\x}\let\YY\pgfmathresult
    \foreach \y in {0,...,\YY} {
      \pgfmathparse{1-\x-\y}\let\z\pgfmathresult
      \fill[color=\coli,fill opacity=\opacityii] ($(A) + 1/\kk*(\x,\y,\z)$) circle (0.8pt);
    }
  }
  \foreach \x in {2,1,...,0} {
    \pgfmathparse{2-\x}\let\YY\pgfmathresult
    \foreach \y in {0,...,\YY} {
      \pgfmathparse{2-\x-\y}\let\z\pgfmathresult
      \fill[color=\colii,fill opacity=\opacityii] ($(A) + 1/\kk*(\x,\y,\z)$) circle (0.8pt);
    }
  }
  \newcount\z
  \foreach \x in {3,2,...,0} {
    \pgfmathparse{3-\x}\let\YY\pgfmathresult
    \foreach \y in {0,...,\YY} {
      \pgfmathsetcount{\z}{3-\x-\y} 
      \coordinate (Axyz) at ($(A) + 1/\kk*(\x,\y,\z)$);
      \node[color=\coliii,below=-0.15] (Nxyz) at (Axyz) {$\haa_{(\x,\y,\the\z)}$};
      \fill[color=\coliii] (Axyz) circle (0.8pt);
    }
  }

  \coordinate (eps0) at ($1/\kk*(0.05,-0.05,0)$); 
  \coordinate (eps) at ($1/\kk*(0,0.05,-0.05)$); 
  \coordinate (A210) at ($(A) + 1/\kk*(2,1,0)$);
  \coordinate (A201) at ($(A) + 1/\kk*(2,0,1)$);
  \draw[color=\coliii,dashed,line width=1.4pt,->,>=latex] (B) -- ($(A210) + (eps0)$);
  \draw[color=\coliii,dashed,line width=0.7pt,->,>=latex] (A210) -- ($(A201) + (eps)$);
  \coordinate (A120) at ($(A) + 1/\kk*(1,2,0)$);
  \coordinate (A102) at ($(A) + 1/\kk*(1,0,2)$);
  \draw[color=\coliii,dashed,line width=0.7pt,->,>=latex] ($(A201) + (eps)$) -- ($(A120) - (eps)$);
  \draw[color=\coliii,dashed,line width=0.7pt,->,>=latex] (A120) -- ($(A102) + (eps)$);
  \coordinate (CD) at ($1/2*(C) + 1/2*(D)$);
  \draw[color=\coliii,dashed,line width=0.7pt,->,>=latex] ($(A102) + (eps)$) -- ($(C) - (eps)$);
  \draw[color=\coliii,dashed,line width=1.4pt,->,>=latex] (C) -- ($(D) + (eps)$);

\end{tikzpicture}
  }
  \caption[Grsymlex ordering]{%
    Grsymlex ordering of~$\calAkd$ when $d\in\{2,3\}$ and $k=3$
    (see Section~\ref{ss:grsymlex}).\\
    The increase in the order is represented by dashed arrows (only in the
    case $l=3$ when $d=3$, see Figure~\ref{f:lag-k3-d2-d3}).\\
    For~$\calAiiidii$ ($d=2$), we have (as in the case of grcolex)
    $(0,0)\ltgrsymlex(1,0)\ltgrsymlex(0,1)\ltgrsymlex(2,0)\ltgrsymlex(1,1)\ltgrsymlex(0,2)\ltgrsymlex(3,0)\ltgrsymlex(2,1)\ltgrsymlex(1,2)\ltgrsymlex(0,3)$.\\
    When $d\geq3$, grcolex and grsymlex differ.
    For instance, for~$\calCiiidiii$ (compare with
    Figure~\ref{f:tet-tria_k3_grcolex}), we have
    $(3,0,0)\ltgrsymlex(2,1,0)\ltgrsymlex(2,0,1)\ltgrsymlex(1,2,0)\ltgrsymlex(1,1,1)\ltgrsymlex(1,0,2)\ltgrsymlex(0,3,0)\ltgrsymlex(0,2,1)\ltgrsymlex(0,1,2)\ltgrsymlex(0,0,3)$.}
  \label{f:tet-tria_k3_grsymlex}
\end{figure}

\clearpage
\subsection{Graded reverse lexicographic order}
\label{ss:grevlex}

The {\em graded reverse lexicographic order}, or simply
{\em ``grevlex'' order}, is defined by
\begin{equation}
  \label{e:grevlex}
  \aalpha \ltgrevlex \bbeta \equivdef
  \left\{
    \begin{array}{l}
      \len{\aalpha} < \len{\bbeta}, \mbox { or}\\
      \len{\aalpha} = \len{\bbeta} \CONJ \aalpha \ltrevlex \bbeta.
    \end{array}
  \right.
\end{equation}
This amounts to first compare the length of multi-indices, and in case of
equality, use the revlex order~\eqref{e:revlex}.
Thus, when $\len{\aalpha}=\len{\bbeta}$, we have
$\aalpha\ltgrevlex\bbeta$ iff $\beta_i<\alpha_i$ for the last index~$i$
where~$\alpha_i$ and~$\beta_i$ differ.

We have the following equivalence, which may be seen as an alternative
recursive definition,
\begin{equation}
  \label{e:grevlex-equiv-1}
  \aalpha \ltgrevlex \bbeta \EQUIV
  \left\{
    \begin{array}{l}
      \len{\aalpha} < \len{\bbeta}, \mbox { or}\\
      \len{\aalpha} = \len{\bbeta}
      \CONJ \beta_d < \alpha_d, \mbox{ or}\\
      \len{\aalpha} = \len{\bbeta} \CONJ \beta_d = \alpha_d
      \CONJ d \geq 2 \CONJ \taalpha \ltgrevlex \tbbeta,
    \end{array}
  \right.
\end{equation}
as when $\len{\aalpha}=\len{\bbeta}$ and $\beta_d=\alpha_d$, we have
$\len{\taalpha}=\len{\tbbeta}$, and revlex and grevlex are identical.
Note that the second case ($\len{\aalpha}=\len{\bbeta}$ and $\beta_d<\alpha_d$)
implies $d\geq2$.

Moreover, the recursive definition may be simplified with the following
equivalence,
\begin{equation}
  \label{e:grevlex-equiv-2}
  \aalpha \ltgrevlex \bbeta \EQUIV
  \left\{
    \begin{array}{l}
      \len{\aalpha} < \len{\bbeta}, \mbox { or}\\
      \len{\aalpha} = \len{\bbeta}
      \CONJ d \geq 2 \CONJ \taalpha \ltgrevlex \tbbeta,
    \end{array}
  \right.
\end{equation}
as when $\len{\aalpha}=\len{\bbeta}$ and $\beta_d<\alpha_d$, we have~$d\geq2$
and $\len{\taalpha}<\len{\tbbeta}$, {\ie} $\taalpha\ltgrevlex\tbbeta$.

This ordering is a monomial order.
Note that the grlex and grevlex orders are identical when $d=1$ or~2, but
differ as soon as $d\geq3$.
Compare on the same~2D and~3D examples the grevlex order in
Figure~\ref{f:tet-tria_k3_grevlex}, and the grlex order in
Figure~\ref{f:tet-tria_k3_grlex}.

\begin{figure}[t]
  \centering
  \resizebox{0.75\linewidth}{!}{
    \begin{tikzpicture}[scale=4,math3d] 

  \def\kk{3}

  \def\colk{black}
  \def\colo{magenta}
  \def\coli{darkgreen}
  \def\colii{red}
  \def\coliii{blue}

  \def\opacity{0.4}
  \def\opacityi{0.7}
  \def\opacityii{0.6}

  \coordinate (AA) at (0,0,-0.25);
  \coordinate (CC) at ($(AA) + (0,1,0)$);
  \coordinate (DD) at ($(AA) + (0,0,1)$);
  \draw (AA) node[above left] {$\hvv_0$};
  \draw (CC) node[above=1.5pt] {$\hvv_1=(1,0)$};
  \draw (DD) node[right=1.5pt] {$\hvv_2=(0,1)$};
  \draw[line width=1.0pt,rounded corners=0.5pt] (AA) -- (CC) -- (DD) -- cycle;
  \node[color=\colk] (K2) at ($(AA) + (0,0.65,0.65)$) {$\hK_{2}$};

  \fill[color=\colo] (AA) circle (0.8pt);
  \node[color=\colo,below] (Nxyz) at (AA) {$\haa_{(0,0)}$};
  \newcount\y
  \foreach \x in {1,0} {
    \pgfmathsetcount{\y}{1-\x} 
    \coordinate (Axyz) at ($(AA) + 1/\kk*(0,\x,\y)$);
    \node[color=\coli,below] (Nxyz) at (Axyz) {$\haa_{(\x,\the\y)}$};
    \fill[color=\coli] (Axyz) circle (0.8pt);
  }
  \newcount\y
  \foreach \x in {2,1,...,0} {
    \pgfmathsetcount{\y}{2-\x} 
    \coordinate (Axyz) at ($(AA) + 1/\kk*(0,\x,\y)$);
    \node[color=\colii,below] (Nxyz) at (Axyz) {$\haa_{(\x,\the\y)}$};
    \fill[color=\colii] (Axyz) circle (0.8pt);
  }
  \newcount\y
  \foreach \x in {3,2,...,0} {
    \pgfmathsetcount{\y}{3-\x} 
    \coordinate (Axyz) at ($(AA) + 1/\kk*(0,\x,\y)$);
    \node[color=\coliii,below] (Nxyz) at (Axyz) {$\haa_{(\x,\the\y)}$};
    \fill[color=\coliii] (Axyz) circle (0.8pt);
  }

  \coordinate (eps0) at ($1/\kk*(0,0,-0.05)$); 
  \coordinate (eps) at ($1/\kk*(0,0.05,-0.05)$); 
  \node[color=\colo,below=12pt] (Nxyz) at (AA) {$\calCod$};
  \coordinate (A10) at ($(AA) + 1/\kk*(0,1,0)$);
  \coordinate (A01) at ($(AA) + 1/\kk*(0,0,1)$);
  \draw[color=\colo,dashed,line width=1.4pt,->,>=latex] (AA) -- ($(A01) + (eps0)$);
  \draw[color=\coli,dashed,line width=1pt,->,>=latex] (A01) -- ($(A10) - (eps)$);
  \node[color=\coli,below=12pt] (Nxyz) at (A10) {$\calCid$};
  \coordinate (A20) at ($(AA) + 1/\kk*(0,2,0)$);
  \coordinate (A02) at ($(AA) + 1/\kk*(0,0,2)$);
  \draw[color=\coli,dashed,line width=1pt,->,>=latex] ($(A10) + (eps)$) -- ($(A02) + (eps)$);
  \draw[color=\colii,dashed,line width=1pt,->,>=latex] (A02) -- ($(A20) - (eps)$);
  \node[color=\colii,below=12pt] (Nxyz) at (A20) {$\calCiid$};
  \coordinate (A30) at ($(AA) + 1/\kk*(0,3,0)$);
  \coordinate (CD) at ($1/2*(CC) + 1/2*(DD)$);
  \draw[color=\colii,dashed,line width=1pt,->,>=latex] ($(A20) + (eps)$) -- ($(DD) + (eps)$);
  \draw[color=\coliii,dashed,line width=1.4pt,->,>=latex]  ($(DD) + (eps)$) -- ($(A30) - (eps)$);
  \node[color=\coliii,below=12pt] (Nxyz) at (A30) {$\calCiiid$};

  \coordinate (A) at (0,2.0,0);
  \coordinate (B) at ($(A) + (1,0,0)$);
  \coordinate (C) at ($(A) + (0,1,0)$);
  \coordinate (D) at ($(A) + (0,0,1)$);
  \draw (A) node[left] {$\hvv_0$}; 
  \draw (B) node[below=12pt] {$\hvv_1=(1,0,0)$}; 
  \draw (C) node[above=2pt] {$\hvv_2=(0,1,0)$}; 
  \draw (D) node[right=1.5pt] {$\hvv_3=(0,0,1)$}; 
  \draw[line width=1.0pt,rounded corners=0.5pt] (A) -- (B) -- (D) -- cycle;
  \draw[line width=1.0pt,rounded corners=0.5pt] (A) -- (C) -- (B) -- cycle;
  \draw[line width=1.0pt,rounded corners=0.5pt] (B) -- (C) -- (D) -- cycle;
  \draw[line width=1.6pt,rounded corners=0.5pt] (A) -- (C) -- (D) -- cycle;

  \node[color=\colk] (K3) at ($(A) + (0,-0.3,0.7)$) {$\hK_{3}$};

  \pgfmathparse{1-1/\kk}\let\kkp\pgfmathresult
  \coordinate (B1) at ($\kkp*(A) + 1/\kk*(B)$);
  \coordinate (C1) at ($\kkp*(A) + 1/\kk*(C)$);
  \coordinate (D1) at ($\kkp*(A) + 1/\kk*(D)$);
  \pgfmathparse{1-2/\kk}\let\kkp\pgfmathresult
  \coordinate (B2) at ($\kkp*(A) + 2/\kk*(B)$);
  \coordinate (C2) at ($\kkp*(A) + 2/\kk*(C)$);
  \coordinate (D2) at ($\kkp*(A) + 2/\kk*(D)$);

  \draw[color=\coli,fill=\coli!40,fill opacity=\opacity] (B1) -- (C1) -- (D1) -- cycle;
  \draw[color=\colii,fill=\colii!20,fill opacity=\opacity] (B2) -- (C2) -- (D2) -- cycle;
  \draw[color=\coliii,fill=\coliii!20,fill opacity=\opacity] (B) -- (C) -- (D) -- cycle;

  \fill[color=\colo,fill opacity=\opacityii] (A) circle (0.8pt);
  \foreach \x in {1,0} {
    \pgfmathparse{1-\x}\let\YY\pgfmathresult
    \foreach \y in {0,...,\YY} {
      \pgfmathparse{1-\x-\y}\let\z\pgfmathresult
      \fill[color=\coli,fill opacity=\opacityii] ($(A) + 1/\kk*(\x,\y,\z)$) circle (0.8pt);
    }
  }
  \foreach \x in {2,1,...,0} {
    \pgfmathparse{2-\x}\let\YY\pgfmathresult
    \foreach \y in {0,...,\YY} {
      \pgfmathparse{2-\x-\y}\let\z\pgfmathresult
      \fill[color=\colii,fill opacity=\opacityii] ($(A) + 1/\kk*(\x,\y,\z)$) circle (0.8pt);
    }
  }
  \newcount\z
  \foreach \x in {3,2,...,0} {
    \pgfmathparse{3-\x}\let\YY\pgfmathresult
    \foreach \y in {0,...,\YY} {
      \pgfmathsetcount{\z}{3-\x-\y} 
      \coordinate (Axyz) at ($(A) + 1/\kk*(\x,\y,\z)$);
      \node[color=\coliii,below=-0.15] (Nxyz) at (Axyz) {$\haa_{(\x,\y,\the\z)}$};
      \fill[color=\coliii] (Axyz) circle (0.8pt);
    }
  }

  \coordinate (eps0) at ($1/\kk*(0,-0.05,+0.05)$); 
  \coordinate (eps) at ($1/\kk*(0,0.05,0)$); 
  \coordinate (A012) at ($(A) + 1/\kk*(0,1,2)$);
  \coordinate (A102) at ($(A) + 1/\kk*(1,0,2)$);
  \coordinate (A021) at ($(A) + 1/\kk*(0,2,1)$);
  \coordinate (A201) at ($(A) + 1/\kk*(2,0,1)$);
  \draw[color=\coliii,dashed,line width=1.5pt,->,>=latex] (D) -- ($(A012) + (eps0)$);
  \draw[color=\coliii,dashed,line width=0.7pt,->,>=latex] (A012) -- ($(A102) + (eps)$);
  \draw[color=\coliii,dashed,line width=0.7pt,->,>=latex] ($(A102) + (eps)$) -- ($(A021) - (eps)$);
  \draw[color=\coliii,dashed,line width=0.7pt,->,>=latex] (A021) -- ($(A201) + (eps)$);
  \coordinate (CD) at ($1/2*(C) + 1/2*(D)$);
  \draw[color=\coliii,dashed,line width=0.7pt,->,>=latex] ($(A201) + (eps)$) -- ($(C) - (eps)$);
  \draw[color=\coliii,dashed,line width=1.5pt,->,>=latex] (C) -- ($(B) + (eps)$);

\end{tikzpicture}
  }
  \caption[Grevlex ordering]{%
    Grevlex ordering of~$\calAkd$ when $d\in\{2,3\}$ and $k=3$
    (see Section~\ref{ss:grevlex}).\\
    The increase in the order is represented by dashed arrows (only in the
    case $l=3$ when $d=3$, see Figure~\ref{f:lag-k3-d2-d3}).\\
    For~$\calAiiidii$ ($d=2$), we have (as in the case of grlex)
    $(0,0)\ltgrevlex(0,1)\ltgrevlex(1,0)\ltgrevlex(0,2)\ltgrevlex(1,1)\ltgrevlex(2,0)\ltgrevlex(0,3)\ltgrlex(1,2)\ltgrevlex(2,1)\ltgrevlex(3,0)$.\\
    When $d\geq3$, grlex and grevlex differ.
    For instance, for~$\calCiiidiii$ (compare with
    Figure~\ref{f:tet-tria_k3_grlex}), we have
    $(0,0,3)\ltgrevlex(0,1,2)\ltgrevlex(1,0,2)\ltgrevlex(0,2,1)\ltgrevlex(1,1,1)\ltgrevlex(2,0,1)\ltgrevlex(0,3,0)\ltgrevlex(1,2,0)\ltgrevlex(2,1,0)\ltgrevlex(3,0,0)$.}
  \label{f:tet-tria_k3_grevlex}
\end{figure}



\clearpage
\subsection{Discussion}
\label{ss:discussion}

Several orders on~$\matNd$ may be used to number multi-indices and the grevlex
order of Section~\ref{ss:grevlex} is known to be well-suited for the division
of multivariate polynomials using Gröbner bases, {\eg}
see~\cite{wp:mo,clo:iva:15}.
However, in the context of the FEM, other desired properties enter the picture.
In this framework, a practical order should be consistent with
\begin{itemize}
\item[(i)] an increase of the degree (from~$k$ to~$k+1$):
  multi-indices of length at most~$k$ should be numbered before those of
  length~$k+1$, {\ie} for all $\aalpha\in\calCkd\subset\calAkd$ and for all
  $\bbeta\in\calCkpid$, we should have $\aalpha<\bbeta$;

\item[(ii)] an increase of the dimension (from~$d-1$ to~$d$):
  ``natural'' bijections between $d-1$-multi-indices of length at most~$k$ and
  $d$-multi-indices of length~$k$ should be increasing,
  \begin{itemize}
  \item[$\bullet$]
    $\fkdo\eqdef(\caalpha\in\calAkdmi
    \longmapsto(k-\len{\caalpha},\caalpha)\in\calCkd)$, or \\
    $\tfkdo\eqdef(\taalpha\in\calAkdmi
    \longmapsto(\taalpha,k-\len{\taalpha})\in\calCkd)$
    from Lemma~\ref{l:card-Ckd-Akdm1}, and

  \item[$\bullet$]
    for any $i\in[1..d]$, $\fkdi\eqdef(\aalphap\in\calAkdmi\longmapsto
    (\alphap_1,\dots,\alphap_{i-1},0,\alphap_i,\dots,\alphap_{d-1})\in\calAkdi)$
    from Lemma~\ref{l:card-Akdi-Akdm1},
  \end{itemize}
  {\ie} $\calAkdmi$, $\calCkd$ ($=\fkdo(\calAkdmi)=\tfkdo(\calAkdmi)$),
  and $\calAkdi(=\fkdi(\calAkdmi))$ share the same numbering;

\item[(iii)] the natural numbering of the multi-indices corresponding to the
  reference vertices~$(\hvv_i)_{i\in[0..d]}$ of the reference simplex,
  {\ie} $\zzero<k\ee_1<k\ee_2<\ldots<k\ee_d$.
\end{itemize}
Indeed, condition~(i) allows to easily sort monomials with respect to their
total degree.
Condition~(ii) allows to easily relate the face nodes and the volume nodes
during the computations.
And condition~(iii) allows to have {\em positive} simplices ({\ie} positively
oriented) when computing the integrals.

\bigskip

Obviously, condition~(i) disqualifies the {\bf lex} order and its variants of
Section~\ref{ss:lex}, but all ``graded'' orders are designed to comply
with it.

\bigskip

Obviously, we also have $\zzero<k\ee_1$ (from condition~(iii)) for all
``graded'' orders.

\bigskip

Conditions~(ii) and~(iii) are not met for the {\bf grlex} order of
Section~\ref{ss:grlex}, which is thus disqualified for our purpose.
To see this, here are some counter-examples, {\cf}
Figure~\ref{f:tet-tria_k3_grlex}.
For~(ii), with~$d=k=3$, we have $(1,0)\ltgrlex(0,2)$, while
$f^3_{3,0}(0,2)=(1,0,2)\ltgrlex(2,1,0)=f^3_{3,0}(1,0)$, and
$\tf^3_{3,0}(0,2)=(0,2,1)\ltgrlex(1,0,2)=\tf^3_{3,0}(1,0)$.
For~(iii), with~$d=2$ and~$k=3$, the vertex~$\hvv_1=\haa_{(3,0)}$ is numbered
after the vertex~$\hvv_2=\haa_{(0,3)}$.

\bigskip

The {\bf grcolex} order of Section~\ref{ss:grcolex} complies with~(iii).
Indeed, if $i<j$, then $\len{k\ee_i}=\len{k\ee_j}=k$, $j$~is the last index
where the components of~$k\ee_i$ and~$k\ee_j$ differ, and the $j$-th component
of~$k\ee_i$ is~0, whereas $(k\ee_j)_j=k>0$, thus $k\ee_i\ltcolex k\ee_j$, {\ie}
$k\ee_i\ltgrcolex k\ee_j$ (see Figure~\ref{f:tet-tria_k3_grcolex} in the
cases~$d\in\{2,3\}$ and~$k=3$).

However, condition~(ii) is not met for the {\bf grcolex} order, which is thus
also disqualified for our purpose.
Indeed, from Figure~\ref{f:tet-tria_k3_grcolex} with~$d=k=3$, we have
$(0,1)\ltgrcolex(2,0)$, while
\begin{gather*}
  f^3_{3,0} (2, 0) = (1, 2, 0) \ltgrcolex (2, 0, 1) = f^3_{3,0} (0, 1), \AND\\
  \tf^3_{3,0} (2, 0) = (2, 0, 1) \ltgrcolex (0, 1, 2) = \tf^3_{3,0} (0, 1).
\end{gather*}

\bigskip

The {\bf grevlex} order of Section~\ref{ss:grevlex} satisfies condition~(ii)
with $\tfkdo$ and $\fkdi$ (for all $i\in[1..d]$).
Indeed, we have the following.
\begin{itemize}
\item
  Let $\taalpha,\tbbeta\in\calAkdmi$ such that $\taalpha\ltgrevlex\tbbeta$.
  Let $\ggamma\eqdef\tfkdo(\taalpha)$ and $\ddelta\eqdef\tfkdo(\tbbeta)$.
  Both~$\ggamma$ and~$\ddelta$ have length~$k$, and by removing their last
  component, we obtain $\tggamma=\taalpha$ and $\tddelta=\tbbeta$.
  Hence, \eqref{e:grevlex-equiv-2} yields
  $\tfkdo(\taalpha)\ltgrevlex\tfkdo(\tbbeta)$.
  Thus, $\tfkdo$ is increasing with respect to {\bf grevlex}.

\item
  Let $i\in[1..d]$, and $\aalphap,\bbetap\in\calAkdmi$ such that
  $\aalphap\ltgrevlex\bbetap$.
  We have $\len{\fkdi(\aalphap)}=\len{\aalphap}$ and
  $\len{\fkdi(\bbetap)}=\len{\bbetap}$.
  Thus, from~\eqref{e:grevlex}, we have two cases:
  \begin{itemize}
  \item
    if $\len{\aalphap}<\len{\bbetap}$, then we have
    $\len{\fkdi(\aalphap)}<\len{\fkdi(\bbetap)}$, and thus
    $\fkdi(\aalphap)\ltgrevlex\fkdi(\bbetap)$.

  \item
    if $\len{\aalphap}=\len{\bbetap}$, then $\aalphap\ltrevlex\bbetap$,
    and $\len{\fkdi(\aalphap)}=\len{\fkdi(\bbetap)}$.
    The insertion of a~``0'' in position~$i$ in~$\aalphap$ and~$\bbetap$ does
    not alter the revlex order, and thus we also have
    $\fkdi(\aalphap)\ltrevlex\fkdi(\bbetap)$, {\ie}
    $\fkdi(\aalphap)\ltgrevlex\fkdi(\bbetap)$.
  \end{itemize}
  Thus, in both cases, $\fkdi$ is increasing with respect to {\bf grevlex}.
\end{itemize}

However, {\bf grevlex}  does not comply with condition~(iii), as the
vertex~$\hvv_1=\haa_{(3,0)}$ is numbered after the
vertex~$\hvv_2=\haa_{(0,3)}$, see Figure~\ref{f:tet-tria_k3_grevlex} (left).
For this reason, we also avoided {\bf grevlex}.

\bigskip

Finally, the {\bf grsymlex} order of Section~\ref{ss:grsymlex} satisfies the
three conditions.
Indeed, for condition~(ii), using~$\fkdo$ instead of~$\tfkdo$ as for grevlex,
we have the following.
\begin{itemize}
\item
  Let $\caalpha,\cbbeta\in\calAkdmi$ such that $\caalpha\ltgrsymlex\cbbeta$.
  Let $\ggamma\eqdef\fkdo(\caalpha)$ and $\ddelta\eqdef\fkdo(\cbbeta)$.
  Both~$\ggamma$ and~$\ddelta$ have length~$k$, and by removing their first
  component, we obtain $\cggamma=\caalpha$ and $\cddelta=\cbbeta$.
  Hence, \eqref{e:grsymlex-equiv-2} yields
  $\fkdo(\caalpha)\ltgrsymlex\fkdo(\cbbeta)$.
  Thus, $\fkdo$ is increasing with respect to {\bf grsymlex}.

\item
  Let $i\in[1..d]$, and $\aalphap,\bbetap\in\calAkdmi$ such that
  $\aalphap\ltgrsymlex\bbetap$.
  We have $\len{\fkdi(\aalphap)}=\len{\aalphap}$ and
  $\len{\fkdi(\bbetap)}=\len{\bbetap}$.
  Thus, from~\eqref{e:grsymlex}, we have two cases:
  \begin{itemize}
  \item
    if $\len{\aalphap}<\len{\bbetap}$, then we have
    $\len{\fkdi(\aalphap)}<\len{\fkdi(\bbetap)}$, and thus
    $\fkdi(\aalphap)\ltgrsymlex\fkdi(\bbetap)$.

  \item
    if $\len{\aalphap}=\len{\bbetap}$, then $\aalphap\ltsymlex\bbetap$,
    and $\len{\fkdi(\aalphap)}=\len{\fkdi(\bbetap)}$.
    The insertion of a~``0'' in position~$i$ in~$\aalphap$ and~$\bbetap$ does
    not alter the symlex order, and thus we also have
    $\fkdi(\aalphap)\ltsymlex\fkdi(\bbetap)$, {\ie}
    $\fkdi(\aalphap)\ltgrsymlex\fkdi(\bbetap)$.
  \end{itemize}
  Thus, in both cases, $\fkdi$ is increasing with respect to {\bf grsymlex}.
\end{itemize}

The {\bf grsymlex} order also complies with (iii).
Indeed, if $i<j$, then we have $\len{k\ee_i}=\len{k\ee_j}=k$, $i$~is the
first index where the components of~$k\ee_i$ and~$k\ee_j$ differ, and the
$i$-th component of~$k\ee_j$ is~0, whereas $(k\ee_i)_i=k>0$, thus
$k\ee_i\ltsymlex k\ee_j$, {\ie} $k\ee_i\ltgrsymlex k\ee_j$ (see
Figure~\ref{f:tet-tria_k3_grsymlex} in the cases~$d\in\{2,3\}$ and~$k=3$).

This is why the {\bf grsymlex} order was chosen for the {\coq} implementation.

\bigskip

Note that the numbering of~$\calAkd$ is not used in the present version of this
document, but is essential in the {\coq} formalization of the present work
in~\cite{mou:phd:24,bol:cfs:24}.

\part{Detailed proofs}
\label{p:detailed-proofs}

\chapter{Introduction}
\label{c:introduction-2}

Statements are displayed inside colored boxes.
Their nature can be identified at a glance by using the following color code:
\begin{center}
  \rmkbox{light gray is for remarks}, \qquad
  \defbox{light green for definitions},\\
  \lembox{light blue for lemmas}, \quad and \quad
  \thmbox{light red for theorems}.
\end{center}
Definitions and results have a number and a name.
Inside the bodies of proofs, pertinent statements are referenced using both
their number and name.
When appropriate, some hints are given about the application of the result,
either to specify arguments, or to provide justification or consequences;
they are \underline{underlined}.
Some useful definitions and results were already stated
in~\cite{cm:lmt:16} or in~\cite{cm:li:21}, which were respectively
devoted to the detailed proofs of the {\LMT} and for Lebesgue integration.
Those are numbered up to~{\theprevtheorem}, and the statements in the present
document are numbered starting from~{\thenexttheorem}.

Furthermore, as in~\cite{cm:lmt:16,cm:li:21}, the most basic results are
supposed to be known and are not detailed further;
they are displayed in \assume{bold dark red}.
These include:
\begin{itemize}
\item Naive set theory:
  definition and results about injective and bijective functions.

\item Numbers:
  \begin{CompactList}
  \item ordered abelian monoid properties of~$\matN$;
  \item ordered valued field properties of~$\matR$.
  \end{CompactList}

\item Linear algebra:
  \begin{CompactList}
  \item
    basic definitions and results about {\vectorspace}s and subspaces, such as
    freedom, generator, linear span, basis, dimension,
    and linear operations over functions;

  \item
    basic definitions and results about linear maps, such as
    distributivity of composition over addition,
    range/rank and kernel/nullity,
    the rank--nullity theorem,
    the characterization of surjectivity with the rank
    and injectivity with the nullity;

  \item
    some results in finite dimension, such as
    the rules for matrix--vector product,
    the incomplete basis theorem,
    and the dimension of the dual space.
  \end{CompactList}

\item
  Topology in metric spaces:
  definitions of interior, open ball, and continuity.

\item
  Calculus:
  basic definitions and results in~$\matRn$, such as
  the usual norms, the differentiability of affine maps ($\Cinf$), and the
  rules of differentiation.

\item
  Real analysis:
  basic definitions and results in~$\matR$, such as
  the rules of derivation.

\item
  Polynomials:
  basic definitions about univariate polynomials, such as the degree.
\end{itemize}

\bigskip

This part is organized as follows.
Chapter~\ref{c:complements} contains some results from various fields of
mathematics (arithmetics, linear and affine algebra, geometry, and univariate
polynomials), that are needed in the proofs for the finite element method.
We recall that the material stated in~\cite{cm:lmt:16,cm:li:21} may also be
used.
Then, Chapter~\ref{c:fe} is devoted to the general definition of finite
element, and Chapter~\ref{c:simplex} addresses simplicial geometry.
Lagrange finite elements on segments are presented in
Chapter~\ref{c:Pk1-lag-fe}, and finally Chapter~\ref{c:Pkd-lag-fe} is dedicated
to the general case of dimension~$d\geq1$.

\chapter{Complements}
\label{c:complements}
\minitoc

\section{Complements on natural numbers}
\label{s:compl-nat-num}

\begin{remark}
  In Figures~\ref{f:double-induction-by-diagonal}
  and~\ref{f:strong-double-induction}, we provide illustrations for the double
  inductions used in this document.

  In Figure~\ref{f:double-induction-by-diagonal}, a  double induction ``by
  diagonal'' is presented, see Lemma~\ref{l:double-induction-by-diagonal}.
  The proof of unisolvence of~$\matPkd$ relies on it, see
  Lemma~\ref{l:lag-lin-forms-Pkd-inj}.
  In Figure~\ref{f:strong-double-induction}, the strong double induction is
  shown, see Lemma~\ref{l:strong-double-induction}.
  It is used in Lemma~\ref{l:prod-2-polynom-alt-proof}.
\end{remark}

\begin{figure}[htb]
  \centering
  \resizebox{0.7\linewidth}{!}{
    \begin{tikzpicture}

  \draw[step=.5cm, very thin, black!30] (0,0) grid (4,2) ;
  \def\colo{magenta}
  \def\coli{darkgreen}
  \def\colii{red}
  \def\coliii{blue}

  \def\kk{2}
  \def\epsi{0.45}
  \def\epsii{0.35}

  \coordinate (A) at (0,0);
  \coordinate (B) at  ($ (A) + (4,0) $);
  \coordinate (C) at  ($ (A) + (0,2) $);

  \coordinate (Nm) at  ($ (A) + (1.5,0) $);
  \coordinate (Nn) at  ($ (A) + (0,1) $);
  \coordinate (N) at  ($ (Nm) + (Nn) - (A) $);

  \draw ($ (B) + (\epsii,0) $) node[above=2pt] {$\matN$};
  \draw ($ (C) + (0,\epsii) $) node[left=2pt] {$\matN$};

  \draw[line width=1.5pt,->,>=latex] (A) -- ($ (B) + (\epsi,0) $);
  \draw[line width=1.5pt,->,>=latex] (A) -- ($ (C) + (0,\epsi) $);

  \foreach \m in {0,1,...,8}  {
    \coordinate (Axy) at  ($ (A) + 1/\kk*(\m,0) $);
    \fill[color=\coli] (Axy) circle (1.5pt);
  }
  \foreach \n in {0,1,...,4}  {
    \coordinate (Axy) at  ($ (A) + 1/\kk*(0,\n) $);
    \fill[color=\coli] (Axy) circle (1.5pt);
  }

  \coordinate (Amn) at  ($ (N) + (0,0) $);
  \coordinate (Amp1n) at  ($ (N) + 1/\kk*(1,0) $);
  \coordinate (Amp2n) at  ($ (N) + 2/\kk*(1,0) $);

  \coordinate (Amm1np1) at  ($ (N) + 1/\kk*(-1,1) $);
  \coordinate (Amnp1) at  ($ (N) + 1/\kk*(0,1) $);
  \coordinate (Amp1np1) at  ($ (N) + 1/\kk*(1,1) $);
  \coordinate (Amp2np1) at  ($ (N) + 1/\kk*(2,1) $);

  \node[color=\coli,below] (Nmn) at ($ (Nm) - 1/\kk*(1,0)$) {%
    {\scriptsize $n$}};
  \node[color=\coli,above] (Nmn) at ($ (Nm) + (0,0)$) {%
    {\scriptsize $n\!+\!1$}};
  \node[color=\coli,below] (Nmn) at ($ (Nm) + 1/\kk*(1,0)$) {%
    {\scriptsize $n\!+\!2$}};
  \node[color=\coli,above] (Nmn) at ($ (Nm) + 2/\kk*(1,0)$) {%
    {\scriptsize $n\!+\!3$}};
  \node[color=\coli,left] (Nmn) at ($ (Nn) + (0,0)$) {{\scriptsize $m$}};
  \node[color=\coli,left] (Nmn) at ($ (Nn) + 1/\kk*(0,1)$) {%
    {\scriptsize $m\!+\!1$}};

  \fill[color=\coli] (Amn) circle (1.5pt);
  \fill[color=\coli] (Amp1n) circle (1.5pt);
  \fill[color=\coli] (Amp2n) circle (1.5pt);
  \fill[color=\coli] (Amm1np1) circle (1.5pt);
  \fill[color=\colii] (Amnp1) circle (1.5pt);
  \fill[color=\colii] (Amp1np1) circle (1.5pt);
  \fill[color=\colii] (Amp2np1) circle (1.5pt);

  \draw[color=\coli,line width=0.7pt,->,>=latex] (Amm1np1) -- (Amnp1);
  \draw[color=\coli,line width=0.7pt,->,>=latex] (Amn) -- (Amnp1);

  \draw[color=\colii,line width=0.7pt,->,>=latex] (Amnp1) -- (Amp1np1);
  \draw[color=\coli,line width=0.7pt,->,>=latex] (Amp1n) -- (Amp1np1);

  \draw[color=\colii,line width=0.7pt,->,>=latex] (Amp1np1) -- (Amp2np1);
  \draw[color=\coli,line width=0.7pt,->,>=latex] (Amp2n) -- (Amp2np1);

\end{tikzpicture}
  }
  \caption[Double induction scheme ``by diagonal'']{%
    Double induction scheme ``by diagonal'' (see
    Lemma~\ref{l:double-induction-by-diagonal}).\\
    Let a predicate~$\PropPP(m,n)$.
    Assume the initializations~$\PropPP(0,n)$ and~$\PropPP(m,0)$ for all~$m$
    and~$n$ (symbolized by green dots along the two axes).
    Then, if~$\PropPP(m,n+1)$ and~$\PropPP(m+1,n)$ implies~$\PropPP(m+1,n+1)$
    (represented by green dots and arrows), then~$\PropPP(m,n)$ holds
    everywhere.
    It is proven by induction on~$m$, then on~$n$.
    Thus, fixing~$m$, one proves successively~$\PropPP(m+1,k)$, for all~$k$
    (red arrows and dots).}
  \label{f:double-induction-by-diagonal}
\end{figure}

\begin{lemma}[double induction by diagonal]
  \label{l:double-induction-by-diagonal}
  \mbox{}\hfill
  Let~$\PropPP$ be a predicate on~$\matN^2$.
  Then, we have
  \begin{align}
    \nonumber
    & \big( \forall m \in \matN,\quad \PropPP(m,0) \big) \CONJ
      \big( \forall n \in \matN,\quad \PropPP(0,n) \big) \CONJ \\
    \nonumber
    & \big( \forall m, n \in \matN,\quad
      \PropPP (m, n + 1) \CONJ \PropPP (m + 1, n)
      \IMPLIES \PropPP (m + 1, n + 1) \big) \\
    \label{e:double-induction-by-diagonal}
    & \qquad \IMPLIES \forall m, n \in \matN,\quad \PropPP (m, n).
  \end{align}
\end{lemma}

\begin{proof}
  Let~$H_0$, $H_1$, and~$H_2$ be the three hypotheses in the
  implication~\eqref{e:double-induction-by-diagonal}.\\
  For all $m\in\matN$, let $\PropQQm\eqdef[\forall n\in\matN,\ \PropPPmn]$.

  \proofparskip{Induction: $\PropQQ(0)$}
  Trivial (hypothesis~$H_1$).

  \proofparskip{Induction: $\PropQQm\Implies\PropQQ(m+1)$}
  Let~$m\in\matN$, assume that $\PropQQm$ holds.\\
  \proofpar{Induction: $\PropPP(m+1,0)$}
  The property $\PropPP(m+1,0)$ holds by hypothesis~$H_0$.\\
  \proofpar{Induction: $\PropPP(m+1,n)\Implies\PropPP(m+1,n+1)$}
  Let~$n\in\matN$, assume that $\PropPP(m+1,n)$ holds.\\
  Then, from $\PropQQm$, we also have $\PropPP(m,n+1)$, and thus, from~$H_2$,
  we have $\PropPP(m+1,n+1)$.\\
  This concludes the induction on~$n$, and we have for all $n\in\matN$,
  $\PropPP(m+1,n)$, {\ie}, $\PropQQ(m+1)$.

  \medskip\noindent
  This concludes the induction on~$m$, and we have for all $m\in\matN$,
  $\PropQQm$.

  \medskip\noindent
  Therefore, we have $\forall m,n\in\matN$, $\PropPPmn$.
\end{proof}

\begin{remark}
  The strong induction principle stipulates that a predicate~$\PropPP$
  on~$\matN$ holds everywhere if it satisfies the property
  \begin{equation*}
    \forall n \in \matN,\quad
    \left( \forall n_1 < n,\quad \PropPP (n_1) \right)
    \IMPLIES \PropPP (n).
  \end{equation*}
  Note that there is no need to assume the base case~$\PropPP(0)$ since it
  follows from the previous hypothesis.
  Indeed, for all $n_1\in\matN$, $n_1<0$, which is false, implies anything,
  such as $\PropPP(n_1)$.

  In the next lemma, the condition
  $m_1\leq m\Conj n_1\leq n\Conj (m_1,n_1)\neq(m,n)$ can be reformulated as
  $(m_1<m\Conj n_1\leq n)\Disj(m_1=m\Conj n_1<n)$.
  As in the strong induction for a single integer, the initialization
  steps~$\PropPP(m,0)$ and~$\PropPP(0,n)$ are not needed for the double strong
  induction.
\end{remark}

\begin{figure}[htb]
  \centering
  \resizebox{0.7\linewidth}{!}{
    \begin{tikzpicture}

  \draw[step=.5cm, very thin, black!30] (0,0) grid (4,2) ;
  \def\colo{magenta}
  \def\coli{darkgreen}
  \def\colii{red}
  \def\coliii{blue}

  \def\kk{2}
  \def\epsi{0.45}
  \def\epsii{0.35}

  \coordinate (A) at (0,0);
  \coordinate (B) at  ($ (A) + (4,0) $);
  \coordinate (C) at  ($ (A) + (0,2) $);

  \coordinate (Nm) at  ($ (A) + (1.5,0) $);
  \coordinate (Nn) at  ($ (A) + (0,1) $);
  \coordinate (N) at  ($ (Nm) + (Nn) - (A) $);

  \draw ($ (B) + (\epsii,0) $) node[above=2pt] {$\matN$};
  \draw ($ (C) + (0,\epsii) $) node[left=2pt] {$\matN$};

  \draw[line width=1.5pt,->,>=latex] (A) -- ($ (B) + (\epsi,0) $);
  \draw[line width=1.5pt,->,>=latex] (A) -- ($ (C) + (0,\epsi) $);

  \foreach \n in {0,1,...,4}  {
    \foreach \m in {0,1,...,2}  {
      \coordinate (Axy) at  ($ (A) + 1/\kk*(\n,\m) $);
    \fill[color=\coli] (Axy) circle (1.5pt);
    }
  }
  \foreach \n in {0,1,...,3}  {
    \coordinate (Axy) at  ($ (A) + 1/\kk*(\n,3) $);
    \fill[color=\coli] (Axy) circle (1.5pt);
  }

  \coordinate (Anm) at  ($ (N) + 1/\kk*(1,1) $);
  \coordinate (Anm1m) at  ($ (N) + 1/\kk*(0,1) $);
  \coordinate (Anmm1) at  ($ (N) + 1/\kk*(1,0) $);
  \coordinate (Anm1mm1) at  ($ (N) + 1/\kk*(0,0) $);
  \fill[color=\colii] (Anm) circle (2pt);

  \node[color=\coli,below] (Nmn) at ($ (Nm) - 1/\kk*(1,0)$) {%
    {\scriptsize $n\!-\!2$}};
  \node[color=\coli,above] (Nmn) at ($ (Nm) + (0,0)$) {%
    {\scriptsize $n\!-\!1$}};
  \node[color=\coli,below] (Nmn) at ($ (Nm) + 1/\kk*(1,0)$) {%
    {\scriptsize $n$}};
  \node[color=\coli,left] (Nmn) at ($ (Nn) + (0,0)$) {{\scriptsize $m\!-\!1$}};
  \node[color=\coli,left] (Nmn) at ($ (Nn) + 1/\kk*(0,1)$) {%
    {\scriptsize $m$}};


\end{tikzpicture}
  }
  \caption[Strong double induction scheme]{%
    Strong double induction scheme (see
    Lemma~\ref{l:strong-double-induction}).\\
    The induction hypothesis is as follows: if the proposition holds on all the
    green nodes, then it holds on the red one.}
  \label{f:strong-double-induction}
\end{figure}

\begin{lemma}[strong double induction]
  \label{l:strong-double-induction}
  \mbox{}\hfill
  Let~$\PropPP$ be a predicate on~$\matN^2$.
  Then, we have
  \begin{align}
    \nonumber
    & \big[ \forall m, n \in \matN,\quad
      \big( \forall m_1 \leq m,\; \forall n_1 \leq n,\;
      (m_1, n_1) \neq (m, n) \Implies \PropPP (m_1, n_1) \big)
      \IMPLIES \PropPP (m, n) \big] \\
    \label{e:strong-double-induction}
    & \qquad \IMPLIES \forall m, n \in \matN,\quad \PropPP (m, n).
  \end{align}
\end{lemma}

\begin{proof}
  For all $m\in\matN$, let $\PropQQm\eqdef[\forall n\in\matN,\ \PropPPmn]$.

  \proofparskip{Strong induction on~$m$:
    $(\forall m_1<m,\;\PropQQ(m_1))\Implies\PropQQm$}\\
  Let~$m\in\matN$, assume that $H_1: \forall m_1<m,\;\PropQQ(m_1)$ holds.
  Let us show that~$\PropQQm$ holds.

  \proofparskip{Strong induction on~$n$:
    $(\forall n_1<n,\;\PropPP(m,n_1))\Implies\PropPPmn$}\\
  Let~$n\in\matN$, assume that $H_2: \forall n_1<n,\;\PropPP(m,n_1)$ holds.
  Let us show that~$\PropPPmn$ holds.\\
  Let~$m_1\leq m$ and~$n_1\leq n$, assume that $(m_1,n_1)\neq(m,n)$.
  \proofpar{Case $m_1<m$}
  Then, $\PropPP(m_1,n_1)$ holds by~$H_1$.
  \proofpar{Case $m_1=m$}
  Then, $n_1<n$, and $\PropPP(m_1,n_1)=\PropPP(m,n_1)$ holds by~$H_2$.
  Thus, in both cases, $\PropPP(m_1,n_1)$ holds, and from the hypothesis
  in~\eqref{e:strong-double-induction}, we have~$\PropPPmn$.\\
  This concludes the strong induction on~$n$, and we have for all $n\in\matN$,
  $\PropPP(m,n)$, {\ie}, $\PropQQm$.

  \medskip\noindent
  This concludes the strong induction on~$m$, and we have for all $m\in\matN$,
  $\PropQQm$.

  \medskip\noindent
  Therefore, we have $\forall m,n\in\matN$, $\PropPPmn$.
\end{proof}

\begin{remark}
  \mbox{}\\
  Note that, for $n,p\in\matN$, we use the convention $n-p\eqdef0$ when $n<p$
  (to have a result in~$\matN$).
\end{remark}

\begin{definition}[binomial coefficient]
  \label{d:binom-coef}
  \mbox{}\\
  Let $n,p\in\matN$.
  Then, the {\em binomial coefficient} is defined as
  \begin{equation}
    \label{e:binom-coef}
    \binomnp \eqdef \frac{\fact{n}}{\fact{p}\;\fact{(n-p)}}
    \mbox{ when } p \leq n,
    \AND \binomnp \eqdef 0 \mbox{ otherwise}.
  \end{equation}
\end{definition}

\begin{remark}
  Note that the formula with factorials on the left of~\eqref{e:binom-coef}
  can be extended to the irregular case when $n<p$ (and yields the value~0)
  provided the use of the Euclidean division.
  Indeed, in this case $\binomnp=\frac{\fact{n}}{\fact{p}\fact{0}}$, whose
  Euclidean division is~0 because $\fact{n}<\fact{p}$.

    In the Pascal triangle formula~\eqref{e:prop-binom-coef-3}, the regular
    case $1\leq p\leq n-1$
    (with $n\geq1$) is extended with the convention
    of~\eqref{e:binom-coef}
    and the remark that $\binom{0-1}{p}=\binom{0}{p}$, for
    $p\in\matN$.
\end{remark}

\begin{lemma}[properties of the binomial coefficient]
  \label{l:prop-binom-coef}
  \mbox{}\hfill
  Let $n,p\in\matN$.
  Then, we have
  \begin{eqnarray}
    \label{e:prop-binom-coef-0}
    && \binom{n}{0} = \binom{n}{n} = 1, \\
    \label{e:prop-binom-coef-1}
    n \geq 1  &\IMPLIES& \binom{n}{1} = \binom{n}{n-1} = n ,\\
    \label{e:prop-binom-coef-2}
    p \leq n  &\IMPLIES& \binom{n}{n-p} = \binom{n}{p},\\
    \label{e:prop-binom-coef-3}
    (n \neq 0 \Disj p \neq 1) \CONJ p \neq 0 &\IMPLIES&
      \binomnp = \binom{n-1}{p-1} +\binom{n-1}{p}, \\
    \label{e:prop-binom-coef-4}
    p \geq 1 &\IMPLIES& \sum_{j=0}^n\binom{j+p-1}{p-1} = \binom{n+p}{p}, \\
    \label{e:prop-binom-coef-5}
    && \binom{n}{p} \in \matN.
  \end{eqnarray}
\end{lemma}

\begin{proof}
  \mbox{}\\
  Properties~\eqref{e:prop-binom-coef-0}, \eqref{e:prop-binom-coef-1}  and
  \eqref{e:prop-binom-coef-2} are direct consequences of
  Definition~\thref{d:binom-coef}.

  \proofparskip{\eqref{e:prop-binom-coef-3}}
  \proofpar{Case $p<n$}
  Then, $n\geq1$ and $1\leq p\leq n-1$, and all binomials are regular.
  Thus, from
  Definition~\thref{d:binom-coef},
  we have
  \[
  \binom{n-1}{p-1} +\binom{n-1}{p} =
  \frac{\fact{(n-1)}}{\fact{(p-1)}\;\fact{(n-p)}} +
  \frac{\fact{(n-1)}}{\fact{p}\;\fact{(n-p-1)}} =
  \frac{\fact{(n-1)} (p + (n-p))}{\fact{p}\;\fact{(n-p)}}  =
  \binomnp.
  \]
  \proofpar{Case $p=n$}
  Then, $n\geq1$, $\binom{n-1}{n}=0$, and from~\eqref{e:prop-binom-coef-0}, the
  equality holds.\\
  \proofpar{Case $n<p$}
  Then, $p\geq2$, all three binomials are equal to~0, and the equality holds.

  \proofparskip{\eqref{e:prop-binom-coef-4}}
  Let $p\geq1$.
  For all $n\in\matN$, let $\PropPP(n)\eqdef
  \left[\sum_{j=0}^n\binom{j+p-1}{p-1}=\binom{n+p}{p}\right]$.\\
  \proofpar{Induction: $\PropPP(0)$}
  Direct consequence of~\eqref{e:prop-binom-coef-0}.\\
  \proofpar{Induction: $\PropPP(n)\Implies\PropPP(n+1)$}
  Let $n\in\matN$, assume that $\PropPP(n)$ holds.
  Then, from
  \assume{abelian monoid properties of~$\matN$}, and
  \eqref{e:prop-binom-coef-3} (with $(n+p+1,p)\neq(0,1)$ and $p\neq0$),
  we have
  \[
    \sum_{j = 0}^{n + 1} \binom{j + p - 1}{p - 1}
    = \sum_{j = 0}^n \binom{j + p - 1}{p - 1} + \binom{n + p}{p - 1}
    = \binom{n + p}{p} + \binom{n + p}{p - 1}
    = \binom{n + 1 + p}{p}.
  \]
  This concludes the induction on~$n$, and we have for all $n\in\matN$,
  $\PropPP(n)$.

  \proofparskip{\eqref{e:prop-binom-coef-5}}
  Let $\PropPP(n)\eqdef\left[\forall p\in\matN,\;\binom{n}{p}\in\matN\right]$.
  \proofpar{Strong induction on~$n$:
    $(\forall n_1<n,\;\PropPP(n_1))\Implies\PropPP(n)$}\\
  \proofpar{Case $n=0$ and $p=1$}
  Then, from
  Definition~\thref{d:binom-coef},
  we have $\binom{n}{p}=0\in\matN$.\\
  \proofpar{Case $p=0$}
  Then, from~\eqref{e:prop-binom-coef-0}, we have
  $\binom{n}{p}=1\in\matN$.\\
  \proofpar{Case $(n\neq0\Disj p\neq1)$ and $p\neq0$}
  Let~$n_1<n$, assume that $\binom{n_1}{p}\in\matN$ for all $p\in\matN$.\\
  Then, from~\eqref{e:prop-binom-coef-3}, and
  \assume{monoid properties of $\matN$},
  we have $\binom{n}{p}=\binom{n-1}{p-1}+\binom{n-1}{p}\in\matN$.\\
  This concludes the strong induction on~$n$, and we have for all $n\in\matN$,
  $\PropPP(n)$.
\end{proof}

\begin{definition}[canonic families]
  \label{d:canon-fam}
  \mbox{}\hfill
  Let~$d\geq1$.\\
  The notation~$\zzero$, resp.~$\oone$, represents the constant family of value~0,
  resp.~1, either in~$\matNd$ or in~$\matRd$.\\
  For all $i\in[1..d]$, $\ee_i$ denotes the element of the canonical ``basis''
  family in $\matNd$ or in $\matRd$, such that $(\ee_i)_j=\kron{i}{j}$ for all
  $j\in[1..d]$.\\
  Let~$n\in\matN$.
  The notation $\famvert{\vv}{n,d}\eqdef(\vv_0,\ldots,\vv_n)$ represents a
  family of $n+1$ points in~$\matRd$.
  When $n\eqdef d$, it is simply written $\famvert{\vv}{d}$.
\end{definition}

\begin{lemma}[circular permutation]
  \label{l:circ-permut}
  \mbox{}\\
  Let~$d\in\matN$, and $i\in[0..d]$.
  Let~$\permcd{i}$ be the function defined by
  \begin{equation}
    \label{e:circ-permut}
    \forall j \in [0..d],\qquad
    \permcd{i} (j) \eqdef
    \left\{
      \begin{array}{ll}
        j + i + 1 & \mbox{when } j < d - i,\\
        j - (d - i) & \mbox{otherwise}.
      \end{array}
    \right.
  \end{equation}
  Then, $\permcd{i}$ is bijective from~$[0..d]$ onto~$[0..d]$, and we have
  $\permcd{i}(d)=i$ and $d=\permcdinv{i}(i)$.

  In particular, we have $\permcd{d}=\identity{[0..d]}$.
\end{lemma}

\begin{proof}
  From
  \assume{monoid properties of~$\matN$},
  \assume{the fact that monotony on disjoint parts implies injectivity}, and
  \assume{the fact that injectivity and cardinal equality imply bijectivity},
  the permutation~$\permcd{i}$ is increasing from~$[0..d-i-1]$ to~$[i+1..d]$
  and from~$[d-i..d]$ to~$[0..i]$,
  both $[0..d-i-1]\cap[d-i..d]$ and $[i+1..d]\cap[0..i]$ are empty, and
  $[0..d]=[0..d-i-1]\cup[d-i..d]$ equals~$[i+1..d]\cup[0..i]$,
  thus it is bijective from~$[0..d]$ onto itself.

  If~$j=d$, then $j\geq d-i$, and we have $\permcd{i}(d)=d-(d-i)=i$.

  If $i=d$, the set $\{j<d-d\}$ is empty, and we have
  $\permcd{d}(j)=j-(d-d)=j$ for all $j\in[0..d]$.
\end{proof}

\begin{lemma}[transposition]
  \label{l:trsp}
  \mbox{}\hfill
  Let~$d\in\matN$, and $i\in[0..d]$.
  Let~$\permtrspd{i}$ be the function defined by
  \begin{equation}
    \label{e:trsp}
    \forall j \in [0..d],\qquad
    \permtrspd{i} (j) \eqdef
    \left\{
      \begin{array}{ll}
        d & \mbox{when } j \eqdef i,\\
        i & \mbox{when } j \eqdef d,\\
        j & \mbox{otherwise}.
      \end{array}
    \right.
  \end{equation}
  Then, $\permtrspd{i}$ is involutive and bijective from~$[0..d]$
  onto~$[0..d]$.

  In particular, we have $\permtrspd{d}=\identity{[0..d]}$.
\end{lemma}

\begin{proof}
  From
  \assume{monoid properties of~$\matN$},
  we have $\permtrspd{i}([0..d])\subset[0..d]$, and
  $\permtrspd{i}\circ\permtrspd{i}=\identity{[0..d]}$,
  thus it is involutive and bijective from~$[0..d]$ onto itself.

  If $i=d$, we have $\permtrspd{d}(d)=d$, and, for all $j\in[0..d-1]$,
  $\permtrspd{d}(j)=j$, hence the result.
\end{proof}

\begin{lemma}[jump enumeration]
  \label{l:jump-enum}
  \mbox{}\\
  Let~$d\in\matN$, and $i\in[0..d+1]$.
  Let~$\thetinjd{i}$ be the function defined by
  \begin{equation}
    \label{e:jump-enum}
    \forall j \in [0..d],\qquad
    \thetinjd{i} (j) \eqdef
    \left\{
      \begin{array}{ll}
        j & \mbox{when } j < i,\\
        j + 1 & \mbox{otherwise}.
      \end{array}
    \right.
  \end{equation}
  Then, $\thetinjd{i}$ is injective from~$[0..d]$ to~$[0..d+1]$, and we have
  $\thetinjd{i}([0..d])=[0..d+1]\setminus\{i\}$.

  In particular, we have $\thetinjd{0}=\Identity_{[0..d]}+1$ and
  $\thetinjd{d+1}=\Identity_{[0..d]}$.
\end{lemma}

\begin{proof}
  \proofpar{Injectivity}
  Let $j,k\in[0..d]$.
  Assume that $r\eqdef\thetinjd{i}(j)=\thetinjd{i}(k)$.\\
  \proofpar{Case $i\leq j,k$}
  Then $r=j+1=k+1$, and thus, from
  \assume{monoid properties of~$\matN$}, $j=k$.\\
  \proofpar{Case $j<i\leq k$}
  Then $r=j=k+1$, which contradicts $j<k$.\\
  \proofpar{Case $k<i\leq j$}
  Then $r=j+1=k$, which contradicts $k<j$.\\
  \proofpar{Case $j,k<i$}
  Then $r=j=k$.\\
  Thus, from
  \assume{the definition of injectivity},
  we have $j=k$ in all cases, and $\thetinjd{i}$ is injective.

  \proofparskip{Image}
  Direct consequence of
  \assume{monotony of image}, and
  \assume{monoid properties of~$\matN$},
  with $[0..d]=[0..i-1]\uplus[i..d]$, $\thetinjd{i}([0..i-1])=[0..i-1]$,
  $\thetinjd{i}([i..d])=[i+1..d+1]$, and
  $[0..i-1]\cup[i+1..d+1]$ equals $[0..d+1]\setminus\{i\}$.

  \proofparskip{Cases $i=0$ and $i=d$}
  Direct consequence of
  \assume{monoid properties of~$\matN$}.
\end{proof}

\section{Complements on monoids}
\label{s:compl-mon}

\begin{lemma}[image of ker is included in ker]
  \label{l:im-ker-incl-ker}
  \mbox{}\hfill
  Let $E,F$ be sets, and $G$ be a monoid.\\
  Let $f:\ArEF$ and $g:\Arrow{F}{G}$ be functions.
  Then, we have
  \begin{equation}
    \label{e:im-ker-incl-ker}
    f (\Ker{g \circ f}) \subset \Ker{g}.
  \end{equation}
\end{lemma}

\begin{proof}
  Let $x\in\Ker{g\circ f}$.
  Then, from
  Definition~\threfc{LM-d:kernel}{%
    extended to functions taking values in a monoid},
  we have $g(f(x))=0$, and thus, $f(x)\in\Ker{g}$.
\end{proof}

\begin{lemma}[image of ker is ker]
  \label{l:im-ker-eq-ker}
  \mbox{}\hfill
  Let~$E,F$ be sets, and~$G$ be a monoid.\\
  Let~$f:\ArEF$ and~$g:\Arrow{F}{G}$ be functions.
  Assume that~$f$ is surjective.
  Then, we have
  \begin{equation}
    \label{e:im-ker-eq-ker}
    f (\Ker{g \circ f}) = \Ker{g}.
  \end{equation}
\end{lemma}

\begin{proof}
  From
  Lemma~\thref{l:im-ker-incl-ker},
  we have $f(\Ker{g\circ f})\subset\Ker{g}$.

  Reciprocally, let $y\in\Ker{g}$.
  Then, from
  Definition~\threfc{LM-d:kernel}{$g(y)=0$}, and
  \assume{the definition of surjectivity (let $x$ such that $y=f(x)$)},
  we have $g(f(x))=g(y)=0$, and $x\in\Ker{g\circ f}$.
  Thus, $y=f(x)$ belongs to $f(\Ker{g\circ f})$.
\end{proof}

\section{Complements on linear algebra}
\label{s:compl-lin-alg}

\begin{remark}
  \mbox{}\\
  In this section, {\vectorspace}s are supposed to be defined over some scalar
  field~$\matK$, such as~$\matR$ or~$\matC$. 
\end{remark}

\begin{lemma}[vector subspace is invariant by translation]
  \label{l:sub-sp-inv-transl}
  \mbox{}\hfill
  Let~$E$ be a {\vectorspace}.
  Let~$\Ep$ be a vector subspace of~$E$.
  Let~$\uupo\in\Ep$.
  Then, we have $\Ep+\uupo\eqdef\{\uup+\uupo\st\uup\in\Ep\}=\Ep$.
\end{lemma}

\begin{proof}
  Direct consequence of
  Lemma~\thref{LM-l:closed-under-vector-operations-is-subspace}.
\end{proof}

\begin{lemma}[range of linear map is vector subspace]
  \label{l:rg-lin-map-is-sub-sp}
  \mbox{}\\
  Let~$E$ and~$F$ be {\vectorspace}s.
  Let~$f\in\LEF$.
  Then, the range~$f(E)$ is a vector subspace of~$F$.
\end{lemma}

\begin{proof}
  Direct consequence of
  Definition~\thref{LM-d:linear-map},
  \assume{the definition of the range of a function}, and
  Lemma~\thref{LM-l:closed-under-vector-operations-is-subspace}.
\end{proof}

\begin{lemma}[injectivity or surjectivity and dimension implies bijectivity]
  \label{l:inj-or-surj-and-dim-implies-bij}
  \mbox{}\\
  Let~$E$ and~$F$ be {\vectorspace}s.
  Let~$f\in\LEF$ be a linear map from~$E$ to~$F$.
  Then, we have
  \begin{align}
    \label{e:inj-or-surj-and-dim-implies-bij-1}
    \left( \infty > \dim E \geq \dim F \CONJ f \mbox{ injective} \right)
    & \IMPLIES f \mbox{ bijective}, \\
    \label{e:inj-or-surj-and-dim-implies-bij-2}
    (\dim E \leq \dim F < \infty \CONJ f \mbox{ surjective})
    & \IMPLIES f \mbox{ bijective}.
\end{align}
\end{lemma}

\begin{proof}
  Direct consequences of
  \assume{the rank--nullity theorem},
  \assume{the characterization of surjectivity with the rank},
  Lemma~\thref{LM-l:injective-linear-map-has-zero-kernel}, and
  \assume{$\dim\{\zzero\}=0$}.
\end{proof}

\begin{lemma}[inverse of isomorphism is linear map]
  \label{l:inverse-of-isomorphism-is-linear-map}
  \mbox{}\hfill
  Let~$E$ and~$F$ be {\vectorspace}s.
  Let $f\in\LEF$ be a linear map from~$E$ to~$F$.
  Assume that~$f$ is an isomorphism from~$E$ onto~$F$.\\
  Then, $f^{-1}$~is a linear map from~$F$ to~$E$.
\end{lemma}

\begin{proof}
  From
  Definition~\threfc{LM-d:isomorphism}{$f$~is bijective},
  there exists a mapping $f^{-1}$ from~$F$ to~$E$.
  Let $\lambda,\lambdap\in\matK$ be scalars.
  Let $v,\vp\in F$ be vectors.
  Let $u,\up\in E$ be the preimages of~$v$ and~$\vp$ ($f(u)=v$ and
  $f(\up)=\vp$).
  Then, from
  Definition~\threfc{LM-d:isomorphism}{$f$ is linear}, and
  Lemma~\thref{LM-l:linear-map-preserves-linear-combinations},
  we have
  \begin{equation*}
    f (\lambda u + \lambdap \up)
    = \lambda f (u) + \lambdap f (\up)
    = \lambda v + \lambdap \vp.
  \end{equation*}
  Thus, $\lambda u+\lambdap\up$ is the preimage of $\lambda v+\lambdap\vp$,
  {\ie}
  \begin{equation*}
    f^{-1} (\lambda v + \lambdap \vp)
    = \lambda u + \lambdap \up
    = \lambda f^{-1} (v) + \lambdap f^{-1} (\vp).
  \end{equation*}

  Therefore, from
  Lemma~\thref{LM-l:linear-map-preserves-linear-combinations},
  $f^{-1}$~is a linear map.
\end{proof}

\begin{lemma}[free family of dim elements is basis]
  \label{l:free-family-of-dim-elements-is-basis}
  \mbox{}\\
  Let $E$ be a {\vectorspace} of dimension $n>0$.
  Then, any free family of $n$ elements is a basis of $E$.
\end{lemma}

\begin{proof}
  Direct consequence of
  \assume{the incomplete basis theorem}
  in the case where the free family is completed with an empty family.
\end{proof}

\section{Complements on affine algebra}
\label{s:compl-aff-alg}

\begin{remark}
  In this section, {\vectorspace}s are supposed to be defined over some scalar
  field~$\matK$, but some results are only stated for~$\matK\eqdef\matR$.
  Even though most of the latter can be extended to any field of characteristic
  not~2, such as~$\matC$, this is not considered here.
\end{remark}

\begin{remark}
  Note that the notion of affine space associated with a {\vectorspace} can be
  defined in an abstract way, either through a function that builds a vector
  from any two points, or equivalently, through a {\em translation} function
  that builds a point from any point and any vector.
  Barycenters correspond to linear combinations with coefficients of sum~1.
  Then, affine subspaces are those closed under barycenter, and affine
  functions (from an affine space to another) are those that preserve
  barycenters.
  Affine functions between two given affine spaces can be proved to form an
  affine space.

  Actually, any {\vectorspace} can be equipped with an affine structure, and
  vice versa.
  Moreover, in the present work, in most cases, we encounter affine structures
  associated to the {\vectorspace}s~$\matRn$, and the abstract vision is not
  mandatory.
  Thus, we choose to only define affine subspaces of a {\vectorspace}, and
  affine maps between two affine subspaces.
  Elements are called {\em vectors} to emphasize their linear properties, or
  {\em points} to emphasize their affine properties.
\end{remark}

\subsection{Affine subspaces}
\label{ss:aff-subspaces}

\begin{definition}[affine subspace]
  \label{d:aff-sub-sp}
  \mbox{}\\
  Let~$E$ be a {\vectorspace}.
  A subset~$\calEp$ of~$E$ is said to be an {\em affine subspace of~$E$} iff
  there exists a vector subspace~$\Ep$ of~$E$ and $\xxo\in E$ such that
  $\calEp=\xxo+\Ep\eqdef\{\xxo+\uup\st\uup\in\Ep\}$.

  Then, $\Ep$ is called the {\em direction} of the affine subspace~$\calEp$,
  and~$\xxo$ its {\em origin}.
  Elements of affine subspaces are called {\em points}.
\end{definition}

\begin{lemma}[origin is in affine subspace]
  \label{l:orig-is-in-aff-sub-sp}
  \mbox{}\\
  Let~$E$ be a {\vectorspace}.
  Let~$\Ep$ be a vector subspace of~$E$, and~$\xxo\in E$.
  Then, we have $\xxo\in\calEp$.
\end{lemma}

\begin{proof}
  Direct consequence of
  Definition~\thref{d:aff-sub-sp}, and
  Lemma~\threfc{LM-l:closed-under-vector-operations-is-subspace}{%
    $0\in\Ep$}.
\end{proof}

\begin{lemma}[equivalent definition of affine subspace]
  \label{l:equiv-def-aff-sub-sp}
  \mbox{}\\
  Let~$E$ be a {\vectorspace}.
  Let~$\calEp\subset E$.
  Let~$\xxpo\in\calEp$.\\
  Then, $\calEp$ is an affine subspace of~$E$ iff
  $\calEp-\xxpo\eqdef\{\xxp-\xxpo\st\xxp\in\calEp\}$ is a vector subspace
  of~$E$.

  Moreover, let~$\Ep$ be a vector subspace of~$E$, and let $\xxo\in E$,
  then we have
  \begin{equation}
    \label{e:equiv-def-aff-sub-sp}
    \forall \xx \in E,\quad
    \xx \in \xxo + \Ep \EQUIV \xx - \xxo \in \Ep.
  \end{equation}
\end{lemma}

\begin{proof}
  \proofpar{First equivalence}
  Direct consequence of
  Definition~\threfc{d:aff-sub-sp}{%
    thus $\calEp$ equals $\xxo+\Ep$ with $\Ep\subset E$ subspace
    and $\xxo\in E$, $\xxpo=\xxo+\uupo$ and $\xxp=\xxo+\uup$ for all
    $\xxp\in\calEp$ with $\uupo,\uup\in\Ep$},
  Lemma~\threfc{l:sub-sp-inv-transl}{with $-\uupo$}, and
  Definition~\threfc{LM-d:space}{%
    additive group properties,
    thus $\xxp-\xxpo=\uup-\uupo$, and $\calEp-\xxpo=\Ep-\uupo=\Ep$}.

  \proofparskip{\eqref{e:equiv-def-aff-sub-sp}}
  Direct consequence of
  Definition~\thref{d:aff-sub-sp},
  Lemma~\thref{LM-l:closed-under-vector-operations-is-subspace}.
\end{proof}

\begin{lemma}[affine subspace is invariant by change of origin]
  \label{l:aff-sub-sp-inv-chg-orig}
  \mbox{}\\
  Let~$E$ be a {\vectorspace}.
  Let~$\Ep$ be a vector subspace of~$E$ and~$\xxo\in E$.
  Let~$\yyp\in\xxo+\Ep$.\\
  Then, we have $\xxo+\Ep=\yyp+\Ep$.
\end{lemma}

\begin{proof}
  Direct consequence of
  Definition~\threfc{d:aff-sub-sp}{thus $\yyp=\xxo+\uupo$ with $\uupo\in\Ep$},
  Lemma~\thref{LM-l:closed-under-vector-operations-is-subspace},
  Lemma~\threfc{l:sub-sp-inv-transl}{with $-\uupo$ thus $\Ep-\uupo=\Ep$}, and
  Lemma~\threfc{l:equiv-def-aff-sub-sp}{%
    thus $\xx\in\yyp+\Ep$ iff
    $(\xx-\xxo)-\uupo\in\Ep=\Ep-\uupo$ iff $\xx-\xxo\in\Ep$
    iff $\xx\in\xxo+\Ep$}.
\end{proof}

\begin{lemma}[vector subspace is affine subspace]
  \label{l:vect_sub-sp-is-aff-sub-sp}
  \mbox{}\\
  Let~$E$ be a {\vectorspace}.
  Let~$\Ep$ be a vector subspace of~$E$.\\
  Then, $\Ep$ is an affine subspace of~$E$, of direction itself and of origin
  any of its points.
\end{lemma}

\begin{proof}
  Direct consequence of
  Lemma~\thref{l:sub-sp-inv-transl}, and
  Lemma~\thref{l:equiv-def-aff-sub-sp}.
\end{proof}

\begin{remark}
  For instance, note that~$E$ is also an affine subspace of itself, {\eg} with
  origin~$\zzero$.
\end{remark}

\begin{lemma}[affine subspace plus vector subspace is affine subspace]
  \label{l:aff-plus-vect-is-aff-sub-sp}
  \mbox{}\hfill
  Let~$E$ be a {\vectorspace}.
  Let~$\Ep$ be a vector subspace of~$E$, let~$\xxo\in E$, and
  let~$\calEp\eqdef\xxo+\Ep$.\\
  Let~$\Fp$ be a vector subspace of~$\Ep$.
  Then, we have
  $\calEp+\Fp\eqdef\{\xxp+\vvp\st\xxp\in\calEp\Conj\vvp\in\Fp\}=\calEp$.
\end{lemma}

\begin{proof}
  \proofpar{$\calEp+\Fp\subset\calEp$}
  Direct consequence of
  Definition~\threfc{LM-d:space}{%
    additive abelian group properties,
    thus $\xxp+\vvp=\xxo+((\xxp-\xxo)+\vvp)$},
  Lemma~\threfc{l:equiv-def-aff-sub-sp}{\eqref{e:equiv-def-aff-sub-sp},
    thus $\xxp-\xxo\in\Ep$}, and
  Lemma~\threfc{LM-l:closed-under-vector-operations-is-subspace}{%
    thus $(\xxp-\xxo)\!+\!\vvp\in\Ep$}.

  \proofparskip{$\calEp+\Fp\supset\calEp$}
  Direct consequence of
  Definition~\thref{d:aff-sub-sp}, and
  Lemma~\threfc{LM-l:closed-under-vector-operations-is-subspace}{%
    thus $0\in\Fp$}.
\end{proof}

\begin{lemma}[closed under barycenter is affine subspace]
  \label{l:closed-under-baryc-is-aff-sub-sp}
  \mbox{}\hfill
  Let~$E$ be a real {\vectorspace}.\\
  Let~$\calEp\subset E$ be nonempty.
  Then, $\calEp$ is an affine subspace of~$E$ iff it is closed under
  barycenter,
  \begin{equation}
    \label{e:aff-sub-sp-is-closed-under-baryc}
      \forall n \in \matN,\;
    \forall (\vv_i)_{i \in [0..n]} \in \calEp,\;
    \forall (\mu_i)_{i \in [0..n]} \in \matR,\quad
    \sum_{i=0}^n \mu_i = 1 \IMPLIES
    \sum_{i=0}^n \mu_i \vv_i \in \calEp.
  \end{equation}
\end{lemma}

\begin{proof}
  \proofpar{From left to right}
  Direct consequence of
  Definition~\threfc{d:aff-sub-sp}{thus $\calEp$ is of the form $\xxo+\Ep$},
  Lemma~\threfc{l:equiv-def-aff-sub-sp}{\eqref{e:equiv-def-aff-sub-sp}},
  Definition~\threfc{LM-d:space}{%
    thus $\sum_{i=0}^n\mu_i\vv_i-\xxo=\sum_{i=0}^n\mu_i(\vv_j-\xxo)$}, and
  Lemma~\threfc{LM-l:closed-under-linear-combination-is-subspace}{%
    with $\vv_i-\xxo\in\Ep$}.

  \proofparskip{From right to left}
  Let~$\xxpo\in\calEp$.
  Assume that $\calEp$ is closed under barycenter.\\
  \proofpar{1. $0\in\calEp-\xxpo$} Trivial.\\
  \proofpar{2. $\calEp-\xxpo$ is closed under scalar multiplication}\\
  Let~$a\in\matR$ and $\uu\in\calEp-\xxpo$, {\ie} $\xxpo+\uu\in\calEp$.
  Then, from
  Definition~\thref{LM-d:space}, and
  hypothesis (with $(1-a)+a=1$),
  we have $\xxpo+a\uu=(1-a)\xxpo+a(\xxpo+\uu)\in\calEp$.
  Thus, $a\uu\in\calEp-\xxpo$.\\
  \proofpar{3. $\calEp-\xxpo$ is closed under addition}\\
  Let~$\uu,\vv\in\calEp-\xxpo$, {\ie} $\xxpo+\uu,\xxpo+\vv\in\calEp$.
  Then, from
  Definition~\thref{LM-d:space}, and
  hypothesis (with $\half+\half=1$),
  we have $\xxpo+\half(\uu+\vv)=\half(\xxpo+\uu)+\half(\xxpo+\vv)\in\calEp$.
  Thus, $\half(\uu+\vv)\in\calEp-\xxpo$, and from~2 (with $a\eqdef2$),
  we have $u+v=2(\half(\uu+\vv))\in\calEp-\xxpo$.\\
  Finally, from
  Lemma~\thref{LM-l:closed-under-vector-operations-is-subspace}, and
  Lemma~\thref{l:equiv-def-aff-sub-sp},
  $\calEp-\xxpo$ is a vector subspace, and $\calEp$ is an affine subspace.
\end{proof}

\begin{lemma}[barycenter closure is affine subspace]
  \label{l:baryc-closure-is-aff-sub-sp}
  \mbox{}\\
  Let~$d\geq1$.
  Let~$n\in\matN$.
  Let~$\famvert{\vv}{n,d}$ be $n+1$ points in~$\matRd$.
  Let~$\calF\eqdef\vv_0+\Span{\vv_j-\vv_0}_{j\in[1..n]}$.\\
  Then, $\calF$ is an affine subspace of~$\matRd$, and we have
  \begin{align}
    \label{e:baryc-closure-is-aff-sub-sp-1}
    \forall i \in [0..n],\quad
    \calF &= \vv_i + \Span{\vv_j - \vv_i}_{j \in [0..n] \setminus \{ i  \}}, \\
    \label{e:baryc-closure-is-aff-sub-sp-2}
    \calF &= \left\{ \sum_{j = 0}^n \mu_j \vv_j \in \matRd \rightst \left.
      (\mu_j)_{j \in [0..n]} \in \matR \CONJ \sum_{j = 0}^n \mu_j = 1 \right\},
  \end{align}
  with the convention that $\Span{\emptyset}\eqdef\{\zzero\}$, thus
  $\calF=\{\vv_0\}$ when $n\eqdef0$.
\end{lemma}

\begin{proof}
  \proofpar{Case $n=0$}
  Direct consequence of
  Definition~\threfc{d:aff-sub-sp}{with $\Ep=\{0\}$}, and
  Lemma~\thref{LM-l:trivial-subspaces}.

  \proofparskip{Case $n\geq1$}
  From
  Definition~\thref{d:aff-sub-sp},
  $\calF$ is an affine subspace.\\
  Let~$i\in[0..n]$.
  Let $\calF_i\eqdef\vv_i+\Span{\vv_j-\vv_i}_{j\in[0..n]\setminus\{i\}}$ and
  \[
    \calH \eqdef
    \left\{ \sum_{j = 0}^n \mu_j \vv_j \in \matRd \rightst \left.
      \forall j \in [0..n],\; \mu_j \in \matR \CONJ \sum_{j = 0}^n \mu_j = 1
    \right\}.
  \]
  Then, from
  Definition~\thref{LM-d:finite-dimensional-subspace}, and
  \assume{ring properties of~$\matR$},
  we have
  \begin{align*}
    \xx \in \calF_i
    &\EQUIV
      \exists (\lambda_j)_{j \in [0..n] \setminus \{ i \}} \in \matRn,\quad
      \xx = \vv_i + \sum_{j = 0, j \neq i}^n \lambda_j (\vv_j - \vv_i)\\
    &\EQUIV
      \exists (\lambda_j)_{j \in [0..n] \setminus \{ i \}} \in \matRn,\quad
      \xx = \left( 1 - \sum_{j = 0, j \neq i}^n \lambda_j \right) \vv_i +
      \sum_{j = 0, j \neq i}^n \lambda_j \vv_j\\
    &\EQUIV
      \exists (\mu_j)_{j \in [0..n]} \in \matRnpi,\quad
      \xx = \sum_{j = 0}^n \mu_j \vv_j \in \matRd
      \CONJ \sum_{j = 0}^n \mu_j = 1\\
    &\EQUIV \xx \in \calH,
  \end{align*}
  where $\mu_j\eqdef\lambda_j$ for all $j\in[0..n]\setminus\{i\}$, and
  $\mu_i\eqdef(1-\sum_{j=0,j\neq i}^n\lambda_j)$, so that $\sum_{j=0}^n\mu_j=1$
  (from left to right), and $\lambda_j\eqdef\mu_j$ for all
  $j\in[0..n]\setminus\{i\}$ (from right to left).\\
  Therefore, we have the two equalities.
\end{proof}

\subsection{Affine maps and submaps}
\label{ss:aff-maps-submaps}

\begin{definition}[affine map]
  \label{d:aff-map}
  \mbox{}\hfill
  Let~$E$ and~$F$ be {\vectorspace}s.\\
  Let~$\Ep$ be a vector subspace of~$E$ and~$\xxo\in E$.
  Let~$\calEp\eqdef\xxo+\Ep$.
  A function $f:\ArcalEpF$ is said to be an {\em affine map} iff
  there exists $\phi\in\LSp{\Ep}{F}$ and $\cc\in F$ such that
  \begin{equation}
    \label{e:aff-map}
    \forall \uup \in \Ep,\quad
    f (\xxo + \uup) = \cc + \phi (\uup),
  \end{equation}
  also denoted $f=\cc+\phi$ with origin~$\xxo$ (and we may omit the origin when
  there is no ambiguity).
\end{definition}

\begin{remark}
  For instance, note that a function $f:\ArEF$ ({\ie} with
  $\calEp\eqdef\Ep\eqdef E$, and origin~$\zzero$) is an affine mapping iff
  there exists~$\phi\in\LEF$ and~$\cc\in F$ such that~$f=\cc+\phi$, {\ie}
  for all vector $\uu\in E$, $f(\zzero+\uu)=\cc+\phi(\uu)$, or equivalently,
  for all point $\xx\in E$, $f(\xx)=\cc+\phi(\xx-\zzero)$.
\end{remark}

\begin{lemma}[equivalent definition of affine map]
  \label{l:equiv-def-aff-map}
  \mbox{}\hfill
  Let~$E$ and~$F$ be {\vectorspace}s.\\
  Let~$\Ep$ be a vector subspace of~$E$ and~$\xxo\in E$.
  Let~$\calEp\eqdef\xxo+\Ep$.
  Let~$f:\ArcalEpF$.\\
  Then, $f$ is an affine map iff
  there exists $\phi\in\LSp{\Ep}{F}$ and $\cc\in F$ such that
  \begin{equation}
    \label{e:equiv-def-aff-map}
    \forall \xxp \in \calEp,\quad
    f (\xxp) = \cc + \phi (\xxp - \xxo).
  \end{equation}
\end{lemma}

\begin{proof}
  Direct consequence of
  Definition~\thref{d:aff-map},
  Definition~\thref{d:aff-sub-sp}, and
  Lemma~\threfc{l:equiv-def-aff-sub-sp}{\eqref{e:equiv-def-aff-sub-sp},
    thus $\xxp-\xxo\in\Ep$}.
\end{proof}

\begin{lemma}[change of origin in affine map]
  \label{l:change-orig-aff-map}
  \mbox{}\hfill
  Let~$E$ and~$F$ be {\vectorspace}s.
  Let~$\Ep$ be a vector subspace of~$E$ and~$\xxo\in E$.
  Let~$\calEp\eqdef\xxo+\Ep$.
  Let~$\phi\in\LSp{\Ep}{F}$ and~$\cc\in F$.
  Let~$f\eqdef\cc+\phi$ be an affine map (with origin~$\xxo$).
  Let~$\yypo\in\calEp$.
  Then, we have
  \begin{equation}
    \label{e:change-orig-aff-map}
    \forall \uup \in \Ep,\quad
    f (\yypo + \uup) = (\cc + \phi(\yypo - \xxo)) + \phi(\uup).
  \end{equation}
  Thus, setting $\ccp\eqdef\cc+\phi(\yypo-\xxo)$, we have $f=\ccp+\phi$ with
  origin~$\yypo$.
\end{lemma}

\begin{proof}
  Direct consequence of
  Definition~\thref{d:aff-map},
  Definition~\thref{d:aff-sub-sp},
  Lemma~\threfc{l:equiv-def-aff-sub-sp}{%
    \eqref{e:equiv-def-aff-sub-sp}, thus $\yypo-\xxo\in\Ep$},
  Lemma~\threfc{LM-l:closed-under-vector-operations-is-subspace}{%
    thus $\yypo-\xxo+\uup\in\Ep$},
  Lemma~\thref{LM-l:linear-map-preserves-linear-combinations}, and
  Definition~\threfc{LM-d:space}{additive abelian group properties}.
\end{proof}

\begin{remark}
  One can thus change the origin of an affine mapping defined over $E$.
  Indeed, let an affine mapping $f:\ArEF$ ({\ie} with
  $\calEp\eqdef\Ep\eqdef E$, and origin~$\zzero$)
  such that~$f=\cc+\phi$. Thus, we have for all point $\xxo\in E$ and for all
  vector $\uu\in E$, $f(\xxo+\uu)=(\cc+\phi(\xxo-\zzero))+\phi(\uu)$.
\end{remark}

\begin{lemma}[range of affine map is affine subspace]
  \label{l:rg-aff-map-is-aff-sub-sp}
  \mbox{}\\
  Let~$E$ and~$F$ be {\vectorspace}s.
  Let~$\Ep$ be a vector subspace of~$E$ and~$\xxo\in E$.
  Let~$\calEp\eqdef\xxo+\Ep$.\\
  Let~$\phi\in\LEpF$ and~$\cc\in F$.
  Let~$f\eqdef\cc+\phi:\ArcalEpF$ be an affine map.\\
  Then, $f(\calEp)=\cc+\phi(\Ep)$ is an affine subspace of~$F$.
\end{lemma}

\begin{proof}
  Direct consequence of
  Definition~\thref{d:aff-map},
  Lemma~\thref{l:rg-lin-map-is-sub-sp}, and
  Definition~\thref{d:aff-sub-sp}.
\end{proof}

\begin{remark}
  Thus, affine maps can be considered as functions from an affine
  subspace~$\calEp$ of~$E$ to another affine subspace~$\calFp$ of~$F$.
\end{remark}

\begin{definition}[set of affine maps]
  \label{d:set-aff-maps}
  \mbox{}\\
  Let~$E$ and~$F$ be {\vectorspace}s.
  Let~$\Ep$ (resp.~$\Fp$) be a vector subspace of~$E$ (resp.~$F$).\\
  Let~$\calEp$ (resp.~$\calFp$) be an affine subspace of~$E$ (resp.~$F$) of
  direction~$\Ep$ (resp.~$\Fp$).\\
  The {\em set of affine maps from~$\calEp$ to~$\calFp$} is
  $\AffcalEpcalFp\eqdef\{f:\ArcalEpcalFp\st f\mbox{ is affine map}\}$.
\end{definition}

\begin{lemma}[space of affine maps]
  \label{l:space-aff-maps}
  \mbox{}\\
  Let~$E$ and~$F$ be {\vectorspace}s.
  Then, $\AffEF$ is a vector subspace of $\FEF$.
\end{lemma}

\begin{proof}
  Direct consequence of
  Definition~\threfc{d:set-aff-maps}{%
    with $\calEp\eqdef\Ep\eqdef E$ and\linebreak $\calFp\eqdef\Fp\eqdef F$},
  Lemma~\threfc{LM-l:space-of-functions-to-space}{$\AffEF\subset\FEF$},
  Lemma~\threfc{LM-l:closed-under-linear-combination-is-subspace}{
    $\lambda f+\lambdap\fp=
    (\lambda \cco+\lambda\ccpo)+(\lambda\phi+\lambdap\phip)$}, and
  Lemma~\thref{LM-l:space-of-linear-maps}.
\end{proof}

\begin{lemma}[output restriction of affine map]
  \label{l:out-restr-aff-map}
  \mbox{}\hfill
  Let~$E$ and~$F$ be {\vectorspace}s.\\
  Let~$\Ep$ (resp.~$\Fp$) be a vector subspace of~$E$ (resp.~$F$).
  Let~$\xxo\in E$ and~$\yyo\in F$.\\
  Let~$\calEp\eqdef\xxo+\Ep$ and~$\calFp\eqdef\yyo+\Fp$.
  Let~$\phi\in\LEpF$ and~$\cc\in F$.
  Let~$f\eqdef\cc+\phi\in\AffcalEpF$.\\
  Assume that~$\phi(\Ep)\subset\Fp$ and~$\cc\in\calFp$.
  Then, we have $f(\calEp)\subset\calFp$ and $f\in\AffcalEpcalFp$.
\end{lemma}

\begin{proof}
  Direct consequence of
  Lemma~\thref{l:rg-aff-map-is-aff-sub-sp},
  Lemma~\threfc{l:aff-sub-sp-inv-chg-orig}{%
    thus $f(\calEp)=\cc+\phi(\Ep)\subset\cc+\Fp=\calFp$}, and
  Definition~\thref{d:set-aff-maps}.
\end{proof}

\begin{lemma}[affine submap]
  \label{l:aff-sub-map}
  \mbox{}\hfill
  Let~$E$ and~$F$ be {\vectorspace}s.\\
  Let~$\phi\in\LEF$ and~$\cc\in F$.
  Let~$f\eqdef\cc+\phi\in\AffEF$ with origin~$\zzero$.\\
  Let~$\Ep$ (resp.~$\Fp$) be a vector subspace of~$E$ (resp.~$F$).
  Let~$\xxo\in E$ and~$\yyo\in F$.
  Let~$\calEp\eqdef\xxo+\Ep$ and~$\calFp\eqdef\yyo+\Fp$.
  Let~$\phip\eqdef\phi_{|\Ep}$ and~$\fp\eqdef f_{|\calEp}$.
  Let $\ccp\eqdef\cc+\phi(\xxo-\zzero)\in F$.\\
  Then, we have $\fp=\ccp+\phip$ with origin~$\xxo$, and
  $\fp\in\AffcalEpF$.

  Moreover, if~$\phip(\Ep)\subset\Fp$ and~$\ccp\in\calFp$, then we have
  $\fp\in\AffcalEpcalFp$, and~$\fp$ is called the
  {\em affine submap of~$f$ on~$\calEp$}.
\end{lemma}

\begin{proof}
  \proofpar{$\fp=\ccp+\phip\in\AffcalEpF$}
  Let~$\uup\in\Ep$.
  Then, from
  Definition~\threfc{d:aff-sub-sp}{thus $\xxo+\uup\in\calEp$},
  Definition~\threfc{d:aff-map}{%
    with $\calEp\eqdef\Ep\eqdef E$ and origin $\zzero$}, and
   Lemma~\threfc{l:change-orig-aff-map}{%
     for $f$, with $\xxo\eqdef\zzero$ and $\yypo\eqdef\xxo$},
  we have
  \[
    \fp (\xxo + \uup) = f (\xxo + \uup)
    = (\cc + \phi (\xxo - \zzero) ) + \phi (\uup)
    = \ccp + \phip (\uup).
  \]
  Then, from
  \assume{restriction to subspace preserves linearity
    (thus $\phip$ belongs to $\LSp{\Ep}{F}$)},
  Definition~\threfc{d:aff-map}{%
    with $\fp$, $\phip$, $\ccp$, $\calEp\eqdef\xxo+\Ep$}, and
  Definition~\threfc{d:set-aff-maps}{with $\calFp\eqdef\Fp\eqdef F$},
  we have $\fp=\ccp+\phip$ (with origin $\xxo$) and $\fp\in\AffcalEpF$.

  \proofparskip{$\fp\in\AffcalEpcalFp$}\\
  Direct consequence of
  Lemma~\threfc{l:out-restr-aff-map}{with $\fp$, $\phip$ and $\ccp$}.
\end{proof}

\begin{lemma}[equivalent definition of affine map (finite dimension)]
  \label{l:equiv-def-aff-map-finite-dim}
  \mbox{}\\
  Let~$p,q\geq1$.
  Let~$f:\ArRpRq$.
  Then, we have
  \begin{equation}
    \label{e:equiv-def-aff-map-finite-dim}
    f \in \AffRpRq \EQUIV
    \exists A \in \calM_{q,p},\;
    \exists \cc \in \matRq,\quad
    \forall \xx \in \matRp,\quad
    f (\xx) = \cc + A \xx.
  \end{equation}
\end{lemma}

\begin{proof}
  Direct consequence of
  Definition~\threfc{d:aff-map}{%
    with $\calEp\eqdef\Ep\eqdef\matRp$, $F\eqdef\matRq$,
    $\xxo\!\eqdef\!\zzero$, and $\phi\eqdef(\uu\longmapsto A\uu)$}, and
  \assume{the definition of matrix--vector product}.
\end{proof}

\begin{lemma}[affine map preserves barycenter]
  \label{l:aff-map-preserves-baryc}
  \mbox{}\hfill
  Let~$E$ and~$F$ be real {\vectorspace}s.\\
  Let~$f\in\AffEF$.
  Let~$n\in\matN$.
  Let~$\famvert{\vv}{n}$ be $n\!+\!1$ points in~$E$.
  Let~$(\mu_i)_{i\in[0..n]}\in\matR$.
  Then, we have
  \begin{equation}
    \label{e:aff-map-preserves-baryc}
    \sum_{i = 0}^n \mu_j = 1 \IMPLIES
    f \left( \sum_{i = 0}^n \mu_j \vv_j \right)
    = \sum_{i = 0}^n \mu_j f (\vv_j).
  \end{equation}
\end{lemma}

\begin{proof}
  Assume that $\sum_{i=0}^n\mu_j=1$.
  Let~$\phi\in\LEF$ and~$\cc\in F$.
  Then, from
  Definition~\thref{d:aff-map},
  Lemma~\thref{LM-l:linear-map-preserves-linear-combinations}, and
  Definition~\threfc{LM-d:space}{%
    $(F,+)$ is an abelian group, distributivity},
  we have
  \begin{align*}
    f \left( \sum_{i = 0}^n \mu_j \vv_j \right)
    &= \cc + \phi \left( \sum_{i = 0}^n \mu_j \vv_j \right)
      = \left( \sum_{i = 0}^n \mu_j \right) \cc
      + \sum_{i = 0}^n \mu_j \phi (\vv_j)\\
    &= \sum_{i = 0}^n \mu_j (\cc + \phi (\vv_j))
      = \sum_{i = 0}^n \mu_j f (\vv_j).
  \end{align*}
\end{proof}

\begin{definition}[isobarycenter]
  \label{d:isobaryc}
  \mbox{}\hfill
  Let~$E$ be a real {\vectorspace}.
  Let~$n\in\matN$.\\
  Let $\famvert{\vv}{n}$ be $n+1$ points in~$E$.
  The {\em isobarycenter of~$\famvert{\vv}{n}$} is
  denoted~$\bisobaryc{\famvert{\vv}{n}}$, and is defined by
  \begin{equation}
    \label{e:isobaryc}
    \bisobaryc{\famvert{\vv}{n}} \eqdef \frac{1}{n + 1} \sum_{i = 0}^n \vv_i.
  \end{equation}
\end{definition}

\begin{lemma}[affine map preserves isobarycenter]
  \label{l:aff-map-preserves-isobaryc}
  \mbox{}\hfill
  Let~$E$ and~$F$ be real {\vectorspace}s.\\
  Let~$f\in\AffEF$.
  Let~$n\in\matN$.
  Let $\famvert{\vv}{n}$ be $n+1$ points in~$E$.
  Let $\famvert{f(\vv)}{n}\eqdef(f(\vv_i))_{i\in[0..n]}\in F$.
  Then, we have
  $f(\bisobaryc{\famvert{\vv}{n}})=\bisobaryc{\famvert{f(\vv)}{n}}$.
\end{lemma}

\begin{proof}
  Direct consequence of
  Definition~\thref{d:isobaryc},
  Definition~\threfc{LM-d:space}{proper\-ties of linear operations},
  \assume{field properties of~$\matR$
    (with $n+1>0$, thus $\sum_{i=0}^n\frac{1}{n+1}=1$)}, and
  Lemma~\threfc{l:aff-map-preserves-baryc}{with $\mu_i\eqdef\frac{1}{n+1}$}.
\end{proof}

\begin{lemma}[affine maps are closed by composition]
  \label{l:aff-map-are-closed-by-composition}
  \mbox{}\\
  Let~$E$, $F$, and~$G$ be {\vectorspace}s.
  Let~$\phi\in\LEF$, $\cc\in F$, $\psi\in\LFG$ and~$\dd\in G$.\\
  Let~$f\eqdef\cc+\phi\in\AffEF$ and~$g\eqdef\dd+\psi\in\AffFG$,
  both with origin~$\zzero$.\\
  Then, we have
  \begin{equation}
    \label{e:aff-map-are-closed-by-composition}
    g \circ f = \dd + \psi (\cc - \zzero) + \psi \circ \phi
    \mbox{ with origin } \zzero
    \CONJ
    g \circ f \in \AffEG.
  \end{equation}
\end{lemma}

\begin{proof}
  Direct consequence of
  \assume{the definition of composition of functions},
  Definition~\threfc{d:aff-map}{with $f$ and $\xxo\eqdef\zzero$},
  Lemma~\threfc{l:equiv-def-aff-map}{with $g$,
    $\xxo\eqdef\zzero$ and $\xxp\eqdef(\ccp-\zzero)+\phi(\uu)$ for any
    $\uu\in E$},
  Definition~\threfc{LM-d:space}{properties of linear operations},
  Definition~\thref{LM-d:linear-map},
  Lemma~\threfc{LM-l:composition-of-linear-maps-is-bilinear}{%
    thus $\psi\circ\phi\in\LEG$}, and
  Definition~\thref{d:set-aff-maps}.
\end{proof}

\begin{lemma}[continuous affine map is continuous linear map]
  \label{l:cont-aff-map-is-cont-linear-map}
  \mbox{}\\
  Let~$(E,\nEdot)$ and~$(F,\nFdot)$ be normed {\vectorspace}s.
  Let~$\phi\in\LEF$, $\xxo\in E$ and~$\cc\in F$.\\
  Let $f\eqdef\cc+\phi\in\AffEF$ with origin~$\xxo$.
  Then, $f$ is continuous iff~$\phi\in\LcEF$.

  Moreover, in this case, $f$ is $\tnEF{\phi}$-Lipchitz continuous.
\end{lemma}

\begin{proof}
  Direct consequence of
  Lemma~\thref{l:equiv-def-aff-map},
  Definition~\threfc{LM-d:space}{%
    properties of linear operations,
    thus $f(\xxp)-f(\xx)=\phi(\xxp-\xxo)-\phi(\xx-\xxo)$,
    and $\xx-\xxo=\xx-\zzero-(\xxo-\zzero)$},
  Definition~\threfc{LM-d:linear-map}{thus
    $f(\xxp)-f(\xx)=\phi(\xxp-\zzero)-\phi(\xx-\zzero)$},
  Definition~\thref{LM-d:pointwise-continuity},
  Definition~\thref{LM-d:continuity-in-a-point},
  Theorem~\thref{LM-t:continuous-linear-map}, and
  Theorem~\thref{LM-t:normed-space-of-continuous-linear-maps}.
\end{proof}


\begin{lemma}[injective affine submap is zero linear kernel]
  \label{l:inj-aff-sub-map-is-zero-linear-ker}
  \mbox{}\\
  Let~$E$ and~$F$ be {\vectorspace}s.
  Let~$\phi\in\LEF$ and~$\cc\in F$.
  Let~$f\eqdef\cc+\phi\in\AffEF$ with origin~$\zzero$.
  Let~$\Ep$ (resp.~$\Fp$) be a vector subspace of~$E$ (resp.~$F$).
  Let~$\xxo\in E$ and~$\yyo\in F$.\\
  Let~$\calEp\eqdef\xxo+\Ep$ and~$\calFp\eqdef\yyo+\Fp$.
  Let~$\phip\eqdef\phi_{|\Ep}$ and~$\fp\eqdef f_{|\calEp}$.
  Let~$\ccp\eqdef\cc+\phi(\xxo-\zzero)\in F$.\\
  Let~$\fp=\ccp+\phip$ with origin~$\xxo$
  be the affine submap of~$f$ on~$\calEp$.\\
  Then, $\fp$ is injective iff $\Ker{\phip}=\{0\}$.
\end{lemma}

\begin{proof}
  Direct consequence of
  Lemma~\thref{l:aff-sub-map},
  Definition~\threfc{LM-d:space}{%
    properties of linear operations,
    thus $\fp(\xxo+\vvp)-\fp(\xxo+\uup)=\phip(\vvp)-\phip(\uup)$},
  Definition~\threfc{LM-d:linear-map}{%
    thus $\phip(\vvp)-\phip(\uup)=\phip(\vvp-\uup)$},
  \assume{the definition of injectivity}, and
  Lemma~\thref{LM-l:injective-linear-map-has-zero-kernel}.
\end{proof}

\begin{lemma}[injective affine map is zero linear kernel]
  \label{l:inj-aff-map-is-zero-linear-ker}
  \mbox{}\hfill
  Let~$E$ and~$F$ be {\vectorspace}s.\\
  Let~$\phi\in\LEF$ and~$\cc\in F$.
  Let $f\eqdef\cc+\phi\in\AffEF$ with origin~$\zzero$.\\
  Then, $f$ is injective iff $\Ker{\phi}=\{0\}$.
\end{lemma}

\begin{proof}
  Direct consequence of
  Lemma~\threfc{l:inj-aff-sub-map-is-zero-linear-ker}{%
    with $\calEp\eqdef\Ep\eqdef E$, $\calFp\eqdef\Fp\eqdef F$ and
    $\xxo=\yyo\eqdef\zzero$, thus $\phip=\phi$, $\fp=f$ and $\ccp=\cc$}.
\end{proof}

\begin{lemma}[surjective affine submap is full linear range]
  \label{l:surj-aff-sub-map-is-full-linear-rg}
  \mbox{}\\
  Let~$E$ and~$F$ be {\vectorspace}s.
  Let~$\phi\in\LEF$ and~$\cc\in F$.
  Let~$f\eqdef\cc+\phi\in\AffEF$ with origin~$\zzero$.
  Let~$\Ep$ (resp.~$\Fp$) be a vector subspace of~$E$ (resp.~$F$).
  Let~$\xxo\in E$ and~$\yyo\in F$.\\
  Let~$\calEp\eqdef\xxo+\Ep$ and~$\calFp\eqdef\yyo+\Fp$.
  Let~$\phip\eqdef\phi_{|\Ep}$ and~$\fp\eqdef f_{|\calEp}$.
  Let~$\ccp\eqdef\cc+\phi(\xxo-\zzero)\in F$.\\
  Assume that~$\phip(\Ep)\subset\Fp$ and~$\ccp\in\calFp$.\\
  Let~$\fp=\ccp+\phip$ with origin~$\xxo$ be the affine submap of~$f$
  on~$\calEp.$\\
  Then, $\fp$ is surjective iff $\phip(\Ep)=\Fp$.
\end{lemma}

\begin{proof}
  \proofparskip{From left to right}
  Assume that~$\fp$ is surjective onto~$\calFp$.\\
  Let~$\vvp\in\Fp$.
  Then, from
  Definition~\threfc{d:aff-sub-sp}{with $\calFp$}, and
  Lemma~\threfc{l:aff-sub-sp-inv-chg-orig}{with $\yyp\eqdef\ccp\in\yyo+\Fp$},
  we have $\ccp+\vvp\in\ccp+\Fp=\calFp$.
  Thus, from
  \assume{the definition of surjectivity},
  Lemma~\thref{l:equiv-def-aff-map},
  Lemma~\thref{LM-l:closed-under-linear-combination-is-subspace}, and
  Lemma~\threfc{l:equiv-def-aff-sub-sp}{\eqref{e:equiv-def-aff-sub-sp}},
  there
  exists~$\xxp\in\calEp$ such that~$\fp(\xxp)=\ccp+\vvp$, and
  $\phip(\xxp-\xxo)=\fp(\xxp)-\ccp=\vvp$
  with~$\xxp-\xxo\in\Ep$.\\
  Therefore, we have~$\Fp\subset\phip(\Ep)$, and thus~$\phip(\Ep)=\Fp$.

  \proofparskip{From right to left}
  Assume that $\phip(\Ep)=\Fp$.
  Let~$\yyp\in\calFp$.
  Then, from
  Definition~\threfc{LM-d:space}{properties of linear operations},
  Lemma~\threfc{l:equiv-def-aff-sub-sp}{%
    \eqref{e:equiv-def-aff-sub-sp} with $\yyp$ and $\ccp$},
  Lemma~\thref{LM-l:closed-under-linear-combination-is-subspace}, and
  Definition~\thref{d:aff-map}, and
  Definition~\threfc{d:aff-sub-sp}{with $\calEp$},
  we have $\yyp-\ccp=(\yyp-\yyo)-(\ccp-\yyo)\in\Fp$, thus there
  exists~$\uup\in\Ep$ such that~$\phip(\uup)=\yyp-\ccp$, and
  $\fp(\xxo+\uup)=\ccp+\phip(\uup)=\yyp$ with $\xxo+\uup\in\calEp$.\\
  Therefore, from
  \assume{the definition of surjectivity},
  $\fp$ is surjective.
\end{proof}

\begin{lemma}[surjective affine map is full linear range]
  \label{l:surj-aff-map-is-full-linear-rg}
  \mbox{}\hfill
  Let~$E$ and~$F$ be {\vectorspace}s.\\
  Let~$\phi\in\LEF$ and~$\cc\in F$.
  Let $f\eqdef\cc+\phi\in\AffEF$ with origin~$\zzero$.\\
  Then, $f$ is surjective iff $\phi(E)=F$.
\end{lemma}

\begin{proof}
  Direct consequence of
  Lemma~\threfc{l:surj-aff-sub-map-is-full-linear-rg}{%
    with $\calEp\eqdef\Ep\eqdef E$, $\calFp\eqdef\Fp\eqdef F$ and
    $\xxo=\yyo\eqdef\zzero$, thus $\phip=\phi$, $\fp=f$, and
    $\ccp=\cc\in F$, and $\phi(E)\subset F$}.
\end{proof}

\begin{lemma}[inverse of affine submap is affine submap]
  \label{l:inv-of-aff-sub-map-is-aff-sub-map}
  \mbox{}\\
  Let~$E$ and~$F$ be {\vectorspace}s.
  Let~$\phi\in\LEF$ and~$\cc\in F$.
  Let~$f\eqdef\cc+\phi\in\AffEF$ with origin~$\zzero$.
  Let~$\Ep$ (resp.~$\Fp$) be a vector subspace of~$E$ (resp.~$F$).
  Let~$\xxo\in E$ and~$\yyo\in F$.\\
  Let~$\calEp\eqdef\xxo+\Ep$ and~$\calFp\eqdef\yyo+\Fp$.
  Let~$\phip\eqdef\phi_{|\Ep}$ and~$\fp\eqdef f_{|\calEp}$.
  Let~$\ccp\eqdef\cc+\phi(\xxo-\zzero)\in F$.
  Assume that~$\phip(\Ep)\subset\Fp$ and~$\ccp\in\calFp$.\\
  Let~$\fp=\ccp+\phip$ with origin~$\xxo$ be the affine submap of~$f$
  on~$\calEp.$\\
  Then, $\fp$ is bijective in~$\AffcalEpcalFp$ iff~$\phip$ is an isomorphism
  in~$\LSp{\Ep}{\Fp}$.

  Moreover, in this case, we have $(\fp)^{-1}\in\AffcalFpcalEp$, and
  with~$\xxpo\eqdef\xxo+(\phip)^{-1}(\yyo-\ccp)\in\calEp$,
  \begin{align}
    \label{e:inv-of-aff-sub-map-is-aff-sub-map-1}
    \forall \yyp \in \calFp,\quad
    (\fp)^{-1} (\yyp) &= \xxo + (\phip)^{-1} (\yyp - \ccp)
      = \xxpo + (\phip)^{-1} (\yyp - \yyo),\\
    \label{e:inv-of-aff-sub-map-is-aff-sub-map-2}
    \forall \vvp \in \Fp,\quad
    (\phip)^{-1} (\vvp) &= (\fp)^{-1} (\ccp + \vvp) - \xxo.
  \end{align}
\end{lemma}

\begin{proof}
  \proofpar{Equivalence}
  Direct consequence of
  \assume{the definition of bijectivity},
  Lemma~\thref{l:inj-aff-sub-map-is-zero-linear-ker},
  Lemma~\thref{l:surj-aff-sub-map-is-full-linear-rg},
  Lemma~\thref{LM-l:injective-linear-map-has-zero-kernel},
  \assume{the characterization of surjectivity with full range}, and
  Definition~\thref{LM-d:isomorphism}.

  \proofparskip{%
    \eqref{e:inv-of-aff-sub-map-is-aff-sub-map-1} and
    $(\fp)^{-1}\in\AffcalFpcalEp$}
  Let~$\yyp\in\calFp$.
  Let~$\xxp\eqdef(\fp)^{-1}(\yyp)\in\calEp$.
  Then, from
  \assume{the definition of the inverse},
  Lemma~\thref{l:equiv-def-aff-map},
  Lemma~\threfc{l:equiv-def-aff-sub-sp}{thus $\ccp-\yyo,\yyp-\yyo\in\Fp$},
  Lemma~\threfc{LM-l:closed-under-vector-operations-is-subspace}{%
    thus $\yyo-\ccp=-(\ccp-\yyo)\in\Fp$
    and $\yyp-\ccp=(\yyo-\ccp)+(\yyp-\yyo)\in\Fp$},
  Definition~\threfc{d:aff-sub-sp}{%
    with $(\phip)^{-1}(\yyo-\ccp)\in\Ep$, thus $\xxpo\in\calEp$},
  Definition~\threfc{LM-d:space}{%
    additive abelian group properties for $\Ep$ and $\Fp$},
  Lemma~\thref{l:inverse-of-isomorphism-is-linear-map}, and
  Definition~\thref{d:aff-map},
  we have successively
  $\yyp=\fp(\xxp)=\ccp+\phip(\xxp-\xxo)$,
  $\phip(\xxp-\xxo)=\yyp-\ccp=(\yyo-\ccp)+(\yyp-\yyo)$, and thus
  \begin{gather*}
    (\fp)^{-1}(\yyp) = \xxp
    = \xxo + (\phip)^{-1} (\yyp - \ccp)
    = \xxpo + (\phip)^{-1} (\yyp - \yyo)
    \AND
    (\fp)^{-1} \in \AffcalFpcalEp.
  \end{gather*}

  \proofparskip{\eqref{e:inv-of-aff-sub-map-is-aff-sub-map-2}}
  Direct consequence of~\eqref{e:inv-of-aff-sub-map-is-aff-sub-map-1}, and
  Lemma~\threfc{LM-l:closed-under-linear-combination-is-subspace}{%
    for both~$E$ and~$F$}.
\end{proof}

\begin{lemma}[inverse of affine map is affine map]
  \label{l:inv-of-aff-map-is-aff-map}
  \mbox{}\hfill
  Let~$E$ and~$F$ be {\vectorspace}s.\\
  Let~$\phi\in\LEF$ and~$\cc\in F$.
  Let $f\eqdef\cc+\phi\in\AffEF$ with origin~$\zzero$.\\
  Then, $f$ is bijective iff~$\phi$ is an isomorphism.

  Moreover, in this case, we have $f^{-1}\in\AffFE$, and
  with~$\xxpo\eqdef\zzero+\phi^{-1}(\zzero-\cc)\in E$,
  \begin{align}
    \label{e:inv-of-aff-map-is-aff-map-2}
    \forall \yy \in F,\quad
    f^{-1} (\yy) &= \zzero + \phi^{-1} (\yy - \cc)
      = \xxpo + \phi^{-1} (\yy - \zzero),\\
    \label{e:inv-of-aff-map-is-aff-map-3}
    \forall \vv \in F,\quad
    \phi^{-1} (\vv) &= f^{-1} (\cc + \vv) - \zzero.
  \end{align}
\end{lemma}

\begin{proof}
  Direct consequence of
  Lemma~\threfc{l:inv-of-aff-sub-map-is-aff-sub-map}{%
    with\linebreak $\calEp\eqdef\Ep\eqdef E$, $\calFp\eqdef\Fp\eqdef F$ and
    $\xxo=\yyo\eqdef\zzero$, thus $\phip=\phi$, $\fp=f$, and
    $\ccp=\cc\in F$, and $\phi(E)\subset F$}.
\end{proof}

\section{Complements on real affine geometry}
\label{s:compl-aff-geom}

\begin{definition}[affinely independent family]
  \label{d:aff-indep-family}
  \mbox{}\\
  Let~$d\geq1$.
  Let~$n\in\matN$.
  Let $\famvert{\vv}{n,d}$ be $n+1$ points in~$\matRd$.\\
  Then, $\famvert{\vv}{n,d}$ is said {\em affinely independent} iff
  $(\vv_1-\vv_0,\ldots,\vv_n-\vv_0)$ is free in the {\vectorspace} $\matRd$.
\end{definition}

\begin{lemma}[equivalent definition of affinely independent family]
  \label{l:equiv-def-aff-indep-family}
  \mbox{}\\
  Let~$d\geq1$.
  Let~$n\in\matN$.
  Let $\famvert{\vv}{n,d}$ be $n+1$ points in~$\matRd$.
  Let $j\in[0..n]$.\\
  Then, $\famvert{\vv}{n,d}$ is affinely independent iff
  $(\vv_i-\vv_j)_{i\in[0..n]\setminus\{j\}}$ is free in the {\vectorspace} $\matRd$.
\end{lemma}

\begin{proof}
  \proofparskip{Case $n=0$}
  Direct consequence of
  Definition~\thref{d:aff-indep-family}, and
  \assume{the definition of freedom
    (the empty family in any {\vectorspace} is free)}.

  \proofparskip{Case $j=0$}
  Direct consequence of
  Definition~\thref{d:aff-indep-family}.

  \proofparskip{Case $n\geq1$ and $j\neq0$}
  \proofpar{(1)}
  Let $k,l\in[0..n]$, such that $k\neq l$.
  Let $(\mu_i)_{i\in[0..n]\setminus\{k\}}\in\matRn$.\\
  Then, from
  Definition~\threfc{LM-d:space}{properties of linear operations},
  we have
  \begin{align*}
    \sum_{i=0,i\neq k}^n \mu_i (\vv_i - \vv_k)
    &= \sum_{\substack{i=0\\i\neq k, i\neq l}}^n \mu_i \vv_i
    - \left( \sum_{i=0, i\neq k}^n \mu_i \right) \vv_k  + \mu_l \vv_l
    = \sum_{i=0, i\neq l}^n \tilde{\mu}_i \vv_i  + \mu_l \vv_l \\
    &= \sum_{i=0, i\neq l}^n \tilde{\mu}_i (\vv_i - \vv_l)
      + \left( \mu_l + \sum_{i=0, i\neq l}^n \tilde{\mu}_i \right) \vv_l \\
    &= \sum_{i=0, i\neq l}^n \tilde{\mu}_i (\vv_i - \vv_l)
      + \left( \tilde{\mu}_k + \sum_{i=0, i\neq k}^n \mu_i \right) \vv_l
    = \sum_{i=0,i\neq l}^n \tilde{\mu}_i (\vv_i - \vv_l),
  \end{align*}
  where we have set
  $\tilde{\mu}_k\eqdef-\sum_{i=0,i\neq k}^n \mu_i$, and
  $\tilde{\mu}_i\eqdef\mu_i$ for $i\in[0..n]\setminus\{k,l\}$.

  \proofparskip{(2) $(\vv_i-\vv_l)_{i\in[0..n]\setminus\{l\}}$ free $\Implies$
    $(\vv_i-\vv_k)_{i\in[0..n]\setminus\{k\}}$ free}\\
  Let $k,l\in[0..n]$, such that $k\neq l$.
  Assume that $(\vv_i-\vv_l)_{i\in[0..n]\setminus\{l\}}$ is free.\\
  Let $(\mu_i)_{i\in[0..n]\setminus\{k\}}\in\matRn$, such that
  $\sum_{i=0,i\neq k}^n\mu_i(\vv_i-\vv_{k})=0$.
  Then, from~(1), and
  \assume{the definition of freedom},
  we have $\sum_{i=0,i\neq l}^n\tilde{\mu}_i(\vv_i-\vv_{l})=0$, and
  \[
    \forall i \in [0..n] \setminus \{k, l\},\;
    \tilde{\mu}_i = \mu_i = 0, \AND
    \tilde{\mu}_{k} = -\sum_{i = 0, i\neq k}^n \mu_i = - \mu_l = 0.
  \]
  Thus, for all $i\in[0..n]\setminus\{k\}$, $\mu_i=0$, and
  the family $(\vv_i-\vv_{k})_{i\in[0..n]\setminus\{k\}}$ is free.

  \proofparskip{(3)}
  The equivalence is a direct consequence of~(2), first with $k\eqdef j$ and
  $l\eqdef0$ (from left to right), then with $k\eqdef0$ and $l\eqdef j$ (from
  right to left).
\end{proof}

\begin{lemma}[affinely independent family of 2 elements]
  \label{l:aff-indep-family-of-two-elements}
  \mbox{}\hfill
  Let~$d\geq1$.\\
  Let $\famvert{\vv}{1,d}\eqdef(\vv_0,\vv_1)$ be 2 points in~$\matRd$.
  Then, $\famvert{\vv}{1,d}$ is affinely independent iff $\vv_0\neq\vv_1$.
\end{lemma}

\begin{proof}
  Direct consequence of
  Definition~\thref{d:aff-indep-family}, and
  \assume{a free vector is {\nonzero}}.
\end{proof}

\begin{lemma}[affinely independent family is closed by sub-family]
  \label{l:aff-indep-closed-by-sub-family}
  \mbox{}\\
  Let~$d\geq1$.
  Let~$n\in\matN$.
  Let $\famvert{\vv}{n,d}$ be $n+1$ affinely independent points in~$\matRd$.\\
  Let $l\in[0..n]$, and let $\indl$ be an injective map from $[0..l]$ into
  $[0..n]$.
  Let $j\in[0..l]$.\\
  Then, $(\vv_{\indl(i)}-\vv_{\indl(j)})_{i\in[0..l]\setminus\{j\}}$ is free in
  the {\vectorspace} $\matRd$.

  In other terms, the family
  $\famvert{\vv_\indl}{l,d}\eqdef(\vv_{\indl(i)})_{i\in[0..l]}$ is affinely
  independent in~$\matRd$.
\end{lemma}

\begin{proof}
  Let $(\mu_{i})_{i\in[0..l]\setminus\{j\}}$ in $\matRl$ such that
  $\sum_{i\in[0..l]\setminus\{j\}}\mu_i(\vv_{\indl(i)}-\vv_{\indl(j)})=0$.\\
  Let $I\eqdef\indl([0..l])\subset[0..n]$, and $J\eqdef[0..n]\setminus I$.
  Thus, from
  \assume{the definition of bijectivity},
  $\indl$ is bijective from $[0..l]$ to~$I$, and $[0..n]=I\uplus J$ is a
  partition.
  Let $\jp\eqdef\indl(j)\in I$.
  Thus, we have
  \begin{equation*}
    0 = \sum_{\ip \in I \setminus \{ \jp \}}
    \mu_{\indlinv (\ip)} (\vv_\ip - \vv_\jp)
    + \sum_{\ip \in J} 0 (\vv_\ip - \vv_\jp)
    = \sum_{\ip \in [0..n] \setminus \{ \jp \}}
    \tilde{\mu}_\ip (\vv_\ip - \vv_\jp),
  \end{equation*}
  where we have set $\forall\ip\in I\setminus\{\jp\}$,
  $\tilde{\mu}_\ip\eqdef\mu_{\indlinv(\ip)}$, and
  $\forall\ip\in J$, $\tilde{\mu}_\ip\eqdef0$.
  Thus, from
  Lemma~\threfc{l:equiv-def-aff-indep-family}{with $j\eqdef\jp$}, and
  \assume{the definition of freedom},
  we deduce that for all $\ip\in I\setminus\{\jp\}$, $\mu_{\indlinv(\ip)}=0$,
  and thus for all $i\in[0..l]\setminus\{j\}$, $\mu_{i}=0$.
  Thus, from
  \assume{the definition of freedom}, and
  Definition~\thref{d:aff-indep-family},
  $(\vv_{\indl(i)}-\vv_{\indl(j)})_{i\in[0..l]\setminus\{j\}}$ is free, and
  $\famvert{\vv_\indl}{l,d}$ is affinely independent.
\end{proof}

\section{Complements on univariate polynomials}
\label{s:compl-univar-pol}

\begin{remark}
  See also Sections~\ref{s:Pk1-lag-pol} and~\ref{s:multivar-pol}.
\end{remark}

\begin{definition}[monomial of a single variable]
  \label{d:monom-k1}
  \mbox{}\hfill
  Let~$\alpha\in\matN$.
  The {\em monomial of degree~$\alpha$ of a single variable} is
  denoted~$X^\alpha$, and is defined by
  $X^\alpha\eqdef(x\in\matR\longmapsto x^\alpha\in\matR)$.
\end{definition}

\begin{definition}[polynomial space $\matPki$]
  \label{d:pol-space-Pk1}
  \mbox{}\hfill
  Let~$k\in\matN$.\\
  The {\em space of polynomials of degree at most~$k$ of a single variable}
  is denoted~$\matPki$, and is defined by
  \begin{equation}
    \label{e:pol-space-Pk1}
    \matPki \eqdef \Span{1, X, X^2, \ldots, X^k}
    = \left\{
      \left( x \mapsto \sum_{\alpha = 0}^k a_\alpha x^\alpha \right) : \ArRR
    \rightst
    \left.
      \vphantom{\sum_{\alpha = 0}^k a_\alpha x^\alpha}
      (a_\alpha)_{\alpha \in [0..k]} \in \matR
    \right\}.
  \end{equation}
\end{definition}

\begin{lemma}[$\matPki$ is space of degree at most $k$]
  \label{l:deg-Pi-leq-k}
  \mbox{}\hfill
  Let~$\calP(\matR)$ be the infinite-dimensional\\ space of polynomials of a
  single variable on~$\matR$.
  Let~$k\in\matN$.
  Then, we have
  \begin{equation}
    \label{e:deg-Pki-leq-k-1}
    \matPki = \{ p \in \calP (\matR) \st \deg p \leq k \}.
  \end{equation}

  For instance, we have
  \begin{align}
    \label{e:deg-Pki-leq-k-2}
    \matPoi &= \Span{1} = \{ (x \mapsto a_0) \st a_0 \in \matR \},\\
    \label{e:deg-Pki-leq-k-3}
    \matPii &= \Span{1, X}
    = \{ (x \mapsto a_0 + a_1 x) \st a_0, a_1 \in \matR \}
    = \AffRR.
  \end{align}
\end{lemma}

\begin{proof}
  Direct consequence of
  \assume{the definition of the degree a univariate polynomial},
  Definition~\thref{d:pol-space-Pk1},
  Definition~\thref{d:monom-k1},
  Lemma~\threfc{l:equiv-def-aff-map-finite-dim}{with $p=q=1$}, and
  Lemma~\thref{l:space-aff-maps}.
\end{proof}

\begin{lemma}[monomials are free in $\matPki$]
  \label{l:monom-free-in-Pk1}
  \mbox{}\hfill
  Let~$k\in\matN$.
  Then, $(1,X,\ldots,X^k)$ is free in~$\matPki$.
\end{lemma}

\begin{proof}
  Let $(a_\alpha)_{\alpha\in[0..k]}$ such that for all $x\in\matR$,
  $\sum_{\alpha=0}^ka_\alpha x^\alpha=0$.
  Let~$l\in[0..k]$.
  Then, from
  \assume{the linearity of derivation},
  \assume{the positivity of factorial},
  \assume{the zero-product property in~$\matR$}, and
  \assume{the definition of freedom},
  the $l$-th derivative taken in $x=0$ provides $\fact{l}\,a_l=0$,
  thus $a_l=0$, and the family is free.
\end{proof}

\begin{lemma}[dimension of $\matPki$]
  \label{l:dim-Pk1}
  \mbox{}\\
  Let~$k\in\matN$.
  Then, $\matPki$ is a {\vectorspace} of dimension $\dim\matPki=k+1$.
\end{lemma}

\begin{proof}
  Direct consequence of
  Definition~\thref{d:pol-space-Pk1},
  Lemma~\thref{l:monom-free-in-Pk1},
  Lemma~\thref{l:free-family-of-dim-elements-is-basis}, and
  \assume{the definition of the dimension of a {\vectorspace}.
}
\end{proof}

\begin{remark}
  Note in the next lemma that when~$p$ or~$q$ is 0, then by convention, its
  degree is~$-\infty$ (which is absorbing for addition in
  $\{-\infty\}\uplus\matN$), and the formula $\deg(pq)=-\infty=\deg(p)+\deg(q)$
  is still valid.
\end{remark}

\begin{lemma}[product of two univariate polynomials]
  \label{l:prod-2-polynom-univ}
  \mbox{}\\
  Let $k,l\in\matN$.
  Let $p\in\matPki$, and $q\in\matPli$.
  Then, we have $pq\in\matPkpli$.

  Moreover, if $\deg(p)=k$ and $\deg(q)=l$, then $\deg(pq)=k+l$.
\end{lemma}

\begin{proof}
  From
  \assume{the definition of univariate polynomial},
  let $(a_i)_{i\in[0..k]}\in\matRkpi$ and $(b_j)_{j\in[0..l]}$ in~$\matRlpi$
  such that we have $p=\sum_{i=0}^ka_iX^i$ and $q=\sum_{j=0}^lb_jX^j$.
  Thus, from
  \assume{commutative ring properties of $\matR$}, and
  \assume{the fact that $a_jb_{n-j}$ is {\it a priori} {\nonzero} iff
    $0\leq j\leq k$ and $0\leq n-j\leq l$ ({\ie} $n-l\leq j\leq n$)},
  we obtain
  \begin{equation}
    \label{e:poly-mul}
    pq = \sum_{n = 0}^{k+l}
    \sum_{j = \max \{ 0, n - l \}}^{\min \{ n, k \}} a_j b_{n-j} X^n.
  \end{equation}
  Thus, $pq\in\matPkpli$.

  \proofparskip{Degree formula}\\
  Direct consequence of
  \assume{the definition of univariate polynomial},
  \assume{the zero-product property in~$\matR$}, and
  formula~\eqref{e:poly-mul} (coefficients of highest degree, $a_k$,
  $b_l$ and $a_kb_l$ are all {\nonzero}).
\end{proof}

\chapter{Finite element}
\label{c:fe}

\begin{remark}
  The next definition is from from Ciarlet~\cite{cia:fem:78}, it is taken
  from~\cite[pp.~50--52]{eg:fe1:21}.
\end{remark}

\begin{definition}[finite element triple]
  \label{d:fe-triple}
  \mbox{}\hfill
  Let~$d\geq1$.
  Let~$\nsh\geq1$.
  Let~$q\in\{1,d\}$.\\
  Let~$K\subset\matRd$.
  Let~$P\subset\calF(K,\matRq)$.
  Let~$\Sigma\eqdef(\sigma_i)_{i\in[1..\nsh]}\in\LSp{P}{\matR}$.\\
  The triple~$(K,P,\Sigma)$ is called {\em finite element} iff
  \begin{enumerate}
  \item $K$~is a {\nontrivial} polyhedron, $\Interior{K}\neq\emptyset$.\\
    The polyhedron~$K$ represents the {\em geometry} of the finite element.
  \item $P$~is a {\nontrivial} finite-dimensional {\vectorspace} of functions,
    $P\neq\{0\}$.\\
    The {\vectorspace}~$P$ is called {\em approximation space} of the finite
    element.
  \item $\Sigma$~is such that $\phi_\Sigma:P\to\matR^{\nsh}$ defined by
    $\phi_\Sigma(p)\eqdef(\sigma_i(p))_{i\in[1..\nsh]}$ is an isomorphism.\\
    The linear forms~$\sigma_i$ are called {\em degrees of freedom} of the
    finite element, $\nsh$ is the {\em number of degrees of freedom}, and
    the bijectivity of~$\phi_\Sigma$ is called {\em unisolvence}.
  \end{enumerate}
\end{definition}

\begin{remark}
  \mbox{}\\
  The functions of~$P$ are typically polynomial functions, possibly composed
  with some smooth diffeomorphism.
  For any finite collection of linear forms~$\Sigma$, the mapping~$\phi_\Sigma$
  is obviously linear.
\end{remark}

\begin{lemma}[injectivity implies unisolvence]
  \label{l:inj-implies-unisolvence}
  \mbox{}\hfill
  Let~$d\geq1$.
  Let~$\nsh\geq1$.
  Let~$q\in\{1,d\}$.\\
  Let~$K\subset\matRd$ be a {\nontrivial} polyhedron.
  Let~$P\subset\calF(K,\matRq)$ be a {\nontrivial} finite-dimensional
  {\vectorspace}.
  Let~$\Sigma\eqdef(\sigma_i)_{i\in[1..\nsh]}\in\LSp{P}{\matR}$ be linear forms
  on~$P$.\\
  If $\dim P\geq\nsh=\card(\Sigma)$
  and~$\phi_\Sigma$ is injective,
  then $(K,P,\Sigma)$ is unisolvent
\end{lemma}

\begin{proof}
  Direct consequence of
  Definition~\thref{d:fe-triple}, and
  Lemma~\thref{l:inj-or-surj-and-dim-implies-bij}.
\end{proof}

\begin{lemma}[dimension of approximation space]
  \label{l:dim-of-P}
  \mbox{}\\
  Let $(K,P,\Sigma)$ be a finite element.
  Then, we have $\dim P=\nsh=\card(\Sigma)$.
\end{lemma}

\begin{proof}
  Direct consequence of
  \assume{the rank--nullity theorem}.
\end{proof}

\begin{lemma}[degrees of freedom are basis]
  \label{l:dof-basis}
  \mbox{}\\
  Let $(K,P,\Sigma)$ be a finite element.
  Then, $\Sigma$ is a basis of the dual {\vectorspace} $\LSp{P}{\matR}$.
\end{lemma}

\begin{proof}
  \mbox{}\\
  From
  Lemma~\thref{l:free-family-of-dim-elements-is-basis},
  Lemma~\thref{l:dim-of-P}, and
  \assume{the fact that $\dim(\LSp{P}{\matR})=\dim P$},
  it suffices to show that~$\Sigma$ is free.\\
  Let~$a\eqdef(a_i)_{i\in[1..\nsh]}\in\matR$ such that
  $\sum_{i\in[1..\nsh]}a_i\sigma_i=0$.
  Then, with $p\eqdef\phi_\Sigma^{-1}(a)$, we have
  \[
    \sum_{i \in [1..\nsh]} a_i \sigma_i (p)
    = \sum_{i \in [1..\nsh]} a_i^2
    = 0,
  \]
  and thus for all $i\in[1..\nsh]$, $a_i=0$.
  Therefore, $\Sigma$ is free, and it is a basis of $\LSp{P}{\matR}$.
\end{proof}

\begin{remark}
  In the next definition, the {\em predual basis} of a
  basis~$\calBp=(\bpi)_{i\in[1..n]}$ of $\Ep\eqdef\LSp{E}{\matK}$, the dual of
  some {\vectorspace}~$E$ of dimension~$n$, is a
  basis~$\calB=(b_i)_{i\in[1..n]}$ of~$E$ such that its dual basis is~$\calBp$.
  This means that, for all $i,j\in[1..n]$, we have $\bpi(b_j)=\kron{i}{j}$.
\end{remark}

\begin{definition}[shape function]
  \label{d:shape-fun}
  \mbox{}\hfill
  Let $(K,P,\Sigma)$ be a finite element.\\
  The elements of the predual basis of~$\Sigma$ are called
  {\em shape functions}.
\end{definition}

\chapter{Simplicial geometry}
\label{c:simplex}

\begin{remark}
  This is an important example, see~\cite[Chap~7, pp.~75--86]{eg:fe1:21}.

  In this section, $d\in\matN^\star$ denotes the dimension and $k\in\matN$ the
  maximal degree of the polynomial approximation.
  The case $k=0$ corresponds to constant functions and is often excluded.
  In the same way, the case $d=0$ (points) is a limit case that is note
  considered (at least in a first step).
\end{remark}

\begin{definition}[family of points]
  \label{d:fam-aff-pts}
  \mbox{}\hfill
  Let~$d\geq1$.
  Let $\calB_d\subset\matNd$ be a finite set.\\
  For all $\bbeta\in\calB_d$, let $\aa_\bbeta\in\matRd$.
  Then, the whole family is denoted
  $\famnode{\aa}{\calB_d}\eqdef(\aa_\bbeta)_{\bbeta\in\calB_d}$.
\end{definition}

\begin{remark}
  \label{r:fam-nodes-fam-verts}
  The notation $\famnode{\aa}{\calB_d}$ is used to collect the {\em nodes}
  ({\ie} points where the nodal degrees of freedom are computed through the
  linear forms $\sigma_\bbeta$, with $\bbeta\in\calB_d$).
  They are numbered with multi-indices in $\matNd$.

  In contrast, the notation $\famvertd{\vv}$ defined in
  Definition~\ref{d:canon-fam} are typically used to collect the {\em vertices}
  of the geometry, with the canonical numbering of $[0..n]$.
\end{remark}

\begin{definition}[family of reference points]
  \label{d:fam-ref-aff-pts}
  \mbox{}\hfill
  Let~$d\geq1$.
  The family of {\em reference points (in~$\matRd$)} is
  denoted~$\famvertd{\hvv}$, and is defined by $\hvv_0\eqdef\zzero$ and for all
  $i\in[1..d]$, $\hvv_i\eqdef\ee_i$.
\end{definition}

\begin{lemma}[reference isobarycenter]
  \label{l:ref-isobaryc}
  \mbox{}\hfill
  Let~$d\geq1$.
  The {\em reference isobarycenter (in~$\matRd$)} is denoted by~$\hbisobarycd$
  and is defined as the isobarycenter of $\famvertd{\hvv}$.
  Thus, we have
  $\hbisobarycd=\frac{1}{d+1}\sum_{i=0}^d\hvv_i=\frac{1}{d+1}\sum_{i=1}^d\ee_i$,
  and for all $i\in[1..d]$, $(\hbisobarycd)_i=\frac{1}{d+1}$.
\end{lemma}

\begin{proof}
  Direct consequence of
  Definition~\thref{d:isobaryc}, and
  Definition~\thref{d:fam-ref-aff-pts}.
\end{proof}

\begin{lemma}[family of reference points is affinely independent]
  \label{l:ref-affine-vert-is-affinely-indep}
  \mbox{}\\
  Let $d\geq 1$.
  Then, the family of reference points $\famvertd{\hvv}$ is affinely
  independent.
\end{lemma}

\begin{proof}
  Direct consequence of
  Definition~\thref{d:fam-ref-aff-pts}
  Definition~\thref{d:canon-fam},
  Definition~\thref{d:aff-indep-family},
  \assume{the canonical family is basis}, and
  \assume{basis is free}.
\end{proof}

\begin{definition}[reference simplex]
  \label{d:ref-simplex}
  \mbox{}\\
  Let $d\geq 1$.
  The {\em reference simplex} is denoted~$\hKd$, and is defined as the unit
  rectangular simplex in~$\matRd$,
  \begin{equation}
    \label{e:ref-simplex}
    \hKd \eqdef \left\{
      \hxx \in \matRd
      \vphantom{\sum_{i=1}^d \hx_i \leq 1}
    \right.
    \leftst
      \forall i \in [1..d],\quad \hx_i \geq 0 \CONJ
      \sum_{i=1}^d \hx_i \leq 1
    \right\}.
  \end{equation}
  Its vertices are the reference points $\famvertd{\hvv}$.
\end{definition}

\begin{lemma}[coordinates in reference simplex are smaller than~1]
  \label{l:coord-in-ref-simplex-smaller-than-1}
  \mbox{}\\
  Let $d\geq 1$.
  Let $\hxx\in\hKd$.
  Then, we have for all $i\in[1..d]$, $\hx_i\leq1$.
\end{lemma}

\begin{proof}
  \mbox{}\\
  By contradiction, assume that there exists~$i$ such that $\hx_i>1$, then from
  Definition~\thref{d:ref-simplex}, and
  \assume{monotony of addition in~$\matRplus$},
  we have $\sum_{j=1}^d\hx_j\geq\hx_i>1$, which is impossible.
\end{proof}

\begin{lemma}[{\nontrivial} reference simplex]
  \label{l:non-trivial-ref-simplex}
  \mbox{}\hfill
  Let~$d\geq1$.
  Then, $\hKd$ has {\nonempty} interior.
\end{lemma}

\begin{proof}
  Let~$\calB\eqdef
  \OpenBall_{\infty}((\frac{1}{2d},\ldots,\frac{1}{2d}),\frac{1}{2d})$.
  Let $\hxx\in\calB$.
  Then, from
  \assume{the definition of the $\infty$-norm},
  \assume{valued field properties of~$\matR$}, and
  Definition~\thref{d:ref-simplex},
  we have $\forall i\in[1..d]$,
  $\left|\hxx_i-\frac{1}{2d}\right|<\frac{1}{2d}$, {\ie}
  $0<\hxx_i<\frac{1}{d}$, thus $\hxx_i>0$ and $\sum_{=0}^d\hxx_i<1$, and
  $\hxx\in\hKd$.
  Therefore, from
  \assume{open ball is open in metric space},
  \assume{open ball with positive radius is {\nonempty}
    (with $\frac{1}{2d}>0$)}, and
  \assume{the definition of the interior},
  $\calB\neq\emptyset$, is open, and included in of~$\hKd$, and thus $\hKd$ has
  a {\nonempty} interior.
\end{proof}

\begin{definition}[simplex]
  \label{d:simplex}
  \mbox{}\hfill
  Let~$d\geq1$.
  Let~$\famvertd{\vv}$ be $d+1$ points in~$\matRd$.\\
  The {\em simplex of vertices $\famvertd{\vv}$} is denoted~$\Kvd$, and is
  defined as the convex envelop of $\famvertd{\vv}$,
  \begin{equation}
    \label{e:simplex}
    \Kvd \eqdef \left\{
      \xx = \sum_{i=0}^d \mu_i \vv_i \in \matRd
      \vphantom{\sum_{i=0}^d \mu_i = 1}
    \right.
    \leftst
      \forall i \in [0..d],\quad \mu_i \geq 0 \CONJ
      \sum_{i=0}^d \mu_i = 1
    \right\}.
  \end{equation}
\end{definition}

\begin{remark}
  The proof that a simplex~$\Kvd$ with affinely independent vertices has
  {\nonempty} interior is given later, once the geometrical transformation that
  maps the reference simplex~$\hKd$ to~$\Kvd$ is defined, see
  Lemma~\ref{l:non-trivial-simplex}.
\end{remark}

\begin{lemma}[coordinates in simplex are smaller than~1]
  \label{l:coord-in-simplex-smaller-than-1}
  \mbox{}\hfill
  Let $d\geq 1$.
  Let~$\famvertd{\vv}$ be $d+1$ points in~$\matRd$.
  Let $\xx=\sum_{i=0}^d\mu_i\vv_i\in\Kvd$.
  Then, we have for all $i=[0..d]$, $\mu_i\leq1$.
\end{lemma}

\begin{proof}
  By contradiction, assume that there exists~$i$ such that $\mu_i>1$, then from
  Definition~\thref{d:simplex}, and
  \assume{monotony of addition in~$\matRplus$},
  we have $\sum_{i=0}^d\mu_i\geq\mu_i>1$, which is impossible.
\end{proof}

\begin{lemma}[simplex of reference vertices is reference simplex]
  \label{l:simplex-of-ref-vert-is-ref-simplex}
  \mbox{}\\
  Let $d\geq 1$.
  Then, we have  $\Khvd=\hKd$.
\end{lemma}

\begin{proof}
  Let $\hxx\in\matRd$.
  Then, since
  \assume{the canon family is basis},
  we have $\hxx=\sum_{i=1}^d\hx_i\ee_i$.\\
  Let $\mu_j\eqdef\hx_j$ for all $j\in[1..d]$, and
  $\mu_0\eqdef1-\sum_{j=1}^d\hx_j$.
  Then, from
  Definition~\thref{d:fam-ref-aff-pts},
  we also have $\hxx=\sum_{i=0}^d\mu_i\hvv_i$ with $\sum_{i=0}^d\mu_i=1$.
  Thus, from
  \assume{ordered field properties of $\matR$},
  we obviously have the equivalence
  \[
    \forall i \in [0..d],\quad \mu_i \geq 0 \EQUIV
    \sum_{j = 1}^d \hx_j \leq 1, \AND
    \forall i \in [1..d],\quad \hx_i \geq 0.
  \]
  And finally, from
  Definition~\thref{d:ref-simplex}, and
  Definition~\thref{d:simplex},
  we have $\hxx\in\Khvd$ iff $\hxx\in\hKd$, hence the equality.
\end{proof}

\chapter{{\ToPFEki} Lagrange finite element on a segment}
\label{c:Pk1-lag-fe}
\minitoc

\begin{remark}
  In this chapter, we set $d\eqdef1$.
  General results for $d\geq1$ are provided in Chapter~\ref{c:Pkd-lag-fe}.
\end{remark}

\section{Multi-indices in dimension $d=1$}
\label{s:multi-ind-in-dim-1}

\begin{definition}[multi-indices $\calAki$]
  \label{d:multi-ind-Ak1}
  \mbox{}\\
  Let $k\in\matN$.
  The {\em set of multi-indices in dimension~1} is denoted~$\calAki$, and is
  defined by
  \begin{equation}
    \label{e:multi-ind-Ak1}
    \calAki \eqdef
    \left\{ \alpha \in \matN \st \alpha \leq k \right\}
    = [0..k].
  \end{equation}
\end{definition}

\begin{lemma}[cardinal of $\calAki$]
  \label{l:card-Ak1}
  \mbox{}\hfill
  Let $k\in\matN$.
  The number of elements of~$\calAki$ is
  \begin{equation}
    \label{e:card-Ak1}
    \card(\calAki) = \binom{k+1}{1} = k+1.
  \end{equation}
\end{lemma}

\begin{proof}
  Direct consequence of
  Definition~\thref{d:multi-ind-Ak1}, and
  Lemma~\threfc{l:prop-binom-coef}{\eqref{e:prop-binom-coef-1}}.
\end{proof}

\section{{\ToPPki} Lagrange polynomials}
\label{s:Pk1-lag-pol}

\begin{remark}
  In the sequel, for a set of indices~$I$, and for a family of
  elements~$(a_i)_{i\in I}$, the expression ``there is no double
  in~$(a_i)_{i\in I}$'' means that the indexation is injective, {\ie} for all
  $i,j \in I$, $i\neq j$ implies $a_i\neq a_j$.
\end{remark}

\begin{definition}[Lagrange polynomials of $\matPki$]
  \label{d:lag-pol-Pk1}
  \mbox{}\\
  Let $k\in\matN$.
  Let $\famnodeki{a}$ be $k+1$ points in~$\matR$.
  Assume that there is no double in~$\famnodeki{a}$.
  The {\em Lagrange polynomials associated with~$\famnodeki{a}$} are denoted
  $\left(\calL^{\famnodeki{a}}_i\right)_{i\in[0..k]}$, and are defined by
  \begin{equation}
    \label{e:lag-pol-Pk1}
    \forall x \in \matR,\quad
    \calL^{\famnodeki{a}}_i (x) \eqdef
    \prod_{j=0, j\neq i}^k \frac{x - a_j}{a_i - a_j},
  \end{equation}
  with the convention that a product indexed by the empty set equals~1
  ({\ie} when $k\eqdef0$).
\end{definition}

\begin{lemma}[Lagrange polynomials is basis of $\matPki$]
  \label{l:lag-pol-is-basis-Pk1}
  \mbox{}\\
  Let $k\in\matN$.
  Let $\famnodeki{a}$ be $k+1$ points in~$\matR$.
  Assume that there is no double in~$\famnodeki{a}$.\\
  Then, the Lagrange polynomials associated with~$\famnodeki{a}$ form a
  basis of~$\matPki$, that satisfies
  \begin{align}
    \label{e:lag-basis-Pk1-1}
    \forall i \in [0..k],\quad
    & \deg \calL^{\famnodeki{a}}_i = k,\\
    \label{e:lag-basis-Pk1-2}
    \forall i, j \in [0..k],\quad
    & \calL^{\famnodeki{a}}_i(a_j) = \kron{i}{j},\\
    \label{e:lag-basis-Pk1-3}
    & \sum_{i=0}^k \calL^{\famnodeki{a}}_i = 1.
  \end{align}
\end{lemma}

\begin{proof}
  \proofparskip{Degree~\eqref{e:lag-basis-Pk1-1}, in $\matPki$,
    and identity~\eqref{e:lag-basis-Pk1-2}}
  Direct consequence of
  Definition~\thref{d:lag-pol-Pk1},
  Lemma~\thref{l:prod-2-polynom-univ},
  Lemma~\thref{l:deg-Pi-leq-k}, and
  \assume{field properties of~$\matR$}.

  \proofparskip{Basis}
  Let $(\lambda_i)_{i\in[0..k]}$ in~$\matR$, such that
  $\sum_{i=0}^k\lambda_i\calL^{\famnodeki{a}}_i=0$.
  This implies that for all $x\in\matR$,
  $\sum_{i=0}^k\lambda_i\calL^{\famnodeki{a}}_i(x)=0$.
  Then, from~\eqref{e:lag-basis-Pk1-2}, and
  \assume{ring properties of~$\matR$},
  taking $x\eqdef a_j$ provides $\sum_{i=0}^k\lambda_i\kron{i}{j}=\lambda_j=0$.
  Thus, from
  \assume{the definition of freedom},
   Lemma~\thref{l:dim-Pk1}, and
   Lemma~\thref{l:free-family-of-dim-elements-is-basis},
   the Lagrange polynomials form a basis of~$\matPki$.

  \proofparskip{Identity~\eqref{e:lag-basis-Pk1-3}}
  From
  Definition~\threfc{d:pol-space-Pk1}{$1$ is in $\matPki$},
  the previous point (basis), and
  \assume{the definition of basis (basis is generator)},
  there exists $(\alpha_i)_{i\in[0..k]}$ in~$\matR$ such that
  $1=\sum_{i=0}^k\alpha_i\calL^{\famnodeki{a}}_i$,
  and thus for all $x\in\matR$,
  $1=\sum_{i=0}^k\alpha_i\calL^{\famnodeki{a}}_i(x)$.
  Let $j\in[0..k]$.
  Then, from~\eqref{e:lag-basis-Pk1-2},
  taking $x\eqdef a_j$ provides $1=\sum_{i=0}^k\alpha_i\kron{i}{j}=\alpha_j$,
  and we obtain~\eqref{e:lag-basis-Pk1-3}.
\end{proof}

\begin{lemma}[decomposition of $\matPki$ polynomial in Lagrange basis]
  \label{l:decomp-Pk1-pol-in-lag-basis}
  \mbox{}\hfill
  Let~$k\in\matN$.
  Let~$\famnodeki{a}$ be $k+1$ points in~$\matR$.
  Assume that there is no double in~$\famnodeki{a}$.
  Then, we have
  \begin{equation}
    \label{e:decomp-Pk1-pol-in-lag-basis}
    \forall p \in \matPki,\quad
    p = \sum_{i=0}^k p(a_i) \, \calL^{\famnodeki{a}}_i.
  \end{equation}
\end{lemma}

\begin{proof}
  Let $p\in\matPki$.
  Then, from
  Lemma~\threfc{l:lag-pol-is-basis-Pk1}{basis},
  there exists a unique $(\alpha_i)_{i\in[0..k]}$ in~$\matR$ such that for all
  $x\in\matR$,
  $p(x)=\sum_{i=0}^k\alpha_i\calL^{\famnodeki{a}}_i(x)$.\\
  Let $j\in[0..k]$.
  Then, from
  Lemma~\threfc{l:lag-pol-is-basis-Pk1}{\eqref{e:lag-basis-Pk1-2}},
  taking $x\eqdef a_j$ provides
  $p(a_j)=\sum_{i=0}^k\alpha_i\kron{i}{j}=\alpha_j$.
  Thus, the result.
\end{proof}

\begin{remark}
  In the following, instead of directly defining the finite element on a
  current cell, we choose to define the finite element on the reference cell
  first, then obtain the one on any current cell by the affine geometric
  transformation.
\end{remark}

\section{{\ToPFErefki} Lagrange finite element on the reference segment}
\label{s:Pk1-lag-fe-ref}

\begin{lemma}[reference simplex is {\nontrivial} in~$\matR$]
  \label{l:ref-simplex-non-trivial-in-R}
  \mbox{}\hfill
  The reference simplex~$\hKi$ is the segment $[0,1]$ of~$\matR$, whose
  vertices are~$\hv_0=0$ and~$\hv_1=1$ in~$\matR$.
  Its interior $\Interior{\hKi}$ is nonempty.
\end{lemma}

\begin{proof}
  Direct consequence of
  Definition~\thref{d:ref-simplex}, and
  Lemma~\thref{l:non-trivial-ref-simplex}.
\end{proof}

\begin{definition}[reference Lagrange nodes of $\matPki$]
  \label{d:lag-nodes-Pk1-ref}
  \mbox{}\\
  Let $k\in\matN$.
  The {\em reference Lagrange nodes of $\matPki$} are denoted
  $\famnodeki{\ha}$, and are defined by
  \begin{align}
    \label{e:lag-nodes-Pk1-ref-0}
    (k = 0)&& \ha_0 &\eqdef \hisobaryci = \frac{1}{2},&&\\
    \label{e:lag-nodes-Pk1-ref-1}
    (k \geq 1)&& \forall i \in [0..k],\quad \ha_i &\eqdef i \, \hat{h}
    \qquad\mbox{where } \hat{h} \eqdef \frac{1}{k}.&&
  \end{align}
\end{definition}

\begin{lemma}[reference Lagrange nodes are distinct]
  \label{l:lag-nodes-distinct-Pk1-ref}
  \mbox{}\hfill
  Let $k\in\matN$.
  Let $i,j\in\matN$.\\
  Then, the $(k+1)$ reference Lagrange nodes $\famnodeki{\ha}$ are distinct,
  and $0\leq i<j\leq k$ implies $\ha_i<\ha_j$.
\end{lemma}

\begin{proof}
  Direct consequence of
  Definition~\thref{d:lag-nodes-Pk1-ref}, and
  \assume{ordered field properties of~$\matR$}.
\end{proof}

\begin{lemma}[reference Lagrange basis of $\matPki$]
  \label{l:lag-basis-Pk1-ref}
  \mbox{}\hfill
  Let $k\in\matN$.\\
  Then, the Lagrange polynomials associated with the reference
  Lagrange nodes form a basis of~$\matPki$ that
  satisfies~\eqref{e:lag-basis-Pk1-1}, \eqref{e:lag-basis-Pk1-2}
  and~\eqref{e:lag-basis-Pk1-3}.

  They are called {\em reference Lagrange polynomials in $\matPki$}, and
  for all $i\in[0..k]$, we define the shorthand notation
  $\hcalL^{k,1}_i\eqdef\calL^{\famnodeki{\ha}}_i$.
\end{lemma}

\begin{proof}
  Direct consequence of
  Lemma~\thref{l:lag-nodes-distinct-Pk1-ref}, and
  Lem\-ma~\thref{l:lag-pol-is-basis-Pk1}.
\end{proof}

\begin{definition}[reference Lagrange linear forms for $\matPki$]
  \label{d:lag-lin-forms-Pk1-ref}
  \mbox{}\\
  Let $k\in\matN$.
  The {\em reference Lagrange linear forms} associated with the reference
  Lagrange nodes of~$\matPki$ are denoted
  $\hSigmaki=(\hsigma_i)_{i\in[0..k]}$, and are defined by
  \begin{equation}
    \label{e:lag-lin-forms-Pk1-ref}
    \forall i \in [0..k],\;
    \forall \hf : \ArRR,\quad
    \hsigma_i (\hf) \eqdef \hf (\ha_i).
  \end{equation}
\end{definition}

\begin{lemma}[reference Lagrange linear forms for $\matPki$ are linear]
  \label{l:lag-lin-forms-are-linear-Pk1-ref}
  \mbox{}\\
  Let $k\in\matN$.
  Let $i\in[0..k]$.
  Then, the reference Lagrange linear form~$\hsigma_i$ is linear.
\end{lemma}

\begin{proof}
  Direct consequence of
  Definition~\thref{d:lag-lin-forms-Pk1-ref}, and
  \assume{the definition of linear operations over functions}.
\end{proof}

\begin{lemma}[reference Lagrange linear forms for $\matPki$ are injective]
  \label{l:lag-lin-forms-Pk1-ref-inj}
  \mbox{}\\\
  Let $k\in\matN$.
  Then, $\phi_{\hSigmaki}$ is injective.
\end{lemma}

\begin{proof}
  Direct consequence of
  Definition~\thref{d:fe-triple},
  Lemma~\thref{l:lag-nodes-distinct-Pk1-ref},
  Definition~\threfc{d:lag-lin-forms-Pk1-ref}{%
    thus for all $\hp\in\matPki$, $\hsigma_i(\hp)$ is zero implies
    $\hp(\ha_i)=0$},
  Lemma~\threfc{l:decomp-Pk1-pol-in-lag-basis}{thus $\hp=0$},
  Definition~\thref{LM-d:kernel}, and
  Lemma~\thref{LM-l:injective-linear-map-has-zero-kernel}.
\end{proof}

\begin{lemma}[unisolvence of $\matPki$ (reference)]
  \label{l:unisolvence-Pk1-ref}
  \mbox{}\\
  Let $k\in\matN$.
  Then, $\left(\hKi,\matPki,\hSigmaki\right)$ satisfies the unisolvence
  property.
\end{lemma}

\begin{proof}
  Direct consequence of
  Lemma~\thref{l:dim-Pk1},
  Lemma~\thref{l:lag-nodes-distinct-Pk1-ref},
  \assume{the definition of order (reflexivity,
    thus $\dim\matPki=k+1$ is greater than or equal to
    $k+1=\card\left(\hSigmaki\right)$)},
  Lemma~\thref{l:lag-lin-forms-Pk1-ref-inj}, and
  Lemma~\thref{l:inj-implies-unisolvence}.
\end{proof}

\begin{theorem}[$\FElagPref{k}{1}$ reference Lagrange finite element]
  \label{t:Pk1-lag-fe-ref}
  \mbox{}\\
  Let $k\in\matN$.
  Then, $\FElagPref{k}{1}\eqdef\left(\hKi,\matPki,\hSigmaki\right)$ is a
  finite element.

  It is called the
  {\em reference Lagrange finite element of degree~$k$ in dimension~1}.
\end{theorem}

\begin{proof}
  Direct consequence of
  Lemma~\thref{l:ref-simplex-non-trivial-in-R},
  Lemma~\thref{l:dim-Pk1},
  Lemma~\thref{l:unisolvence-Pk1-ref}, and
  Definition~\threfc{d:fe-triple}{with $q\eqdef1$}.
\end{proof}

\section{{\ToPFEki} Lagrange finite element on a current segment}
\label{s:Pk1-lag-fe-cur}

\begin{definition}[geometric mapping in dimension 1]
  \label{d:geo-mapping-1d}
  \mbox{}\hfill
  Let $\famverti{v}$ be two points in~$\matR$.\\
  The {\em geometric mapping associated with $\famverti{v}$} is
  denoted~$\phigeoiv$, and is defined by
  \begin{equation}
    \label{e:geo-mapping-1d}
    \forall \hx \in \matR,\quad
    \phigeoiv (\hx) \eqdef (v_1 - v_0) \hx + v_0.
  \end{equation}
\end{definition}

\begin{lemma}[properties of geometric mapping in dimension 1]
  \label{l:prop-geo-mapping-1d}
  \mbox{}\\
  Let $\famverti{v}$ be two points in~$\matR$.
  Then, $\phigeoiv\in\matPii$, {\ie} is affine, and we have
  \begin{align}
    \label{e:prop-geo-mapping-1d-1}
    \forall \hx \in \matR,\quad
    & \phigeoiv (\hx) = (1 - \hx) \, v_0 + \hx \, v_1
      = \hcalL^{1,1}_0(\hx) \, v_0 + \hcalL^{1,1}_1(\hx) \, v_1,\\
    \label{e:prop-geo-mapping-1d-2}
    \forall \hx \in \matR,\quad
    & \left( \phigeoiv \right)^\prime (\hx) = v_1 - v_0.
  \end{align}

  Moreover, if $v_0\neq v_1$, then $\phigeoiv$ is bijective, $\invphigeoiv$
  is affine, and we have
  \begin{equation}
    \label{e:prop-geo-mapping-1d-4}
    \forall x \in \matR,\quad
    \invphigeoiv (x) = \frac{x - v_0}{v_1 - v_0}.
  \end{equation}
\end{lemma}

\begin{proof}
  Direct consequence of
  Definition~\thref{d:geo-mapping-1d},
  Lemma~\thref{l:equiv-def-aff-map-finite-dim},
  Lemma~\thref{l:deg-Pi-leq-k},
  Definition~\threfc{d:lag-nodes-Pk1-ref}{%
    with $k\eqdef1$, thus $\famnodeki{\ha}\eqdef\famverti{\hv}$},
  Lemma~\threfc{l:lag-basis-Pk1-ref}{with $k\eqdef1$},
  \assume{the rules of derivation in~$\matR$}, and
  Lemma~\threfc{l:inv-of-aff-map-is-aff-map}{%
    with $\phi^{-1}(y)\eqdef\frac{y}{v_1-v_0}$}.
\end{proof}

\begin{lemma}[current simplex is {\nontrivial} in~$\matR$]
  \label{l:curr-simplex-non-trivial-in-R}
  \mbox{}\hfill
  Let $\famverti{v}$ be two points in~$\matR$.\\
  Then, the current simplex~$\Kvi$ is the segment
  $[\min\{v_0,v_1\},\max\{v_0,v_1\}]$ of~$\matR$.

  Moreover, we have
  \begin{equation}
    \label{e:curr-simplex-non-trivial-in-R}
    v_0 \neq v_1 \EQUIV \Interior{\Kvi} \neq \emptyset.
  \end{equation}
\end{lemma}

\begin{proof}
  \proofpar{$\Kvi$}
  Direct consequence of
  Definition~\thref{d:simplex}, and
  \assume{the definition of a segment}.

  \proofparskip{Equivalence}
  Let $w_0\eqdef\min\{v_0,v_1\}$ and $w_1\eqdef\max\{v_0,v_1\}$.
  Then, from
  \assume{ordered field properties of~$\matR$}, and
  \assume{the interior of a segment},
  we have
  \[
    v_0 \neq v_1 \EQUIV w_0 < w_1 \EQUIV (w_0, w_1) \neq \emptyset
    \EQUIV \Interior{\Kvi} \neq \emptyset.
  \]
\end{proof}

\begin{lemma}[current simplex is image of reference in~$\matR$]
  \label{l:curr-simplex-im-ref-in-R}
  \mbox{}\\
  Let $\famverti{v}$ be two points in~$\matR$.
  Then, we have $\Kvi=\phigeoiv(\hKi)$.
\end{lemma}

\begin{proof}
  Direct consequence of
  Definition~\thref{d:geo-mapping-1d},
  Lemma~\thref{l:ref-simplex-non-trivial-in-R}, and
  Lemma~\thref{l:curr-simplex-non-trivial-in-R}.
\end{proof}

\begin{definition}[Lagrange nodes of $\matPki$]
  \label{d:lag-nodes-Pk1}
  \mbox{}\hfill
  Let $k\in\matN$.
  Let $\famverti{v}$ be two points in~$\matR$.\\
  The {\em Lagrange nodes of $\matPki$} are denoted $\famnodeki{a}$, and are
  defined by
  \begin{align}
    \label{e:lag-nodes-Pk1-0}
    (k = 0)&&
    a_0 &\eqdef \isobarycv = \frac{v_0 + v_1}{2},&&\\
    \label{e:lag-nodes-Pk1-1}
    (k \geq 1)&&
    \forall i \in [0..k],\quad
    a_i &\eqdef v_0 + i \, h
    \qquad\mbox{where } h \eqdef \frac{v_1 - v_0}{k}.&&
  \end{align}
\end{definition}

\begin{remark}
  The isobarycenter $\isobarycv$ is defined in Definition~\ref{d:isobaryc}.
\end{remark}

\begin{lemma}[Lagrange nodes are distinct]
  \label{l:lag-nodes-distinct-Pk1}
  \mbox{}\\
  Let $k\in\matN$.
  Let $i,j\in\matN$.
  Let $\famverti{v}$ be two distinct points in~$\matR$.\\
  Then, the $(k+1)$ Lagrange nodes $\famnodeki{a}$ are distinct,
  and $0\leq i<j\leq k$ implies $a_i\neq a_j$.

  Moreover, if $v_0<v_1$, then $h>0$, and $0\leq i<j\leq k$ implies $a_i<a_j$,
  and if $v_0>v_1$, then $h<0$, and $0\leq i<j\leq k$ implies $a_i>a_j$.
\end{lemma}

\begin{proof}
  Direct consequence of
  Definition~\thref{d:lag-nodes-Pk1}, and
  \assume{ordered field properties of~$\matR$}.
\end{proof}

\begin{lemma}[Lagrange nodes of $\matPki$ are images of reference]
  \label{l:lag-nodes-Pk1-im-ref}
  \mbox{}\\
  Let $k\in\matN$.
  Let $\famverti{v}$ be two points in~$\matR$.
  Then, we have $\forall i\in[0..k]$, $a_i=\phigeoiv(\ha_i)$.
\end{lemma}

\begin{proof}
  Direct consequence of
  Definition~\thref{d:geo-mapping-1d},
  Definition~\thref{d:lag-nodes-Pk1-ref}, and
  Definition~\thref{d:lag-nodes-Pk1}.
\end{proof}

\begin{lemma}[Lagrange basis of $\matPki$]
  \label{l:lag-basis-Pk1}
  \mbox{}\hfill
  Let $k\in\matN$.
  Let $\famverti{v}$ be two distinct points in~$\matR$.\\
  Then, the Lagrange polynomials associated with the Lagrange nodes form a
  basis of~$\matPki$ that satisfies~\eqref{e:lag-basis-Pk1-1},
  \eqref{e:lag-basis-Pk1-2} and~\eqref{e:lag-basis-Pk1-3}.
  They are called {\em current Lagrange polynomials in $\matPki$}.
\end{lemma}

\begin{proof}
  Direct consequence of
  Lemma~\thref{l:lag-nodes-distinct-Pk1}, and
  Lemma~\thref{l:lag-pol-is-basis-Pk1}.
\end{proof}

\begin{lemma}[Lagrange polynomials of $\matPki$ are images of reference]
  \label{l:lag-pol-Pk1-im-ref}
  \mbox{}\\
  Let $k\in\matN$.
  Let $\famverti{v}$ be two distinct points in~$\matR$.
  Then, we have for all $i\in[0..k]$,
  \begin{equation}
    \label{e:lag-pol-Pk1-im-ref}
    \calL^{\famnodeki{a}}_i = \hcalL^{k,1}_i \circ \invphigeoiv
    \AND
    \hcalL^{k,1}_i = \calL^{\famnodeki{a}}_i \circ \phigeoiv.
  \end{equation}
\end{lemma}

\begin{proof}
  Direct consequence of
  Definition~\thref{d:lag-pol-Pk1},
  Lemma~\thref{l:lag-basis-Pk1-ref},
  Lemma~\thref{l:lag-basis-Pk1},
  Lemma~\thref{l:lag-nodes-Pk1-im-ref},
  Definition~\thref{d:geo-mapping-1d}, and
  \assume{field properties of~$\matR$ (providing the following identities}
  \begin{gather*}
    \assumemm{%
    \frac{x - a_j}{a_i - a_j}
    = \frac{(v_1 - v_0) \, \hx + v_0 - (v_1 - v_0) \, \ha_j - v_0}%
    {(v_1 - v_0) \, \ha_i + v_0 - (v_1 - v_0) \, \ha_j - v_0}
    = \frac{\hx - \ha_j}{\ha_i - \ha_j},}\\
    \assumemm{\forall f : \ArRR,\quad 1 = 1 \circ f)}.
  \end{gather*}
\end{proof}

\begin{lemma}[geometric mapping of $\matPki$ is $\matPki$]
  \label{l:geo-mapping-of-Pk1-is-Pk1}
  \mbox{}\hfill
  Let $k\in\matN$.
  Let $\famverti{v}$ be two points in~$\matR$.\\
  Then, for all $p\in\matPki$, $p\circ\phigeoiv\in\matPki$.
  Moreover, if $v_0\neq v_1$, then, $p\circ\invphigeoiv\in\matPki$.
\end{lemma}

\begin{proof}
  Let $p\in\matPki$.

  \proofparskip{Case $v_0=v_1$}
  Direct consequence of
  Definition~\threfc{d:geo-mapping-1d}{thus, $\phigeoiv=v_0$ is constant}, and
  Definition~\threfc{d:pol-space-Pk1}{%
    thus, we have $p\circ\phigeoiv=p(v_0)$ lives
    in $\matPoi\subset\matPki$}.

  \proofparskip{Case $v_0\neq v_1$}
  Then, from
  Lemma~\threfc{l:lag-basis-Pk1}{basis},
  Lemma~\thref{l:lag-pol-Pk1-im-ref}, and
  Lemma~\threfc{l:lag-basis-Pk1-ref}{basis},
  there exists $(\alpha_i)_{i\in[0..k]}$ in~$\matR$ such that
  $p=\sum_{i=0}^k\alpha_i\calL^{\famnodeki{a}}_i$, and we have
  \[
    p \circ \phigeoiv
    = \sum_{i=0}^k \alpha_i \calL^{\famnodeki{a}}_i \circ \phigeoiv
    = \sum_{i=0}^k \alpha_i \hcalL^{k,1}_i \quad \in \matPki.
  \]
  Similarly, from
  Lemma~\threfc{l:lag-basis-Pk1-ref}{basis},
  Lemma~\thref{l:lag-pol-Pk1-im-ref}, and
  Lemma~\threfc{l:lag-basis-Pk1}{basis},
  there exists  $(\alpha_i)_{i\in[0..k]}$ in~$\matR$ such that
  $p=\sum_{i=0}^k\alpha_i\hcalL^{k,1}_i$, and we have
  \[
    p \circ \invphigeoiv
    = \sum_{i=0}^k \alpha_i \hcalL^{k,1}_i \circ \invphigeoiv
    = \sum_{i=0}^k \alpha_i \calL^{\famnodeki{a}}_i \quad \in \matPki.
  \]
\end{proof}

\begin{definition}[Lagrange linear forms for $\matPki$]
  \label{d:lag-lin-forms-Pk1}
  \mbox{}\\
  Let $k\in\matN$.
  Let $\famverti{v}$ be two points in~$\matR$.
  The {\em Lagrange linear forms} associated with the Lagrange nodes
  of~$\matPki$ are denoted
  $\Sigma^{\famnodeki{a}}=(\sigma_i)_{i\in[0..k]}$, and are defined by
  \begin{equation}
    \label{e:lag-lin-forms-Pk1}
    \forall i \in [0..k],\;
    \forall f : \ArRR,\quad
    \sigma_i (f) \eqdef f (a_i).
  \end{equation}
\end{definition}

\begin{lemma}[$\matPki$ Lagrange linear forms are linear]
  \label{l:lag-lin-forms-are-linear-Pk1}
  \mbox{}\hfill
  Let $k\in\matN$.\\
  Let $\famverti{v}$ be two points in~$\matR$.
  Let $i\in[0..k]$.
  Then, the Lagrange linear form~$\sigma_i$ is linear.
\end{lemma}

\begin{proof}
  Direct consequence of
  Definition~\thref{d:lag-lin-forms-Pk1}, and
  \assume{the definition of linear operations over functions.}
\end{proof}

\begin{lemma}[$\matPki$ Lagrange linear forms are images of reference]
  \label{l:lag-lin-forms-Pk1-im-ref}
  \mbox{}\hfill
  Let $k\in\matN$.
  Let $\famverti{v}$ be two points in~$\matR$.
  Then, we have
  $\forall i\in[0..k]$, $\forall f:\ArRR$, $\sigma_i(f)=\hsigma_i(\hf)$,
  where $\hf=f\circ\phigeoiv$.
\end{lemma}

\begin{proof}
  Direct consequence of
  Definition~\thref{d:lag-lin-forms-Pk1},
  Lemma~\thref{l:lag-nodes-Pk1-im-ref}, and
  Definition~\thref{d:lag-lin-forms-Pk1-ref}.
\end{proof}

\begin{remark}
  Note that Lemmas~\ref{l:curr-simplex-im-ref-in-R},
  \ref{l:geo-mapping-of-Pk1-is-Pk1}, and~\ref{l:lag-lin-forms-Pk1-im-ref}
  imply that the current~$\matPki$ FE $(\Kvi,\matPki,\Sigma^{\famnodeki{a}})$
  can be considered as the image of the reference~$\matPki$ FE
  $(\hKi,\matPki,\hSigmaki)$ by~$\phigeoiv$.
\end{remark}

\begin{lemma}[Lagrange linear forms for $\matPki$ are injective]
  \label{l:lag-lin-forms-Pk1-inj}
  \mbox{}\\
  Let $k\in\matN$.
  Let $\famverti{v}$ be two distinct points in~$\matR$.
  Then, $\phi_{\Sigma^{\famnodeki{a}}}$ is injective.
\end{lemma}

\begin{proof}
  Direct consequence of
  Definition~\thref{d:fe-triple},
  Lemma~\thref{l:lag-nodes-distinct-Pk1},
  Definition~\threfc{d:lag-lin-forms-Pk1}{%
    thus for all $p\in\matPki$, $\sigma_i(p)=0$ implies $p(a_i)=0$},
  Lemma~\threfc{l:decomp-Pk1-pol-in-lag-basis}{thus $p=0$},
  Definition~\thref{LM-d:kernel}, and
  Lemma~\thref{LM-l:injective-linear-map-has-zero-kernel}.
\end{proof}

\begin{lemma}[unisolvence of $\matPki$]
  \label{l:unisolvence-Pk1}
  \mbox{}\hfill
  Let $k\in\matN$.
  Let $\famverti{v}$ be two distinct points in~$\matR$.\\
  Then, $\left(\Kvi,\matPki,\Sigma^{\famnodeki{a}}\right)$ satisfies the
  unisolvence property.
\end{lemma}

\begin{proof}
  Direct consequence of
  Lemma~\thref{l:dim-Pk1},
  Lemma~\thref{l:lag-nodes-distinct-Pk1},
  \assume{the definition of order (reflexivity,
    thus $\dim\matPki=k+1$ is greater than or equal to
    $k+1=\card\left(\Sigma^{\famnodeki{a}}\right)$)},
  Lemma~\thref{l:lag-lin-forms-Pk1-inj}, and
  Lemma~\thref{l:inj-implies-unisolvence}.
\end{proof}

\begin{theorem}[$\FElagP{k}{1}$ Lagrange finite element]
  \label{t:Pk1-lag-fe}
  \mbox{}\hfill
  Let $k\in\matN$.
  Let $\famverti{v}$ be two distinct points in~$\matR$.
  Then, $\FElagP{k}{1}\eqdef\left(\Kvi,\matPki,\Sigma^{\famnodeki{a}}\right)$ is
  a finite element.

  It is called the
  {\em Lagrange finite element of degree~$k$ in dimension~1 associated
    with vertices~$\famverti{v}$}.
\end{theorem}

\begin{proof}
  Direct consequence of
  Lemma~\thref{l:curr-simplex-non-trivial-in-R},
  Lemma~\thref{l:dim-Pk1},
  Lemma~\thref{l:unisolvence-Pk1}, and
  Definition~\threfc{d:fe-triple}{with $q\eqdef1$}.
\end{proof}

\chapter{{\ToPFEkd} Lagrange finite element in dimension $d\geq1$}
\label{c:Pkd-lag-fe}
\minitoc

\begin{remark}
  In this chapter, we consider a simplex in~$\matRd$ with~$d\geq1$.
  The case~$d=1$ corresponds to a segment, and is treated in detail in
  Chapter~\ref{c:Pk1-lag-fe}.

  In Figure~\ref{f:FE-lag-P1234}, we present the Lagrange nodes on simplices
  for several dimensions and degrees.
\end{remark}

\begin{figure}[htb]
  \centering
  \resizebox{1\linewidth}{!}{
      \begin{tikzpicture}[scale=2,math3d] 
        \def\kkmax{4}

        \def\colo{magenta}
        \def\coli{green}
        \def\colii{red}
        \def\coliii{blue}

	\coordinate (AA) at ($ (0,0,1.5) $);
	\coordinate (CC) at  ($ (AA) + (0,1,0) $);
	\draw[line width=1.0pt] (AA) -- (CC);
        \coordinate (Axyz) at  ($ 0.5*(AA) + 0.5*(CC) $) ;
        \fill[color=\colii] (Axyz) circle (1.5pt);
        \node[color=\colii,below] (Nxyz) at (Axyz) {{\tiny ${(0)}$}};  
        \foreach \K in {1,...,\kkmax}  {
	  \coordinate (AA) at ($ (0,0,1.5) + 1.5*(0,\K,0) $);
	  \coordinate (CC) at  ($ (AA) + (0,1,0) $);
	  \draw[line width=1.0pt] (AA) -- (CC);
          
          \foreach \x in {0,1,...,\K}  {
            \coordinate (Axyz) at  ($ (AA) + 1/\K*(0,\x,0) $) ;
            \ifthenelse{\x<\K}
                       { \fill[color=\coliii] (Axyz) circle (1.5pt);
                         \node[color=\coliii,below] (Nxyz) at (Axyz) {{\tiny ${(\x)}$}}; }
                       { \fill[color=\colii] (Axyz) circle (1.5pt);
                 \node[color=\colii,below] (Nxyz) at (Axyz) {{\tiny ${(\x)}$}}; } ; 
          }
        }
          
	\coordinate (AA) at ($ (0,0,0) $);
	\coordinate (CC) at  ($ (AA) + (0,1,0) $);
	\coordinate (DD) at  ($ (AA) + (0,0,1) $);
	\draw[line width=1.0pt,rounded corners=0.5pt] (AA) -- (CC) -- (DD) -- cycle;
        \coordinate (Axyz) at  ($ 1/3*(AA) + 1/3*(CC) + 1/3*(DD) $) ;
        \fill[color=\colii] (Axyz) circle (1.5pt);
        \node[color=\colii,below] (Nxyz) at (Axyz) {{\tiny ${(0,0)}$}};  
        \foreach \K in {1,...,\kkmax}  {
	  \coordinate (AA) at ($ (0,0,0) + 1.5*(0,\K,0) $);
	  \coordinate (CC) at  ($ (AA) + (0,1,0) $);
	  \coordinate (DD) at  ($ (AA) + (0,0,1) $);
	  \draw[line width=1.0pt,rounded corners=0.5pt] (AA) -- (CC) -- (DD) -- cycle;

          \foreach \k in {0,1,...,\K}  {
            \newcount\y
            \foreach \x in {\k,...,0}  {
              \pgfmathsetcount{\y}{\k-\x} 
              \coordinate (Axyz) at  ($ (AA) + 1/\K*(0,\x,\y) $) ;
              \ifthenelse{\k<\K}
                         {\fill[color=\coliii] (Axyz) circle (1.5pt);
                         \node[color=\coliii,below] (Nxyz) at (Axyz) {{\tiny ${(\x,\the\y)}$}}; }
                         {\fill[color=\colii] (Axyz) circle (1.5pt);
                           \node[color=\colii,below] (Nxyz) at (Axyz) {{\tiny ${(\x,\the\y)}$}}; } ; 
            }
          }
        }
        
	\coordinate (A) at ($ (0,0,-2) $);
	\coordinate (B) at  ($ (A) + (1,0,0) $);
	\coordinate (C) at  ($ (A) + (0,1,0) $);
	\coordinate (D) at  ($ (A) + (0,0,1) $);
	\draw[line width=1.0pt,rounded corners=0.5pt] (A) -- (B) -- (D) -- cycle;
	\draw[line width=1.0pt,rounded corners=0.5pt] (A) -- (C) -- (B) -- cycle;
	\draw[line width=1.0pt,rounded corners=0.5pt] (B) -- (C) -- (D) -- cycle;
	\draw[line width=1.6pt,rounded corners=0.5pt] (A) -- (C) -- (D) -- cycle;
        \coordinate (Axyz) at  ($ 1/4*(A) + 1/4*(B) + 1/4*(C) + 1/4*(D) $) ;
        \fill[color=\colii] (Axyz) circle (1.5pt);
        \node[color=\colii,above] (Nxyz) at (Axyz) {{\tiny ${(0,0,0)}$}};  

        \foreach \K in {1,...,\kkmax}  {
	  \coordinate (A) at ($ (0,0,-2) + 1.5*(0,\K,0) $);
	  \coordinate (B) at  ($ (A) + (1,0,0) $);
	  \coordinate (C) at  ($ (A) + (0,1,0) $);
	  \coordinate (D) at  ($ (A) + (0,0,1) $);
	  \draw[line width=1.0pt,rounded corners=0.5pt] (A) -- (B) -- (D) -- cycle;
	  \draw[line width=1.0pt,rounded corners=0.5pt] (A) -- (C) -- (B) -- cycle;
	  \draw[line width=1.0pt,rounded corners=0.5pt] (B) -- (C) -- (D) -- cycle;
	  \draw[line width=1.6pt,rounded corners=0.5pt] (A) -- (C) -- (D) -- cycle;

          \newcount\z
          \foreach \k in {0,1,...,\K}  {
            \foreach \x in {\k,...,0}  {
              \pgfmathparse{\k-\x}\let\YY\pgfmathresult
              \foreach \y in {0,...,\YY} {
                \pgfmathsetcount{\z}{\k-\x-\y} 
                \coordinate (Axyz) at  ($ (A) + 1/\K*(\x,\y,\z) $) ;
                \ifthenelse{\k<\K}
                           {\fill[color=\coliii] (Axyz) circle (1.5pt); }
                           {\fill[color=\colii] (Axyz) circle (1.5pt) ;
                             \node[color=\colii,below] (Nxyz) at  (Axyz) {{\tiny ${(\x,\y,\the\z)}$}}; };
              }
            }
          }
        }
  \end{tikzpicture}}
  \caption[Lagrange nodes of the reference simplex]{%
    Lagrange nodes of the reference simplex for $d=1,2,3$ (from top to
    bottom), and
    $k=0,1,2,3,4$ (from left to right).
    The nodes corresponding to the highest degree are depicted in red,
    and the others in blue.
    The multi-indices are indicated in all cases when $d=1,2$, and only for the
    highest degree when $d=3$.}
  \label{f:FE-lag-P1234}
\end{figure}
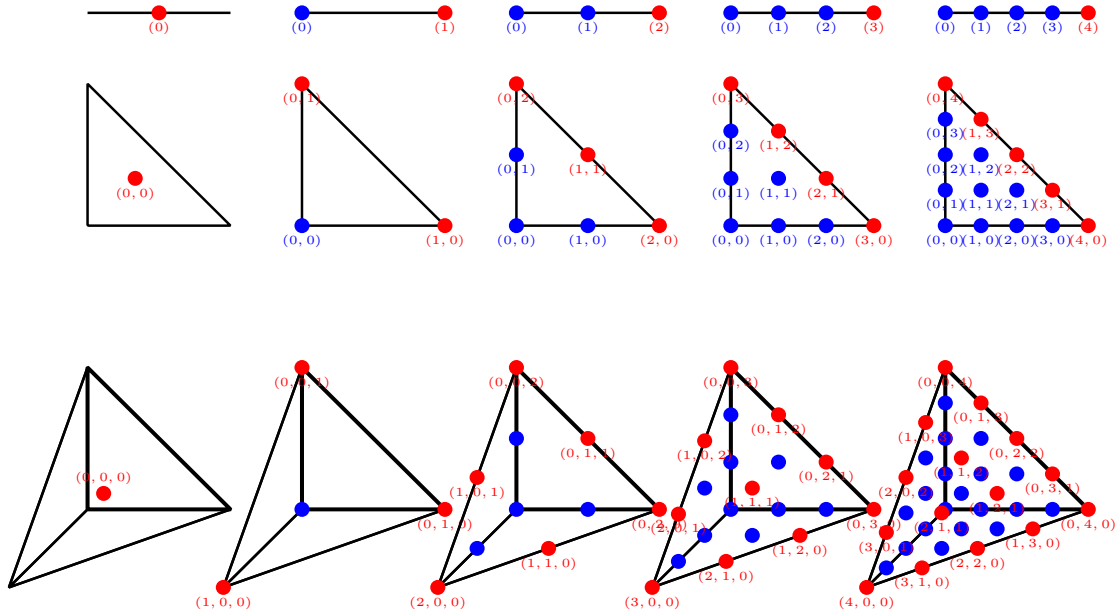

\section{Multi-indices}
\label{s:multi-indices}

\subsection{Some useful notations for multi-indices}
\label{ss:notations-mutli-indices}

\begin{definition}[length of multi-indices]
  \label{d:len-multi-ind}
  \mbox{}\hfill
  Let~$d\geq1$.
  Let~$\aalpha\in\matNd$ be a multi-index.\\
  The {\em length} of~$\aalpha$ is denoted~$\len{\aalpha}$, and is defined by
  $\len{\aalpha}\eqdef\sum_{i=1}^d\alpha_i\in\matN$.
\end{definition}

\begin{lemma}[length of multi-indices is additive]
  \label{l:len-multi-ind-is-add}
  \mbox{}\\
  Let~$d\geq1$.
  Then, the length of multi-indices is additive,
  for all $\aalpha,\bbeta\in\matNd$,
  $\len{\aalpha+\bbeta}=\len{\aalpha}+\len{\bbeta}$.
\end{lemma}

\begin{proof}
  Direct consequence of
  Definition~\thref{d:len-multi-ind}, and
  \assume{abelian monoid properties of~$\matN$}.
\end{proof}

\begin{definition}[factorial of multi-indices]
  \label{d:fact-multi-ind}
  \mbox{}\hfill
  Let~$d\geq1$.
  Let~$\aalpha\in\matNd$ be a multi-index.\\
  The {\em factorial} of~$\aalpha$ is denoted~$\ffact{\aalpha}$, and is defined
  by $\ffact{\aalpha}\eqdef\prod_{i=1}^d(\fact{\alpha_i})\in\matNstar$.
\end{definition}

\begin{lemma}[factorial of multi-index is positive]
  \label{l:fact-multi-ind-pos}
  \mbox{}\hfill
  Let~$d\geq1$.
  Let~$\aalpha\in\matNd$.
  Then, $\ffact{\aalpha}>0$.
\end{lemma}

\begin{proof}
  Direct consequence of
  Definition~\thref{d:fact-multi-ind},
  \assume{the positivity of factorial}, and
  \assume{closedness of multiplication in~$\matNstar$}.
\end{proof}

\begin{definition}[Kronecker delta of multi-indices]
  \label{d:kron-multi-ind}
  \mbox{}\hfill
  Let~$d\geq1$.
  Let~$\aalpha,\bbeta\in\matNd$.\\
  The {\em Kronecker delta} of~$\aalpha$ and~$\bbeta$ is
  denoted~$\kkron{\aalpha}{\bbeta}$, and is defined by
  $\kkron{\aalpha}{\bbeta}\eqdef\prod_{i=1}^d\kron{\alpha_i}{\beta_i}$.
\end{definition}

\begin{lemma}[value of Kronecker delta of multi-indices]
  \label{l:kron-multi-ind-val}
  \mbox{}\hfill
  Let~$d\geq1$.
  Let~$\aalpha,\bbeta\in\matNd$.\\
  Then, we have $\kkron{\aalpha}{\bbeta}$ equals~1 when $\aalpha=\bbeta$, and~0
  otherwise.
\end{lemma}

\begin{proof}
  Direct consequence of
  Definition~\thref{d:kron-multi-ind},
  \assume{the definition of the (scalar) Kronecker delta}, and
  \assume{the zero-product property in $\matN$}.
\end{proof}

\subsection{Sets $\calAkd$ and $\calCkd$ of multi-indices}
\label{ss:sets-Akd-Ckd-mutli-indices}

\begin{definition}[sets of multi-indices $\calAkd$ and $\calCkd$]
  \label{d:multi-ind-Akd-Ckd}
  \mbox{}\hfill
  Let~$d\geq1$.
  Let~$k\in\matN$.
  The {\em set of multi-indices of length at most~$k$ (resp. of length~$k$)} is
  denoted~$\calAkd$ (resp.~$\calCkd$), and is defined by
  \begin{align}
    \label{e:multi-ind-Akd-Ckd-1}
    \calAkd &\eqdef \{ \aalpha \in \matNd \st \len{\aalpha} \leq k \},\\
    \label{e:multi-ind-Akd-Ckd-2}
    \calCkd &\eqdef \{ \aalpha \in \matNd \st \len{\aalpha} = k \}.
  \end{align}
  Let~$i\in[1..d]$.
  The {\em subset of~$\calAkd$ of multi-indices with zero $i$-th component} is
  defined by
  \begin{equation}
    \label{e:multi-ind-Akd-Ckd-3}
    \calAkdi \eqdef \{ \aalpha \in \calAkd \st \alpha_i = 0 \}.
  \end{equation}
\end{definition}

\begin{lemma}[multi-indices $\calAkd$ for $d=1$ is $\calAki$]
  \label{l:multi-ind-Akd-for-d-eq-1-is-Ak1}
  \mbox{}\\
  Let~$k\in\matN$.
  Then, $\calAkd$ for~$d\eqdef1$ and~$\calAki$ from
  Definition~\ref{d:multi-ind-Ak1} coincide.
\end{lemma}

\begin{proof}
  Direct consequence of
  Definition~\threfc{d:multi-ind-Akd-Ckd}{%
    \eqref{e:multi-ind-Akd-Ckd-1} with $d\eqdef1$},
  Definition~\threfc{d:len-multi-ind}{%
    with $d\eqdef1$, thus $\len{\aalpha}=\alpha_1$}, and
  Definition~\thref{d:multi-ind-Ak1}.
\end{proof}

\begin{lemma}[indices are smaller than maximal length]
  \label{l:ind-smaller-than-max-len}
  \mbox{}\\
  Let~$d\geq1$.
  Let~$k\in\matN$.
  Let~$\aalpha\in\calAkd\cup\calCkd$.
  Let~$i\in[1..d]$.
  Then, we have $\alpha_i\leq k$.
\end{lemma}

\begin{proof}
  \mbox{}\\
  Direct consequence of
  Definition~\threfc{d:len-multi-ind}{thus $\alpha_i\leq\len{\aalpha}$},
  Definition~\threfc{d:multi-ind-Akd-Ckd}{thus $\len{\aalpha}\leq k$}, and
  \assume{the definition of order (transitivity)}.
\end{proof}

\begin{lemma}[first $\calCkd$]
  \label{l:first-Ckd}
  \mbox{}\hfill
  Let~$d\geq1$.
  Let~$k\in\matN$.
  Then, we have
  \begin{align}
    \label{e:first-multi-ind-Ckd-1}
    (d = 1)&&
    \calCki = \{ k \},
    &\AND \card (\calCki) = 1,&&\\
    \label{e:first-multi-ind-Ckd-2}
    (d = 2)&&
    \calCkii = \{ (k - i, i) \in \matN^2 \st i \in [0..k] \},
    &\AND \card (\calCkii) = k + 1,&&\\
    \label{e:first-multi-ind-Ckd-3}
    (k = 0)&&
    \calCod = \{ \zzero \in \matNd \},
    &\AND \card (\calCod) = 1,&&\\
    \label{e:first-multi-ind-Ckd-4}
    (k = 1)&&
    \calCid = \{ \ee_i \in \matNd \st i \in [1..d] \},
    &\AND \card (\calCid) = d.&&
  \end{align}
\end{lemma}

\begin{proof}
  Direct consequence of
  Definition~\thref{d:canon-fam},
  Definition~\thref{d:multi-ind-Akd-Ckd}, and
  Lemma~\thref{l:ind-smaller-than-max-len}.
\end{proof}

\subsection{Slices of $\calCkd$ and its cardinal}
\label{ss:slices-Ckd-cardinal}

\begin{remark}
  \label{r:tilde-and-check-notations}
  In the sequel, when $d\geq2$, the check notation~$\caalpha$ denotes the
  {\em last} $d-1$ components of the multi-index~$\aalpha\in\matNd$, and the
  tilde notation~$\taalpha$ denotes the {\em first} ones.\\
  Thus, for all $\aalpha\eqdef(\alpha_1,\ldots,\alpha_d)\in\matNd$, we have
  $\aalpha=(\alpha_1,\caalpha)=(\taalpha,\alpha_d)$.

  The notation naturally extends to vectors in~$\matRd$, and other objects such
  as subsets and functions may also inherit the check or tilde decoration,
  {\eg} see Definition~\ref{d:slices-Sckdi-and-Stkdi} and
  Lemma~\ref{l:card-slices-of-Ckd}.
\end{remark}

\begin{definition}[slices $\calSckdi$ and $\calStkdi$]
  \label{d:slices-Sckdi-and-Stkdi}
  \mbox{}\hfill
  Let~$d\geq2$.
  Let~$k\in\matN$.
  Let~$i\in[0..k]$.\\
  The {\em $i$-th ``vertical'' (resp. ``horizontal'') slice of~$\calCkd$}
  is denoted~$\calSckdi$ (resp.~$\calStkdi$), and is defined by
  \begin{equation}
    \label{e:slices-Sckdi-and-Stkdi}
    \calSckdi \eqdef
    \left\{
      (i, \caalpha) \in \matNd \st \caalpha \in \calCdmi{k-i}
    \right\},
    \AND
    \calStkdi \eqdef
    \left\{
      (\taalpha, i) \in \matNd \st \taalpha \in \calCdmi{k-i}
    \right\}.
  \end{equation}
\end{definition}

\begin{remark}
  Geometrically, the ``vertical'' slice~$\calSckdi$ is the intersection
  of~$\calCkd$ (subset in~$\matNd$ of the hyperplane of equation
  $\sum_{i=1}^d\alpha_i=k$) with the (vertical) first canonical
  hyperplane of equation~$\alpha_1=i$, see Figure~\ref{f:lag-k3-d3-slices}.
  The order discussed in Section~\ref{ss:discussion} is consistent with the
  decomposition of~$\calCkd$ in slices as $\calCkd=\biguplus_{i=k}^0\calSckdi$
  (decreasing from~$k$ to~0), see Lemma~\ref{l:slices-of-multi-ind-Ckd}.

  In the same way, the ``horizontal'' slice~$\calStkdi$ is the intersection
  of~$\calCkd$ with the (horizontal) last canonical hyperplane of
  equation~$\alpha_d=i$.
\end{remark}

\begin{figure}[htb]
  \centering
  \begin{tikzpicture}[scale=4,math3d] 

        \def\kk{3}

        \def\colk{black}
        \def\colo{magenta}
        \def\coli{darkgreen}
        \def\colii{red}
        \def\coliii{blue}

        \def\opacity{0.3}
        \def\opacityi{0.7}
        \def\opacityii{0.6}

	\coordinate (A) at (0,0,0);
	\coordinate (B) at  ($ (A) + (1,0,0) $);
	\coordinate (C) at  ($ (A) + (0,1,0) $);
	\coordinate (D) at  ($ (A) + (0,0,1) $);
        \draw (A) node[left] {$\hvv_0$} ;  
        \draw (B) node[below=11pt] {$\hvv_1=(1,0,0)$} ; 
        \draw (C) node[above=2pt] {$\hvv_2$} ; 
        \draw (D) node[right=1.5pt] {$\hvv_3=(0,0,1)$} ; 
	\draw[line width=1.0pt,rounded corners=0.5pt] (A) -- (B) -- (D) -- cycle;
	\draw[line width=1.0pt,rounded corners=0.5pt] (A) -- (C) -- (B) -- cycle;
	\draw[line width=1.0pt,rounded corners=0.5pt] (B) -- (C) -- (D) -- cycle;
	\draw[line width=1.6pt,rounded corners=0.5pt] (A) -- (C) -- (D) -- cycle;

        \node[color=\colk] (K3) at ($ (A) + (0,-0.3,0.7) $) {$\hK_{3}$};

        \pgfmathparse{1-1/\kk}\let\kkp\pgfmathresult
        \coordinate (B1) at ($ \kkp*(A) + 1/\kk*(B) $) ;
        \coordinate (C1) at ($ \kkp*(A) + 1/\kk*(C) $) ;
        \coordinate (D1) at ($ \kkp*(A) + 1/\kk*(D) $) ;
        \pgfmathparse{1-2/\kk}\let\kkp\pgfmathresult
        \coordinate (B2) at ($ \kkp*(A) + 2/\kk*(B) $) ;
        \coordinate (C2) at ($ \kkp*(A) + 2/\kk*(C) $) ;
        \coordinate (D2) at ($ \kkp*(A) + 2/\kk*(D) $) ;

        \coordinate (eps) at  ($ 1/\kk*(0,0.05,-0.05) $); 
	\coordinate (A210) at  ($ (A) + 1/\kk*(2,1,0) $);
	\coordinate (A201) at  ($ (A) + 1/\kk*(2,0,1) $);
	\coordinate (A120) at  ($ (A) + 1/\kk*(1,2,0) $);
	\coordinate (A102) at  ($ (A) + 1/\kk*(1,0,2) $);
        \coordinate (CD) at ($ 1/2*(C) + 1/2*(D) $);

        \draw[color=\colk,fill=\colo!10,fill opacity=\opacity]   (A) -- (C) -- (D) -- cycle;
        \draw[color=\colk,fill=\colo!10,fill opacity=\opacity]  (B1) -- (A120) -- (A102) -- cycle;
        \draw[color=\colk,fill=\colo!10,fill opacity=\opacity]  (B2) -- (A210) -- (A201) -- cycle;

        \fill[color=\colk,fill opacity=\opacityii] (A) circle (0.6pt);
        \foreach \x in {1,0}  {
          \pgfmathparse{1-\x}\let\YY\pgfmathresult
          \foreach \y in {0,...,\YY} {
            \pgfmathparse{1-\x-\y}\let\z\pgfmathresult
            \fill[color=\colk,fill opacity=\opacityii] ($ (A) + 1/\kk*(\x,\y,\z) $) circle (0.6pt);
          }
        }
        \foreach \x in {2,1,...,0}  {
          \pgfmathparse{2-\x}\let\YY\pgfmathresult
          \foreach \y in {0,...,\YY} {
            \pgfmathparse{2-\x-\y}\let\z\pgfmathresult
            \fill[color=\colk,fill opacity=\opacityii] ($ (A) + 1/\kk*(\x,\y,\z) $) circle (0.6pt);
          }
        }
        \draw[color=\colo,dashed,line width=2.5pt,<->,>=latex]  ($ (C) - (eps) $) -- ($ (D) + (eps) $);

        \draw[color=\colo,dashed,line width=2.5pt,<->,>=latex] ($ (A210) - (eps) $) -- ($ (A201) + (eps) $);
        \draw[color=\colo,dashed,line width=2.5pt,<->,>=latex] ($ (A120) - (eps) $) -- ($ (A102) + (eps) $);
        \newcount\z
        \foreach \x in {3,2,...,0}  {
          \pgfmathparse{3-\x}\let\YY\pgfmathresult
          \foreach \y in {0,...,\YY} {
            \pgfmathsetcount{\z}{3-\x-\y} 
            \coordinate (Axyz) at  ($ (A) + 1/\kk*(\x,\y,\z) $) ;
            \node[color=\coliii,below=-0.1] (Nxyz) at (Axyz) {$\haa_{(\x,\y,\the\z)}$};
            \fill[color=\coliii] (Axyz) circle (0.8pt);
          }
        }

        \draw (B) node[color=\colo,above left] {${\bf \calSckd{3}}$} ;  
        \draw (A201) node[color=\colo,above left] {${\bf \calSckd{2}}$} ; 
        \draw (A102) node[color=\colo,above left] {${\bf \calSckd{1}}$} ; 
        \draw (D) node[color=\colo,above left] {${\bf \calSckd{0}}$} ;

	\coordinate (A) at (0,1.55,0);
	\coordinate (B) at  ($ (A) + (1,0,0) $);
	\coordinate (C) at  ($ (A) + (0,1,0) $);
	\coordinate (D) at  ($ (A) + (0,0,1) $);
        \draw (A) node[left] {$\hvv_0$} ;  
        \draw (B) node[below=11pt] {$\hvv_1=(1,0,0)$} ; 
        \draw (C) node[above=2pt] {$\hvv_2$} ; 
        \draw (D) node[right=1.5pt] {$\hvv_3=(0,0,1)$} ; 
	\draw[line width=1.0pt,rounded corners=0.5pt] (A) -- (B) -- (D) -- cycle;
	\draw[line width=1.0pt,rounded corners=0.5pt] (A) -- (C) -- (B) -- cycle;
	\draw[line width=1.0pt,rounded corners=0.5pt] (B) -- (C) -- (D) -- cycle;
	\draw[line width=1.6pt,rounded corners=0.5pt] (A) -- (C) -- (D) -- cycle;

        \node[color=\colk] (K3) at ($ (A) + (0,-0.3,0.7) $) {$\hK_{3}$};

        \pgfmathparse{1-1/\kk}\let\kkp\pgfmathresult
        \coordinate (B1) at ($ \kkp*(A) + 1/\kk*(B) $) ;
        \coordinate (C1) at ($ \kkp*(A) + 1/\kk*(C) $) ;
        \coordinate (D1) at ($ \kkp*(A) + 1/\kk*(D) $) ;
        \pgfmathparse{1-2/\kk}\let\kkp\pgfmathresult
        \coordinate (B2) at ($ \kkp*(A) + 2/\kk*(B) $) ;
        \coordinate (C2) at ($ \kkp*(A) + 2/\kk*(C) $) ;
        \coordinate (D2) at ($ \kkp*(A) + 2/\kk*(D) $) ;

        \coordinate (eps) at  ($ 1/\kk*(0.05,-0.05,0) $); 
        \coordinate (BC) at ($ 1/2*(B) + 1/2*(C) $);
        \coordinate (A201) at  ($ (A) + 1/\kk*(2,0,1) $);
        \coordinate (A021) at  ($ (A) + 1/\kk*(0,2,1) $);
	\coordinate (A102) at  ($ (A) + 1/\kk*(1,0,2) $);
	\coordinate (A012) at  ($ (A) + 1/\kk*(0,1,2) $);

        \draw[color=\colk,fill=\coli!30,fill opacity=\opacity]   (A) -- (B) -- (C) -- cycle;
        \draw[color=\colk,fill=\coli!30,fill opacity=\opacity]  (D1) -- (A201) -- (A021) -- cycle;
        \draw[color=\colk,fill=\coli!30,fill opacity=\opacity]  (D2) -- (A102) -- (A012) -- cycle;

        \fill[color=\colk,fill opacity=\opacityii] (A) circle (0.6pt);
        \foreach \x in {1,0}  {
          \pgfmathparse{1-\x}\let\YY\pgfmathresult
          \foreach \y in {0,...,\YY} {
            \pgfmathparse{1-\x-\y}\let\z\pgfmathresult
            \fill[color=\colk,fill opacity=\opacityii] ($ (A) + 1/\kk*(\x,\y,\z) $) circle (0.6pt);
          }
        }
        \foreach \x in {2,1,...,0}  {
          \pgfmathparse{2-\x}\let\YY\pgfmathresult
          \foreach \y in {0,...,\YY} {
            \pgfmathparse{2-\x-\y}\let\z\pgfmathresult
            \fill[color=\colk,fill opacity=\opacityii] ($ (A) + 1/\kk*(\x,\y,\z) $) circle (0.6pt);
          }
        }
        \draw[color=\coli,dashed,line width=2.5pt,<->,>=latex]  ($ (B) - (eps) $) -- ($ (C) + (eps) $);

        \draw[color=\coli,dashed,line width=2.5pt,<->,>=latex] ($ (A201) - (eps) $) -- ($ (A021) + (eps) $);
        
        \draw[color=\coli,dashed,line width=2.5pt,<->,>=latex] ($ (A102) - (eps) $) -- ($ (A012) + (eps) $);
        \newcount\z
        \foreach \x in {3,2,...,0}  {
          \pgfmathparse{3-\x}\let\YY\pgfmathresult
          \foreach \y in {0,...,\YY} {
            \pgfmathsetcount{\z}{3-\x-\y} 
            \coordinate (Axyz) at  ($ (A) + 1/\kk*(\x,\y,\z) $) ;
            \node[color=\coliii,below right] (Nxyz) at (Axyz) {$\haa_{(\x,\y,\the\z)}$};
            \fill[color=\coliii] (Axyz) circle (0.8pt);
          }
        }

        \draw (C) node[color=\coli,above right] {${\bf \calStkd{0}}$} ;  
        \draw (A021) node[color=\coli,above right] {${\bf \calStkd{1}}$} ; 
        \draw (A012) node[color=\coli,above right] {${\bf \calStkd{2}}$} ; 
        \draw (D) node[color=\coli,above right] {${\bf \calStkd{3}}$} ;

\end{tikzpicture}
  \caption[Vertical and horizontal slices]{%
    Vertical slices~$\calSckdi$ (left) and horizontal slices~$\calStkdi$
    (right), in the case $d=k=3$ (see
    Definition~\ref{d:slices-Sckdi-and-Stkdi}).\\
    The reference Lagrange node~$\haa_{(\alpha_1,\alpha_2,\alpha_3)}$ in blue
    corresponds to the element $(\alpha_1,\alpha_2,\alpha_3)\in\calCkd$.
    For instance, the set $\calSckd{1}=\{(1,2,0),(1,1,1),(1,0,2)\}$ is depicted
    by the nodes linked by a dashed arrow.}
  \label{f:lag-k3-d3-slices}
\end{figure}
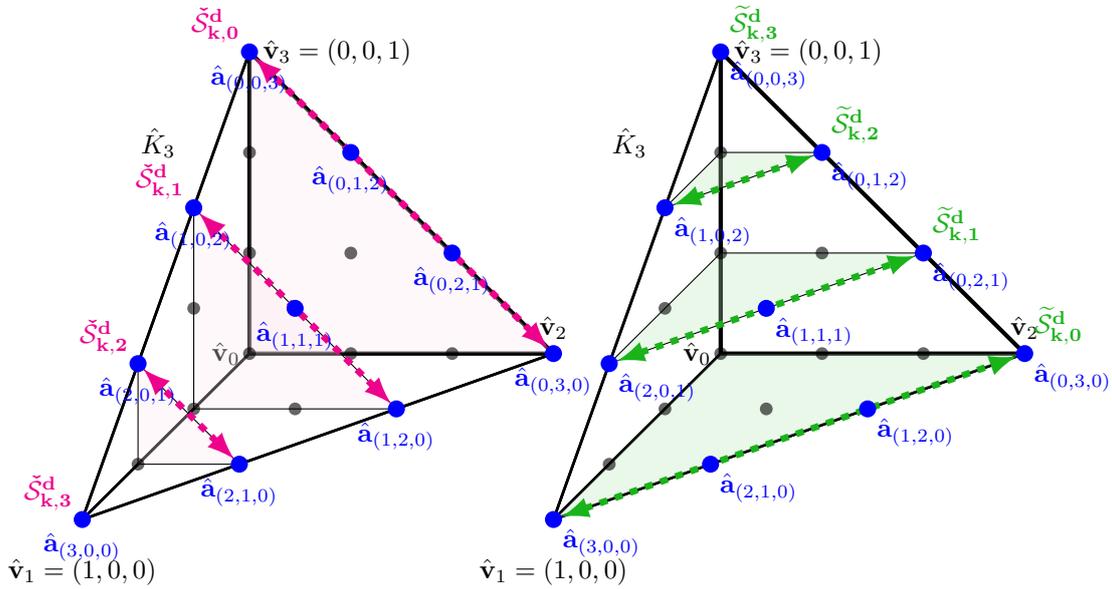

\begin{lemma}[slices of $\calCkd$]
  \label{l:slices-of-multi-ind-Ckd}
  \mbox{}\\
  Let~$d\geq 2$.
  Let~$k\in\matN$.
  Let $i\in[0..k]$.
  Then, $\calStkdi$ and~$\calSckdi$ are subsets of~$\calCkd$, and we have
  \begin{equation}
    \label{e:slices-of-multi-ind-Ckd}
    \calCkd = \biguplus_{i=0}^k \calSckdi = \biguplus_{i=0}^k \calStkdi.
  \end{equation}
\end{lemma}

\begin{proof}
  \proofparskip{Case of~$\calSckdi$}
  Let~$\aalpha\in\calSckdi$.
  Then, from
  Definition~\thref{d:slices-Sckdi-and-Stkdi}, and
  Definition~\thref{d:multi-ind-Akd-Ckd},
  there exists~$\caalpha\in\calCdmi{k-i}$ such that $\aalpha=(i,\caalpha)$,
  and we have $\len{\aalpha}=i+\len{\caalpha}=i+k-i=k$, {\ie}
  $\aalpha\in\calCkd$, and $\calSckdi\subset\calCkd$.
  Thus, we have $\bigcup_{i=0}^k\calSckdi\subset\calCkd$.

  Let $i,j\in[0..k]$.
  Let~$\aalpha\in\calSckdi\cap\calSckd{j}$.
  Then, from
  Definition~\thref{d:slices-Sckdi-and-Stkdi},
  there exist~$\caalpha^i\in\calCdmi{k-i}$ and~$\caalpha^j\in\calCdmi{k-j}$
  such that $\aalpha=(i,\caalpha_i)=(j,\caalpha_j)$.
  Thus, we have $i=j$ and $\caalpha_i=\caalpha_j$, and by contrapositive, we
  have $\calSckdi\cap\calSckd{j}=\emptyset$ when $i\neq j$.

  Let~$\aalpha\in\calCkd$.
  Let~$\caalpha\eqdef(\alpha_2,\ldots,\alpha_d)\in\matN^{d-1}$.
  Then, from
  Lemma~\thref{l:ind-smaller-than-max-len}, and
  Definition~\thref{d:multi-ind-Akd-Ckd},
  we have $0\leq\alpha_1\leq k$, $\caalpha\in\calCdmi{k-\alpha_1}$, and
  $\aalpha=(\alpha_1,\caalpha)\in\calSckd{\alpha_1}$.
  Thus, we have $\calCkd\subset\biguplus_{i=0}^k\calSckdi$.

  Therefore, we have $\calCkd = \biguplus_{i=0}^k \calSckdi$.

  \proofparskip{Case of~$\calStkdi$} The proof is very similar.
  Let~$\aalpha\in\calStkdi$.
  Then, from
  Definition~\thref{d:slices-Sckdi-and-Stkdi}, and
  Definition~\thref{d:multi-ind-Akd-Ckd},
  there exists~$\taalpha\in\calCdmi{k-i}$ such that $\aalpha=(\taalpha,i)$,
  and we have $\len{\aalpha}=\len{\taalpha}+i=k-i+i=k$, {\ie}
  $\aalpha\in\calCkd$, and $\calStkdi\subset\calCkd$.
  Thus, we have $\bigcup_{i=0}^k\calStkdi\subset\calCkd$.

  Let $i,j\in[0..k]$.
  Let~$\aalpha\in\calStkdi\cap\calStkd{j}$.
  Then, from
  Definition~\thref{d:slices-Sckdi-and-Stkdi},
  there exist~$\caalpha^i\in\calCdmi{k-i}$ and~$\caalpha^j\in\calCdmi{k-j}$
  such that $\aalpha=(\caalpha_i,i)=(\caalpha_j,j)$.
  Thus, we have $\caalpha_i=\caalpha_j$ and $i=j$, and by contrapositive, we
  have $\calStkdi\cap\calStkd{j}=\emptyset$ when $i\neq j$.

  Let~$\aalpha\in\calCkd$.
  Let~$\caalpha\eqdef(\alpha_1,\ldots,\alpha_{d-1})\in\matN^{d-1}$.
  Then, from
  Lemma~\thref{l:ind-smaller-than-max-len}, and
  Definition~\thref{d:multi-ind-Akd-Ckd},
  we have $0\leq\alpha_d\leq k$, $\taalpha\in\calCdmi{k-\alpha_d}$, and
  $\aalpha=(\taalpha,\alpha_d)\in\calStkd{\alpha_d}$.
  Thus, we have $\calCkd\subset\biguplus_{i=0}^k\calStkdi$.

  Therefore, we have $\calCkd=\biguplus_{i=0}^k\calStkdi$.
\end{proof}

\begin{lemma}[cardinal of slices of $\calCkd$]
  \label{l:card-slices-of-Ckd}
  \mbox{}\hfill
  Let~$d\geq2$.
  Let~$k\in\matN$.
  Let~$i\in[0..k]$.\\
  Let $\cfhi_{k,d}^i\eqdef
  (\aalpha=(i,\caalpha)\in\calSckdi\longmapsto\caalpha\in\calCdmi{k-i})$, and
  $\tfhi_{k,d}^i\eqdef
  (\aalpha=(\taalpha,i)\in\calStkdi\longmapsto\taalpha\in\calCdmi{k-i})$.\\
  Then, $\cfhi_{k,d}^i$ and $\tfhi_{k,d}^i$ are bijections, and we have
  $\card(\calSckdi)=\card(\calStkdi)=\card(\calCdmi{k-i})$.
\end{lemma}

\begin{proof}
  From
  Definition~\thref{d:slices-Sckdi-and-Stkdi},
  the applications~$\cfhi_{k,d}^i$ and~$\tfhi_{k,d}^i$ are well defined, and
  obviously surjective.
  Moreover, let $\aalpha,\bbeta\in\calSckdi$, such that
  $\caalpha=\cfhi_{k,d}^i(\aalpha)=\cfhi_{k,d}^i(\bbeta)=\cbbeta$.
  Thus, we have $(i,\caalpha)=(i,\cbbeta)$, and~$\cfhi_{k,d}^i$ is injective.
  In the same manner, let $\aalpha,\bbeta\in\calStkdi$, such that
  $\taalpha=\tfhi_{k,d}^i(\aalpha)$ equals $\tfhi_{k,d}^i(\bbeta)=\tbbeta$.
  Thus, we have $(\taalpha,i)=(\tbbeta,i)$, and~$\tfhi_{k,d}^i$ is injective.

  Therefore, $\cfhi_{k,d}^i$ and~$\tfhi_{k,d}^i$ are bijective, and we have the
  result.
\end{proof}

\begin{lemma}[cardinal of $\calCkd$]
  \label{l:card-Ckd}
  \mbox{}\hfill
  Let~$d\geq1$.
  Let~$k\in\matN$.
  Then, the number of elements of~$\calCkd$ is
  \begin{equation}
    \label{e:card-Ckd}
    \card(\calCkd)
    = \binom{k + d - 1}{d - 1} = \binom{k + d - 1}{k}
    = \frac{\fact{(k + d - 1)}}{\fact{k} \; \fact{(d - 1)}}.
  \end{equation}
\end{lemma}

\begin{proof}
  From
  Lemma~\threfc{l:prop-binom-coef}{\eqref{e:prop-binom-coef-2}}, and
  Definition~\thref{d:binom-coef},
  we have the last two equalities.

  \proofparskip{First equality}
  For all~$d\in\matNstar$, let
  $\PropPP(d)\eqdef[\card(\calCkd)=\binom{k+d-1}{d-1}]$.\\
  \proofpar{Induction:~$\PropPP(1)$}
  Direct consequence of
  Lemma~\threfc{l:first-Ckd}{\eqref{e:first-multi-ind-Ckd-1}}, and
  Lemma~\threfc{l:prop-binom-coef}{\eqref{e:prop-binom-coef-0}}.\\
  \proofpar{Induction: $\PropPP(d)\Implies \PropPP(d+1)$}\\
  Assume that~$\PropPP(d)$ holds.
  Then, from
  Lemma~\thref{l:slices-of-multi-ind-Ckd},
  Lemma~\thref{l:card-slices-of-Ckd}, and
  Lemma~\threfc{l:prop-binom-coef}{\eqref{e:prop-binom-coef-4}},
  we have
  \begin{align*}
    \card(\calCkdpi)
    &= \sum_{i = 0}^k \card (\calStkdpii)
    = \sum_{i = 0}^k \card (\calCd{k - i})\\
    &= \sum_{i = 0}^k \binom{k - i + d - 1}{d - 1}
    = \sum_{j = 0}^k \binom{j + d - 1}{d - 1}
    = \binom{k + d}{d}.
  \end{align*}

  \medskip\noindent
  This concludes the induction on~$d$, and we have for all $d\geq1$,
  $\PropPP(d)$.
\end{proof}

\subsection{Layers of $\calAkd$ and its cardinal}
\label{ss:layers-Akd-cardinal}

\begin{lemma}[$\calCkd$ are layers of $\calAkd$]
  \label{l:Ckd-layers-Akd}
  \mbox{}\hfill
  Let~$d\geq1$.
  Let~$k\in\matN$.
  Then, we have
  \begin{equation}
    \label{e:Ckd-layers-Akd}
    \calAod = \calCod
    \AND
    \calAkpid = \calAkd \uplus \calCkpid,
  \end{equation}
  thus, $\calAkd=\biguplus_{l=0}^k\calCld$ and the sequence
  $(\calAkd)_{k\in\matN}$ is increasing.
\end{lemma}

\begin{proof}
  Direct consequence of
  Definition~\thref{d:multi-ind-Akd-Ckd}, and
  \assume{the definition of (disjoint) union}.
\end{proof}

\begin{lemma}[first multi-indices $\calAkd$]
  \label{l:first-multi-ind-Akd}
  \mbox{}\hfill
  Let~$d\geq1$.
  Let~$k\in\matN$.
  Then, we have
  \begin{align}
    \label{e:first-multi-ind-Akd-0}
    (d = 1)&&
    \calAki = [0..k],
    &\AND \card (\calAki) = k + 1,&&\\
    \label{e:first-multi-ind-Akd-1}
    (k = 0)&&
    \calAod = \{ \zzero \in \matNd \},
    &\AND \card (\calAod) = 1,&&\\
    \label{e:first-multi-ind-Akd-2}
    (k = 1)&&
    \calAid = \{ \zzero, \ee_1, \ldots, \ee_d \},
    &\AND \card (\calAid) = d + 1.&&
  \end{align}
\end{lemma}

\begin{proof}
  \mbox{}\\
  Direct consequence of
  Definition~\thref{l:Ckd-layers-Akd}, and
  Definition~\thref{l:first-Ckd}.
\end{proof}

\begin{lemma}[cardinal of $\calAkd$]
  \label{l:card-Akd}
  \mbox{}\hfill
  Let~$d\geq1$.
  Let~$k\in\matN$.
  Then, the number of elements of~$\calAkd$ is
  \begin{equation}
    \label{e:card-Akd}
    \card(\calAkd)
    = \binom{k + d}{d} = \binom{k + d}{k}
    = \frac{\fact{(k + d)}}{\fact{k} \; \fact{d}}.
  \end{equation}
  By extension, for $d=0$, we also set
  $\card(\calAko)\eqdef\binom{k}{0}=\binom{k}{k}=1$.
\end{lemma}

\begin{proof}
  Direct consequence of
  Lemma~\threfc{l:Ckd-layers-Akd}{%
    thus, $\card(\calAkd)=\sum_{l=0}^k\card(\calCld)$},
  Lemma~\thref{l:card-Ckd},
  Lemma~\threfc{l:prop-binom-coef}{\eqref{e:prop-binom-coef-4}}, and
  Lemma~\threfc{l:prop-binom-coef}{%
    \eqref{e:prop-binom-coef-2}, \eqref{e:prop-binom-coef-0},
    thus last two equalities hold,
    and with $d=0$ too}.
\end{proof}

\subsection{Other cardinals}
\label{ss:other-cardinals}

\begin{remark}
  In Figure~\ref{f:fkdi_d32}, we plot examples of the mappings~$\fkdo$
  and~$\fkd{1}$, defined over multi-index sets in Lemmas~\ref{l:card-Ckd-Akdm1}
  and~\ref{l:card-Akdi-Akdm1}.

  These functions are used for the transfer of nodes from a $(d-1)$-simplex to
  the hyperface of a $d$-simplex, see Lemmas~\ref{l:face-hyperpl-lag-nodes-Pkd}
  and~\ref{l:im-nodes-by-geo-hyperface-mapping}.
  The transfer is made via the geometric applications~$\phindv{\thetinjdmi{i}}$
  (for $i\in[0..d]$), that map an hyperface to a simplex, see
  Lemma~\ref{l:geo-hyperface-mapping}, and Figure~\ref{f:geo-hyper-map_d32}.
\end{remark}

\begin{lemma}[cardinal of $\calCkd$ and $\calAkdmi$]
  \label{l:card-Ckd-Akdm1}
  \mbox{}\hfill
  Let~$d\geq1$.
  Let~$k\in\matN$.\\
  Let $\fkdo\eqdef
  (\caalpha\in\calAkdmi\longmapsto(k-\len{\caalpha},\caalpha)\in\calCkd)$, and
  $\tfkdo\eqdef
  (\taalpha\in\calAkdmi\longmapsto(\taalpha,k-\len{\taalpha})\in\calCkd)$.\\
  Then, $\fkdo$ and~$\tfkdo$ are bijections, and we have
  $\card(\calCkd)=\card(\calAkdmi)$.
\end{lemma}

\begin{proof}
  \proofparskip{Identity on the cardinals}\\
  Direct consequence of
  Lemma~\thref{l:card-Ckd}, and
  Lemma~\thref{l:card-Akd}.

  \proofparskip{Well-defined applications}
  Let~$\caalpha,\taalpha\in\calAkdmi$.
  Then, from
  Definition~\thref{d:multi-ind-Akd-Ckd}, and
  Definition~\thref{d:len-multi-ind},
  we have
  $k-\len{\caalpha},k-\len{\taalpha}\geq0$, and
  $\len{\fkdo(\caalpha)}=(k-\len{\caalpha})+\len{\caalpha}=k
  =\len{\taalpha}+(k-\len{\taalpha})=\len{\tfkdo(\taalpha)}$.
  Thus, $\fkdo$ and~$\tfkdo$ are well defined.

  \proofparskip{Bijections}
  Direct consequence of
  the definition of~$\fkdo$ and~$\tfkdo$ (obviously injective), and
  \assume{the fact that injectivity and cardinal equality imply bijectivity}.
\end{proof}

\begin{figure}[htb]
  \centering
  \resizebox{0.8\linewidth}{!}{
    \input{figtikz_fem_TriaToTetx2_k3}
  }
  \caption[Mappings~$\fkdo$ and~$\fkd{1}$]{%
    Mappings~$\fkdo$ and~$\fkd{1}$ in the case $d=k=3$ (see
    Lemmas~\ref{l:card-Ckd-Akdm1} and~\ref{l:card-Akdi-Akdm1}).\\
    The multi-indices in~$\calAiiidii$ are mapped to multi-indices
    of~$\calCiiidiii$ (for~$\fkdo$) or to multi-indices of~$\calAiiidiiione$
    (for~$\fkd{1}$).
    This is illustrated geometrically with the representation of the nodes in
    the triangle or the tetrahedra.\\
    For~$\fkdo$, the reference triangle nodes are mapped onto the nodes of the
    blue face of the tetrahedron.
    This face, opposite vertex~$\hvv_0$, contains the nodes having indices
    in~$\calCkd$.
    The coloring of the nodes is intended to help see the mapping:
    for all $(i,j)\in\calAiidiii$, we have $\fkdo(i,j)=(3-(i+j),i,j)$.\\
    For~$\fkd{1}$, the reference triangle nodes are mapped onto the
    nodes of the magenta face of the tetrahedron.
    This face, opposite vertex~$\hvv_1$, contains the nodes having indices
    in~$\calAkdone$.
    For all $(i,j)\in\calAiidiii$, we have $\fkd{1}(i,j)=(0,i,j)$.}
  \label{f:fkdi_d32}
\end{figure}

\begin{lemma}[cardinal of $\calAkdi$ and $\calAkdmi$]
  \label{l:card-Akdi-Akdm1}
  \mbox{}\hfill
  Let~$d\geq1$.
  Let~$k\in\matN$.
  Let~$i\in[1..d]$.\\
  Let $\fkdi\eqdef(\aalphap\in\calAkdmi\longmapsto
  (\alphap_1,\dots,\alphap_{i-1},0,\alphap_i,\dots,\alphap_{d-1})\in\calAkdi)$.\\
  Then, $\fkdi$ is a bijection that preserves length, and we have
  $\card(\calAkdi)=\card(\calAkdmi)$.
\end{lemma}

\begin{proof}
  \proofparskip{(1) Well-defined application and preservation of length}\\
  Let $\aalphap\in\calAkdmi$.
  Then, from
  Definition~\thref{d:len-multi-ind}, and
  Definition~\thref{d:multi-ind-Akd-Ckd},
  we have
  $\len{\fkdi(\aalphap)}=\len{\aalphap}\leq k$, and
  $\fkdi(\aalphap)\in\calAkdi$.

  \proofparskip{(2) Injectivity}
  Direct consequence of
  the definition of~$\fkdi$.

  \proofparskip{(3) Surjectivity}\\
  Let~$\aalpha\in\calAkdi$.
  Let~$\aalphap\eqdef
  (\alpha_1,\ldots,\alpha_{i-1},\alpha_{i+1},\ldots,\alpha_d)\in\matN^{d-1}$.
  Then, from~(1), and
  Definition~\thref{d:multi-ind-Akd-Ckd},
  we have $\fkdi(\aalphap)=\aalpha$ and
  $\len{\aalphap}=\len{\aalpha}\leq k$.
  Thus, $\aalphap\in\calAkdmi$.

  Therefore, from~(2), (3),
  \assume{the definition of bijectivity}, and
  \assume{the definition of cardinal},
  $\fkdi$~is bijective, and we have the equality of cardinals.
\end{proof}

\section{Multivariate polynomials}
\label{s:multivar-pol}

\begin{remark}
  See also univariate polynomials in Section~\ref{s:compl-univar-pol}.

  In the following statements, $\matPkd$ is defined as the space of polynomials
  of total degree at most~$k$.
  Note that the notion of degree is fully defined \emph{afterwards}.
\end{remark}

\subsection{Monomials and polynomials of $\matPkd$}
\label{ss:monom-Pkd}

\begin{definition}[monomial in $d$ variables]
  \label{d:monom-kd}
  \mbox{}\hfill
  Let~$d\geq1$.
  Let~$k\in\matN$.
  Let~$\aalpha\in\calCkd$.\\
  The {\em monomial of degree~$k$ in~$d$ variables of
    multi-exponent~$\aalpha$} is denoted~$\XX^\aalpha$, and is defined by
  \begin{equation}
    \label{e:monom-kd}
    \XX^\aalpha \eqdef
    \left(
      \xx \in \matRd \longmapsto \prod_{i = 1}^d x_i^{\alpha_i} \in \matR
    \right).
  \end{equation}

  Moreover, 1~is a shortcut for~$\XX^\zzero$ (the constant function of value~1,
  which may be omitted in a multiplicative context), and for all $i\in[1..d]$,
  $X_i$ is a shortcut for~$\XX^{\ee_i}=(\xx\mapsto x_i)$.
\end{definition}

\begin{lemma}[monomial in $d$ variables for $d=1$ is monomial of a single
  variable]
  \label{l:monom-kd-for-d-eq-1-is-monom-k1}
  \mbox{}\\
  Let~$k\in\matN$.
  Then, $\XX^\aalpha$ for $\aalpha\in\calCkd$ and $d\eqdef1$, and~$X^k$
  from Definition~\ref{d:monom-k1} coincide.
\end{lemma}

\begin{proof}
  \mbox{}\\
  Direct consequence of
  Definition~\threfc{d:monom-kd}{%
    \eqref{e:monom-kd} with $d\eqdef1$},
  Lemma~\threfc{l:first-Ckd}{with $d\eqdef1$}, and
  Definition~\threfc{d:monom-k1}{with $\alpha\eqdef k$}.
\end{proof}

\begin{definition}[polynomial space $\matPkd$]
  \label{d:pol-space-Pkd}
  \mbox{}\hfill
  Let~$d\geq1$.
  Let~$k\in\matN$.\\
  The {\em space of polynomials of degree at most~$k$ of~$d$ variables} is
  denoted~$\matPkd$, and is defined by
  \begin{equation}
    \label{e:pol-space-Pkd}
    \matPkd \eqdef \Span{\XX^\aalpha}_{\aalpha \in \calAkd}
    = \left\{
      \left( \xx \longmapsto
        \sum_{\aalpha \in \calAkd}
        a_{\aalpha} x_1^{\alpha_1} \ldots x_d^{\alpha_d}
      \right) : \ArRdR
    \rightst
    \left.
      \vphantom{\sum_{\aalpha \in \calAkd}
        a_\aalpha x_1^{\alpha_1} \ldots x_d^{\alpha_d}}
      (a_\aalpha)_{\aalpha \in \calAkd} \in \matR
    \right\}.
  \end{equation}
\end{definition}

\begin{lemma}[polynomial space $\matPkd$ for $d=1$ is $\matPki$]
  \label{l:pol-space-Pkd-for-d-eq-1-is-Pk1}
  \mbox{}\\
  Let~$k\in\matN$.
  Then, $\matPkd$ for $d\eqdef1$, and~$\matPki$ from
  Definition~\ref{d:pol-space-Pk1} coincide.
\end{lemma}

\begin{proof}
  Direct consequence of
  Definition~\threfc{d:pol-space-Pkd}{\eqref{e:pol-space-Pkd} with $d\eqdef1$},
  Lem\-ma~\threfc{l:first-multi-ind-Akd}{with $d\eqdef1$}, and
  Definition~\thref{d:pol-space-Pk1}.
\end{proof}

\begin{lemma}[$\matPkd$ is {\vectorspace}]
  \label{l:Pkd-sp}
  \mbox{}\hfill
  Let~$d\geq1$.
  Let~$k\in\matN$.
  Then, $\matPkd$ is a {\vectorspace}.
\end{lemma}

\begin{proof}
  Direct consequence of
  Definition~\thref{d:pol-space-Pkd}, and
  \assume{the definition of the linear span}.
\end{proof}

\begin{lemma}[$\matPkd$ is nondecreasing sequence in $k$]
  \label{l:Pkd-nondecr-k}
  \mbox{}\\
  Let~$d\geq1$.
  Then, the sequence $(\matPkd)_{k\in\matN}$ is nondecreasing for the inclusion.
\end{lemma}

\begin{proof}
  Direct consequence of
  Definition~\thref{d:pol-space-Pkd},
  Lemma~\threfc{l:Ckd-layers-Akd}{increasing sequence}, and
  \assume{monotonicity of the linear span}.
\end{proof}

\begin{lemma}[constant and affine spaces $\matPod$ and  $\matPid$]
  \label{l:pol-space-P0d-P1d}
  \mbox{}\\
  Let~$d\geq1$.
  The spaces of polynomials of~$d$ variables of degree at most~0 and~1 are
  respectively the spaces of constant and affine maps,
  \begin{align}
    \label{e:pol-space-P0d-P1d-1}
    \matPod
    &= \Span{1}
      = \left\{ (\xx \longmapsto a_\zzero) \st a_\zzero \in \matR \right\},\\
    \label{e:pol-space-P0d-P1d-2}
    \matPid
    &= \Span{1, X_1, X_2, \ldots, X_d}\\
    \nonumber
    &= \left\{
      (\xx \longmapsto a_\zzero + a_{\ee_1} x_1 + \dots + a_{\ee_d} x_d)
      \st a_\zzero, a_{\ee_1}, \ldots, a_{\ee_d} \in \matR
    \right\}\\
    \nonumber
    &= \AffRdR.
  \end{align}
  Thus, $p\in\matPid$ iff
  there exists $a_\zzero\in\matR$ and $A\in\calM_{1,d}(\matR)$
  such that,
  for all $\xx\in\matRd$, $p(\xx)=a_\zzero+A\xx$.
\end{lemma}

\begin{proof}
  Direct consequence of
  Definition~\thref{d:pol-space-Pkd},
  Definition~\thref{d:canon-fam},
  Lemma~\thref{l:first-multi-ind-Akd},
  Lemma~\thref{l:equiv-def-aff-map-finite-dim}.
\end{proof}

\begin{definition}[degree of polynomial]
  \label{d:deg-pol}
  \mbox{}\hfill
  Let~$d\geq1$.
  Let~$k\in\matN$.
  Let~$(a_\aalpha)_{\aalpha\in\calAkd}\in\matR$.\\
  Let $p\eqdef\sum_{\aalpha\in\calAkd}a_\aalpha\XX^\aalpha\in\matPkd$.
  The {\em degree of~$p$} is denoted~$\deg p$, and is defined by
  \begin{equation}
    \label{e:deg-pol}
    \deg p \eqdef
    \max (\len{\aalpha})_{a_\aalpha \neq 0},
  \end{equation}
  with the convention that the maximum of an empty family is~$-\infty$, {\ie}
  when~$p\eqdef0$.
\end{definition}

\begin{lemma}[values of degree of polymial]
  \label{l:val-deg-pol}
  \mbox{}\\
  Let~$d\geq1$.
  Let~$k\in\matN$.
  Let~$p\in\matPkd$.
  Then, we have $\deg p\in\{-\infty\}\cup[0..k]$.
\end{lemma}

\begin{proof}
  Direct consequence of
  Definition~\thref{d:deg-pol}.
\end{proof}

\begin{lemma}[monomials of $\calCkd$ have degree $k$]
  \label{l:deg-monom-Ckd-is-k}
  \mbox{}\\
  Let~$d\geq1$.
  Let~$k\in\matN$.
  Let~$\aalpha\in\calCkd$.
  Then, we have $\deg\XX^\aalpha=\len{\aalpha}=k$.
\end{lemma}

\begin{proof}
  Direct consequence of
  Definition~\thref{d:deg-pol},
  Definition~\thref{d:monom-kd},
  Definition~\thref{d:multi-ind-Akd-Ckd}.
\end{proof}

\begin{lemma}[$\matPkd$ is space of degree at most $k$]
  \label{l:deg-Pkd-leq-k}
  \mbox{}\hfill
  Let~$d\geq1$.\\
  Let~$\calP(\matRd)$ be the infinite-dimensional space of polynomials of~$d$
  variables on~$\matR$.
  Let~$k\in\matN$.\\
  Then, $\matPkd$ is the space of polynomials of degree at most $k$,
  \begin{equation}
    \label{e:deg-Pkd-leq-k}
    \matPkd = \left\{ p \in \calP(\matRd) \st \deg p \leq k \right\}.
  \end{equation}
\end{lemma}

\begin{proof}
  Direct consequence of
  Definition~\thref{d:pol-space-Pkd},
  Definition~\thref{d:deg-pol}, and
  Lemma~\thref{l:Ckd-layers-Akd}.
\end{proof}

\subsection{Product of polynomials}
\label{ss:prod-polynom}

\begin{lemma}[product of monomials]
  \label{l:prod-monom}
  \mbox{}\\
  Let~$d\geq1$.
  Let~$k,l\in\matN$.
  Let~$\aalpha\in\calCkd$ and~$\bbeta\in\calCld$.
  Then, we have
  \begin{equation}
    \label{e:prod-monom}
    \XX^\aalpha \XX^\bbeta =  \XX^{\aalpha + \bbeta},
    \AND
    \deg (\XX^\aalpha \XX^\bbeta)
    = \deg \XX^\aalpha + \deg \XX^\bbeta
    = \len{\aalpha} + \len{\bbeta} = k + l.
  \end{equation}
\end{lemma}

\begin{proof}
  Direct consequence of
  Definition~\thref{d:monom-kd},
  \assume{commutative ring properties of~$\matR$
    (thus $\prod_{i=1}^dx_i^{\alpha_i}\times\prod_{i=1}^dx_i^{\beta_i}=
    \prod_{i=1}^dx_i^{\alpha_i}x_i^{\beta_i}$)},
  Lemma~\thref{l:deg-monom-Ckd-is-k}, and
  Lemma~\thref{l:len-multi-ind-is-add}.
\end{proof}

\begin{lemma}[product of monomial and polynomial]
  \label{l:prod-monom-polynom}
  \mbox{}\\
  Let~$d\geq1$.
  Let~$k,l\in\matN$.
  Let~$\aalpha\in\calAkd$.
  Let~$q\in\matPld$.
  Then, we have $\XX^\aalpha\,q\in\matPkpld$.
\end{lemma}

\begin{proof}
  From
  Definition~\thref{d:pol-space-Pkd},
  \assume{commutative ring properties of~$\matR$},
  Lem\-ma~\thref{l:prod-monom}, and
  Definition~\threfc{d:multi-ind-Akd-Ckd}{%
    thus $\len{\aalpha}\leq k$ and $\len{\bbeta}\leq l$},
  there exists~$(b_\bbeta)_{\bbeta\in\calAld}\in\matR$ such that
  \[
  \XX^\aalpha \, q
  = \XX^\aalpha \left( \sum_{\bbeta \in \calAld} b_\bbeta \XX^\bbeta \right)
  = \sum_{\bbeta \in \calAld} b_\bbeta \XX^{\aalpha + \bbeta},
  \AND
  \len{\aalpha + \bbeta} \leq k + l.
  \]
  Thus, from
  Definition~\thref{d:pol-space-Pkd},
  we have $\XX^\aalpha\,q\in\matPkpld$.
\end{proof}

\begin{lemma}[product of two polynomials]
  \label{l:prod-2-polynom}
  \mbox{}\\
  Let~$d\geq1$.
  Let~$k,l\in\matN$.
  Let~$p\in\matPkd$ and $q\in\matPld$.
  Then, we have $pq\in\matPkpld$.
\end{lemma}

\begin{proof}
  From
  Definition~\thref{d:pol-space-Pkd}, and
  \assume{ring properties of $\matR$},
  there exists coefficients~$(a_\aalpha)_{\aalpha\in\calAkd}\in\matR$ such that
  $pq=\sum_{\aalpha\in\calAkd}a_\aalpha(\XX^\aalpha\,q)$.
  Thus, from
  Lemma~\thref{l:prod-monom-polynom}, and
  Lemma~\thref{l:Pkd-sp},
  we have $pq\in\matPkpld$.
\end{proof}

\subsection{Linear independence of monomials}
\label{ss:free-monom}

\begin{lemma}[partial derivative of monomials]
  \label{l:pder-monom}
  \mbox{}\hfill
  Let~$d\geq1$.
  Let~$k,l\in\matN$.\\
  Let $\aalpha\in\calAkd$ and $\bbeta\in\calAld$.
  Then, the partial derivative of order~$\bbeta$ of the
  monomial~$\XX^\aalpha$ is
  \begin{equation}
    \label{e:pder-monom}
    \pder^\bbeta \XX^\aalpha
    = \left\{
      \begin{array}{ll}
        \dps
        \prod_{i = 1}^d \left(
          \prod_{j = 0}^{\beta_i - 1} (\alpha_i - j)
        \right) X_i^{\alpha_i - \beta_i}
        & \mbox{when } \forall i \in [1..d],\; \beta_i \leq \alpha_i,\\
        0 & \mbox{otherwise}.
      \end{array}
    \right.
  \end{equation}
  Thus, $\pder^\bbeta\XX^\aalpha\in\matPkd$.

  Moreover, if for all $i\in[1..d]$, $\beta_i\leq\alpha_i$, then we have
  $\deg(\pder^\bbeta\XX^\aalpha)=\len{\aalpha}-\len{\bbeta}(=k-l)$.
\end{lemma}

\begin{proof}
  Direct consequence of
  Definition~\thref{d:monom-kd},
  Lemma~\thref{l:deg-monom-Ckd-is-k}, and
  \assume{the differentiation rules} for univariate polynomials,
    \[
      \forall i \in [1..d],\quad
      \frac{\p^{\beta_i} (x_i^{\alpha_i})}{\p x_i^{\beta_i}}
      = \left\{
        \begin{array}{ll}
          x_i^{\alpha_i}
          & \mbox{when } \beta_i = 0,\\
          \dps \left(
            \prod_{j = 0}^{\beta_i - 1} (\alpha_i - j)
          \right) X_i^{\alpha_i - \beta_i}
          & \mbox{when } 0 < \beta_i \leq \alpha_i,\\
          0 & \mbox{when } \alpha_i < \beta_i.
        \end{array}
      \right.
    \]
\end{proof}

\begin{remark}
  \mbox{}\\
  Note that in the first clause of~\eqref{e:pder-monom}, we use the convention
  that a product indexed by the empty set equals~1, the identity
  element for multiplication, {\ie} when some $\beta_i$ is zero.

  Note also that the second clause may be omitted, provided the use of the
  convention that subtraction is closed in~$\matN$ ({\ie} $n-p=0$ when $n<p$).
  Indeed, when there exists some~$i\in[1..d]$ such that $\alpha_i<\beta_i$,
  then $x_i^{\alpha_i-\beta_i}=1$ and $\prod_{j=0}^{\beta_i-1}(\alpha_i-j)=0$.
\end{remark}

\begin{lemma}[partial derivative is linear]
  \label{l:pder-is-lin}
  \mbox{}\hfill
  Let~$d\geq1$.
  Let $k,l\in\matN$.
  Let $\bbeta\in\calAld$.\\
  Then, $\pder^\bbeta\in\calL(\matPkd,\matPkd)$, {\ie} $\pder^\bbeta$ is linear
  from~$\matPkd$ to~$\matPkd$, and for all
  $(a_\aalpha)_{\aalpha\in\calAkd}\in\matR$, we have
  \begin{equation}
    \label{e:pder-is-lin}
    \pder^\bbeta \left(
      \sum_{\aalpha \in \calAkd} a_\aalpha \XX^\aalpha \right)
    = \sum_{\aalpha \in \calAkd} a_\aalpha \pder^\bbeta \XX^\aalpha
    = \sum_{\substack{\aalpha \in \calAkd\\
        \forall i \in [1..d],\ \beta_i \leq \alpha_i}}
      a_\aalpha \prod_{i = 1}^d \left(
        \prod_{j = 0}^{\beta_i - 1} (\alpha_i - j)
      \right) X_i^{\alpha_i - \beta_i}.
    \end{equation}
\end{lemma}

\begin{proof}
  Direct consequence of
  \assume{the linearity of partial derivative},
  Definition~\thref{d:pol-space-Pkd}, and
  Lemma~\thref{l:pder-monom}.
\end{proof}

\begin{lemma}[partial derivative of 0]
  \label{l:pder-0}
  \mbox{}\\
  Let~$d\geq1$.
  Let~$l\in\matN$.
  Let~$\bbeta\in\calAld$.
  Then, we have $\pder^\bbeta0=0$.
\end{lemma}

\begin{proof}
  Direct consequence of
  Lemma~\thref{l:pder-is-lin}.
\end{proof}

\begin{lemma}[derivating more than degree is 0]
  \label{l:pdef-more-than-deg-is-0}
  \mbox{}\\
  Let~$d\geq1$.
  Let~$k,l\in\matN$.
  Let $\aalpha\in\calCkd$ and $\bbeta\in\calCld$.
  Assume that $k<l$.
  Then, we have $\pder^\bbeta\XX^\aalpha=0$.
\end{lemma}

\begin{proof}
  \proofpar{Case $\forall i\in[1..d],\ \beta_i\leq\alpha_i$}
  Then, from
  Definition~\thref{d:multi-ind-Akd-Ckd}, and
  \assume{the monotonicity of addition in $\matN$},
  we have $l=\len{\bbeta}\leq\len{\aalpha}=k$, which is impossible.

  \proofparskip{Case $\exists i\in[1..d],\ \alpha_i<\beta_i$}
  Direct consequence of
  Lemma~\thref{l:pder-monom}.
\end{proof}

\begin{lemma}[partial derivative of monomials at 0]
  \label{l:pder-monom-at-0}
  \mbox{}\\
  Let~$d\geq1$.
  Let~$k,l\in\matN$.
  Let $\aalpha\in\calAkd$ and $\bbeta\in\calAld$.
  Then, we have
  $\pder^\bbeta\XX^\aalpha(\zzero)=\ffact{\aalpha}\;\kkron{\aalpha}{\bbeta}$.
\end{lemma}

\begin{proof}
  From
  Definition~\thref{d:fact-multi-ind}, and
  Definition~\thref{d:kron-multi-ind},
  we have
  $\ffact{\aalpha}\;\kkron{\aalpha}{\bbeta}=\prod_{i=1}^d\fact{\alpha_i}$ when
  $\aalpha=\bbeta$, and~0 otherwise.

  \proofparskip{Case~$\aalpha=\bbeta$}
  Direct consequence of
  Lemma~\threfc{l:pder-monom}{%
    since for all $i\in[1..d]$, $x_i^{\alpha_i-\beta_i}=1$ and
    $\prod_{j=0}^{\beta_i-1}(\alpha_i-j)=\fact{\alpha_i}$}.

  \proofparskip{Case~$\exists i\in[1..d],\ \alpha_i\neq\beta_i$}\\
  \proofpar{Case~$\alpha_i<\beta_i$}
  Direct consequence of
  Lemma~\threfc{l:pder-monom}{$\pder^{\bbeta} \XX^\aalpha=0$}.\\
  \proofpar{Case~$\beta_i<\alpha_i$}
  Direct consequence of
  Lemma~\threfc{l:pder-monom}{%
    since $x_i^{\alpha_i-\beta_i}=0$ when $x_i\eqdef0$}.
\end{proof}

\begin{lemma}[monomials are free in $\matPkd$]
  \label{l:monom-free-in-Pkd}
  \mbox{}\\
  Let~$d\geq1$.
  Let~$k\in\matN$.
  Then, $(\XX^\aalpha)_{\aalpha\in\calAkd}$ is free in~$\matPkd$.
\end{lemma}

\begin{proof}
  Let $(\lambda_\aalpha)_{\aalpha\in\calAkd}\in\matR$ such that
  $\sum_{\aalpha\in\calAkd}\lambda_\aalpha\XX^\aalpha=0$.
  Let~$\bbeta\in\calAkd$.\\
  Then, from
  Lemma~\thref{l:pder-is-lin},
  Lemma~\thref{l:pder-0}, and
  Lemma~\threfc{l:pder-monom-at-0}{with $\xx\eqdef\zzero$},
  we have
  \[
    0 = \sum_{\aalpha \in \calAkd}
      \lambda_\aalpha \pder^\bbeta \XX^\aalpha (\zzero)
    = \sum_{\aalpha \in \calAkd}
      \lambda_\aalpha \ffact{\bbeta} \; \kkron{\aalpha}{\bbeta}
    = \lambda_\bbeta \ffact{\bbeta}.
  \]
  Thus, from
  Lemma~\thref{l:fact-multi-ind-pos},
  \assume{the zero-product property in~$\matR$},
  Definition~\threfc{d:pol-space-Pkd}{thus $\XX^\aalpha\in\matPkd$}, and
  \assume{the definition of freedom},
  we have $\lambda_\bbeta=0$, and $(\XX^\aalpha)_{\aalpha\in\calAkd}$ is free
  in~$\matPkd$.
\end{proof}

\begin{lemma}[monomials are a basis of $\matPkd$]
  \label{l:monom-basis-Pkd}
  \mbox{}\\
  Let~$d\geq1$.
  Let~$k\in\matN$.
  Then, $(\XX^\aalpha)_{\aalpha\in\calAkd}$ is a basis of~$\matPkd$.
\end{lemma}

\begin{proof}
  Direct consequence of
  Definition~\threfc{d:pol-space-Pkd}{thus monomials are genera\-tors},
  Lemma~\thref{l:monom-free-in-Pkd}, and
  \assume{the definition of basis}.
\end{proof}

\begin{lemma}[dimension of $\matPkd$]
  \label{l:dim-Pkd}
  \mbox{}\\
  Let~$d\geq1$.
  Let~$k\in\matN$.
  Then, $\matPkd$ is a {\vectorspace} of dimension
  \begin{equation}
    \label{e:dim-Pkd-1}
    \dim \matPkd
    = \binom{k + d}{d} = \binom{k + d}{k}
    = \frac{\fact{(k + d)}}{\fact{k} \; \fact{d}}
    \qquad (=\card(\calAkd)).
  \end{equation}

  For instance, we have
  \begin{equation}
    \label{e:dim-Pkd-2}
    \left.\begin{array}{ll}
       \left\{
      \begin{array}{l}
        \dim \matPod = 1,\\
        \dim \matPid = d + 1,\\
        \dim \matPiid = \frac12 (d + 1) (d + 2),\\
      \end{array}
    \right. \mbox{ and }
       \left\{
      \begin{array}{l}
        \dim \matPki = (k + 1),\\
        \dim \matPkii = \frac12 (k + 1) (k + 2),\\
        \dim \matPkiii = \frac16 (k + 1) (k + 2) (k + 3).
      \end{array}
    \right.
    \end{array}\right.
  \end{equation}
\end{lemma}

\begin{proof}
  Direct consequence of
  Lemma~\thref{l:monom-basis-Pkd},
  \assume{the definition of the dimension}, and
  Lemma~\thref{l:card-Akd}.
\end{proof}

\subsection{Decomposition of polynomials, isomorphisms of $\matPkd$}
\label{ss:decomp-polynom-isomporph}

\begin{lemma}[isomorphism between $\matPod$ and $\matPodmi$]
  \label{l:isom-P0d-P0dm1}
  \mbox{}\\
  Let~$d\geq2$.
  Let $\xi^d\eqdef
  ((\xx\mapsto a_0)\in\matPod\longmapsto(\txx\mapsto a_0)\in\matPodmi)$.\\
  Then, for all~$p\in\matPod$, for all~$\xx\in\matRd$, we have
  $\xi^d(p)(\txx)=p(\xx)$, and~$\xi^d$ is an isomorphism.
\end{lemma}

\begin{proof}
  Direct consequence of
  Lemma~\threfc{l:pol-space-P0d-P1d}{%
    thus $\xi^d$~is well-defined, and the identity holds},
  Lemma~\threfc{l:dim-Pkd}{$\dim\matPod=\dim\matPodmi=1$},
  \assume{the definition of order (reflexivity)},
  Lemma~\threfc{l:inj-or-surj-and-dim-implies-bij}{%
    $\xi^d$ is obviously linear and injective}, and
  Definition~\thref{LM-d:isomorphism}.
\end{proof}

\begin{remark}
  \mbox{}\\
  Note in the next lemma the use of the tilde notation of
  Remark~\ref{r:tilde-and-check-notations}, such that $\xx=(\txx,x_d)$.

  Note also that the decomposition extracts~$x_d$, and thus uses the
  horizontal slices~$\calStkdi$.
  A similar result, exhibiting~$x_1$ and using vertical slices~$\calSckdi$,
  could be written, $p(\xx)=\check{p}_0(\check{\xx})+x_1\,p_1(\xx)$, with
  $\xx=(x_1,\check{\xx})$ and obvious modifications in the notations.
\end{remark}

\begin{remark}
  See the sketch of the next proof in
  Section~\ref{s:sketch-of-the-proof-of-Euclid-div-Pkd}.
\end{remark}

\begin{lemma}[decomposition of $\matPkd$]
  \label{l:decomp-Pkd}
  \mbox{}\hfill
  Let~$d\geq2$.
  Let~$k\geq1$.
  Let $p\in\matPkd$.\\
  Then, there exist unique $\tp_0\in\matPkdmi$ and $p_1\in\matPkmid$ such that
  \begin{equation}
    \label{e:decomp-Pkd}
    \forall \xx \in \matRd,\quad
    p (\xx) = \tp_0 (\txx) + x_d \, p_1 (\xx)
    \qquad (\mbox{{\ie} } p = \tp_0 + X_d \, p_1),
  \end{equation}
  which also writes
  $p (x_1, \ldots, x_d)
    = \tp_0 (x_1, \ldots, x_{d - 1}) + x_d \, p_1 (x_1, \ldots, x_d)$.
\end{lemma}

\begin{proof}
  \proofpar{Existence}
  For all $k\geq1$, let
  $P(k)\eqdef
  [\forall p\in\matPkd,\exists\tp_0\in\matPkdmi,\exists p_1\in\matPkmid,
  p=\tp_0+\XX_d\,p_1]$.\\
  \proofpar{Induction:~$P(1)$}
  Let~$p\in\matPid$.
  Then, from
  Lemma~\thref{l:pol-space-P0d-P1d}, and
  \assume{commutative ring properties of~$\matR$}
  there exist
  \[
    a_\zzero, (a_{\ee_i})_{i \in [1..d]} \in \matR,\quad
    \tp_0 \eqdef a_\zzero + \sum_{i = 1}^{d - 1} a_{\ee_i} X_i
    \in \matPidmi,
    \AND
    p_1 \eqdef a_{\ee_d} \in \matPod,
  \]
  such that
  $p=a_\zzero+\sum_{i=1}^da_{\ee_i}X_i=\tp_0+X_d\,p_1$.\\
  \proofpar{Induction: $P(k)\Implies P(k+1)$}
  Assume that~$P(k)$ holds.
  Let $p\in\matPkpid$.\\
  Then, from
  Definition~\thref{d:pol-space-Pkd}, and
  Lemma~\thref{l:Ckd-layers-Akd},
  there exist $(a_\aalpha)_{\aalpha\in\calAkpid}\in\matR$,
  $q\eqdef\sum_{\aalpha\in\calAkd}a_\aalpha\XX^\aalpha\in\matPkd$, and
  $r\eqdef\sum_{\aalpha\in\calCkpid}a_\aalpha\XX^\aalpha\in\matPkpid$ such that
  $p=\sum_{\aalpha\in\calAkpid}a_\aalpha\XX^\aalpha=q+r$.
  By induction hypothesis, and from
  Lemma~\thref{l:Pkd-nondecr-k},
  there exist
  $\tq_0\in\matPkdmi\subset\matPkpidmi$ and $q_1\in\matPkmid\subset\matPkd$,
  such that $q=\tq_0+X_d\,q_1$.\\
  For all $i\in[0..k+1]$, for all $\taalpha_i \in \calCdmi{k+1-i}$, let
  $b_{\taalpha_i}\eqdef a_{(\taalpha_i,i)}=a_{(\tfhi_{k+1,d}^i)^{-1}(\taalpha_i)}$.
  Then, from
  Lemma~\thref{l:slices-of-multi-ind-Ckd},
  Lemma~\thref{l:card-slices-of-Ckd}, and
  Definition~\thref{d:slices-Sckdi-and-Stkdi},
  we have
  \begin{align*}
    r
    &= \sum_{\aalpha \in \calStkpido} a_\aalpha \XX^\aalpha
      + \sum_{i = 1}^{k + 1} \sum_{\aalpha \in \calStkpidi}
        a_\aalpha \XX^\aalpha\\
    &= \underbrace{\sum_{\taalpha_0 \in \calCkpidmi}
        b_{\taalpha_0} \tXX^{\taalpha_0}}_{\eqdef \tr_0}
      + X_d \; \underbrace{\sum_{i = 1}^{k + 1}
        \sum_{\taalpha_i \in \calCdmi{k + 1 - i}}
          b_{\taalpha_i} \tXX^{\taalpha_i} X_d^{i - 1}}_{\eqdef r_1}.
  \end{align*}
  Let $i\in[1..k+1]$.
  Let $\taalpha_i\in\calCdmi{k+1-i}$.
  Then, from
  Lemma~\thref{l:deg-monom-Ckd-is-k},
  we have $\deg(\tXX^{\taalpha_i}X_d^{i-1})=(k+1-i)+(i-1)=k$.
  Thus, from
  Lemma~\thref{l:Ckd-layers-Akd}, and
  Definition~\thref{d:pol-space-Pkd},
  we have $\tr_0\in\matPkpidmi$ and $r_1\in\matPkd$.

  Let~$\tp_0\eqdef\tq_0+\tr_0$ and~$p_1\eqdef q_1+r_1$.
  Then, from
  \assume{ring properties of $\matR$}, and
  Lemma~\thref{l:Pkd-sp},
  we have $\tp_0\in\matPkpidmi$, $p_1\in\matPkd$, and
  $p=\tp_0+X_d\,p_1$, {\ie} $P(k+1)$ holds.\\
  This concludes the induction on~$k$, and we have existence for all
  $k\geq1$.

  \proofparskip{Uniqueness}
  Let $p\in\matPkd$, $\tq_0,\tr_0\in\matPkdmi$, and $q_1,r_1\in\matPkmid$, such
  that,
  \[
    \forall \xx \in \matRd,\quad
    p (\xx) = \tq_0 (\txx) + x_d \, q_1 (\xx)
    = \tr_0 (\txx) + x_d \, r_1 (\xx).
  \]
  Let~$\txx\in\matRdmi$.
  Let~$\xx\eqdef(\txx,0)$.
  Then, we have $\tq_0(\txx)=\tr_0(\txx)$, {\ie} $\tq_0=\tr_0$, thus
  $X_d(q_1-r_1)=0$.
  Thus, from
  Lemma~\thref{l:Pkd-sp},
  Definition~\thref{d:pol-space-Pkd},
  Lemma~\thref{l:monom-free-in-Pkd}, and
  \assume{the definition of freedom},
  we have $q_1-r_1\in\matPkmid$ and there exists
  $(a_\aalpha)_{\aalpha\in\calAkmid}\in\matR$ such that
  $q_1-r_1=\sum_{\aalpha\in\calAkmid}a_\aalpha\XX^\aalpha$, and thus
  \[
    X_d (q_1 - r_1)
    = \sum_{\aalpha \in \calAkmid} a_\aalpha \XX^{\aalpha + \ee_d}
    = 0 \; (\mbox{in } \matPkd),
  \]
  which implies $a_\aalpha=0$ for all $\aalpha\in\calAkmid$, {\ie} $q_1=r_1$.\\
  Therefore, we have uniqueness.
\end{proof}

\begin{remark}
  \label{r:isom-dim-Pkd}
  Lemmas~\ref{l:decomp-Pkd} and~\ref{l:isom-Pkd-Pkdm1xPkm1d} allow to write
  $\dim(\matPkd)=\dim(\matPkdmi)+\dim(\matPkmid)$.
  Actually, this could be a way to prove (by induction) that
  $\dim(\matPkd)=\binomkpdd$.
  But, we preferred to prove directly that $\card(\calCkd)=\binom{k+d-1}{d-1}$,
  by working on slices of~$\calCkd$.
\end{remark}

\begin{lemma}[isomorphism between $\matPkd$ and $\matPkdmi\times\matPkmid$]
  \label{l:isom-Pkd-Pkdm1xPkm1d}
  \mbox{}\\
  Let~$d\geq2$.
  Let~$k\geq1$.
  Let $\zkd\eqdef\left(p\in\matPkd\longmapsto
    (\tp_0,p_1)\in\matPkdmi\times\matPkmid\right)$ with $p=\tp_0+X_d\,p_1$.\\
  Then, $\zkd$ is an isomorphism.
\end{lemma}

\begin{proof}
  From
  Lemma~\thref{l:decomp-Pkd},
  $\zkd$~is well-defined for all $p\in\matPkd$.

  \proofparskip{Linearity}
  Let $p,q\in\matPkd$, and $\lambda,\mu\in\matR$.
  Then, from
  Lemma~\thref{l:decomp-Pkd},
  \assume{commutative ring properties of~$\matR$},
  Lemma~\thref{l:Pkd-sp}, and
  Lemma~\thref{LM-l:closed-under-linear-combination-is-subspace},
  there exist unique $\tp_0,\tq_0\in\matPkdmi$ and $p_1,q_1\in\matPkmid$ such
  that $p=\tp_0+X_dp_1$ and $q=\tq_0+X_dq_1$, thus
  $\lambda p+\mu q=\lambda\tp_0+\mu\tq_0+X_d(\lambda p_1+\mu q_1)$ with
  $\lambda\tp_0+\mu\tq_0\in\matPkdmi$ and $\lambda p_1+\mu q_1\in\matPkmid$.
  Hence, from
  Lemma~\threfc{l:decomp-Pkd}{for $\lambda p+\mu q$, uniqueness},
  Definition~\thref{LM-d:product-vector-operations},
  Lemma~\thref{LM-l:product-is-space}, and
  Lemma~\thref{LM-l:linear-map-preserves-linear-combinations},
  we have
  \[
    \zkd (\lambda p + \mu q)
    = (\lambda \tp_0 + \mu \tq_0, \lambda p_1 + \mu q_1)
    = \lambda (\tp_0, p_1) + \mu (\tq_0, q_1)
    = \lambda \zkd (p) + \mu \zkd (q),
  \]
  and~$\zkd$ is a linear map from~$\matPkd$ to the product {\vectorspace}
  $\matPkdmi\times\matPkmid$.

  \proofparskip{Injectivity}
  Let~$p\in\matPkd$.
  Assume that $\zkd(p)=0$.
  Then, from
  Lemma~\thref{l:decomp-Pkd},
  Definition~\thref{LM-d:kernel}, and
  Lemma~\thref{LM-l:injective-linear-map-has-zero-kernel},
  there exist unique $\tp_0\in\matPkdmi$ and $p_1\in\matPkmid$ such that
  $p=\tp_0+X_dp_1$, thus $\zkd(p)=(\tp_0,p_1)=0$, $p=0+X_d0=0$, and~$\zkd$ is
  injective.

  \proofparskip{Dimension}
  From
  Lemma~\threfc{l:dim-Pkd}{three times},
  Lemma~\threfc{l:prop-binom-coef}{\eqref{e:prop-binom-coef-3}}, and
  \assume{the rule of dimension for product of spaces},
  we have
  \[
    \dim(\matPkd)
    = \binomkpdd
    = \binom{k + d - 1}{d - 1} + \binom{k - 1 + d}{d}
    = \dim (\matPkdmi) + \dim (\matPkmid)
    = \dim (\matPkdmi \times \matPkmid).
  \]

  \proofparskip{Isomorphism}
  Direct consequence of
  \assume{the definition of order (reflexivity)},
  Lemma~\thref{l:inj-or-surj-and-dim-implies-bij}, and
  Definition~\thref{LM-d:isomorphism}.
\end{proof}

\begin{remark}
  Note in the next lemma that, by an abuse of notation, for all $p\in\matPkd$,
  the function that sends~$(\xx,x_{d+1})$ to~$p(\xx)$ is sometimes still
  denoted~$p$.
\end{remark}


\begin{lemma}[$\matPkd$ is nondecreasing sequence in $d$]
  \label{l:Pkd-nondecr-d}
  \mbox{}\hfill
  Let~$k\in\matN$.\\
  Then, the sequence~$(\matPkd)_{d\geq1}$ is nondecreasing for the inclusion,
  in the sense that, for all $d\geq1$, for all $p\in\matPkd$, the function
  $((\xx,x_{d+1})\in\matRd\times\matR\longmapsto p(\xx)\in\matR)$ belongs
  to~$\matPkdpi$.
\end{lemma}

\begin{proof}
  Let~$d\geq1$.
  Let~$p\in\matPkd$.

  \proofparskip{Case $k=0$}
  Direct consequence of
  Lemma~\thref{l:pol-space-P0d-P1d}.

  \proofparskip{Case $k\geq1$}
  Let~$\tp_0\eqdef p\in\matPkd$ and~$p_1\eqdef0$.
  Let~$\xx\in\matRd$ and $x_{d+1}\in\matR$.
  Then, from
  Lemma~\threfc{l:Pkd-sp}{thus $p_1=0\in\matPkmidpi$},
  Lemma~\thref{l:isom-Pkd-Pkdm1xPkm1d}, and
  Definition~\thref{d:pol-space-Pkd},
  we have $\hp=(\Zkd{k}{d+1})^{-1}(\tp_0,p_1)\in\matPkdpi$, and
  $\hp(\xx,x_{d+1})=p(\xx)+x_{d+1}0=p(\xx)$.
\end{proof}

\begin{lemma}[expression of $\matPkd$ as polynomial of $x_d$]
  \label{l:Pkd-as-pol-xd}
  \mbox{}\hfill
  Let~$d\geq2$.
  Let~$k\in\matN$.
  Let~$p\in\matPkd$.\\
  Then, for all $i\in[0..k]$, there exists a unique $\tr_i\in\matPkmidmi$, such
  that
  \begin{equation}
    \label{e:Pkd-as-pol-xd}
    \forall \xx \in \matRd,\quad
    p (\xx) = \sum_{i = 0}^k \tr_i (\txx) x_d^i
    \qquad \left( \mbox{{\ie} } p = \sum_{i = 0}^k \tr_i X_d^i \right).
  \end{equation}
  Thus, for all $\txx\in\matRdmi$, the function
  $(x_d\longmapsto p(\txx,x_d))$ belongs to~$\matPki$.
\end{lemma}

\begin{proof}
  \proofpar{Existence}
  For all $k\in\matN$, let
  $P(k)\eqdef
  [\forall p\in\matPkd,(\forall i\in[0..k],\exists\tr_i\in\matPkmidmi),
  p=\sum_{i=0}^k\tr_iX_d^i]$.\\
  \proofpar{Induction:~$P(0)$}
  Let $p\in\matPod$.
  Then, from
  Lemma~\thref{l:isom-P0d-P0dm1},
  $\tr_0\eqdef\xi^d(p)\in\matPodmi$ is such that, for all $\txx\in\matRdmi$,
  for all $x_d\in\matR$, $p(\txx,x_d)=\tr_0(\txx)$, {\ie}~$P(0)$ holds.\\
  \proofpar{Induction:~$P(1)$}
  Let $p\in\matPid$.
  Then, from
  Lemma~\thref{l:decomp-Pkd},
  there exists $\tr_0\in\matPidmi$ and $q_1\in\matPod$, such that
  $p=\tr_0+X_dq_1$.
  Thus, from
   Lemma~\thref{l:isom-P0d-P0dm1},
  $\tr_1\eqdef\xi^d(q_1)\in\matPodmi$ is such that, for all $\txx\in\matRdmi$,
  for all $x_d\in\matR$, $q_1(\txx,x_d)=\tr_1(\txx)$, hence
  $p=\tr_0+X_d\tr_1$ with $\tr_0\in\matPidmi$ and $\tr_1\in\matPodmi$,
  {\ie}~$P(1)$ holds.\\
  \proofpar{Induction: $P(k)\Implies P(k+1)$}
  Let~$k\geq1$.
  Assume that~$P(k)$ holds.
  Let~$p\in\matPkpid$.
  Then, from
  Lemma~\thref{l:decomp-Pkd},
  there exist $\tr_0\in\matPkpidmi$ and $q_1\in\matPkd$, such that
  $p=\tr_0+X_dq_1$.
  By induction hypothesis, for all $i\in[0..k]$, there exist
  $\ts_i\in\matPkmidmi$, such that $q_1=\sum_{i=0}^{k}\ts_iX_d^i$.
  For all $i\in[1..k+1]$, let $\tr_i\eqdef\ts_{i-1}\in\matPkpimidmi$.
  Then, from
  \assume{commutative ring properties of~$\matR$},
  we have $p=\sum_{i=0}^{k+1}\tr_iX_d^i$, {\ie}~$P(k+1)$ holds.\\
  This concludes the induction on~$k$, and we have existence for all
  $k\in\matN$.

  \proofparskip{Uniqueness}
  Let~$p\in\matPkd$, and for all $i\in[0..k]$, let $\tr_i,\ts_i\in\matPkmidmi$
  such that
  \[
    p = \sum_{i = 0}^k \tr_i X_d^i = \sum_{i = 0}^k \ts_i X_d^i.
  \]
  Let~$\txx\in\matRdmi$.
  Then, from
  \assume{ring properties of~$\matR$}, and
  Lemma~\thref{l:monom-free-in-Pk1},
  we have $\sum_{i=0}^k(\tr_i(\txx)-\ts_i(\txx))X_d^i=0$, and thus, for all
  $i\in[0..k]$, $\tr_i(\txx)-\ts_i(\txx)=0$, {\ie} $\tr_i=\ts_i$.

  Therefore, we have uniqueness.
\end{proof}

\begin{remark}
  Lemma~\ref{l:Pkd-as-pol-xd} expresses a $d$-multivariate polynomial as an
  univariate polynomial with coefficients that are $(d-1)$-multivariate
  polynomials.
  The proof uses Horner's rule,
  \[
    p = \tr_0 + X_d (\tr_1 + X_d (\tr_2 + X_d
    (\ldots + X_d (\tr_{k - 1} + X_d \tr_k) \ldots))).
  \]
\end{remark}

\subsection{Product of polynomials (alternate)}
\label{ss:prod-polynom-alt}

\begin{remark}
  The next lemma provides an alternate proof for Lemma~\ref{l:prod-2-polynom}.
  This second proof avoids the double sum on multi-indices that is not
  explicitly formulated in the first proof.
  It relies on the decomposition of polynomials of Lemma~\ref{l:decomp-Pkd} and
  on a strong double induction.
\end{remark}

\begin{lemma}[product of two polynomials (alternate proof)]
  \label{l:prod-2-polynom-alt-proof}
  \mbox{}\\
  Let~$d\geq1$.
  Let~$k,l\in\matN$.
  Let~$p\in\matPkd$ and~$q\in\matPld$.
  Then, we have $pq\in\matPkpld$.
\end{lemma}

\begin{proof}
  For all~$d\geq1$, for all~$n\in\matN$, let~$\PropPP$ be the property defined
  by
  \[
  \PropPP (d, n)
  \eqdef \left[
    \forall k, l \in \matN,\;
    \forall p \in \matPkd,\;
    \forall q \in \matPld,\quad
    k + l \leq n \IMPLIES pq \in \matPd{n}
    \right].
  \]
  Then, from
  Lemma~\threfc{l:Pkd-nondecr-k}{%
    thus $k+l\leq n\Implies\matPkpld\subset\matPnd$},
  the result to prove is equivalent to establish~$\PropPP$.
  Let~$d\geq1$ and~$n\in\matN$.

  \proofparskip{Strong double induction}\\
  Assume that for all $d_1,n_1\in\matN$, $1\leq d_1\leq d$, $n_1\leq n$, and
  $(d_1,n_1)\neq(d,n)$ implies $\PropPP(d_1,n_1)$.
  Let us show~$\PropPP(d,n)$.
  Let~$k,l\in\matN$.
  Let~$p\in\matPkd$ and~$q\in\matPld$.
  Assume that~$k+l\leq n$.
  Let us show $pq\in\matPnd$.

  \proofparskip{Case $d=1$}
  Direct consequence of
  Lemma~\thref{l:prod-2-polynom-univ}.

  \proofparskip{Case $k=0$ or $l=0$}
  Direct consequence of
  Lemma~\threfc{l:pol-space-P0d-P1d}{%
    thus $p=a_0\XX^\zzero$ (resp. $q=b_0\XX^\zzero$) is constant},
  \assume{commutative ring properties of $\matR$
    (thus $pq=a_0q$ (resp. $pq=b_0p$))},
  Lemma~\thref{l:Pkd-sp},
  Lemma~\threfc{LM-l:closed-under-vector-operations-is-subspace}{%
    closed under scalar multiplication, thus $pq$ is in $\matPkpld$}, and
  Lemma~\threfc{l:Pkd-nondecr-k}{thus $\matPkpld\subset\matPnd$}.

  \proofparskip{Case $k+l<n$}
  Then, by hypothesis, $\PropPP(d,k+l)$ holds, {\ie} $pq\in\matPnd$.

  \proofparskip{Case $d\geq2$, $k\neq0$, $l\neq0$ and $k+l=n\geq1$}\\
  Then, from
  Lemma~\thref{l:decomp-Pkd}, and
  \assume{commutative ring properties of $\matR$},
  there exist $\tp_0\in\matPkdmi$, $p_1\in\matPkmid$, $\tq_0\in\matPldmi$ and
  $q_1\in\matPlmid$, such that
  \[
    p = \tp_0 + X_d \, p_1,\quad
    q = \tq_0 + X_d \, q_1,\quad \mbox{thus} \quad
    pq = \tp_0 \tq_0 + X_d \, (\tp_0 q_1 + p_1 \tq_0) + X_d^2 \, p_1 q_1.
  \]
  Then, by hypothesis (with $d_1\eqdef d-1<d$ and $n_1\eqdef n$), and from
  Lemma~\thref{l:Pkd-nondecr-d},
  $\PropPP(d-1,n)$ holds, {\ie} $\tp_0\tq_0\in\matPndmi$, and thus
  \[
    \left(
      \xx \eqdef (\txx, x_d) \in \matRd \longmapsto (\tp_0 \tq_0) (\txx)
    \right)
    \in \matPnd.
  \]
  Next, from
  Lemma~\thref{l:Pkd-nondecr-d}, and
  by hypothesis (with $d_1\eqdef d$ and $n_1\eqdef n-1<n$),
  we have
  $(\xx\mapsto\tp_0(\txx))\in\matPkd$,
  $(\xx\mapsto\tq_0(\txx))\in\matPld$, and
  $\PropPP(d,n-1)$ holds, {\ie} $(\xx\mapsto\tp_0(\txx)q_1(\xx))$ and
  $(\xx\mapsto p_1(\xx)\tq_0(\txx))$ belong to~$\matPnmid$.
  Hence, from
  Lemma~\thref{l:Pkd-sp},
  Lemma~\thref{LM-l:closed-under-vector-operations-is-subspace}, and
  Lemma~\threfc{l:prod-monom-polynom}{with $\aalpha\eqdef\ee_d$},
  we have $\tp_0q_1+p_1\tq_0\in\matPnmid$, and
  $X_d\,(\tp_0q_1+p_1\tq_0)\in\matPnd$.\\
  Moreover, by hypothesis (with $d_1\eqdef d$ and $n_1\eqdef n-2<n$), and from
  Lemma~\threfc{l:prod-monom-polynom}{with $\aalpha\eqdef2\ee_d$},
  $\PropPP(d,n-2)$ holds, {\ie} $p_1q_1\in\matPnmiid$, and
  $X_d^2\,p_1q_1\in\matPnd$.\\
  Finally, from
  Lemma~\thref{l:Pkd-sp}, and
  Lemma~\thref{LM-l:closed-under-vector-operations-is-subspace},
  we have $pq\in\matPnd$.

  \medskip\noindent
  Therefore, we always have $pq\in\matPnd$, and from
  Lemma~\thref{l:strong-double-induction},
  the property holds for all~$d\geq1$ and for all~$n\in\matN$.
\end{proof}

\begin{lemma}[product of polynomials]
  \label{l:prod-polynom}
  \mbox{}\hfill
  Let~$d,n\geq1$.
  Let~$(k_i)_{i\in[1..n]}\in\matN$.
  Let~$k\eqdef\sum_{i=1}^nk_i$.
  For all $i\in[1..n]$, let~$p_i\in\matPd{k_i}$.
  Then, we have $\prod_{i=1}^np_i\in\matPkd$.
\end{lemma}

\begin{proof}
  Direct consequence (by induction on~$n$) of
  Lemma~\thref{l:prod-2-polynom}, or
  Lemma~\thref{l:prod-2-polynom-alt-proof}.
\end{proof}

\subsection{Composition of polynomials and affine mappings}
\label{ss:compose-polynom-aff-map}

\begin{lemma}[affine mapping of monomials is $\matPkl$]
  \label{l:aff-mapping-of-monom-is-Pkl}
  \mbox{}\\
  Let~$l,d\geq1$.
  Let~$k\in\matN$.
  Let~$f\in\AffRlRd$.
  Let~$\aalpha\in\calCkd$.
  Then, we have $\XX^\aalpha\circ f\in\matPkl$.
\end{lemma}

\begin{proof}
  Let~$q\eqdef\XX^\aalpha\circ f$.\\
  For all $i\in[1..d]$, let~$f_i$ be the $i$-th component of~$f$.
  Let~$\yy\in\matRl$.
  Then, from
  Lemma~\thref{l:equiv-def-aff-map-finite-dim},
  Lemma~\thref{l:pol-space-P0d-P1d}, and
  Definition~\thref{d:monom-kd},
  $f_i$ belongs to~$\matPil$, and we have
  \[
    q (\yy) = \prod_{i = 1}^d \left[ f_i (\yy) \right]^{\alpha_i},
    \AND
    \left[ f_i (\yy) \right]^{\alpha_i} = \prod_{j = 1}^{\alpha_i} f_i (\yy).
  \]
  Thus, from
  Lemma~\threfc{l:prod-polynom}{%
    with $n\eqdef\alpha_i$ and $p_j\eqdef f_i\in\matPil$ for all
    $j\in[1..\alpha_i]$, then with $n\eqdef d$ and
    $p_i\eqdef f_i^{\alpha_i}\in\matPl{\alpha_i}$ for all $i\in[1..d]$,
    thus $q\in\matPl{\sum_{i=1}^d\alpha_i}$},
  Definition~\threfc{d:multi-ind-Akd-Ckd}{thus $\len{\aalpha}=k$}, and
  Definition~\thref{d:len-multi-ind},
  we have $q\in\matPkl$.
\end{proof}

\begin{lemma}[affine mapping of $\matPkd$ is $\matPkl$]
  \label{l:aff-mapping-of-Pkd-is-Pkl}
  \mbox{}\\
  Let~$l,d\geq1$.
  Let~$k\in\matN$.
  Let~$f\in\AffRlRd$.
  Let~$p\in\matPkd$.
  Then, we have $p\circ f\in\matPkl$.
\end{lemma}

\begin{proof}
  From
  Definition~\thref{d:pol-space-Pkd},
  \assume{distributivity of composition over addition},
  Lemma~\thref{l:Ckd-layers-Akd},
  Lemma~\threfc{l:aff-mapping-of-monom-is-Pkl}{%
    thus $\XX^\aalpha\circ f\in\matPl{i}$ for all $i\in[0..k]$ and for all
    $\aalpha\in\calCd{i}$}, and
  Lemma~\threfc{l:Pkd-nondecr-k}{%
    thus $\matPl{i}\subset\matPkl$ for all $i\in[0..k]$}, and
  Lemma~\thref{l:Pkd-sp},
  we have
  \[
    q \eqdef p \circ f
    = \left( \sum_{\aalpha \in \calAkd} a_\aalpha \XX^\aalpha \right) \circ f
    = \sum_{\aalpha \in \calAkd} a_\aalpha (\XX^\aalpha \circ f)
    = \sum_{i = 0}^k \sum_{\aalpha \in \calCd{i}} a_\aalpha (\XX^\aalpha \circ f),
  \]
  and $q\in\matPkl$.
\end{proof}

\section{{\ToPPid} affine polynomials and affine geometric mapping}
\label{s:P1d-aff-pol-aff-geom-map}

\subsection{Reference affine Lagrange polynomials}
\label{ss:ref-aff-lag-pol}

\begin{definition}[reference Lagrange polynomials of $\matPid$]
  \label{d:lag-pol-P1d-ref}
  \mbox{}\hfill
  Let~$d\geq1$.
  Let~$i\in[0..d]$.
  Let~$\hxx\in\matRd$.
  The $i$-th {\em reference Lagrange polynomials of~$\matPid$} is denoted
  $\hcalL^{1,d}_i$, and is defined by
  \begin{align}
    \label{e:lag-pol-P1d-ref-0}
    (i = 0)&&
    \hcalL^{1, d}_0 (\hxx) &\eqdef 1 - \sum_{j = 1}^d \hx_j,&&\\
    \label{e:lag-pol-P1d-ref-1}
    (i \in [1..d])&&
    \hcalL^{1, d}_i (\hxx) &\eqdef \hx_i.&&
  \end{align}
\end{definition}

\begin{lemma}[reference Lagrange polynomials of $\matPid$ for $d=1$
    are reference Lagrange polynomials of $\matPki$ for $k=1$]
  \label{l:ref-lag-pol-P1d-for-d-eq-1-are-ref-lag-pol-Pk1-for-k}
  \mbox{}\\
  The reference Lagrange polynomials of~$\matPid$ for~$d\eqdef1$ and the
  reference Lagrange polynomials of~$\matPki$ for~$k\eqdef1$ from
  Lemma~\thref{l:lag-basis-Pk1-ref} coincide.
\end{lemma}

\begin{proof}
  Direct consequence of
  Definition~\threfc{d:lag-pol-P1d-ref}{with $d\eqdef1$},
  Definition~\threfc{d:lag-pol-Pk1}{with $k\eqdef1$},
  Definition~\threfc{d:lag-nodes-Pk1-ref}{%
    \eqref{e:lag-nodes-Pk1-ref-1} with $k\eqdef1$},
  Lemma~\threfc{l:lag-basis-Pk1-ref}{with $k\eqdef1$}, and
  \assume{field properties of~$\matR$
    (with $1-0\neq0$ and $0-1\neq0$ for $k\eqdef1$)}.
\end{proof}

\begin{lemma}[reference Lagrange polynomials is basis of $\matPid$]
  \label{l:lag-pol-is-basis-P1d-ref}
  \mbox{}\hfill
  Let~$d\geq1$.\\
  Then, the reference Lagrange polynomials~$(\hcalL^{1,d}_i)_{i\in[0..d]}$
  form a basis of~$\matPid$, that satisfies
  \begin{align}
    \label{e:lag-pol-is-basis-P1d-ref-1}
    \forall i \in [0..d],\quad
    & \deg \hcalL^{1, d}_i = 1,\\
    \label{e:lag-pol-is-basis-P1d-ref-2}
    \forall i, j \in [0..d],\quad
    & \hcalL^{1, d}_i (\hvv_j) = \kron{i}{j},\\
    \label{e:lag-pol-is-basis-P1d-ref-3}
    & \sum_{i = 0}^d \hcalL^{1, d}_i = 1.
  \end{align}
\end{lemma}

\begin{proof}
  \proofparskip{Degree~\eqref{e:lag-pol-is-basis-P1d-ref-1}, in~$\matPid$, and
    identities~\eqref{e:lag-pol-is-basis-P1d-ref-2}
    and~\eqref{e:lag-pol-is-basis-P1d-ref-3}}\\
  Direct consequence of
  Definition~\thref{d:lag-pol-P1d-ref},
  Lemma~\thref{l:pol-space-P0d-P1d},
  Definition~\thref{d:deg-pol},
  \assume{additive group properties of~$\matR$}, and
  Definition~\thref{d:fam-ref-aff-pts}.

  \proofparskip{Basis}
  Let $(\lambda_i)_{i\in[0..d]}\in\matR$ such that, for all $\hxx\in\matRd$,
  $\sum_{i=0}^d \lambda_{i}\hcalL^{1,d}_i(\hxx)=0$.
  Then, from
  \eqref{e:lag-pol-is-basis-P1d-ref-2}, and
  \assume{ring properties of~$\matR$},
  taking $\hxx\eqdef\hvv_j$ provides
  $\sum_{i=0}^d\lambda_i\kron{i}{j}=\lambda_j=0$.
  Thus, from
  \assume{the definition of freedom},
  Lemma~\threfc{l:dim-Pkd}{with $k\eqdef1$},  and
  Lemma~\thref{l:free-family-of-dim-elements-is-basis},
  the reference Lagrange polynomials form a basis of~$\matPid$.
\end{proof}

\begin{remark}
  In the next lemma, the linear maps from~$\matRn$ to~$\matRp$ are assimilated
  to their matrix relative to the canonical bases.
\end{remark}

\begin{lemma}[differential of reference Lagrange polynomials]
  \label{l:diff-lag-pol-ref}
  \mbox{}\hfill
  Let~$d\geq1$.
  Let~$\hxx\in\matRd$.
  Let~$i\in[0..d]$.
  Then, we have $\hcalL^{1,d}_i\in\CinfRdR$ and
  \begin{align}
    \label{e:diff-lag-pol-ref-0}
    (i = 0)&&
    \Diff \hcalL^{1, d}_0 (\hxx) &= -\oone \in \calM_{1, d} (\matR),&&\\
    \label{e:diff-lag-pol-ref-1}
    (i \in [1..d])&&
    \Diff \hcalL^{1, d}_i (\hxx) &= \ee_i \in \calM_{1, d} (\matR).&&
  \end{align}
\end{lemma}

\begin{proof}
  Direct consequence of
  Lemma~\thref{l:lag-pol-is-basis-P1d-ref},
  Lem\-ma~\threfc{l:pol-space-P0d-P1d}{%
    thus reference Lagrange polynomials are affine maps},
  \assume{affine maps in $\matRd$ are~$\Cinf$}, and
  \assume{the rules of derivation (thus affine maps have constant differential
    and the differential is the linear part)}.
\end{proof}

\subsection{Affine geometric mapping from $\matRd$ to $\matRd$}
\label{ss:aff-geomap-Rd}

\begin{remark}
  In Figure~\ref{f:geo-map-d=2}, we plot an example of~$\phigeodv$, the affine
  geometric mapping from the reference simplex in~$\matRd$ to the current
  simplex in~$\matRd$.
  This mapping transforms the $i$-th hyperface ({\ie} opposite the vertex
  $i\in[0..d]$) of the reference simplex to the $i$-th hyperface of the current
  simplex, see Figure~\ref{f:geo-map-d=3-hyperplane}.

  Other affine geometric mappings are also introduced.
  First, in Definition~\ref{d:geo-l-face-mapping}, the affine geometric
  mapping~$\phindlv:\ArRlRd$ (for $l\in[1..d]$) is a generalization
  of~$\phigeodv$, that allows to pass from the reference $l$-simplex to the
  current $d$-simplex.
  The injective mapping~$\indl:[0..l]\to[0..d]$ is used to select the chosen
  $l$-face of the $d$-simplex.

  Next, one can define the geometric mapping~$\phindv{\thetinjdmi{i}}:\ArRdmiR$
  (for $i\in[0..d]$), see Lemma~\ref{l:geo-hyperface-mapping}.
  It transfers the reference $(d-1)$-simplex to the $i$-th hyperface of the
  current $d$-simplex, see Figure~\ref{f:geo-hyper-map_d32}.

  Finally, the geometric mapping~$\phinddv:\ArRdRd$ is defined in
  Lemma~\ref{l:geo-mapping-permut}.
  It transforms the reference $d$-simplex to a ``permutation'' of the current
  $d$-simplex ({\ie} with a permutation of the vertices), see
  Figure~\ref{f:geo-map-perm-circ0-d=3-hyperplane} in the case of a circular
  permutation (defined in~Lemma~\ref{l:circ-permut}).
  It is used in the proof of the factorization of polynomial
  Lemma~\ref{l:factor-zero-pol-hyperpl-Pkd}.
\end{remark}

\begin{definition}[geometric mapping]
  \label{d:geo-mapping}
  \mbox{}\hfill
  Let~$d\geq1$.
  Let~$\famvertd{\vv}$ be $d+1$ points in~$\matRd$.\\
  The {\em geometric mapping associated with $\famvertd{\vv}$} is
  denoted~$\phigeodv$, and is defined by
  \begin{equation}
    \label{e:geo-mapping}
    \forall \hxx \in \matRd,\quad
    \phigeodv (\hxx)
    = \sum_{i = 0}^d \hcalL^{1, d}_i (\hxx) \, \vv_i \quad\in \matRd.
  \end{equation}
\end{definition}

\begin{lemma}[geometric mapping for $d=1$ is geometric mapping in dimension~1]
  \label{l:geo-mapping-for-d-eq-1-is-geo-mapping-1d}
  \mbox{}\\
  Let~$\famverti{\vv}=\famverti{v}$ be two points in~$\matR^1$.\\
  Then, $\phigeodv$ for $d\eqdef1$ and~$\phigeoi{v}$ from
  Definition~\ref{d:geo-mapping-1d} coincide.
\end{lemma}

\begin{proof}
  Direct consequence of
  Definition~\threfc{d:geo-mapping}{with $d\eqdef1$},
  Definition~\threfc{d:lag-pol-P1d-ref}{with $d\eqdef1$},
  Definition~\thref{d:geo-mapping-1d}, and
  \assume{commutative ring properties of~$\matR$}.
\end{proof}

\begin{lemma}[reference geometric mapping is identity]
  \label{l:ref-geo-mapping-is-id}
  \mbox{}\\
  Let~$d\geq1$.
  Then, the geometric mapping for reference vertices~$\famvertd{\hvv}$ is the
  identity, $\phigeod{\hvv}=\identity{\matRd}$.
\end{lemma}

\begin{proof}
  Direct consequence of
  Definition~\thref{d:geo-mapping},
  Definition~\thref{d:lag-pol-P1d-ref}, and
  Definition~\thref{d:fam-ref-aff-pts}.
\end{proof}

\begin{figure}[htb]
  \centering
  \resizebox{0.85\linewidth}{!}{
    \begin{tikzpicture}
\draw[step=.5cm, very thin, black!20] (-4,-4) ; 
\coordinate[label=left : $\hvv_0$] (a) at (-3,-3) ;
\coordinate[label=above : $\hvv_2$] (c) at (-3,0) ;
\coordinate[label=right : $\hvv_1$] (b) at (0,-3) ;

\coordinate[label=left : $0$] (0) at (-2.7,-3.3) ;
\coordinate[label=above : $1$] (2) at (-3.3,-0.3) ;
\coordinate[label=right : $1$] (1) at (-0.3,-3.3) ;

\draw[thick] (a)--(b)--(c)--cycle ;

\coordinate[label=left : \textcolor{black}{$\vv_0 = \phigeodv (\hvv_0)$}] (A) at (4.5,-1.5) ;
\coordinate[label=above : \textcolor{black}{$\vv_2 = \phigeodv (\hvv_2) $}] (C) at (5.5,1) ;
\coordinate[label=right : \textcolor{black}{$\vv_1 = \phigeodv (\hvv_1) $}] (B) at (6.5,-3) ;
\draw[thick] (A)--(B)--(C)--cycle ;

\node (arrow) [] {};
    \draw[-latex,blue, very thick] ($(arrow.north west)+(-1,-1)$) arc
    [start angle=160,
     end angle=70,
     x radius=3.8cm,
     y radius =2cm] ;
     
\coordinate[label=above : \textcolor{blue}{$\phigeodv$}] (T) at (1,0.5) ;
\coordinate[label=left : \textcolor{blue}{$\hKd$}] (Khat) at (-1.7,-2) ;
\coordinate[label=right : \textcolor{blue}{$\Kvd$}] (Khat) at (5,-1) ;

\end{tikzpicture}
  }
  \caption[Geometric mapping]{Geometric mapping $\phigeodv$ in the
    case $d=2$, see Definition~\ref{d:geo-mapping}.
  }
  \label{f:geo-map-d=2}
\end{figure}

\begin{lemma}[properties of geometric mapping]
  \label{l:prop-geo-mapping}
  \mbox{}\hfill
  Let~$d\geq1$.\\
  Let~$\famvertd{\vv}$ be $d+1$ points in~$\matRd$.
  Let $\Ageodv\eqdef\begin{pmatrix}\vv_1-\vv_0&\vv_2-\vv_0&\dots&
    \vv_d-\vv_0\end{pmatrix}$ in~$\calM_{d,d}$.\\
  Then, $\phigeodv$ belongs to~$\AffRdRd$, and we have
  \begin{align}
    \label{e:prop-geo-mapping-1}
    \forall j\in[0..d],\quad
    \forall \hxx \in \matRd,\quad
    &\phigeodv (\hxx) = \vv_j +
      \sum_{i = 0, i \neq j}^d \hcalL^{1, d}_i (\hxx) (\vv_i - \vv_j),\\
    \label{e:prop-geo-mapping-2}
    \forall \hxx \in \matRd,\quad
    &\phigeodv (\hxx)
      = \vv_0 +  \sum_{i = 1}^d \hx_i (\vv_{i} - \vv_{0})
      = \vv_0 + \Ageodv \hxx,\\
    \label{e:prop-geo-mapping-3}
    \forall i\in[0..d],\quad
    &\phigeodv (\hvv_i) = \vv_i,\\
    \label{e:prop-geo-mapping-3bis}
    &\phigeodv (\hKd) = \Kvd.
  \end{align}

  Moreover, if $\famvertd{\vv}$ is affinely independent, then $\phigeodv$~is
  bijective, $\invphigeodv$~is affine, and
  \begin{align}
    \label{e:prop-geo-mapping-4}
    \forall \xx \in \matRd,\quad
    &\invphigeodv (\xx) = \invAgeodv (\xx - \vv_0),\\
    \label{e:prop-geo-mapping-5}
    \forall i \in [0..d],\quad
    &\invphigeodv (\vv_i) = \hvv_i.
  \end{align}
\end{lemma}

\begin{proof}
  \proofpar{Identities~\eqref{e:prop-geo-mapping-1}
    and~\eqref{e:prop-geo-mapping-2}, affine}
  Let $j\in[0..d]$.
  Then, from
  Definition~\thref{d:geo-mapping},
  Lemma~\threfc{l:lag-pol-is-basis-P1d-ref}{%
    \eqref{e:lag-pol-is-basis-P1d-ref-3}, extracting~$j$}, and
  \assume{ring properties of~$\matR$},
  identity~\eqref{e:prop-geo-mapping-1} holds.
  Thus, from
  \eqref{e:prop-geo-mapping-1} (with $j\eqdef0$),
  Definition~\threfc{d:lag-pol-P1d-ref}{\eqref{e:lag-pol-P1d-ref-1}},
  \assume{the columnwise rule for matrix-vector product},
  \assume{commutative ring properties of~$\matR$}, and
  Definition~\threfc{d:aff-map}{%
    with $\phi\eqdef\Ageodv$ and $\yyo\eqdef\vv_0$},
  identities~\eqref{e:prop-geo-mapping-2} hold, and~$\phigeodv$ is affine.

  \proofparskip{Identity~\eqref{e:prop-geo-mapping-3}}
  Direct consequence of
  Definition~\thref{d:geo-mapping}, and
  Lemma~\threfc{l:lag-pol-is-basis-P1d-ref}{%
    \eqref{e:lag-pol-is-basis-P1d-ref-2}}.

  \proofparskip{Identity~\eqref{e:prop-geo-mapping-3bis}}\\
  \proofpar{$\phigeodv(\hKd)\subset\Kvd$}
  Let~$\hxx\in\hKd$.
  Then, from
  Definition~\thref{d:lag-pol-P1d-ref},
  Definition~\thref{d:ref-simplex},
  \assume{ordered field properties of~$\matR$},
  Lemma~\threfc{l:lag-pol-is-basis-P1d-ref}{%
    \eqref{e:lag-pol-is-basis-P1d-ref-3}},
  Definition~\thref{d:geo-mapping}, and
  Definition~\thref{d:simplex},
  we have, for all $i\in[1..d]$, $\hcalL^{1,d}_i(\hxx)=\hx_i\geq0$,
  and $\hcalL^{1,d}_0(\hxx)=1-\sum_{i=0}^d\hcalL^{1,d}_i\geq0$,
  thus~$\phigeodv(\hxx)$ is a convex combination of the~$\vv_i$'s, and it
  belongs to~$\Kvd$.\\
  \proofpar{$\Kvd\subset\phigeodv(\hKd)$}
  Let~$\xx\in\Kvd$.
  Then, from
  Definition~\thref{d:simplex},
  there exists $(\mu_i)_{i\in[0..d]}\in\matR$ such that, for all $i\in[0..d]$,
  $\mu_i\geq0$, $\sum_{i=0}^d\mu_i=1$ and $\xx=\sum_{i=0}^d\mu_i\vv_i$.\\
  Let~$\hxx\eqdef(\mu_i)_{i\in[1..d]}\in\matRd$.
  Then, from
  Definition~\thref{d:lag-pol-P1d-ref},
  \assume{ordered field properties of~$\matR$},
  Definition~\thref{d:ref-simplex}, and
  Definition~\thref{d:geo-mapping},
  we have, for all $i\in[1..d]$, $\hcalL^{1,d}_i(\hxx)=\mu_i\geq0$ and
  $\hcalL^{1,d}_0(\hxx)=1-\sum_{j=1}^d\mu_j=\mu_0\geq0$, thus
  $\hxx\in\hKd$ and $\phigeodv(\hxx)=\xx$, {\ie} $\xx\in\phigeodv(\hKd)$.\\
  Therefore, we have equality.

  \proofparskip{Bijection, $\invphigeodv$ affine, and
    identity~\eqref{e:prop-geo-mapping-4}}
  Direct consequence of
  \eqref{e:prop-geo-mapping-2},
  Lemma~\thref{l:inj-aff-map-is-zero-linear-ker},
  Definition~\thref{LM-d:kernel},
  Definition~\thref{d:aff-indep-family},
  Lemma~\thref{l:inj-or-surj-and-dim-implies-bij}, and
  Lemma~\thref{l:inv-of-aff-map-is-aff-map}.

  \proofparskip{Identity~\eqref{e:prop-geo-mapping-5}}
  Direct consequence of
  \eqref{e:prop-geo-mapping-3}.
\end{proof}

\begin{lemma}[differential of geometric mapping]
  \label{l:diff-geo-mapping}
  \mbox{}\\
  Let~$d\geq1$.
  Let~$\famvertd{\vv}$ be $d+1$ points in~$\matRd$.
  Then, $\phigeodv$ belongs to~$\CinfRdRd$, and we have
  \begin{equation}
    \label{e:diff-geo-mapping-1}
    \forall \hxx \in \matRd,\quad \Diff \phigeodv (\hxx) = \Ageodv.
  \end{equation}

  Moreover, if~$\famvertd{\vv}$ is affinely independent, $\invphigeodv$ also
  belongs to~$\CinfRdRd$, and we have
  \begin{equation}
    \label{e:diff-geo-mapping-2}
    \forall \xx \in \matRd,\quad \Diff \invphigeodv (\xx) = \invAgeodv.
  \end{equation}
\end{lemma}

\begin{proof}
  Direct consequence of
  Lemma~\thref{l:prop-geo-mapping},
  \assume{affine maps in $\matRd$ are $\Cinf$}, and
  \assume{the rules of derivation (thus affine maps have constant differential
    and the differential is the linear part)}.
\end{proof}

\subsection{Current affine Lagrange polynomials}
\label{ss:cur-aff-lag-pol}

\begin{lemma}[Lagrange polynomials of $\matPid$]
  \label{l:lag-pol-P1d}
  \mbox{}\\
  Let~$d\geq1$.
  Let~$\famvertd{\vv}$ be $d+1$ affinely independent points in~$\matRd$.
  Let~$i\in[0..d]$.\\
  Then, $\calL^{\famvertd{\vv}}_i\eqdef\hcalL^{1,d}_i\circ\invphigeodv$ is
  well-defined, and we have
  $\hcalL^{1,d}_i=\calL^{\famvertd{\vv}}_i\circ\phigeodv$.

  $\calL^{\famvertd{\vv}}_i$ is called the {\em $i$-th Lagrange polynomial
    of~$\matPid$ associated with points~$\famvertd{\vv}$}.
\end{lemma}

\begin{proof}
  \proofpar{Existence}
  Direct consequence of
  Lemma~\thref{l:prop-geo-mapping}.

  \proofparskip{Property}
  Direct consequence of
  \assume{associativity of composition of functions}, and
  \assume{the definition of the inverse}.
\end{proof}

\begin{lemma}[Lagrange polynomials of reference vertices are reference Lagrange
  polynomials of $\matPid$]
  \label{l:lag-pol-of-ref-vert-are-ref-lag-pol-P1d}
  \mbox{}\hfill
  Let~$d\geq1$.
  Let~$i\in[0..d]$.
  Then, we have $\calL^{\famvertd{\hvv}}_i=\hcalL^{1,d}_i$.
\end{lemma}

\begin{proof}
  Direct consequence of
  Lemma~\thref{l:lag-pol-P1d}, and
  Lemma~\thref{l:ref-geo-mapping-is-id}.
\end{proof}

\begin{lemma}[Lagrange polynomials is basis of $\matPid$]
  \label{l:lag-pol-is-basis-P1d}
  \mbox{}\\
  Let~$d\geq1$.
  Let~$\famvertd{\vv}$ be $d+1$ affinely independent points in~$\matRd$.\\
  Then, the Lagrange polynomials $(\calL^{\famvertd{\vv}}_i)_{i\in[0..d]}$ form
  a basis of~$\matPid$, that satisfies
  \begin{align}
    \label{e:lag-pol-is-basis-P1d-1}
    \forall i, j \in [0..d],\quad
    &\calL^{\famvertd{\vv}}_i (\vv_j) = \kron{i}{j},\\
    \label{e:lag-pol-is-basis-P1d-2}
    &\sum_{i = 0}^d \calL^{\famvertd{\vv}}_i = 1.
  \end{align}
\end{lemma}

\begin{proof}
  \proofpar{In $\matPid$}
  Direct consequence of
  Lemma~\thref{l:lag-pol-P1d},
  Lemma~\threfc{l:pol-space-P0d-P1d}{thus $\matPid=\AffRdR$},
  Lemma~\threfc{l:lag-pol-is-basis-P1d-ref}{thus $\hcalL^{1,d}_i$ is affine},
  Lemma~\threfc{l:prop-geo-mapping}{thus $\invphigeodv$ is affine}, and
  Lemma~\thref{l:aff-map-are-closed-by-composition}.

  \proofparskip{Identities~\eqref{e:lag-pol-is-basis-P1d-1}
    and~\eqref{e:lag-pol-is-basis-P1d-2}}\\
  Direct consequence of
  Lemma~\thref{l:lag-pol-P1d},
  Lemma~\threfc{l:prop-geo-mapping}{\eqref{e:prop-geo-mapping-5}},
  Lemma~\threfc{l:lag-pol-is-basis-P1d-ref}{%
    \eqref{e:lag-pol-is-basis-P1d-ref-2}
    and~\eqref{e:lag-pol-is-basis-P1d-ref-3}}, and
  \assume{distributivity of composition over addition}.

  \proofparskip{Basis}
  Let $(\lambda_i)_{i\in[0..d]}\in\matR$, such that for all $\xx\in\matRd$,
  $\sum_{i=0}^d\lambda_i\calL^{\famvertd{\vv}}_i(\xx)=0$.
  Then, from~\eqref{e:lag-pol-is-basis-P1d-1}, and
  \assume{ring properties of~$\matR$},
  taking~$\xx\eqdef\vv_j$ provides
  $\sum_{i=0}^d\lambda_i\kron{i}{j}=\lambda_j=0$.
  Thus, from
  \assume{the definition of freedom},
  Lemma~\threfc{l:dim-Pkd}{with $k\eqdef1$}, and
  Lemma~\thref{l:free-family-of-dim-elements-is-basis},
  the Lagrange polynomials form a basis of~$\matPid$.
\end{proof}

\begin{lemma}[decomposition of $\matPid$ polynomial in Lagrange basis]
  \label{l:decomp-P1d-pol-in-lag-basis}
  \mbox{}\\
  Let~$d\geq1$.
  Let~$\famvertd{\vv}$ be $d+1$ affinely independent points in~$\matRd$.
  Let~$p\in\matPid$.
  Then, we have
  \begin{equation}
    \label{e:decomp-P1d-pol-in-lag-basis}
    p = \sum_{i = 0}^d p (\vv_i) \, \calL^{\famvertd{\vv}}_i.
  \end{equation}
\end{lemma}

\begin{proof}
  Direct consequence of
  Lemma~\threfc{l:lag-pol-is-basis-P1d}{%
    basis, then \eqref{e:lag-pol-is-basis-P1d-1}}, and
  \assume{ring properties of~$\matR$}.
\end{proof}

\begin{lemma}[differential of Lagrange polynomials $\matPid$]
  \label{l:diff-lag-pol-P1d}
  \mbox{}\\
  Let~$d\geq1$.
  Let~$\famvertd{\vv}$ be $d+1$ affinely independent points in~$\matRd$.
  Let~$\xx\in\matRd$.
  Let~$i\in[0..d]$.\\
  Then, we have $\calL^{\famvertd{\vv}}_i\in\CinfRdR$ and
  \begin{align}
    \label{e:diff-lag-pol-P1d-0}
    (i = 0)&&
    \Diff \calL^{\famvertd{\vv}}_0 (\xx) &=
      -\sum_{j = 1}^d \underline{(\Ageodv)^{-1}}_j \in \calM_{1, d}(\matR),&&\\
    \label{e:diff-lag-pol-P1d-1}
    (i \in [1..d])&&
    \Diff \calL^{\famvertd{\vv}}_i (\xx) &=
      \underline{(\Ageodv)^{-1}}_i \in \calM_{1, d}(\matR),&&
  \end{align}
  where, we recall that for any matrix~$B$, $\underline{B}_i$ denotes its
  $i$-th line.
\end{lemma}

\begin{proof}
  Direct consequence of
  Lemma~\thref{l:lag-pol-is-basis-P1d},
  Lemma~\threfc{l:pol-space-P0d-P1d}{%
    thus Lagrange polynomials are affine maps},
  \assume{affine maps in $\matRd$ are~$\Cinf$},
  Lemma~\thref{l:diff-lag-pol-ref},
  Lemma~\thref{l:diff-geo-mapping}, and
  \assume{the rules derivation for composition}.
\end{proof}

\subsection{Nontrivial current simplex}
\label{ss:nontrivial-curr-simplex}

\begin{lemma}[{\nontrivial} simplex]
  \label{l:non-trivial-simplex}
  \mbox{}\hfill
  Let~$d\geq1$.
  Let~$\famvertd{\vv}$ be $d+1$ affinely independent points in~$\matRd$.
  Then, $\Kvd$ has a {\nonempty} interior.
  It is then said to be {\em {\nontrivial}}.
\end{lemma}

\begin{proof}
  From
  Lemma~\thref{l:non-trivial-ref-simplex}, and
  \assume{the definition of the interior},
  let~$\hU$ be a {\nonempty} open set included in~$\hKd$.
  Then, from
  Lemma~\thref{l:diff-geo-mapping},
  \assume{$\Cinf$ is continuous (thus $\phigeodv$ is an homeomorphism)}, and
  \assume{homeomorphisms are open maps},
  $U\eqdef\phigeodv(\hU)$ is open.
  Therefore, from
  \assume{image of $\emptyset$ is $\emptyset$ (contrapositive)},
  \assume{image is nondecreasing},
  Lemma~\threfc{l:prop-geo-mapping}{\eqref{e:prop-geo-mapping-3bis}}, and
  \assume{the definition of the interior},
  we have $\emptyset\neq U\subset\Kvd$, and~$\Kvd$ has a {\nonempty} interior.
\end{proof}

\section{Barycentric coordinates}
\label{s:baryc-coord}

\begin{remark}
  In the next lemma, the uniqueness of the barycentric coordinates actually
  makes them functions of the point~$\xx$.
  In particular, in~\eqref{e:baryc-coor-2}, they do not depend on the choice
  for~$j$.
\end{remark}

\begin{lemma}[barycentric coordinate]
  \label{l:baryc-coor}
  \mbox{}\\
  Let~$d\geq1$.
  Let~$\famvertd{\vv}$ be $d+1$ affinely independent points in~$\matRd$.
  Let~$\xx\in\matRd$.\\
  Then, there exists unique~$(\barcvd_i)_{i\in[0..d]}\in\matR$ satisfying the
  next two equivalent decompositions
  \begin{align}
    \label{e:baryc-coor-1}
    \xx &= \sum_{i = 0}^d \barcvd_i \, \vv_i
    &&\mbox{with} \quad \sum_{i = 0}^d \barcvd_i = 1,\\
    \label{e:baryc-coor-2}
    \forall j \in [0..d],\quad
    \xx &= \vv_j + \sum_{i = 0, i \neq j}^d \barcvd_i (\vv_i - \vv_j)
    &&\mbox{with} \quad \barcvd_j = 1 - \sum_{i = 0, i \neq j}^d \barcvd_i.
  \end{align}
  The $\barcvd_i$'s are called
  {\em barycentric coordinates of~$\xx$ with respect to~$\famvertd{\vv}$},
  they are functions of~$\xx$.
\end{lemma}

\begin{proof}
  \proofpar{Existence}
  Let $j\in[0..d]$.
  Then, from
  Lemma~\threfc{LM-l:closed-under-vector-operations-is-subspace}{%
    thus $\xx-\vv_j$ belongs to the {\vectorspace}~$\matRd$},
  Lemma~\thref{l:equiv-def-aff-indep-family},
  Lemma~\thref{l:free-family-of-dim-elements-is-basis}, and
  \assume{the definition of basis},
  there exist unique real numbers $(\mu_i^j)_{i\in[0..d]\setminus\{j\}}$ such that
  $\xx-\vv_j=\sum_{i=0,i\neq j}\mu_i^j(\vv_i-\vv_j)$, {\ie} the
  decomposition~\eqref{e:baryc-coor-2} exists, and is unique (but still
  depends on~$j$).

  Let~$\mu_j^j\eqdef1-\sum_{i=0,i\neq j}\mu_i^j$.
  Then, from
  \assume{commutative ring properties of~$\matR$},
  we have $\xx=\sum_{i=0}\mu_i^j\vv_i$ and $\sum_{i=0}^d\mu_i^j=1$, {\ie} the
  decomposition~\eqref{e:baryc-coor-1} exists.

  \proofparskip{Uniqueness}
  Let $(\lambda_i)_{i\in[0..d]},(\mu_i)_{i\in[0..d]}\in\matR$, such that
  \[
    \xx = \sum_{i = 0}^d \lambda_i \vv_i = \sum_{i = 0}^d \mu_i \vv_i
    \AND
    \sum_{i = 0}^d \lambda_i =\sum_{i = 0}^d \mu_i = 1.
  \]
  For all~$i\in[0..d]$, let~$\xi_i\eqdef\lambda_i-\mu_i$.
  Then, from
  \assume{ring properties of~$\matR$},
  we have successively $\sum_{i=0}^d\xi_i\vv_i=0$, $\sum_{i=0}^d\xi_i=0$,
  $\xi_0=-\sum_{i=1}^d\xi_i$, and $\sum_{i=1}^d\xi_i(\vv_i-\vv_0)=0$.
  Thus, from
  Definition~\thref{d:aff-indep-family},
  \assume{the definition of freedom}, and
  \assume{additive group properties of~$\matR$},
  we have, for all~$i\in[1..d]$, $\xi_i=0$, and $\xi_0=0$.
  Therefore, the decomposition~\eqref{e:baryc-coor-1} is unique.

  \proofparskip{Equivalence and independence on~$j$}\\
  Let~$j\in[0..d]$.
  Then, from~\eqref{e:baryc-coor-1}, and
  \assume{ring properties of~$\matR$},
  we have
  \[
    \xx = \left( 1 - \sum_{i = 0, i \neq j} \barcvd_i \right) \vv_j
    + \sum_{i = 0, i \neq j} \barcvd_i \vv_i,
  \]
  and thus $\xx-\vv_j=\sum_{i=0,i\neq j}^d\barcvd_i(\vv_i-\vv_j)$, {\ie}
  equivalence of the two decompositions.
  Moreover, from the uniqueness of the $\mu_i^j$'s defined above, we have, for
  all $i\in[0..d]\setminus\{j\}$, $\mu_i^j=\barcvd_i$, {\ie} the
  decomposition~\eqref{e:baryc-coor-2} is independent on~$j$.
\end{proof}

\begin{lemma}[Lagrange polynomials of $\matPid$ are barycentric coordinate]
  \label{l:lag-pol-P1d-is-baryc-coor}
  \mbox{}\\
  Let~$d\geq1$.
  Let~$\famvertd{\vv}$ be $d+1$ affinely independent points in~$\matRd$.\\
  Then, barycentric coordinates and Lagrange polynomials coincide, for all
  $i\in[0..d]$, $\barcvd_i=\calL^{\famvertd{\vv}}_i$.
  Thus, we have, for all $\xx\in\matRd$,
  $\xx=\sum_{i=0}^d\calL^{\famvertd{\vv}}_i(\xx)\,\vv_i$ with
  $\sum_{i=0}^d\calL^{\famvertd{\vv}}_i(\xx)=1$.

  Moreover, let $\bbarcvd\eqdef(\barcvd_i)_{i\in[1..d]}$.
  Then, we have $\bbarcvd=\invphigeodv\in\AffRdRd$, {\ie}
  \begin{equation}
    \label{e:lag-pol-P1d-is-baryc-coor-2}
    \forall \xx \in \matRd,\quad
    \xx = \phigeodv (\bbarcvd (\xx)),
    \quad\mbox{or equivalently}\quad
    \forall \hxx \in \matRd,\quad
    \hxx = \bbarcvd (\phigeodv (\hxx)).
  \end{equation}
\end{lemma}

\begin{proof}
  \proofpar{(i) Equality $\bbarcvd=\invphigeodv$}
  Let~$\xx\in\matRd$.
  Then, from
  Lemma~\threfc{l:baryc-coor}{\eqref{e:baryc-coor-2} with $j\eqdef0$},
  \assume{additive group properties of~$\matR$},
  \assume{the columnwise rule for matrix-vector product},
  \assume{the rules of matrix-vector product (inverse matrix)}, and
  Lemma~\thref{l:prop-geo-mapping},
  we have
  \begin{gather*}
    \xx - \vv_0 = \sum_{i = 1}^d \barcvd_i (\xx) (\vv_i - \vv_0)
    = \Ageodv \bbarcvd (\xx),\\
    \bbarcvd (\xx) = \invAgeodv (\xx - \vv_0) = \invphigeodv (\xx),
  \end{gather*}
  and thus $\bbarcvd=\invphigeodv\in\AffRdRd$.

  \proofparskip{(ii) Identities~\eqref{e:lag-pol-P1d-is-baryc-coor-2}}
  Direct consequence of (i), and
  \assume{the definition of the inverse}.

  \proofparskip{(iii) Identity $\barcvd_i=\calL^{\famvertd{\vv}}_i$}\\
  Let~$\xx\in\matRd$.
  Then, from
  Lemma~\threfc{l:baryc-coor}{\eqref{e:baryc-coor-1}, then uniqueness},
  (ii),
  Definition~\thref{d:geo-mapping},
  (i), and
  Lemma~\thref{l:lag-pol-P1d},
  we have
  \begin{align*}
    \sum_{i = 0}^d \barcvd_i (\xx) \, \vv_i
    &= \xx
    = \phigeodv (\bbarcvd (\xx))\\
    &= \sum_{i = 0}^d \hcalL^{1, d}_i (\bbarcvd (\xx)) \, \vv_i
    = \sum_{i = 0}^d \hcalL^{1, d}_i \left(
      \invphigeodv (\xx) \right) \, \vv_i
    = \sum_{i = 0}^d \calL^{\famvertd{\vv}}_i (\xx) \, \vv_i,
  \end{align*}
  and thus, for all $i\in[0..d]$, $\barcvd_i=\calL^{\famvertd{\vv}}_i$.
\end{proof}

\begin{lemma}[decomposition of $\matPid$ polynomial with barycentric
  coordinates]
  \label{l:decomp-aff-map-with-baryc-coor}
  \mbox{}\\
  Let~$d\geq1$.
  Let~$\famvertd{\vv}$ be $d+1$ affinely independent points in~$\matRd$.
  Let~$p\in\matPid$.
  Then, we have
  \begin{equation}
    \label{e:decomp-aff-map-with-baryc-coor}
    p = \sum_{i = 0}^d p (\vv_i) \, \barcvd_i.
  \end{equation}
\end{lemma}

\begin{proof}
  Direct consequence of
  Lemma~\thref{l:decomp-P1d-pol-in-lag-basis}, and
  Lemma~\thref{l:lag-pol-P1d-is-baryc-coor}.
\end{proof}

\section{Hyperplanes and $l$-faces}
\label{s:hyperplanes-l-faces}

\subsection{Hyperplanes and $d-1$-faces}
\label{ss:hyperplanes-d-1-faces}

\begin{definition}[face hyperplane]
  \label{d:face-hyperpl}
  \mbox{}\\
  Let~$d\geq1$.
  Let~$\famvertd{\vv}$ be $d+1$ points in~$\matRd$.
  Let $i\in[0..d]$.
  The {\em $i$-th face hyperplane} opposite vertex~$\vv_i$ is
  denoted~$\calHvd_i$, and is defined as the affine subspace of~$\matRd$
  \begin{align}
    \label{e:face-hyperpl-0}
    (i = 0)&&
    \calHvd_0 &\eqdef \vv_1 + \Span{\vv_j - \vv_1}_{j \in [2..d]},&&\\
    \label{e:face-hyperpl-1}
    (i \in [1..d]&&
    \calHvd_i &\eqdef \vv_0
      + \Span{\vv_j - \vv_0}_{j \in [1..d] \setminus \{ i \}}.&&
  \end{align}
  For the reference vertices~$\famvertd{\hvv}$, the $i$-th reference face
  hyperplane is denoted~$\hcalH^d_i\eqdef\calH^{\famvertd{\hvv}}_i$.
\end{definition}

\begin{figure}[htb]
  \centering
  \resizebox{0.85\linewidth}{!}{
    \input{figtikz_fem_TetToTet_k3}
  }
  \caption[Geometric transformation]{%
    Geometric transformation~$\phigeodv$ in the case $d=k=3$.\\
    The reference simplex~$\hKd$ is mapped onto the current simplex~$\Kvd$,
    and reference nodes~$\haa_\aalpha$ onto current nodes~$\aa_\aalpha$, see
    Lemmas~\ref{l:prop-geo-mapping} and~\ref{l:lag-nodes-Pkd-im-ref}.
    For instance, we show that the reference hyperplane~$\hcalH^d_0$
    containing the face~$\hH^d_0$ is mapped onto~$\calHvd_0$ containing the
    face~$\Hvd_0$.
    The nodes in these two faces are colored in order to help see the
    correspondence: for all $\aalpha\in\calCiiidiii$, we have
    $\aa_\aalpha=\phigeodv(\haa_\aalpha)$.}
  \label{f:geo-map-d=3-hyperplane}
\end{figure}

\begin{lemma}[equivalent definition of face hyperplane]
  \label{l:equiv-def-face-hyperpl}
  \mbox{}\hfill
  Let~$d\geq1$.
  Let~$\famvertd{\vv}$ be $d+1$ points in~$\matRd$.
  Let $i\in[0..d]$.
  Then, the $i$-th face hyperplane is characterized by
  \begin{equation}
    \label{e:equiv-def-face-hyperpl-1}
    \calHvd_i = \left\{
      \sum_{j = 0, j \neq i}^d \mu_j \vv_j \in \matRd
    \rightst \left.
      \forall j \in [0..d] \setminus \{ i \},\quad \mu_j \in \matR
      \CONJ \sum_{j = 0, j \neq i}^d \mu_j = 1
      \right\} .
  \end{equation}

  Moreover, if $\famvertd{\vv}$ is affinely independent, then we have
  \begin{equation}
    \label{e:equiv-def-face-hyperpl-2}
    \calHvd_i = \Ker{\barcvd_i},
  \end{equation}
  where~$\barcvd_i$ is the $i$-th barycentric coordinate application.
\end{lemma}

\begin{proof}
  \proofpar{Identity~\eqref{e:equiv-def-face-hyperpl-1}}
  Direct consequence of
  Lemma~\threfc{l:baryc-closure-is-aff-sub-sp}{%
    with $n\eqdef d-1\geq0$, and $\vv_j$ for $j\in[0..d]\setminus\{i\}$}.

  \proofparskip{Identity~\eqref{e:equiv-def-face-hyperpl-2}}
  \proofpar{Case $i=0$}
  Then, from
  Definition~\thref{d:face-hyperpl},
  \assume{ring properties of~$\matR$},
  Lemma~\threfc{l:baryc-coor}{\eqref{e:baryc-coor-2} with $j\eqdef1$},
  we have
  \begin{align*}
    \xx \in \calHvd_0
    &\EQUIV \exists (\tilde{\mu}_j)_{j \in [2..d]} \in \matR,\quad
      \xx = \vv_1 + 0 (\vv_0 - \vv_1)
      + \sum_{j = 2}^d \tilde{\mu}_j (\vv_j - \vv_1)\\
    &\EQUIV \barcvd_0 (\xx) = 0.
  \end{align*}
  \proofpar{Case $i\in[1..d]$}
  Then, from
  Definition~\thref{d:face-hyperpl},
  \assume{ring properties of~$\matR$},
  Lem\-ma~\threfc{l:baryc-coor}{\eqref{e:baryc-coor-2} with $j\eqdef0$},
  we have
  \begin{align*}
    \xx \in \calHvd_i
    &\EQUIV
      \exists (\tilde{\mu}_j)_{j \in [1..d] \setminus \{ i \}} \in \matR,\quad
      \xx = \vv_0 + 0 (\vv_i - \vv_0)
      + \sum_{j = 1, j \neq i}^d \tilde{\mu}_j (\vv_j - \vv_0)\\
    &\EQUIV \barcvd_i (\xx) = 0.
  \end{align*}

  \medskip\noindent
  Therefore, from
  Definition~\thref{LM-d:kernel},
  the identity always holds.
\end{proof}

\begin{lemma}[reference face hyperplane]
  \label{l:ref-face-hyperpl}
  \mbox{}\hfill
  Let~$d\geq1$.
  Let $i\in[0..d]$.
  Then, we have
  \begin{align}
    \label{e:ref-face-hyperpl-0}
    (i = 0)&&
    \hcalH^d_0 &= \left\{
      \vphantom{\sum_{j = 1}^d \hx_j = 1}
      \hxx \in \matRd \rightst
      \left. \sum_{j = 1}^d \hx_j = 1 \right\}
    = \Ker{\hcalL^{1,d}_0},&&\\
    \label{e:ref-face-hyperpl-1}
    (i \in [1..d])&&
    \hcalH^d_i &= \left\{ \hxx \in \matRd \st \hx_i = 0 \right\}
    = \Ker{\hcalL^{1,d}_i}.&&
  \end{align}
\end{lemma}

\begin{proof}
  Direct consequence of
  Lemma~\thref{l:ref-affine-vert-is-affinely-indep},
  Lemma~\threfc{l:equiv-def-face-hyperpl}{\eqref{e:equiv-def-face-hyperpl-2}},
  Lemma~\thref{l:lag-pol-P1d-is-baryc-coor},
  Lemma~\thref{l:lag-pol-of-ref-vert-are-ref-lag-pol-P1d},
  Definition~\thref{d:lag-pol-P1d-ref}, and
  Definition~\thref{LM-d:kernel}.
\end{proof}

\begin{lemma}[face hyperplane is image of reference face hyperplane]
  \label{l:face-hyperpl-im-ref-face-hyperpl}
  \mbox{}\\
  Let~$d\geq1$.
  Let~$\famvertd{\vv}$ be $d+1$ affinely independent points in~$\matRd$.
  Let $i\in[0..d]$.
  Then, we have
  \begin{equation}
    \label{e:face-hyperpl-im-ref-face-hyperpl}
    \calHvd_i = \phigeodv \left( \hcalH^d_i \right).
  \end{equation}
\end{lemma}


\begin{proof}
  \proofpar{Case $i=0$}
  Then, from
  Definition~\thref{d:face-hyperpl},
  \assume{ring properties of~$\matR$},
  Definition~\threfc{d:lag-pol-P1d-ref}{%
    thus $\hcalL^{1,d}_j(\hmmu)=\mu_j$ for all $j\in[2..d]$,
    $\hcalL^{1,d}_1(\hmmu)$ equals $(1-\sum_{j=2}^d\mu_j)$ and
    $\hcalL^{1,d}_0(\hmmu)=0$},
  Definition~\thref{d:geo-mapping}, and
  Definition~\thref{d:fam-ref-aff-pts},
  we have
  \begin{align*}
    \xx \in \calHvd_0
    &\EQUIV \exists (\mu_j)_{j \in [2..d]} \in \matR,\quad
      \xx = \vv_1 + \sum_{j = 2}^d \mu_j (\vv_j - \vv_1)\\
    &\EQUIV \exists (\mu_j)_{j \in [2..d]} \in \matR,\quad
      \xx = 0 \vv_0 + \left( 1 - \sum_{j = 2}^d \mu_j \right) \vv_1
      + \sum_{j = 2}^d \mu_j \vv_j\\
    & \EQUIV \exists (\mu_j)_{j \in [2..d]} \in \matR,\quad
      \mu_1 \eqdef 1 - \sum_{j = 2}^d \mu_j \;\Conj\;
      \xx = \left( 1 - \sum_{j = 1}^d \mu_j \right) \vv_0
      + \sum_{j = 1}^d \mu_j \vv_j \\
    & \EQUIV \exists (\mu_j)_{j \in [2..d]} \in \matR,\quad
      \hmmu \eqdef \left( 1 - \sum_{j = 2}^d \mu_j \right) \ee_1
      + \sum_{j = 2}^d \mu_j \ee_j \in \matRd
      \;\Conj\; \xx = \phigeodv (\hmmu)\\
    &\EQUIV \exists (\mu_j)_{j \in [2..d]} \in \matR,\quad
      \hmmu \eqdef \hvv_1 + \sum_{j = 2}^d \mu_j (\hvv_j - \hvv_1)
      \in \hcalH^d_0 \CONJ \xx = \phigeodv (\hmmu).
  \end{align*}
  Therefore, we have the equality.

  \proofparskip{Case $i\in[1..d]$}
  Then, from
  Definition~\thref{d:face-hyperpl},
  \assume{ring properties of~$\matR$},
  Definition~\threfc{d:lag-pol-P1d-ref}{%
    thus $\hcalL^{1,d}_j(\hmmu)=\mu_j$ for all $j\in[1..d]\setminus\{i\}$,
    $\hcalL^{1,d}_0(\hmmu)=(1-\sum_{l=1,l\neq i}^d\mu_l)$ and
    $\hcalL^{1,d}_i(\hmmu)=0$},
  Definition~\thref{d:geo-mapping}, and
  Definition~\thref{d:fam-ref-aff-pts},
  we have
  \begin{align*}
    \xx \in \calHvd_i
    &\EQUIV \exists (\mu_j)_{j \in [1..d] \setminus \{ i \}} \in \matR,\quad
      \xx = \vv_0 + \sum_{j = 1, j \neq i}^d \mu_j (\vv_j - \vv_0)\\
    &\EQUIV \exists (\mu_j)_{j \in [1..d] \setminus \{ i \}} \in \matR,\quad
      \xx = \left( 1 - \sum_{j = 1, j \neq i}^d \mu_j \right) \vv_0
      + 0 \vv_i + \sum_{j = 1, j \neq i}^d \mu_j \vv_j\\
    &\EQUIV \exists (\mu_j)_{j \in [1..d] \setminus \{ i \}} \in \matR,\quad
      \mu_i \eqdef 0 \;\Conj\;
      \xx = \left( 1 - \sum_{j = 1}^d \mu_j \right) \vv_0
      + \sum_{j = 1}^d \mu_j \vv_j\\
    &\EQUIV \exists (\mu_j)_{j \in [1..d] \setminus \{ i \}} \in \matR,\quad
      \hmmu \eqdef 0\ee_i + \sum_{j = 1, j \neq i}^d \mu_j \ee_j \in \matRd
      \;\Conj\; \xx = \phigeodv (\hmmu)\\
    &\EQUIV \exists (\mu_j)_{j \in [1..d] \setminus \{ i \}} \in \matR,\quad
      \hmmu \eqdef \hvv_0 + \sum_{j = 1, j \neq i}^d \mu_j (\hvv_j - \hvv_0)
      \in \hcalH^d_i \;\Conj\; \xx = \phigeodv (\hmmu).
  \end{align*}
  Therefore, we have the equality.
\end{proof}

\begin{definition}[hyperface]
  \label{d:hyperface}
  \mbox{}\\
  Let~$d\geq1$.
  Let~$\famvertd{\vv}$ be $d+1$ points in~$\matRd$.
  Let $i\in[0..d]$.
  The {\em $i$-th hyperface} opposite vertex~$\vv_i$, having the
  vertices~$(\vv_j)_{j\in[0..d]\setminus\{i\}}$, is denoted~$\Hvd_i$, and is
  defined by
  \begin{equation}
    \label{e:hyperface}
    \Hvd_i \eqdef \left\{
      \sum_{j = 0, j \neq i}^d \mu_j \vv_j \in \matRd
    \rightst \left.
      \forall j \in [0..d] \setminus \{ i \},\quad \mu_j \geq 0
      \CONJ \sum_{j = 0, j \neq i}^d \mu_j = 1
    \right\}.
  \end{equation}
\end{definition}

\begin{lemma}[hyperface is included in face hyperplane]
  \label{l:hyperface-is-incl-face-hyperplane}
  \mbox{}\\
  Let~$d\geq1$.
  Let~$\famvertd{\vv}$ be $d+1$ points in~$\matRd$.
  Let $i\in[0..d]$.
  Then, we have $\Hvd_i\subset\calHvd_i$.
\end{lemma}

\begin{proof}
  Direct consequence of
  Definition~\thref{d:hyperface}, and
  Lemma~\thref{l:equiv-def-face-hyperpl}.
\end{proof}

\subsection{$l$-face affine spaces and $l$-faces}
\label{ss:l-aff-spaces-l-faces}

\begin{definition}[$l$-face affine space]
  \label{d:l-face-aff-space}
  \mbox{}\hfill
  Let~$d\geq1$.
  Let~$\famvertd{\vv}$ be $d+1$ points in~$\matRd$.\\
  Let $l\in[0..d]$.
  Let~$\indl$ be an injective map from $[0..l]$ to $[0..d]$.
  The {\em $l$-face affine space associated with~$\famvertd{\vv}$ and having
    the vertices
    $\famvert{\vv_\indl}{l,d}\eqdef(\vv_{\indl(j)})_{j\in[0..l]}$}
  is denoted~$\calFvdindl$, and is defined by
  \begin{equation}
    \label{e:l-face-aff-space}
    \calFvdindl \eqdef \left\{
      \sum_{j = 0}^l \mu_j \vv_{\indl (j)} \in \matRd
    \rightst \left.
      \forall j \in [0..l],\quad \mu_j \in \matR \CONJ \sum_{j = 0}^l \mu_j = 1
    \right\}.
  \end{equation}
\end{definition}

\begin{lemma}[equivalent definition of $l$-face affine space]
  \label{l:equiv-def-l-face-aff-space}
  \mbox{}\hfill
  Let~$d\geq1$.\\
  Let~$\famvertd{\vv}$ be $d+1$ points in~$\matRd$.
  Let $l\in[0..d]$.
  Let~$\indl$ be an injective map from $[0..l]$ to $[0..d]$.\\
  Then, $\calFvdindl$ is an affine space, that can be equivalently defined for
  any $i\in[0..l]$ by
  \begin{equation}
    \label{e:equiv-def-l-face-aff-space}
    \calFvdindl = \vv_{\indl (i)}
    + \Span{\vv_{\indl (j)} - \vv_{\indl (i)}}_{j \in [0..l] \setminus \{ i \}}.
  \end{equation}
\end{lemma}

\begin{proof}
  Direct consequence of
  Definition~\thref{d:l-face-aff-space}, and
  Lemma~\threfc{l:baryc-closure-is-aff-sub-sp}{%
    with $n\eqdef l$ and $\vv_j\eqdef\vv_{\indl(j)}$ for $j\in[0..l]$}.
\end{proof}

\begin{lemma}[$d$-face affine space is full space]
  \label{l:d-face-aff-space-is-full-space}
  \mbox{}\hfill
  Let~$d\geq1$.
  Let~$\famvertd{\vv}$ be $d+1$ affinely independent points in~$\matRd$.
  Let~$\indd$ be an injective map from $[0..d]$ to $[0..d]$.
  Then, $\calFvdindd=\matRd$.
\end{lemma}

\begin{proof}
  Direct consequence of
  Lemma~\threfc{l:equiv-def-l-face-aff-space}{with $i\eqdef0$},
  Lemma~\threfc{l:aff-indep-closed-by-sub-family}{with $j\eqdef0$},
  Definition~\threfc{d:aff-indep-family}{with $\famvert{\vv_\indd}{d}$},
  Lemma~\thref{l:free-family-of-dim-elements-is-basis}, and
  \assume{$\matRd$ is affine space with any origin}.
\end{proof}

\begin{lemma}[$0$-face affine space is vertex]
  \label{l:0-face-aff-space-is-vertex}
  \mbox{}\hfill
  Let~$d\geq1$.
  Let~$\famvertd{\vv}$ be $d+1$ points in~$\matRd$.
  Let~$\ind{0}$ be a map from~$\{0\}$ to $[0..d]$.
  Then, $\calFvdindo$ is the affine space
  $\{\vv_{\ind{0}(0)}\}=\vv_{\ind{0}(0)}+\{\zzero\}$.
\end{lemma}

\begin{proof}
  Direct consequence of
  \assume{the injectivity of a function defined on singleton}, and
  Lem\-ma~\threfc{l:equiv-def-l-face-aff-space}{with $l\eqdef0$}.
\end{proof}

\begin{definition}[$l$-face]
  \label{d:l-face}
  \mbox{}\hfill
  Let~$d\geq1$.
  Let~$\famvertd{\vv}$ be $d+1$ points in~$\matRd$.\\
  Let $l\in[0..d]$.
  Let~$\indl$ an injective map from $[0..l]$ to $[0..d]$.
  The {\em $l$-face associated with~$\famvertd{\vv}$ and having the vertices
    $\famvert{\vv_\indl}{l,d}$} is denoted~$\Fvdindl$, and is defined by
  \begin{equation}
    \label{e:l-face}
    \Fvdindl \eqdef \left\{
      \sum_{j = 0}^l \mu_j \vv_{\indl (j)} \in \matRd
    \rightst \left.
      \forall j \in [0..l],\quad \mu_j \geq 0 \CONJ \sum_{j = 0}^l \mu_j = 1
    \right\}.
  \end{equation}
\end{definition}

\begin{lemma}[$l$-face is included in $l$-face affine space]
  \label{l:l-face-is-incl-l-face-aff-space}
  \mbox{}\hfill
  Let~$d\geq1$.\\
  Let~$\famvertd{\vv}$ be $d+1$ points in~$\matRd$.
  Let $l\in[0..d]$.
  Let~$\indl$ an injective map from $[0..l]$ to $[0..d]$.\\
  Then, we have $\Fvdindl\subset\calFvdindl$.
\end{lemma}

\begin{proof}
  Direct consequence of
  Definition~\thref{d:l-face}, and
  Definition~\thref{d:l-face-aff-space}.
\end{proof}

\begin{lemma}[$d$-face is simplex]
  \label{l:l-face-is-simplex}
  \mbox{}\\
  Let~$d\geq1$.
  Let~$\famvertd{\vv}$ be $d+1$ points in~$\matRd$.
  Let~$\indd$ be an injective map from $[0..d]$ to $[0..d]$.\\
  Then, $\indd$ is a bijective permutation of $[0..d]$, and we have
  $\Fvdindd=\Kvd$.
\end{lemma}

\begin{proof}
  \proofpar{Bijectivity}\\
  Direct consequence of
  \assume{the fact that injectivity and cardinal equality imply bijectivity}.

  \proofparskip{Equality}
  From
  Definition~\thref{d:l-face}, and
  Definition~\thref{d:simplex},
  we have
  \begin{align*}
    \Fvdindd
    &= \left\{
      \sum_{j = 0}^d \mu_j \vv_{\indd (j)} \in \matRd
      \rightst \left.
      \forall j \in [0..d],\quad \mu_j \geq 0 \CONJ \sum_{j = 0}^d \mu_j = 1
      \right\}\\
    &= \left\{
      \sum_{i = 0}^d \lambda_i \vv_i \in \matRd
      \rightst \left.
      \forall i \in [0..d],\quad \lambda_i \geq 0
      \CONJ \sum_{i = 0}^d \lambda_i = 1
      \right\} = \Kvd,
  \end{align*}
  with $i\eqdef\indd(j)\Equiv j=\inddinv(i)$ ($\indd$ bijective in $[0..d]$)
  and $\lambda_i\eqdef\mu_{\inddinv(i)}$.
\end{proof}

\begin{remark}
  In the next lemma, the injection $\thetinjdmi{i}$ from~$[0..d-1]$ to~$[0..d]$
  is the ``jump'' enumeration function defined in Lemma~\ref{l:jump-enum}.
  The $(d-1)$-face is used to build the geometric hyperface mapping in
  Lemma~\ref{l:geo-hyperface-mapping}, see also
  Figure~\ref{f:geo-hyper-map_d32}.
\end{remark}

\begin{lemma}[$(d-1)$-face is hyperface]
  \label{l:d-1-face-aff-sp-is-hyperface-aff-sp}
  \mbox{}\\
  Let~$d\geq1$, and $i\in[0..d]$.
  Let~$\famvertd{\vv}$ be $d+1$ points in~$\matRd$.\\
  Then, we have
  $\famvert{\vv_{\thetinjdmi{i}}}{d-1,d}=(\vv_{j})_{j\in[0..d]\setminus\{i\}}$.

  Moreover, if $\famvertd{\vv}$ is affinely independent, then we have
  $\calHvd_i=\calFvd{\thetinjdmi{i}}$ and $\Hvd_i=\Fvd{\thetinjdmi{i}}$.
\end{lemma}

\begin{proof}
  \mbox{}\\
  Direct consequence of
  Lemma~\threfc{l:jump-enum}{%
    with $n\eqdef d-1$, injectivity and image of~$\thetinjdmi{i}$},
  Definition~\threfc{d:l-face-aff-space}{%
    with $l\eqdef d-1$ and $\indl\eqdef\thetinjdmi{i}$},
  Lemma~\thref{l:equiv-def-face-hyperpl},
  Definition~\thref{d:hyperface}, and
  Definition~\thref{d:l-face}.
\end{proof}

\subsection{Geometric $l$-face mappings}
\label{ss:geomap-l-face}

\begin{remark}
  In this section, we assume that~$l\geq1$ to avoid treating polynomials over
  $0$-dimensional affine spaces.
\end{remark}

\begin{definition}[geometric $l$-face mapping]
  \label{d:geo-l-face-mapping}
  \mbox{}\hfill
  Let~$d\geq1$.
  Let~$\famvertd{\vv}$ be $d+1$ points in~$\matRd$.\\
  Let $l\in[1..d]$.
  Let~$\indl$ be an injective map from $[0..l]$ to $[0..d]$.
  The {\em geometric $l$-face mapping associated with~$\famvertd{\vv}$ and
    the $l$-face~$\Fvdindl$} is denoted $\phindlv$, and is defined by
  \begin{equation}
    \label{e:geo-l-face-mapping}
    \forall \hxx \in \matRl,\quad
    \phindlv (\hxx) \eqdef
    \sum_{j = 0}^l \hcalL^{1, l}_j (\hxx) \, \vv_{\indl (j)} \in \matRd.
  \end{equation}
\end{definition}

\begin{lemma}[geometric $d$-face mapping is geometric mapping]
  \label{l:geo-d-face-mapping-is-geo-mapping}
  \mbox{}\\
  Let~$d\geq1$.
  Let~$\famvertd{\vv}$ be $d+1$ points in~$\matRd$.
  Then, we have $\phind{\identity{\matRd}}{\vv}=\phigeodv$.
\end{lemma}

\begin{proof}
  Direct consequence of
  \assume{injectivity of the identity map},
  Definition~\thref{d:geo-l-face-mapping}, and
  Definition~\thref{d:geo-mapping}.
\end{proof}

\begin{remark}
  The next lemma (on~$\phindlv$) is an extension of
  Lemma~\ref{l:prop-geo-mapping} (on~$\phigeodv$).
\end{remark}

\begin{lemma}[properties of geometric $l$-face mapping]
  \label{l:prop-geo-l-face-mapping}
  \mbox{}\hfill
  Let~$d\geq1$.
  Let~$\famvertd{\vv}$ be $d+1$ points in~$\matRd$.
  Let $l\in[1..d]$.
  Let~$\indl$ an injective map from $[0..l]$ to $[0..d]$.\\
  Let $\Bindlv\eqdef\begin{pmatrix}\vv_{\indl(1)}-\vv_{\indl(0)}
    &\vv_{\indl(2)}-\vv_{\indl(0)}&\dots&\vv_{\indl(l)}-\vv_{\indl(0)}\end{pmatrix}$
  in~$\calM_{d,l}$.\\
  Then, $\phindlv$ belongs to~$\AffRlRd$, and we have
  \begin{align}
    \label{e:prop-geo-l-face-mapping-1}
    \forall j \in [0..l],\quad
    \forall \hxx \in \matRl,\quad
    &\phindlv (\hxx) = \vv_{\indl (j)} +
      \sum_{k = 0, k \neq j}^l \hcalL^{1, l}_k (\hxx)
      (\vv_{\indl (k)} - \vv_{\indl (j)}),\\
    \label{e:prop-geo-l-face-mapping-2}
    \forall \hxx \in \matRl,\quad
    &\phindlv (\hxx)
      = \vv_{\indl (0)}
      + \sum_{j = 1}^l \hx_j (\vv_{\indl (j)} - \vv_{\indl (0)})
      = \vv_{\indl (0)} + \Bindlv \hxx,\\
    \label{e:prop-geo-l-face-mapping-3}
    \forall j \in [0..l],\quad
    &\phindlv (\hvv_j^l) = \vv_{\indl (j)}
      \mbox{ (where $\hvv_j^{l}$ is a reference vertex in $\matRlmi$)},\\
    \label{e:prop-geo-l-face-mapping-4}
    &\phindlv (\matRl) = \calFvdindl,\\
    \label{e:prop-geo-l-face-mapping-4bis}
    &\phindlv (\hKl) = \Fvdindl.
  \end{align}

  Moreover, if $\famvertd{\vv}$ is affinely independent, then $\phindlv$~is
  bijective from~$\matRl$ onto~$\calFvdindl$, its inverse~$\invphindlv$
  from~$\calFvdindl$ onto~$\matRl$ is affine, and we have
  \begin{equation}
    \label{e:prop-geo-l-face-mapping-5}
    \forall j \in [0..l],\quad \invphindlv (\vv_{\indl (j)}) = \hvv_j^{l}.
  \end{equation}
\end{lemma}

\begin{proof}
  \proofpar{Identities~\eqref{e:prop-geo-l-face-mapping-1}
    and~\eqref{e:prop-geo-l-face-mapping-2}, $\phindlv$ affine}
  Let $j\in[0..l]$.
  Then, from
  Definition~\thref{d:geo-l-face-mapping},
  Lemma~\threfc{l:lag-pol-is-basis-P1d-ref}{%
    \eqref{e:lag-pol-is-basis-P1d-ref-3} with $d\eqdef l$, extracting $j$}, and
  \assume{ring properties of~$\matR$},
  identity~\eqref{e:prop-geo-l-face-mapping-1} holds.
  Thus, from~\eqref{e:prop-geo-l-face-mapping-1} (with $j\eqdef0$),
  Definition~\threfc{d:lag-pol-P1d-ref}{\eqref{e:lag-pol-P1d-ref-1}},
  \assume{the columnwise rule for matrix-vector product},
  \assume{commutative ring properties of~$\matR$}, and
  Definition~\threfc{d:aff-map}{%
    with $\phi\eqdef\Bindlv$ and $\yyo\eqdef\vv_{\indl(0)}$},
  identities~\eqref{e:prop-geo-l-face-mapping-2} hold, and $\phindlv$~is
  affine.

  \proofparskip{Identity~\eqref{e:prop-geo-l-face-mapping-3}}
  Direct consequence of
  Definition~\thref{d:geo-l-face-mapping}, and
  Lem\-ma~\threfc{l:lag-pol-is-basis-P1d-ref}{%
    \eqref{e:lag-pol-is-basis-P1d-ref-2}}.

  \proofparskip{Identity~\eqref{e:prop-geo-l-face-mapping-4}}
  \proofpar{$\phindlv(\matRl)\subset\calFvdindl$}\\
  Direct consequence of
  Definition~\thref{d:geo-l-face-mapping},
  Lemma~\threfc{l:lag-pol-is-basis-P1d-ref}{%
    \eqref{e:lag-pol-is-basis-P1d-ref-3}}, and
  Definition~\thref{d:l-face-aff-space}.\\
  \proofpar{$\calFvdindl\subset\phindlv(\matRl)$}
  From
  Definition~\thref{d:l-face-aff-space},
  let $\xx=\sum_{j=0}^l\mu_j\vv_{\indl(j)}\in\calFvdindl$ with
  $\sum_{j=0}^l \mu_j=1$.
  Let $\hxx\eqdef(\mu_j)_{j\in[1..l]}\in\matRl$.
  Then, from
  Definition~\thref{d:lag-pol-P1d-ref}, and
  Definition~\thref{d:geo-l-face-mapping},
  we have for all $j\in[1..l]$, $\hcalL^{1,l}_j(\hxx)=\mu_j$ and
  $\hcalL^{1,l}_0(\hxx)=1-\sum_{j=1}^l\mu_j=\mu_0$, thus
  $\xx=\phindlv(\hxx)\in\phindlv(\matRl)$.\\
  Therefore, we have the equality.

  \proofparskip{Identity~\eqref{e:prop-geo-l-face-mapping-4bis}}
  \proofpar{$\phindlv(\hKl)\subset\Fvdindl$}
  Direct consequence of
  Definition~\thref{d:geo-l-face-mapping},
  Definition~\thref{d:ref-simplex},
  Definition~\threfc{d:lag-pol-P1d-ref}{%
    thus for all $\hxx\in\hKl$, for all $i\in[0..l]$,
    $\hcalL^{1,l}_i(\hxx)\geq0$}, and
  Lemma~\threfc{l:lag-pol-is-basis-P1d-ref}{%
    \eqref{e:lag-pol-is-basis-P1d-ref-3}}, and
  Definition~\thref{d:l-face}.\\
  \proofpar{$\Fvdindl\subset\phindlv(\hKl)$}
  From
  Definition~\thref{d:l-face},
  let $\xx=\sum_{j=0}^l \mu_j\vv_{\indl(j)}\in\Fvdindl$ with for all
  $j\in[0..l]$, $\mu_j\geq0$ and $\sum_{j=0}^l\mu_j=1$.
  Let $\hxx\eqdef(\mu_j)_{j\in[1..l]}\in\matRl$.
  Then, from
  Definition~\thref{d:lag-pol-P1d-ref},
  Definition~\thref{d:ref-simplex}, and
  Definition~\thref{d:geo-l-face-mapping},
  we have
  \[
    \forall j \in [1..l],\; \hx_j = \hcalL^{1, l}_j (\hxx) = \mu_j \geq 0
    \AND
    1 - \sum_{j = 1}^l \hx_j = \hcalL^{1, l}_0 (\hxx)
    = 1 - \sum_{j = 1}^l \mu_j = \mu_0 \geq 0,
  \]
  thus $\hxx\in\hKl$ and $\xx=\phindlv(\hxx)\in\phindlv(\hKl)$.\\
  Therefore, we have the equality.

  \proofparskip{Bijection}
  Direct consequence of
  \eqref{e:prop-geo-l-face-mapping-2},
  Lemma~\thref{l:inj-aff-map-is-zero-linear-ker},
  Definition~\thref{LM-d:kernel},
  Lemma~\threfc{l:aff-indep-closed-by-sub-family}{%
    with $j\eqdef0$, thus $(\vv_{\indl(i)-}\vv_{\indl(0)})_{i\in[1..l]}$ is free},
  \eqref{e:prop-geo-l-face-mapping-4}, and
  \assume{the definition of bijectivity}.

  \proofparskip{$\invphindlv$ affine}
  Direct consequence of~\eqref{e:prop-geo-l-face-mapping-2},
  Lemma~\threfc{l:inv-of-aff-sub-map-is-aff-sub-map}{%
    with $E=\Ep=\calEp\eqdef\matRl$ ($\xxo\eqdef\zzero$), $F\eqdef\matRd$,
    $\calFp\eqdef\calFvdindl$, and
    $\yypo\eqdef\vv_{\indl(0)}\in\calFvdindl$}, and
  Lemma~\thref{l:equiv-def-l-face-aff-space}.

  \proofparskip{Equation~\eqref{e:prop-geo-l-face-mapping-5}}
  Direct consequence of
  Equation~\eqref{e:prop-geo-l-face-mapping-3}.
\end{proof}

\begin{lemma}[geometric $l$-face mapping of $\matPkd$ is $\matPkl$]
  \label{l:geo-l-face-mapping-of-Pkd-is-Pkl}
  \mbox{}\\
  Let~$d\geq1$.
  Let~$k\in\matN$.
  Let~$\famvertd{\vv}$ be $d+1$ points in~$\matRd$.
  Let $l\in[1..d]$.
  Let~$\indl$ an injective map from $[0..l]$ to $[0..d]$.
  Let $p\in\matPkd$.
  Then, we have $p\circ\phindlv\in\matPkl$.
\end{lemma}

\begin{proof}
  Direct consequence of
  Lemma~\threfc{l:prop-geo-l-face-mapping}{%
    thus $\phindlv$ belongs to $\AffRlRd$}, and
  Lemma~\thref{l:aff-mapping-of-Pkd-is-Pkl}.
\end{proof}

\begin{lemma}[geometric mapping of $\matPkd$ is $\matPkd$]
  \label{l:geo-mapping-of-Pkd-is-Pkd}
  \mbox{}\\
  Let~$d\geq1$.
  Let~$k\in\matN$.
  Let~$\famvertd{\vv}$ be $d+1$ points in~$\matRd$.
  Let~$p\in\matPkd$.
  Then, we have $p\circ\phigeodv\in\matPkd$.

  Moreover, if $\famvertd{\vv}$ is affinely independent,
  then~$p\circ\invphigeodv$ also belongs to~$\matPkd$.
\end{lemma}

\begin{proof}
  Direct consequence of
  Lemma~\thref{l:geo-d-face-mapping-is-geo-mapping},
  Lemma~\threfc{l:geo-l-face-mapping-of-Pkd-is-Pkl}{%
    with $l\eqdef d$ and $\indl\eqdef\identity{\matRd}$},
  Lemma~\threfc{l:prop-geo-mapping}{thus $\invphigeodv$ is affine}, and
  Lemma~\threfc{l:aff-mapping-of-Pkd-is-Pkl}{with $l\eqdef d$}.
\end{proof}

\subsection{Geometric hyperface mapping}
\label{ss:geomap-hyperface}

\begin{figure}[htb]
  \centering
  \resizebox{0.85\linewidth}{!}{
    \input{figtikz_fem_ImTriaFaceTet_k3}
  }
  \caption[Geometric hyperface mapping]{%
    Geometric hyperface mapping in the case $d=k=3$ (see
    Lemma~\ref{l:geo-hyperface-mapping}).\\
    The reference triangle~$\hK_2$ is mapped onto~$\Hvd_0$, the $0$-th face
    opposite vertex~$\vv_0$ that is depicted in blue.
    The correspondence between the reference nodes in the triangle ($d=2$) and
    the face nodes of the tetrahedron ($d=3$) is illustrated by the colors:
    we have $\phindv{\thetinjdmi{0}}(\haa_{i,j})=\aa_{(3-(i+j),i.j)}$, for all
    $(i,j)\in\calAiidiii$ (see
    Lemma~\ref{l:im-nodes-by-geo-hyperface-mapping}).}
  \label{f:geo-hyper-map_d32}
\end{figure}

\begin{lemma}[geometric hyperface mapping]
  \label{l:geo-hyperface-mapping}
  \mbox{}\hfill
  Let~$d\geq1$.
  Let~$\famvertd{\vv}$ be $d+1$ points in~$\matRd$.\\
  Let~$i\in[0..d]$.
  Let~$\hxx \in \matRdmi$.
  Then, we have $\phindv{\thetinjdmi{i}}\in\AffRdmiRd$ and
  \begin{align}
    \label{e:geo-hyperface-mapping-0}
    (i = 0)&&
    \phindv{\thetinjdmi{0}} (\hxx) &= \vv_1 +
      \sum_{j = 1}^{d - 1} \hx_j (\vv_{j + 1} - \vv_1),&&\\
    \label{e:geo-hyperface-mapping-1}
    (i \in [1..d - 1])&&
    \phindv{\thetinjdmi{i}} (\hxx) &= \vv_0 +
      \sum_{j = 1}^{i - 1} \hx_j (\vv_j - \vv_0) +
      \sum_{j = i}^{d - 1} \hx_j (\vv_{j + 1} - \vv_0),&&\\
    \label{e:geo-hyperface-mapping-2}
    \forall j \in [0..d] && \phindv{\thetinjdmi{i}} (\hvv_j^{d - 1}) &= \left\{
      \begin{array}{ll}
        \vv_j & \mbox{if } j < i,\\
        \vv_{j + 1} & \mbox{if } j \geq i.
      \end{array}
      \right.&&
  \end{align}

  Moreover, if $\famvertd{\vv}$ is affinely independent, then
  $\phindv{\thetinjdmi{i}}$~is bijective from~$\matRdmi$ onto~$\calHvd_i$, and
  its inverse $\left(\phindv{\thetinjdmi{i}}\right)^{-1}$ from~$\calHvd_i$
  onto~$\matRdmi$ is affine.
\end{lemma}

\begin{proof}
  Direct consequence of
  Lemma~\thref{l:d-1-face-aff-sp-is-hyperface-aff-sp}, and
  Lemma~\thref{l:prop-geo-l-face-mapping}.
\end{proof}

\begin{lemma}[hyperface geometric mapping of $\matPkd$ is $\matPkdmi$]
  \label{l:hyperface-geo-mapping-of-Pkd-is-Pkdmi}
  \mbox{}\hfill
  Let $d\geq2$.
  Let~$k\in\matN$.\\
  Let~$\famvertd{\vv}$ be $d+1$ points in~$\matRd$.
  Let $i\in[0..d]$.
  Let $p\in\matPkd$.
  Then, we have $p\circ\phindv{\thetinjdmi{i}}\in\matPkdmi$.
\end{lemma}

\begin{proof}
  Direct consequence of
  Lemma~\thref{l:d-1-face-aff-sp-is-hyperface-aff-sp}, and
  Lemma~\threfc{l:geo-l-face-mapping-of-Pkd-is-Pkl}{%
    with $l\eqdef d-1$ and $\indl\eqdef\thetinjdmi{i}$}.
\end{proof}

\subsection{Geometric mapping with permutation}
\label{ss:geomap-perm}

\begin{figure}[htb]
  \centering
  \resizebox{0.85\linewidth}{!}{
    \begin{tikzpicture}[scale=4,math3d] 

        \def\kk{3}

        \def\colk{black}
        \def\colo{magenta}
        \def\coli{darkgreen}
        \def\colii{red}
        \def\coliii{blue}

        \def\opacity{0.4}
        \def\opacityi{0.6}
        \def\opacityo{0.2}

	\coordinate (A) at (0,0,0);
	\coordinate (B) at  ($ (A) + (1,0,0) $);
	\coordinate (C) at  ($ (A) + (0,1,0) $);
	\coordinate (D) at  ($ (A) + (0,0,1) $);
        \draw (A) node[left] {$\hvv_0$} ;  
        \draw (B) node[above left] {$\hvv_1$} ; 
        \draw (C) node[above=2pt] {$\hvv_2$} ; 
        \draw (D) node[right=1.5pt] {$\hvv_3$} ; 
	\draw[line width=1.0pt,rounded corners=0.5pt] (A) -- (B) -- (D) -- cycle;
	\draw[line width=1.0pt,fill=\coliii!40,fill opacity=\opacity,rounded corners=0.5pt] (A) -- (C) -- (B) -- cycle;
	\draw[line width=1.0pt,rounded corners=0.5pt] (B) -- (C) -- (D) -- cycle;
	\draw[line width=1.6pt,rounded corners=0.5pt] (A) -- (C) -- (D) -- cycle;

        \pgfmathparse{1-1/\kk}\let\kkp\pgfmathresult
        \coordinate (B1) at ($ \kkp*(A) + 1/\kk*(B) $) ;
        \coordinate (C1) at ($ \kkp*(A) + 1/\kk*(C) $) ;
        \coordinate (D1) at ($ \kkp*(A) + 1/\kk*(D) $) ;
        \pgfmathparse{1-2/\kk}\let\kkp\pgfmathresult
        \coordinate (B2) at ($ \kkp*(A) + 2/\kk*(B) $) ;
        \coordinate (C2) at ($ \kkp*(A) + 2/\kk*(C) $) ;
        \coordinate (D2) at ($ \kkp*(A) + 2/\kk*(D) $) ;

        %

        \coordinate (E) at ($ -0.15*(A) - 0.15*(B) + 1.3*(C) $) ;
        \coordinate (F) at ($ -2.6*(A) + 1.3*(B) + 1.3*(C) $) ;
        \coordinate (G) at ($ -0.15*(A) + 1.3*(B) - 0.15*(C) $) ;
        \coordinate (H) at ($ 1.4*(A) - 0.2*(B) - 0.2*(C) $) ;

        \draw[color=\coliii,line width=0.2pt,fill=\coliii!40,fill
          opacity=\opacityo]  (E) -- (F) -- (G) -- (H) -- cycle;

        \node [color=\coliii] (K3) at ($ (A) + (0.15,0.15,0) $) {$\hH^d_d$};
        \node [color=\coliii] (K33) at ($ (F) - (0.3,0.3,0) $) {$\hcalH^d_d$};
        \node [color=\colk] (K33) at ($ -0.3*(C) + 0.6*(D) $) {$\hKd$};

        \fill[color=\colk,fill opacity=1.0] (A) circle (0.6pt);
        \foreach \x in {3,2,...,0}  {
          \pgfmathparse{\kk-\x}\let\YY\pgfmathresult
          \foreach \y in {0,...,\YY} {
            \pgfmathparse{\kk-\x-\y}\let\z\pgfmathresult
            \fill[color=\colk,fill opacity=1.0] ($ (A) + 1/\kk*(\x,\y,\z) $) circle (0.6pt);
          }
        }
	\coordinate (nA) at (0,1.8,0 );
	\coordinate (nB) at  ($ (nA) + (-1.7,0.3,-0.4) $);
	\coordinate (nC) at  ($ (nA) + (0,0.9,-0.4) $);
	\coordinate (nD) at  ($ (nA) + (0.05,0,1.3) $);
        \draw (nA) node[below=2.5pt] {$\vv_0=\phipermcd{0}(\hvv_3)$} ;  
        \draw (nB) node[above] {$\vv_1=\phipermcd{0}(\hvv_0)$} ; 
        \draw (nC) node[below] {$\vv_2=\phipermcd{0}(\hvv_1)$} ; 
        \draw (nD) node[right=1.5pt] {$\vv_3=\phipermcd{0}(\hvv_2)$} ; 
	\draw[line width=1.0pt,rounded corners=0.5pt] (nA) -- (nB) -- (nD) -- cycle;
	\draw[line width=1.0pt,rounded corners=0.5pt] (nA) -- (nC) -- (nB) -- cycle;
	\draw[line width=1.0pt,fill=\coliii!40,fill opacity=\opacity,rounded corners=0.5pt] (nB) -- (nC) -- (nD) -- cycle;
	
        \coordinate (nE) at ($ 0.3*(nB) - 0.45*(nC) + 1.15*(nD) $) ;
        \coordinate (nF) at ($ -0.4*(nB) + 0.2*(nC) + 1.15*(nD) $) ;
        \coordinate (nG) at ($ -0.15*(nB) + 1.3*(nC) - 0.15*(nD) $) ;
        \coordinate (nH) at ($ 1.3*(nB) - 0.15*(nC) - 0.15*(nD) $) ;

        \draw[color=\coliii,line width=0.2pt,fill=\coliii!40,fill
          opacity=\opacityo]  (nE) -- (nF) -- (nG) -- (nH) -- cycle;

        \node [color=\coliii] (nK3) at ($ (nD) + (0,0.1,0.2) $) {$\calHvd_0$};

        \pgfmathparse{1-1/\kk}\let\kkp\pgfmathresult
        \coordinate (nB1) at ($ \kkp*(nA) + 1/\kk*(nB) $) ;
        \coordinate (nC1) at ($ \kkp*(nA) + 1/\kk*(nC) $) ;
        \coordinate (nD1) at ($ \kkp*(nA) + 1/\kk*(nD) $) ;
        \pgfmathparse{1-2/\kk}\let\kkp\pgfmathresult
        \coordinate (nB2) at ($ \kkp*(nA) + 2/\kk*(nB) $) ;
        \coordinate (nC2) at ($ \kkp*(nA) + 2/\kk*(nC) $) ;
        \coordinate (nD2) at ($ \kkp*(nA) + 2/\kk*(nD) $) ;


        \node [color=\coliii] (nK33) at ($ (nD) + (0,0.55,-0.5) $) {$\Hvd_0$};
        \node [color=\colk] (nK333) at ($ 0.6*(nA) + 0.4*(nD) + (0,-0.15,0) $) {$\Kvd$};

        \foreach \x in {3,2,1,...,0}  {
          \pgfmathparse{3-\x}\let\YY\pgfmathresult
          \foreach \y in {0,...,\YY} {
            \pgfmathparse{3-\x-\y}\let\z\pgfmathresult
            \pgfmathparse{\x/\kk}\let\xkk\pgfmathresult
            \pgfmathparse{\y/\kk}\let\ykk\pgfmathresult
            \pgfmathparse{\z/\kk}\let\zkk\pgfmathresult
            \pgfmathparse{1-\xkk-\ykk-\zkk}\let\okk\pgfmathresult
            \fill[color=\colk,fill opacity=1.0]
            ($ \okk*(nA) + \xkk*(nB) + \ykk*(nC) + \zkk*(nD) $) circle (0.6pt);
          }
        }
        \newcount\x
        \pgfmathsetcount{\x}{0}
        \foreach \k in {0,...,\kk}  {
          \foreach \y in {0,...,\k} {
            \pgfmathparse{\k-\x-\y}\let\z\pgfmathresult
            \pgfmathparse{\x/\kk}\let\xkk\pgfmathresult
            \pgfmathparse{\y/\kk}\let\ykk\pgfmathresult
            \pgfmathparse{\z/\kk}\let\zkk\pgfmathresult
            \pgfmathparse{1-\xkk-\ykk-\zkk}\let\okk\pgfmathresult
            \fill[color=\colk,fill opacity=1.0]
            ($ \okk*(nA) + \xkk*(nB) + \ykk*(nC) + \zkk*(nD) $) circle (0.6pt);
          }
        }

        \tikzstyle{fl}=[->,>=latex,very thick];
        \node (K2DD) at  ($  1/2*(C) + 1/2*(D) $) {} ;
        \node (nK3D) at  ($ 0.7*(nD) + 0.3*(nA) $) {};
        \draw[fl] (K2DD) to[bend left] (nK3D);
        \node (fkd1) at  ($ (C) + (0,0.1,0.8) $) {$\phipermcd{0}$};

\end{tikzpicture}
  }
  \caption[Geometric transformation with a circular permutation]{%
    Geometric transformation~$\phipermcd{0}$ in the case $d=k=3$, with the
    circular permutation~$\permcdo$ that maps $(0,1,2,3)$ to
    $(1,2,3,0)$, see Lemma~\ref{l:circ-permut}.\\
    The reference simplex~$\hKd$ is mapped onto the current simplex~$\Kvd$, see
    Lemma~\ref{l:geo-mapping-permut}.
    The reference vertices are mapped onto the current vertices with a change
    of indices: for all $j\in[0..d]$, we have
    $\phipermcd{0}(\hvv_j)=\vv_{\permcdo(j)}$, see
    Lemma~\ref{l:prop-geo-l-face-mapping}.
    Note that as~$\permcdo(d)=0$, the reference hyperplane~$\hcalH^d_d$
    (containing the face~$\hH^d_d$ opposite vertex~$\hvv_d$) is mapped
    onto~$\calHvd_0$ (containing the face~$\Hvd_0$ opposite vertex~$\vv_0$).\\
    In this figure, only the visible nodes are depicted.}
  \label{f:geo-map-perm-circ0-d=3-hyperplane}
\end{figure}


\begin{lemma}[geometric mapping with permutation]
  \label{l:geo-mapping-permut}
  \mbox{}\hfill
  Let~$d\geq1$.
  Let~$\famvertd{\vv}$ be $d+1$\\ affinely independent points in~$\matRd$.
  Let~$\indd$ be an injective mapping from $[0..d]$ to~$[0..d]$.
  Let~$\hxx\in\matRd$.\\
  Then, $\indd$~is a bijective permutation of $[0..d]$, $\phinddv\in\AffRdRd$
  is bijective, and we have
  \begin{align}
    \label{e:geo-mapping-permut-1}
    \phinddv (\hxx)
    &= \sum_{i = 0}^d \hcalL^{1, d}_{\inddinv (i)} (\hxx) \, \vv_i\\
    \label{e:geo-mapping-permut-1bis}
    \forall i \in [0..d],\quad
    \phinddv (\hxx)
    &= \vv_i +
      \sum_{k = 0, k \neq i}^d \hcalL^{1, d}_{\inddinv (k)} (\hxx)
      (\vv_k - \vv_i),\\
    \label{e:geo-mapping-permut-2}
    \forall i \in [0..d],\quad
    \barcvd_i = \calL^{\famvertd{\vv}}_i
    &= \hcalL^{1, d}_{\inddinv (i)} \circ \invphinddv,\\
    \label{e:geo-mapping-permut-3}
    \forall j \in [0..d],\quad
    \phinddv (\hcalH^d_j) &= \calHvd_{\indd (j)},\\
    \label{e:geo-mapping-permut-4}
    \phinddv (\hKd) &= \Kvd.
  \end{align}
\end{lemma}

\begin{proof}
  \proofpar{Bijections, $\phinddv$ affine,
    identities~\eqref{e:geo-mapping-permut-1},
    \eqref{e:geo-mapping-permut-1bis}}
  Direct consequence of
  Lemma~\thref{l:l-face-is-simplex},
  Lemma~\threfc{l:prop-geo-l-face-mapping}{with $l\eqdef d$},
  Lem\-ma~\thref{l:d-face-aff-space-is-full-space}, and
  Definition~\thref{d:geo-l-face-mapping}.

  \proofparskip{Identities~\eqref{e:geo-mapping-permut-2}}
  Let~$\xx\in\matRd$.
  Let $\hyy\eqdef\invphinddv(\xx)\in\matRd$, {\ie} such that
  $\phinddv(\hyy)=\xx$.
  From
  Lemma~\threfc{l:baryc-coor}{decomposition, then uniqueness},
  \eqref{e:geo-mapping-permut-1},
  Lemma~\threfc{l:lag-pol-is-basis-P1d-ref}{%
    \eqref{e:lag-pol-is-basis-P1d-ref-3}}, and
  Lemma~\thref{l:lag-pol-P1d-is-baryc-coor},
  we have
  \[
    \xx = \sum_{i = 0}^d \barcvd_i (\xx) \, \vv_i
    = \sum_{i = 0}^d \hcalL^{1, d}_{\inddinv (i)} (\hyy) \, \vv_{i}
    = \sum_{i = 0}^d \hcalL^{1, d}_{\inddinv(i)}
    \circ \invphinddv (\xx) \, \vv_{i}
  \]
  and $\sum_{i = 0}^d \hcalL^{1, d}_{\inddinv (i)} = 1$.
  Thus, for all $i\in[0..d]$,
  \(
    \hcalL^{1,d}_{\inddinv(i)}\circ\invphinddv
    = \barcvd_i
    = \calL^{\famvertd{\vv}}_i.
  \)

  \proofparskip{Identity~\eqref{e:geo-mapping-permut-3}}
  Let $j\in[0..d]$.
  Let $i\eqdef\indd(j)$, {\ie} such that $\inddinv(i)=j$.
  Then, from
  Lemma~\thref{l:ref-face-hyperpl},
  \eqref{e:geo-mapping-permut-2} (thus
  $\hcalL^{1,d}_j=\barcvd_i\circ\phinddv$),
  \assume{the definition of bijectivity},
  Lemma~\threfc{l:im-ker-eq-ker}{%
    with $f\eqdef\phinddv$ surjective, and $g\eqdef\barcvd_i$}, and
  Lemma~\thref{l:equiv-def-face-hyperpl},
  we have
  \[
    \phinddv \left( \hcalH^d_j \right)
    = \phinddv \left( \Ker{\hcalL^{1, d}_j} \right)
    = \Ker{\barcvd_i}
    = \calHvd_i.
  \]

  \proofparskip{Identity~\eqref{e:geo-mapping-permut-4}}
  Direct consequence of
  Lemma~\threfc{l:prop-geo-l-face-mapping}{%
    \eqref{e:prop-geo-l-face-mapping-4bis} with $l\eqdef d$}, and
  Lemma~\thref{l:l-face-is-simplex}.
\end{proof}

\section{Lagrange nodes for $\FElagP{k}{d}$}
\label{s:lag-nodes-Pkd}

\begin{remark}
  We recall the distinct notations for vertices~$\famvertd{\vv}$, and for
  nodes~$\famnodekd{\aa}$ (see Remark~\ref{r:fam-nodes-fam-verts}), as well as
  Definition~\ref{d:isobaryc} for isobarycenters.
\end{remark}

\subsection{Lagrange nodes for the current element}
\label{ss:lag-nodes-curr}

\begin{definition}[Lagrange nodes of $\matPkd$]
  \label{d:lag-nodes-Pkd}
  \mbox{}\hfill
  Let~$d\geq1$.
  Let~$k\in\matN$.
  Let~$\famvertd{\vv}$ be $d+1$ points in~$\matRd$.
  The {\em Lagrange nodes of~$\matPkd$} are denoted
  $\famnodekd{\aa}=(\aa_\aalpha)_{\aalpha\in\calAkd}$, and are defined by
  \begin{align}
    \label{e:lag-nodes-Pkd-0}
    (k = 0)&&
    \aa_\zzero &\eqdef
      \bisobarycv \left( = \frac{1}{d + 1} \sum_{i = 0}^d \vv_i \right),&&\\
    \label{e:lag-nodes-Pkd-1}
    (k \geq 1)&&
    \forall \aalpha \in \calAkd,\quad
    \aa_\aalpha &\eqdef
      \vv_0 + \sum_{i = 1}^d \frac{\alpha_i}{k} (\vv_i -  \vv_0).&&
  \end{align}
\end{definition}

\begin{lemma}[Lagrange nodes of $\matPkd$ for $d=1$ are Lagrange nodes of
  $\matPki$]
  \label{l:lag-nodes-Pkd-for-d-eq-1-are-lag-nodes-Pk1}
  \mbox{}\\
  Let~$k\in\matN$.
  Let~$\famverti{\vv}=\famverti{v}$ be two points in~$\matR^1$.\\
  Then, $\famnodekd{\aa}$ for $d\eqdef1$, and~$\famnodeki{a}$ from
  Definition~\ref{d:lag-nodes-Pk1} coincide.
\end{lemma}

\begin{proof}
  Direct consequence of
  Definition~\threfc{d:lag-nodes-Pkd}{with $d\eqdef1$},
  Lemma~\threfc{l:first-multi-ind-Akd}{with $d\eqdef1$}, and
  Definition~\thref{d:lag-nodes-Pk1}.
\end{proof}

\begin{lemma}[number of Lagrange nodes of $\matPkd$]
  \label{l:num-lag-nodes-Pkd}
  \mbox{}\\
  Let~$d\geq1$.
  Let~$k\in\matN$.
  Let~$\famvertd{\vv}$ be $d+1$ affinely independent points in~$\matRd$.
  Let $\aalpha,\bbeta\in\calAkd$.\\
  Then, $\aalpha\neq\bbeta$ implies $\aa_\aalpha\neq\aa_\bbeta$.
  Thus, there are~$\binomkpdd$ distinct Lagrange nodes in~$\famnodekd{\aa}$.
\end{lemma}

\begin{proof}
  Direct consequence of
  Definition~\thref{d:lag-nodes-Pkd},
  Lemma~\threfc{l:baryc-coor}{%
    uniqueness in~\eqref{e:baryc-coor-2} with $j\eqdef0$}, and
  Lemma~\thref{l:card-Akd}.
\end{proof}

\begin{lemma}[barycentric coordinates of Lagrange nodes of $\matPkd$]
  \label{l:baryc-coor-lag-nodes-Pkd}
  \mbox{}\hfill
  Let~$d\geq1$.
  Let~$k\in\matN$.\\
  Let~$\famvertd{\vv}$ be $d+1$ affinely independent points in~$\matRd$.
  Let~$\aalpha\in\calAkd$.
  Let~$i\in[0..d]$.\\
  Then, the barycentric coordinates of the Lagrange node~$\aa_\aalpha$ are
  \begin{align}
    \label{e:baryc-coor-lag-nodes-Pkd-0}
    (k = 0)&&
    &\barcvd_i (\aa_\zzero) = \frac{1}{d + 1} \in [0,1],&&\\
    \label{e:baryc-coor-lag-nodes-Pkd-1}
    (k \geq 1)&&
    &\left\{
      \begin{array}{rl}
        (i = 0) \qquad & \displaystyle
        \barcvd_0 (\aa_\aalpha) = 1 - \frac{\len{\aalpha}}{k} \in [0,1],\\
        (i \in [1..d]) \qquad & \displaystyle
        \barcvd_i (\aa_\aalpha) = \frac{\alpha_i}{k} \in [0,1].
      \end{array}
    \right.&&
  \end{align}

  Moreover, if $k\geq1$, then the Lagrange node is equivalently defined by
  \begin{equation}
    \label{e:baryc-coor-lag-nodes-Pkd-2}
    \aa_\aalpha =  \left(1 - \frac{\len{\aalpha}}{k} \right) \vv_0
    + \sum_{i = 1}^d \frac{\alpha_i}{k} \vv_i
    \quad \mbox{with }
    \left( 1 - \frac{\len{\aalpha}}{k} \right)
    + \sum_{i = 1}^d \frac{\alpha_i}{k} = 1.
  \end{equation}
\end{lemma}

\begin{proof}
  Direct consequence of
  Definition~\thref{d:lag-nodes-Pkd},
  Definition~\threfc{d:isobaryc}{for $k=0$},
  Lemma~\threfc{l:baryc-coor}{%
    \eqref{e:baryc-coor-1} for $k=0$,
    \eqref{e:baryc-coor-1} and \eqref{e:baryc-coor-2} with $j\eqdef0$ for
    $k\geq1$, and uniqueness},
  Definition~\threfc{d:multi-ind-Akd-Ckd}{for $k\geq1$},
  Definition~\threfc{d:len-multi-ind}{for $k\geq1$},
  Lemma~\threfc{l:ind-smaller-than-max-len}{for $k\geq1$}, and
  \assume{ordered field properties of $\matR$
    (with $d+1>0$ for $k=0$, and $k>0$ for $k\geq1$)}.
\end{proof}

\begin{lemma}[vertices are Lagrange nodes of $\matPkd$]
  \label{l:vert-lag-nodes-Pkd}
  \mbox{}\\
  Let~$d\geq1$.
  Let~$k\geq1$.
  Let~$\famvertd{\vv}$ be $d+1$ points in~$\matRd$.
  Then, the vertices are nodes of~$\famnodeid{\aa}$,
  \begin{equation}
    \label{e:vert-lag-nodes-Pkd}
    \vv_0 = \aa_\zzero \AND
    \forall i \in [1..d],\quad \vv_i = \aa_{k \ee_i}.
  \end{equation}
\end{lemma}

\begin{proof}
  \proofpar{Case $i=0$}
  Direct consequence of
  Definition~\threfc{d:lag-nodes-Pkd}{with $k>0$}.

  \proofparskip{Case $i\in[1..d]$}
  From
  Definition~\thref{d:canon-fam},
  Definition~\thref{d:len-multi-ind},
  Definition~\thref{d:multi-ind-Akd-Ckd}, and
  Definition~\threfc{d:lag-nodes-Pkd}{with $k>0$},
  we have $\len{k\ee_i}=k\sum_{j=1}^d\kron{i}{j}=k\leq k$, thus
  $k\ee_i\in\calAkd$, and
  \[
    \aa_{k \ee_i}
    = \vv_0 + \sum_{j = 1}^d \frac{k\kron{i}{j}}{k} (\vv_j - \vv_0)
    = \vv_0 + (\vv_i - \vv_0) = \vv_i.
  \]

  \medskip\noindent
  Therefore, the identities hold.
\end{proof}

\begin{lemma}[Lagrange nodes of $\matPid$ are vertices]
  \label{l:lag-nodes-Pid-are-vert}
  \mbox{}\hfill
  Let~$d\geq1$.
  Let~$\famvertd{\vv}$ be $d+1$\\ affinely independent points in~$\matRd$.
  Then, the nodes of~$\famnodeid{\aa}$ are the vertices,
  \begin{equation}
    \label{e:lag-nodes-Pid-are-vert}
    \aa_\zzero = \vv_0 \AND
    \forall i \in [1..d],\quad \aa_{\ee_i} = \vv_i.
  \end{equation}
\end{lemma}

\begin{proof}
  Direct consequence of
  Lemma~\threfc{l:first-multi-ind-Akd}{with $k=1$}, and
  Lemma~\threfc{l:vert-lag-nodes-Pkd}{with $k=1$}.
\end{proof}

\subsection{Sub-vertices and sub-nodes}
\label{ss:sub-vertices-sub-nodes}

\begin{definition}[sub-vertices of Lagrange nodes of $\matPkd$]
  \label{d:sub-vert-lag-nodes-Pkd}
  \mbox{}\\
  Let~$d\geq1$.
  Let~$k\geq1$.
  Let~$\famvertd{\vv}$ be $d+1$ points in~$\matRd$.
  Let~$i\in[0..d]$.
  The {\em sub-vertices} of the Lagrange nodes of~$\matPkd$ with respect
  to~$\vv_0$ are denoted~$\famvertd{\uvv}=(\uvv_i)_{i\in[0..d]}$, and are
  defined by
  \begin{align}
    \label{e:sub-vert-lag-nodes-Pkd-0}
    (i = 0)&&
    \uvv_0 &\eqdef \vv_0,&&\\
    \label{e:sub-vert-lag-nodes-Pkd-1}
    (i \in [1..d])&&
    \uvv_i &\eqdef \aa_{(k - 1) \ee_i}.&&
  \end{align}
\end{definition}

\begin{lemma}[equivalent definition of sub-vertices of Lagrange nodes of
    $\matPkd$]
  \label{l:equiv-def-sub-vert-lag-nodes-Pkd}
  \mbox{}\\
  Let~$d\geq1$.
  Let~$k\geq1$.
  Let~$\famvertd{\vv}$ be $d+1$ points in~$\matRd$.
  Let $i\in[0..d]$.
  Then, we have
  \begin{equation}
    \label{e:equiv-def-sub-vert-lag-nodes-Pkd}
    \uvv_i = \frac{1}{k} \vv_0 + \frac{k - 1}{k} \vv_i.
  \end{equation}
\end{lemma}

\begin{proof}
  Direct consequence of
  Definition~\thref{d:sub-vert-lag-nodes-Pkd},
  \assume{ordered field properties of~$\matR$ (with $k>0$),
    thus $\frac1k+\frac{k-1}{k}=1$ (case $i=0$),
    $1-\frac{k-1}{k}=\frac{1}{k}$, and $0\leq k-1<k$},
  Definition~\thref{d:canon-fam},
  Definition~\threfc{d:len-multi-ind}{%
    thus $\len{(k-1)\ee_i}=k-1$},
  Definition~\threfc{d:multi-ind-Akd-Ckd}{thus $(k-1)\ee_i\in\calAkd$}, and
  Lemma~\threfc{l:baryc-coor-lag-nodes-Pkd}{%
    \eqref{e:baryc-coor-lag-nodes-Pkd-2}}.
\end{proof}

\begin{remark}
  \label{r:sub-vertices}
  An illustration of the sub-vertices (see
  Definition~\ref{d:sub-vert-lag-nodes-Pkd}) and of $\matPkmid$~sub-nodes (see
  Lemma~\ref{l:Pkm1d-sub-nodes-sub-vert-are-some-nodes-Pkd}) in dimension $d=3$
  and for $k=3$ can be seen in Figure~\ref{f:sub-vert-sub-nodes_d3-k3}.

  The sub-vertices~$\famvertd{\uvv}$ are~$\vv_0$ and the nodes of~$\matPkd$
  that are the closest to the vertices along the axes.
  The $\matPkmid$~sub-nodes are the nodes of~$\matPkd$, except the nodes
  indexed by~$\aalpha\in\calCkd$ (that lie in the face~$\Hvd_0$).
  These sub-vertices and sub-nodes are used in the step~2 of the proof of the
  unisolvence Lemma~\ref{l:lag-lin-forms-Pkd-inj}.
\end{remark}

\begin{figure}[htb]
  \centering
  \resizebox{0.5\linewidth}{!}{
    \begin{tikzpicture}[scale=4,math3d] 

        \def\kk{3}

        \def\colo{magenta}
        \def\coli{darkgreen}
        \def\colii{red}
        \def\coliii{blue}

	\coordinate (A) at (0,0,0);
	\coordinate (B) at  ($ (A) + (1,0,0) $);
	\coordinate (C) at  ($ (A) + (0,1,0) $);
	\coordinate (D) at  ($ (A) + (0,0,1) $);
        \draw (A) node[above right] {$\uvv_0$} ;  
        \draw (B) node[above left] {$\vv_1$} ; 
        \draw (C) node[above=2pt] {$\vv_2$} ; 
        \draw (D) node[right=1.5pt] {$\vv_3$} ; 
	\draw[line width=1.0pt,rounded corners=0.5pt] (A) -- (B) -- (D) -- cycle;
	\draw[line width=1.0pt,rounded corners=0.5pt] (A) -- (C) -- (B) -- cycle;
	\draw[line width=1.6pt,rounded corners=0.5pt] (A) -- (C) -- (D) -- cycle;

        \pgfmathparse{1-1/\kk}\let\kkp\pgfmathresult
        \coordinate (B1) at ($ \kkp*(A) + 1/\kk*(B) $) ;
        \coordinate (C1) at ($ \kkp*(A) + 1/\kk*(C) $) ;
        \coordinate (D1) at ($ \kkp*(A) + 1/\kk*(D) $) ;
        \coordinate (B2) at ($ \kkp*(A) + 2/\kk*(B) $) ;
        \coordinate (C2) at ($ \kkp*(A) + 2/\kk*(C) $) ;
        \coordinate (D2) at ($ \kkp*(A) + 2/\kk*(D) $) ;

        \draw[color=\coli,line width=0.8pt]   (B1) -- (C1) -- (D1) -- cycle;
        \draw[color=\colii,line width=0.8pt]  (B2) -- (C2) -- (D2) -- cycle;
        \draw[color=\coliii,line width=0.4pt]  (B) -- (C) -- (D) -- cycle;

        \coordinate (Bp) at  ($ (B) - 1/\kk*(1,0,0) $) ;
        \coordinate (Cp) at  ($ (C) - 1/\kk*(0,1,0) $) ;
        \coordinate (Dp) at  ($ (D) - 1/\kk*(0,0,1) $) ;
        \draw (Bp) node[above] {$\uvv_1$} ; 
        \draw (Cp) node[above right] {$\uvv_2$} ; 
        \draw (Dp) node[above right] {$\uvv_3$} ; 
        
        \fill[color=\colo] (A) circle (1.0pt);
        \newcount\z
        \foreach \x in {1,0}  {
          \pgfmathparse{1-\x}\let\YY\pgfmathresult
          \foreach \y in {0,...,\YY} {
            \pgfmathsetcount{\z}{1-\x-\y} 
            \coordinate (Axyz) at  ($ (A) + 1/\kk*(\x,\y,\z) $) ;
            \node[color=\coli,below] (Nxyz) at (Axyz) {$\uaa_{(\x,\y,\the\z)}$};
            \fill[color=\coli] (Axyz) circle (1.0pt);
          }
        }
        \newcount\z
        \foreach \x in {2,1,...,0}  {
          \pgfmathparse{2-\x}\let\YY\pgfmathresult
          \foreach \y in {0,...,\YY} {
            \pgfmathsetcount{\z}{2-\x-\y} 
            \coordinate (Axyz) at  ($ (A) + 1/\kk*(\x,\y,\z) $) ;
            \node[color=\colii,below] (Nxyz) at (Axyz) {$\uaa_{(\x,\y,\the\z)}$};
            \fill[color=\colii] (Axyz) circle (1.0pt);
          }
        }
        \newcount\z
        \foreach \x in {3,2,...,0}  {
          \pgfmathparse{3-\x}\let\YY\pgfmathresult
          \foreach \y in {0,...,\YY} {
            \pgfmathsetcount{\z}{3-\x-\y} 
            \coordinate (Axyz) at  ($ (A) + 1/\kk*(\x,\y,\z) $) ;
            \fill[color=\coliii] (Axyz) circle (0.6pt);
          }
        }

\end{tikzpicture}
  }
  \caption[Sub-vertices and Lagrange sub-nodes]{%
    Sub-vertices~$\famvertd{\uvv}$ with respect to $\vv_0$ and Lagrange
    $\matPkmid$~sub-nodes $\famnodekmid{\uaa}$, in the case $d=k=3$.\\
    The Lagrange~$\matPkd$ nodes~$\famnodekd{\aa}$ are set in the tetrahedron
    defined by the vertices $\famvert{\vv}{3}=§(\vv_0,\vv_1,\vv_2,\vv_3)$.\\
    The sub-vertices are
    $\famvert{\uvv}{3}\eqdef(\vv_0,\aa_{2\ee_1},\aa_{2\ee_2},\aa_{2\ee_3})$,
    see Definition~\ref{d:sub-vert-lag-nodes-Pkd} and
    Remark~\ref{r:sub-vertices}.\\
    The $\matPkmid$~sub-nodes~$\famnode{\uaa}{\calAiidiii}$ are defined with
    respect to the tetrahedron whose vertices are~$\famvert{\uvv}{3}$.
    These are the nodes~$\famnode{\aa}{\calAiiidiii}$, except the (small blue)
    nodes corresponding to~$k=3$ (with indices in~$\calCiiidiii$), see
    Lemma~\ref{l:Pkm1d-sub-nodes-sub-vert-are-some-nodes-Pkd}.\\
    Thus, geometrically, passing from~$\matPkd$ to~$\matPkmid$ amounts to
    remove the face~$\Hvd_0$ from the tetrahedron defined by~$\famvertd{\vv}$.}
  \label{f:sub-vert-sub-nodes_d3-k3}
\end{figure}

\begin{lemma}[sub-vertices are affinely independent]
  \label{l:sub-vert-aff-indep}
  \mbox{}\\
  Let~$d\geq1$.
  Let~$k\geq2$.
  Let~$\famvertd{\vv}$ be $d+1$ affinely independent points in~$\matRd$.\\
  Then, the sub-vertices~$\famvertd{\uvv}$ are affinely independent.
\end{lemma}

\begin{proof}
  Let~$(\mu_i)_{i\in[1..d]}$ such that $\sum_{i=1}^d\mu_i(\uvv_i-\uvv_0)=0$.
  Then, from
  Lemma~\thref{l:equiv-def-sub-vert-lag-nodes-Pkd}, and
  \assume{field properties of~$\matR$ (with $k\neq0$)},
  \[
    \sum_{i = 1}^d \mu_i (\uvv_i - \uvv_0)
    = \sum_{i = 1}^d \mu_i \left(
      \frac{k - 1}{k} \vv_i + \left( \frac{1}{k} - 1 \right) \vv_0 \right)
    = \sum_{i = 1}^d \mu_i \frac{k - 1}{k} (\vv_i - \vv_0) = 0.
  \]
  Thus, from
  Definition~\thref{d:aff-indep-family},
  \assume{the zero-product property in~$\matR$
    (with $\frac{k-1}{k}\neq0$)}, and
  \assume{the definition of freedom},
  we have $\mu_i=0$ for all $i\in[1..d]$, thus $(\uvv_i-\uvv_0)_{i\in[1..d]}$
  is a free family, and $\famvertd{\uvv}$ is affinely independent.
\end{proof}

\begin{lemma}[$\matPkmid$ sub-nodes of sub-vertices are some nodes of
  $\matPkmid$]
  \label{l:Pkm1d-sub-nodes-sub-vert-are-some-nodes-Pkd}
  \mbox{}\\
  Let~$d\geq1$.
  Let~$k\geq2$.
  Let~$\famvertd{\vv}$ be $d+1$ affinely independent points in~$\matRd$.
  The Lagrange nodes of~$\matPkmid$ with respect to the
  sub-vertices~$\famvertd{\uvv}$ are called {\em $\matPkmid$~sub-nodes} and
  denoted as~$\famnodekmid{\uaa}$.\\
  Then, these sub-nodes are some of the nodes~$\famnodekd{\aa}$ of $\matPkd$
  (defined with respect to $\famvertd{\vv}$),
  \begin{equation}
    \label{e:Pkm1d-sub-nodes-sub-vert-are-some-nodes-Pkd}
    \forall \aalpha \in \calAkmid,\qquad \uaa_\aalpha = \aa_\aalpha.
  \end{equation}
\end{lemma}

\begin{proof}
  Let $\aalpha\in\calAkmid$.
  Then, from
  Lemma~\thref{l:sub-vert-aff-indep},
  Lem\-ma~\threfc{l:baryc-coor-lag-nodes-Pkd}{%
    \eqref{e:baryc-coor-lag-nodes-Pkd-2},
    first for $\famnodekmid{\uaa}$ with $k-1\geq1$ and $\famvertd{\uvv}$,
    then for $\famnodekd{\aa}$ with $k\geq1$ and $\famvertd{\vv}$},
  Lemma~\thref{l:equiv-def-sub-vert-lag-nodes-Pkd},
  \assume{field properties of~$\matR$ (with $k\geq2$, thus $k-1\neq0$)},
  Definition~\thref{d:len-multi-ind}, and
  Lemma~\threfc{l:Ckd-layers-Akd}{thus $\aalpha\in\calAkd$},
  we have
  \begin{align*}
    \uaa_\aalpha
    &= \left( 1 - \frac{\len{\aalpha}}{k - 1} \right) \uvv_0
      + \sum_{i = 1}^d \frac{\alpha_i}{k - 1} \uvv_i
    = \left( 1 - \frac{\len{\aalpha}}{k - 1} \right) \vv_0
      + \sum_{i = 1}^d \frac{\alpha_i}{k - 1}
      \left( \frac{1}{k} \vv_0 + \frac{k - 1}{k} \vv_i \right)\\
    &= \left( 1 - \frac{\len{\aalpha}}{k - 1}
      + \frac{\len{\aalpha}}{(k - 1) k} \right) \vv_0
      + \sum_{i = 1}^d \frac{\alpha_i}{k} \vv_i
    = \left( 1 - \frac{\len{\aalpha}}{k} \right) \vv_0
      + \sum_{i = 1}^d \frac{\alpha_i}{k} \vv_i = \aa_\aalpha.
  \end{align*}
\end{proof}

\subsection{Lagrange nodes for the reference element}
\label{ss:lag-nodes-ref}

\begin{lemma}[reference Lagrange nodes of $\matPkd$]
  \label{l:lag-nodes-Pkd-ref}
  \mbox{}\hfill
  Let~$d\geq1$.
  Let~$k\in\matN$.\\
  The {\em reference Lagrange nodes of~$\matPkd$} are
  denoted~$\famnodekd{\haa}$, and are defined as the Lagrange nodes
  of~$\matPkd$ for the reference vertices.
  Let~$\aalpha\in\calAkd$.
  Then, we have
  \begin{align}
    \label{e:lag-nodes-Pkd-ref-0}
    (k = 0)&&
    \haa_\zzero &= \hbisobarycd,&&\\
    \label{e:lag-nodes-Pkd-ref-1}
    (k \geq 1)&&
    \haa_\aalpha &=
    \hvv_0 + \sum_{i = 1}^d \frac{\alpha_i}{k} (\hvv_i - \hvv_0)
    = \sum_{i = 1}^d \frac{\alpha_i}{k} \ee_i&&\\
    \nonumber
    && &= \left( 1 - \frac{\len{\aalpha}}{k} \right) \hvv_0
    + \sum_{i = 1}^d \frac{\alpha_i}{k} \hvv_i
    \quad \mbox{with }
    \left( 1 - \frac{\len{\aalpha}}{k} \right)
    + \sum_{i = 1}^d \frac{\alpha_i}{k} = 1.&&
  \end{align}
\end{lemma}

\begin{proof}
  Direct consequence of
  Definition~\thref{d:lag-nodes-Pkd},
  Lemma~\threfc{l:ref-isobaryc}{for $k=0$},
  Lemma~\threfc{l:baryc-coor-lag-nodes-Pkd}{%
    \eqref{e:baryc-coor-lag-nodes-Pkd-2} for $k\geq1$}, and
  Definition~\threfc{d:fam-ref-aff-pts}{for $k\geq1$}.
\end{proof}

\begin{lemma}[reference Lagrange nodes of $\matPkd$ for $d=1$ are reference
  Lagrange nodes of $\matPki$]
  \label{l:lag-nodes-Pkd-ref-for-d-eq-1-are-lag-nodes-Pk1-ref}
  \mbox{}\hfill
  Let~$k\in\matN$.
  Then, $\famnodekd{\haa}$ for $d\eqdef1$ and~$\famnodeki{\ha}$ from
  Definition~\ref{d:lag-nodes-Pk1-ref} coincide.
\end{lemma}

\begin{proof}
  Direct consequence of
  Lemma~\threfc{l:lag-nodes-Pkd-ref}{with $d\eqdef1$},
  Definition~\threfc{d:canon-fam}{$\ee_i$},
  Lemma~\threfc{l:first-multi-ind-Akd}{with $d\eqdef1$}, and
  Definition~\thref{d:lag-nodes-Pk1-ref}.
\end{proof}

\begin{lemma}[equivalent definition of reference Lagrange nodes of $\matPkd$]
  \label{l:equiv-def-lag-nodes-Pkd-ref}
  \mbox{}\\
  Let~$d\geq1$.
  Let~$k\in\matN$.
  Let~$\aalpha\in\calAkd$.
  Let~$i\in[1..d]$.
  Then, we have
  \begin{align}
    \label{e:equiv-def-lag-nodes-Pkd-ref-0}
    (k = 0)&& (\haa_\zzero)_i &= \frac{1}{d + 1},&&\\
    \label{e:equiv-def-lag-nodes-Pkd-ref-1}
    (k \geq 1)&& (\haa_\aalpha)_i &= \frac{\alpha_i}{k}.&&
  \end{align}
\end{lemma}

\begin{proof}
  Direct consequence of
  Lemma~\thref{l:lag-nodes-Pkd-ref},
  Lemma~\threfc{l:ref-isobaryc}{for $k=0$}, and
  Definition~\threfc{d:canon-fam}{for $k\geq1$}.
\end{proof}

\begin{lemma}[number of reference Lagrange nodes of $\matPkd$]
  \label{l:num-lag-nodes-Pkd-ref}
  \mbox{}\hfill
  Let~$d\geq1$.
  Let~$k\in\matN$.\\
  Then, there are~$\binomkpdd$ distinct reference Lagrange
  nodes~$\famnodekd{\haa}$.
\end{lemma}

\begin{proof}
  Direct consequence of
  Lemma~\thref{l:lag-nodes-Pkd-ref},
  Lemma~\thref{l:ref-affine-vert-is-affinely-indep},
  Lemma~\threfc{l:baryc-coor}{%
    uniqueness in~\eqref{e:baryc-coor-2} with $j\eqdef0$ and
    $\famvertd{\hvv}$}, and
  Lemma~\thref{l:card-Akd}.
\end{proof}

\begin{remark}
  In the next lemma, the proof could also have been done using
  Lemma~\ref{l:aff-map-preserves-baryc}.
\end{remark}

\begin{lemma}[Lagrange nodes of $\matPkd$ are image of reference]
  \label{l:lag-nodes-Pkd-im-ref}
  \mbox{}\hfill
  Let~$d\geq1$.
  Let~$k\in\matN$.\\
  Let~$\famvertd{\vv}$ be $d+1$ points in~$\matRd$.
  Let $\aalpha\in\calAkd$.
  Then, we have $\aa_\aalpha=\phigeodv(\haa_\aalpha)$.
\end{lemma}

\begin{proof}
  Direct consequence of
  Lemma~\threfc{l:prop-geo-mapping}{%
    \eqref{e:prop-geo-mapping-2}},
  Lemma~\thref{l:equiv-def-lag-nodes-Pkd-ref},
  \assume{field properties of~$\matR$ (with $d+1\neq0$) for $k=0$}, and
  Definition~\thref{d:lag-nodes-Pkd}.
\end{proof}

\subsection{Lagrange nodes and face hyperplanes}
\label{ss:lag-nodes-face-hyperp}

\begin{lemma}[face hyperplanes of Lagrange nodes of $\matPkd$]
  \label{l:face-hyperpl-lag-nodes-Pkd}
  \mbox{}\hfill
  Let~$d\geq1$.
  Let~$k\geq1$.
  Let~$\famvertd{\vv}$ be $d+1$ affinely independent points in~$\matRd$.
  Let~$\aalpha\in\calAkd$.
  Let~$i\in[0..d]$.
  Then, we have
  \begin{align}
    \label{e:face-hyperpl-lag-nodes-Pkd-0}
    (i = 0)&&
    \aa_\aalpha \in \calHvd_0
    \EQUIV \len{\aalpha} &= k
    \EQUIV \aalpha \in \calCkd,&&\\
    \label{e:face-hyperpl-lag-nodes-Pkd-1}
    (i \in [1..d])&&
    \aa_\aalpha \in \calHvd_i
    \EQUIV \alpha_i &= 0
    \EQUIV \aalpha \in \calAkdi.&&
  \end{align}
  Thus, we have
  $\card\left\{\aa_\aalpha\in\calHvd_i\rightst
    \left.\vphantom{\calHvd_i}\aalpha\in\calAkd\right\}=\binom{k+d-1}{d-1}$.
\end{lemma}

\begin{proof}
  \proofparskip{Equivalences}
  Direct consequence of
  Lemma~\threfc{l:equiv-def-face-hyperpl}{\eqref{e:equiv-def-face-hyperpl-2}},
  Lemma~\threfc{l:baryc-coor-lag-nodes-Pkd}{
    \eqref{e:baryc-coor-lag-nodes-Pkd-1}},
  \assume{field properties of~$\matR$ (with $k>0$)},
  Definition~\thref{d:multi-ind-Akd-Ckd},
  Lemma~\thref{l:card-Akdi-Akdm1}, and
  Lemma~\thref{l:card-Akd}.

  \proofparskip{Cardinal}
  \proofpar{Case $i=0$}
  Direct consequence of
  Definition~\thref{d:multi-ind-Akd-Ckd}, and
  Lemma~\thref{l:card-Ckd}.\\
  \proofpar{Case $i\in[1..d]$}
  Direct consequence of
  Lemma~\threfc{l:num-lag-nodes-Pkd}{%
    injectivity, thus
    $\card\{\aa_{\aalpha}\in\calHvd_i\st\aalpha\in\calAkd\}=\card\{\calAkdi\}$},
  Lemma~\thref{l:card-Akdi-Akdm1},
  Lem\-ma~\thref{l:card-Ckd-Akdm1}, and
  Lemma~\thref{l:card-Ckd}.
\end{proof}

\begin{remark}
  In the next lemma, we recall that~$\thetinjdmi{i}$ is defined in
  Lemma~\ref{l:jump-enum}, $\fkdi$ in Lemmas~\ref{l:card-Ckd-Akdm1} ($i=0$)
  and~\ref{l:card-Akdi-Akdm1} ($i\in[1..d]$), and~$\phindlv$ in
  Definition~\ref{d:geo-l-face-mapping}.
  See also Figure~\ref{f:geo-hyper-map_d32}.
\end{remark}

\begin{lemma}[image of nodes by geometric hyperface mapping]
  \label{l:im-nodes-by-geo-hyperface-mapping}
  \mbox{}\hfill
  Let~$d\geq2$.
  Let~$k\geq1$.\\
  Let~$\famvertd{\vv}$ be $d+1$ affinely independent points in~$\matRd$.
  Let~$\famnodekdmi{\haa^{d-1}}$ be the reference Lagrange nodes of~$\matPkdmi$
  and~$\famnodekd{\aa}$ be the Lagrange nodes of~$\matPkd$
  for~$\famvertd{\vv}$.
  Let~$i\in[0..d]$.
  Let~$\aalphap\in\calAkdmi$.\\
  Then, we have
  \begin{equation}
    \label{e:im-nodes-by-geo-hyperface-mapping}
    \phindv{\thetinjdmi{i}} (\haa_{\aalphap}^{d - 1}) = \aa_{\fkdi (\aalphap)}.
  \end{equation}
  %
\end{lemma}

\begin{proof}
  From
  Lemma~\threfc{l:lag-nodes-Pkd-ref}{%
    \eqref{e:lag-nodes-Pkd-ref-1} with $d-1\geq1$},
  we have
  \[
    \haa_{\aalphap}^{d - 1} =
    \left( 1 - \frac{\len{\aalphap}}{k} \right) \hvv_0^{d - 1}
    + \sum_{j = 1}^{d - 1} \frac{\alphap_j}{k} \hvv_j^{d - 1}
    \quad \mbox{with }
    \left( 1 - \frac{\len{\aalphap}}{k} \right)
    + \sum_{j = 1}^{d - 1} \frac{\alphap_j}{k} = 1.
  \]
  \proofpar{Case $i=0$}
  Then, from
  Lemma~\thref{l:card-Ckd-Akdm1},
  Definition~\thref{d:multi-ind-Akd-Ckd},
  Lemma~\threfc{l:baryc-coor-lag-nodes-Pkd}{%
    \eqref{e:baryc-coor-lag-nodes-Pkd-2}}, and
  \assume{field properties of~$\matR$ (with $k\neq0$)},
  we have $\fkdo(\aalphap)=
  (k-\len{\aalphap},\aalphap)\in\calCkd\subset\calAkd$,
  $\len{\fkdo(\aalphap)}=k$,
  \[
    \aa_{\fkdo (\aalphap)} =
    \left( 1 - \frac{\len{\aalphap}}{k} \right) \vv_1
    + \sum_{j = 1}^{d - 1} \frac{\alphap_j}{k} \vv_{j + 1}.
  \]
  \proofpar{Case $i\in[1..d]$}
  Then, from
  Lemma~\thref{l:card-Akdi-Akdm1},
  Definition~\thref{d:multi-ind-Akd-Ckd}, and
  Lemma~\threfc{l:baryc-coor-lag-nodes-Pkd}{%
    \eqref{e:baryc-coor-lag-nodes-Pkd-2}},
  we have $\fkdi(\aalphap)=(\alphap_1,\ldots,\alphap_{i-1},0,
  \alphap_i,\ldots,\alphap_{d-1})\in\calAkdi\subset\calAkd$,
  $\len{\fkdi(\aalphap)}=\len{\aalphap}$, and
  \[
    \aa_{\fkdi (\aalphap)} =
    \left( 1 - \frac{\len{\aalphap}}{k} \right) \vv_0
    + \sum_{j = 1}^{i - 1} \frac{\alphap_j}{k} \vv_j
    + \sum_{j = i}^{d - 1} \frac{\alphap_j}{k} \vv_{j + 1}.
  \]
  Therefore, in both cases, $\haa_{\aalphap}^{d-1}$ is the barycenter of
  $(\hvv_j^{d-1})_{j\in[0..d-1]}$ and~$\aa_{\fkdi(\aalphap)}$ is the barycenter
  of $(\vv_j)_{j\in[0..d]\setminus\{i\}}$ with the same coefficients.
  Thus, from
  Lemma~\threfc{l:geo-hyperface-mapping}{%
    affineness and \eqref{e:geo-hyperface-mapping-2},
    thus $\phindv{\thetinjdmi{i}}(\hvv_j^{d-1})=\vv_j$ for $j<i$ and
      $\phindv{\thetinjdmi{i}}(\hvv_j^{d-1})=\vv_{j+1}$ for $j\geq i$},  and
  Lemma~\thref{l:aff-map-preserves-baryc},
  we have the equality.
\end{proof}

\section{Lagrange linear forms for $\FElagP{k}{d}$}
\label{s:lag-lin-forms-Pkd}

\begin{definition}[Lagrange linear forms for $\matPkd$]
  \label{d:lag-lin-forms-Pkd}
  \mbox{}\\
  Let~$d\geq1$.
  Let~$k\in\matN$.
  Let~$\famvertd{\vv}$ be $d+1$ points in~$\matRd$.
  The {\em Lagrange linear forms} associated with the Lagrange nodes
  of~$\matPkd$ are denoted
  $\Sigma^{\famnodekd{\aa}}=(\sigma_\aalpha)_{\aalpha\in\calAkd}$, and are
  defined by
  \begin{equation}
    \label{e:lag-lin-forms-Pkd}
    \forall \aalpha \in \calAkd,\quad
    \forall f : \ArRdR,\quad
    \sigma_\aalpha (f) \eqdef f (\aa_\aalpha).
  \end{equation}
\end{definition}

\begin{lemma}[Lagrange linear forms for $\matPkd$ for $d\!\!=\!\!1$ are
  Lagrange linear forms for $\matPki$]
  \label{l:lag-lin-forms-Pkd-for-d-eq-1-are-lag-lin-forms-Pk1}
  Let~$k\in\matN$.
  Let~$\famverti{\vv}=\famverti{v}$ be two points in~$\matR^1$.\\
  Then, $\Sigma^{\famnodekd{\aa}}$ for $d\eqdef1$ and~$\Sigma^{\famnodeki{\aa}}$
  from Definition~\ref{d:lag-lin-forms-Pk1} coincide.
\end{lemma}

\begin{proof}
  Direct consequence of
  Definition~\threfc{d:lag-lin-forms-Pkd}{with $d\eqdef1$},
  Lem\-ma~\thref{l:lag-nodes-Pkd-for-d-eq-1-are-lag-nodes-Pk1},
  Lemma~\threfc{l:first-multi-ind-Akd}{with $d\eqdef1$}, and
  Definition~\thref{d:lag-lin-forms-Pk1}.
\end{proof}

\begin{lemma}[Lagrange linear forms of $\matPkd$ are linear]
  \label{l:lag-lin-forms-are-linear-Pkd}
  \mbox{}\\
  Let~$d\geq1$.
  Let~$k\in\matN$.
  Let~$\famvertd{\vv}$ be $d+1$ points in~$\matRd$.
  Let~$\aalpha\in\calAkd$.
  Then, $\sigma_\aalpha$ is linear.
\end{lemma}

\begin{proof}
  Direct consequence of
  Definition~\thref{d:lag-lin-forms-Pkd}, and
  \assume{the definition of linear operations over functions.}
\end{proof}

\begin{lemma}[cardinal of Lagrange linear forms of $\matPkd$]
  \label{l:card-lag-lin-forms-Pkd}
  \mbox{}\hfill
  Let~$d\geq1$.
  Let~$k\in\matN$.\\
  Let~$\famvertd{\vv}$ be $d+1$ affinely independent points in~$\matRd$.
  Then, $\card\left(\Sigma^{\famnodekd{\aa}}\right)=\binomkpdd$.
\end{lemma}

\begin{proof}
  Direct consequence of
  Definition~\thref{d:lag-lin-forms-Pkd}, and
  Lemma~\thref{l:num-lag-nodes-Pkd}.
\end{proof}

\begin{definition}[reference Lagrange linear forms for $\matPkd$]
  \label{d:lag-lin-forms-Pkd-ref}
  \mbox{}\\
  Let~$d\geq1$.
  Let~$k\in\matN$.
  The {\em reference Lagrange linear forms} associated with the reference
  Lagrange nodes of~$\matPkd$ are denoted
  $\hSigmakd=(\hsigma_\aalpha)_{\aalpha\in\calAkd}$, and are defined by
  \begin{equation}
    \label{e:lag-lin-forms-Pkd-ref}
    \forall \aalpha \in \calAkd,\quad
    \forall \hf : \ArRdR,\quad
    \hsigma_\aalpha (\hf) \eqdef \hf (\haa_\aalpha).
  \end{equation}
\end{definition}

\begin{lemma}[reference Lagrange linear forms for $\matPkd$ for $d=1$ are
  reference Lagrange linear forms for $\matPki$]
  \label{l:lag-lin-forms-Pkd-ref-for-d-eq-1-are-lag-lin-forms-Pk1-ref}
  \mbox{}\hfill
  Let~$k\in\matN$.
  Then, $\hSigmakd$ for $d\eqdef1$ and~$\hSigmaki$
  from Definition~\ref{d:lag-lin-forms-Pk1-ref} coincide.
\end{lemma}

\begin{proof}
  Direct consequence of
  Definition~\threfc{d:lag-lin-forms-Pkd-ref}{with $d\eqdef1$},
  Lemma~\thref{l:lag-nodes-Pkd-ref-for-d-eq-1-are-lag-nodes-Pk1-ref},
  Lemma~\threfc{l:first-multi-ind-Akd}{with $d\eqdef1$}, and
  Definition~\thref{d:lag-lin-forms-Pk1-ref}.
\end{proof}

\begin{lemma}[Lagrange linear forms of $\matPkd$ are images of reference]
  \label{l:lag-lin-forms-Pkd-im-ref}
  \mbox{}\\
  Let~$d\geq1$.
  Let~$k\in\matN$.
  Let~$\famvertd{\vv}$ be $d+1$ points in~$\matRd$.
  Let~$\aalpha\in\calAkd$.
  Let~$f:\ArRdR$.\\
  Then, we have $\sigma_\aalpha(f)=\hsigma_\aalpha(\hf)$ where
  $\hf=f\circ\phigeodv$.
\end{lemma}

\begin{proof}
  Direct consequence of
  Definition~\thref{d:lag-lin-forms-Pkd},
  Lemma~\thref{l:lag-nodes-Pkd-im-ref}, and
  Definition~\thref{d:lag-lin-forms-Pkd-ref}.
\end{proof}

\section{Unisolvence for $\FElagP{0}{d}$}
\label{s:unisolvence-P0d}

\begin{lemma}[Lagrange linear forms for $\matPod$ are injective]
  \label{l:lag-lin-forms-P0d-inj}
  \mbox{}\\
  Let $d\geq1$.
  Let $\famvertd{\vv}$ be $d+1$ points of~$\matRd$.
  Then, $\phi_{\Sigma^{\famnodeod{\aa}}}$ is injective.
\end{lemma}

\begin{proof}
  Direct consequence of
  Definition~\thref{d:fe-triple},
  Lemma~\threfc{l:pol-space-P0d-P1d}{%
    thus for all $p\in\matPod$, there exists $a_0\in\matR$ such that
      $p=(\xx\mapsto a_0)$},
  Definition~\threfc{d:lag-lin-forms-Pkd}{%
    thus $\sigma_{\aalpha}(p)=p(\aa_{\aalpha})=a_0=0$, and
    $p=0$},
  Definition~\thref{LM-d:kernel}, and
  Lemma~\thref{LM-l:injective-linear-map-has-zero-kernel}.
\end{proof}

\begin{remark}
  In the next Lemma, the hypothesis of affinely independent points is not
  necessary.
  In fact, Lemmas~\ref{l:card-lag-lin-forms-Pkd} and~\ref{l:num-lag-nodes-Pkd}
  do not require this hypothesis when~$k\eqdef0$, as there is a single node set
  at the isobarycenter of $\famvertd{\vv}$, which is always defined.

  However, to have a proper FE, it is required to have a {\nondegenerate}
  simplex, and thus the hypothesis remains necessary at that stage.
  This is why we did not relax it here (otherwise, the statements in
  Lemmas~\ref{l:card-lag-lin-forms-Pkd} and~\ref{l:num-lag-nodes-Pkd} would
  have been more intricate).
\end{remark}

\begin{lemma}[unisolvence of $\matPod$]
  \label{l:unisolvence-P0d}
  \mbox{}\\
  Let~$d\geq1$.
  Let $\famvertd{\vv}$ be $d+1$ affinely independent points of~$\matRd$.\\
  Then, $\left(\Kvd,\matPod,\Sigma^{\famnodeod{\aa}}\right)$ satisfies the
  unisolvence property.
\end{lemma}

\begin{proof}
  Direct consequence of
  Lemma~\thref{l:dim-Pkd},
  Lemma~\threfc{l:card-lag-lin-forms-Pkd}{with $k\eqdef0$},
  Lemma~\threfc{l:prop-binom-coef}{\eqref{e:prop-binom-coef-0}},
  \assume{the definition of order (reflexivity,
    thus $\dim\matPod=1\geq1=\card\left(\Sigma^{\famnodeod{\aa}}\right)$)},
  Lemma~\thref{l:lag-lin-forms-P0d-inj}, and
  Lemma~\thref{l:inj-implies-unisolvence}.
\end{proof}

\section{Unisolvence for $\FElagP{1}{d}$}
\label{s:unisolvence-P1d}

\begin{lemma}[decomposition of $\matPid$ polynomial with $\sigma_\aalpha$]
  \label{l:decomp-P1d-pol-with-sigma}
  \mbox{}\\
  Let~$d\geq1$.
  Let~$\famvertd{\vv}$ be $d+1$ affinely independent points in~$\matRd$.
  Let $p\in\matPid$.
  Then, we have
  \begin{equation}
    \label{e:decomp-P1d-pol-with-sigma}
    p = \sigma_\zzero (p) \, \barcvd_0
    + \sum_{i = 1}^d \sigma_{\ee_i} (p) \, \barcvd_i
    = \sigma_\zzero (p) \, \calL^{\famvertd{\vv}}_0
    + \sum_{i = 1}^d \sigma_{\ee_i} (p) \, \calL^{\famvertd{\vv}}_i.
  \end{equation}
\end{lemma}

\begin{proof}
  Direct consequence of
  Definition~\thref{d:lag-lin-forms-Pkd},
  Lemma~\threfc{l:lag-nodes-Pid-are-vert}{%
    thus $\sigma_\zzero(p)=p(\vv_0)$, and $\sigma_{\ee_i}(p)=p(\vv_i)$
    for all $i\in[1..d]$},
  Lem\-ma~\thref{l:decomp-aff-map-with-baryc-coor}, and
  Lemma~\thref{l:decomp-P1d-pol-in-lag-basis}.
\end{proof}

\begin{lemma}[Lagrange linear forms for $\matPid$ are injective]
  \label{l:lag-lin-forms-P1d-inj}
  \mbox{}\\
  Let $d\geq1$.
  Let~$\famvertd{\vv}$ be $d+1$ affinely independent points in~$\matRd$.
  Then, $\phi_{\Sigma^{\famnodeid{\aa}}}$ is injective.
\end{lemma}

\begin{proof}
  Direct consequence of
  Definition~\thref{d:fe-triple},
  Lemma~\threfc{l:first-multi-ind-Akd}{with $k\eqdef1$},
  Lemma~\threfc{l:decomp-P1d-pol-with-sigma}{thus
    $\forall\aalpha\in\calAid$, $\sigma_{\aalpha}(p)=0$
    implies $p=0$},
  Definition~\thref{LM-d:kernel}, and
  Lemma~\thref{LM-l:injective-linear-map-has-zero-kernel}.
\end{proof}

\begin{lemma}[unisolvence of $\matPid$]
  \label{l:unisolvence-P1d}
  \mbox{}\hfill
  Let~$d\geq1$.
  Let~$\famvertd{\vv}$ be $d+1$ affinely independent points in~$\matRd$.
  Then, $\left(\Kvd,\matPid,\Sigma^{\famnodeid{\aa}}\right)$ satisfies the
  unisolvence property.
\end{lemma}

\begin{proof}
  Direct consequence of
  Lemma~\thref{l:dim-Pkd},
  Lemma~\threfc{l:card-lag-lin-forms-Pkd}{with $k\eqdef1$},
  Lemma~\threfc{l:prop-binom-coef}{%
    \eqref{e:prop-binom-coef-1} with $n\eqdef d+1$},
  \assume{the definition of order (reflexivity,
    thus $\dim\matPid=d+1$ is greater than or equal to
    $d+1=\card\left(\Sigma^{\famnodeid{\aa}}\right)$)},
  Lemma~\thref{l:lag-lin-forms-P1d-inj}, and
  Lemma~\thref{l:inj-implies-unisolvence}.
\end{proof}

\section{Unisolvence for $\FElagP{k}{d}$}
\label{s:unisolvence-Pkd}

\begin{lemma}[factorization of zero polynomial on last reference hyperplane]
  \label{l:factor-zero-pol-last-ref-hyperpl}
  \mbox{}\\
  Let~$d\geq1$.
  Let~$k\geq1$.
  Let $\hp\in\matPkd$.
  Then, we have the equivalence
  \begin{equation}
    \label{e:factor-zero-pol-last-ref-hyperpl}
    \hp_{|\hcalH^d_d} = 0
    \EQUIV
    \exists \hq \in \matPkmid,\quad \hp = \hcalL^{1,d}_d \, \hq.
  \end{equation}
\end{lemma}

\begin{proof}
  \proofpar{From right to left}
  Direct consequence of
  Lemma~\threfc{l:ref-face-hyperpl}{%
    \eqref{e:ref-face-hyperpl-1} with $i\eqdef d$},
  Definition~\thref{LM-d:kernel}, and
  \assume{the zero-product property in~$\matR$}.

  \proofparskip{From left to right}
  From
  Lemma~\thref{l:decomp-Pkd},
  Lemma~\threfc{l:ref-face-hyperpl}{%
    \eqref{e:ref-face-hyperpl-1} with $i\eqdef d$}, and
  Definition~\threfc{d:lag-pol-P1d-ref}{%
    \eqref{e:lag-pol-P1d-ref-1} with $i\eqdef d$},
  there exists $\tp_0\in\matPkdmi$ and $\hq\in\matPkmid$ such that
  $\hp=\tp_0+X_d\,\hq$, thus for all
  $\hxx=(\hx_1,\ldots,\hx_{d-1},0)\in\hcalH^d_d$, we have
  \[
    0 = \hp (\hxx) = \tp_0 (\hx_1, \ldots, \hx_{d - 1})
    + 0 \, \hq (\hx_1, \ldots, \hx_{d - 1}, 0)
    = \tp_0 (\hx_1, \ldots, \hx_{d - 1}),
  \]
  {\ie} $\tp_0=0$, and $\hp=\hcalL^{1,d}_d\,\hq$.
\end{proof}

\begin{remark}
  In the second part of the proof of the next lemma, we used the circular
  permutation~$\permcd{i}$, see Lemma~\ref{l:circ-permut}.
  The geometric mapping with permutation~$\phipermcd{i}$ is illustrated in
  Figure~\ref{f:geo-map-perm-circ0-d=3-hyperplane} when~$i=0$ and~$d=3$.

  Note that one could also use
  for~$\indd$ the transposition~$\permtrspd{i}$ exchanging~$i$
  and~$d$, see Lemma~\ref{l:trsp},
  which is also injective and satisfies $\permtrspd{i}(d)=i$.
\end{remark}

\begin{lemma}[factorization of zero polynomial on hyperplane $\matPkd$]
  \label{l:factor-zero-pol-hyperpl-Pkd}
  \mbox{}\hfill
  Let~$d\geq1$.
  Let~$k\geq1$.
  Let~$\famvertd{\vv}$ be $d+1$ affinely independent points in~$\matRd$.
  Let~$p\in\matPkd$.
  Let~$i\in[0..d]$.\\
  Then, we have the equivalence
  \begin{equation}
    \label{e:factor-zero-pol-hyperpl-Pkd}
    p_{|\calHvd_i} = 0
    \EQUIV
    \exists q \in \matPkmid,\quad p = \barcvd_i \, q.
  \end{equation}
\end{lemma}

\begin{proof}
  \proofpar{From right to left}
  Direct consequence of
  Lemma~\threfc{l:equiv-def-face-hyperpl}{%
    \eqref{e:equiv-def-face-hyperpl-2} with $i\eqdef0$},
  Definition~\thref{LM-d:kernel}, and
  \assume{the zero-product property in~$\matR$}.

  \proofparskip{From left to right}
  Assume that $p_{|\calHvd_i}=0$.
  Let~$\hp\eqdef p\circ\phipermcd{i}$.\\
  Then, from
  Lemma~\threfc{l:circ-permut}{%
    thus $\permcd{i}$ bijective, $\permcd{i}(d)=i$ and $\permcdinv{i}(i)=d$},
  \assume{the definition of bijectivity (implies injectivity)},
  Lemma~\threfc{l:geo-mapping-permut}{%
    with $\indd\eqdef\permcd{i}$,
    \eqref{e:geo-mapping-permut-3} with $j\eqdef d$,
    thus $\phipermcd{i}(\hcalH^d_d)=\calHvd_i$ and $\hp_{|\hcalH^d_d}=0$},
  Lemma~\thref{l:factor-zero-pol-last-ref-hyperpl},
  Lemma~\threfc{l:geo-mapping-permut}{%
    thus $\phipermcd{i}$ is affine and bijective, and
    \eqref{e:geo-mapping-permut-2} with $i\eqdef i$},
  \assume{the rules of composition with a bijective function},
  Lemma~\thref{l:inv-of-aff-map-is-aff-map}, and
  Lemma~\threfc{l:aff-mapping-of-Pkd-is-Pkl}{%
    with $l\eqdef d$, $k\eqdef k-1$ and $f\eqdef\invphipermcd{i}$ affine},
  there exists~$\hq\in\matPkmid$ such that $\hp=\hcalL^{1,d}_d\,\hq$, and
  thus we have
  \[
    p = \left( \hcalL^{1, d}_{\permcdinv{i}(i)} \circ \invphipermcd{i} \right)
    \left( \hq \circ \invphipermcd{i} \right) = \barcvd_i \, q,
  \]
  with $q\eqdef\hq\circ\invphipermcd{i}\in\matPkmid$.
\end{proof}

\begin{remark}
  See the sketch of the next two proofs in
  Section~\ref{s:sketch-of-the-proof-of-unisolvence-Pkd}.

  The first step dealing with the induction from dimension~$d-1$ to~$d$ (with
  constant degree~$k$) is illustrated in Figure~\ref{f:geo-hyper-map_d32}.
  The second step that passes from degree~$k-1$ to~$k$ (with constant
  dimension~$d$) is illustrated in Figure~\ref{f:sub-vert-sub-nodes_d3-k3}.
\end{remark}

\begin{lemma}[Lagrange linear forms for $\matPkd$ are injective]
  \label{l:lag-lin-forms-Pkd-inj}
  \mbox{}\hfill
  Let~$d\geq1$.
  Let~$k\geq1$.
  Let~$\famvertd{\vv}$ be $d+1$ affinely independent points in~$\matRd$.
  Then, $\phi_{\Sigma^{\famnodekd{\aa}}}$ is injective.
\end{lemma}

\begin{proof}
  For all~$d\geq1$ and~$k\geq1$, let~$P$ be the predicate defined by
  \[
    \PropPP (d, k)
    \eqdef \left[
    \forall \famvertd{\vv} \in \matRd,\quad
      \famvertd{\vv} \mbox{ affinely  independent} \IMPLIES
    \phi_{\Sigma^{\famnodekd{\aa}}} \mbox{ is injective} \right]
  \]
  where
  the application~$\phi_{\Sigma^{\famnodekd{\aa}}}$ is defined in
  Definition~\thref{d:fe-triple},
  the Lagrange linear forms~$\Sigma^{\famnodekd{\aa}}$ are defined in
  Definition~\thref{d:lag-lin-forms-Pkd}, and
  the Lagrange nodes~$\famnodekd{\aa}$ are defined in
  Definition~\thref{d:lag-nodes-Pkd}.

  \proofparskip{Double induction:~$\PropPP(1,k)$, for all~$k\geq1$}\\
  Direct consequence of
  Lemma~\threfc{l:aff-indep-family-of-two-elements}{thus $v_0\neq v_1$},
  Lemma~\thref{l:lag-lin-forms-Pk1-inj}, and
  Lemma~\thref{l:lag-lin-forms-Pkd-for-d-eq-1-are-lag-lin-forms-Pk1}.

  \proofparskip{Double induction:~$\PropPP(d,1)$, for all~$d\geq1$}\\
  Direct consequence of
  Lemma~\thref{l:lag-lin-forms-P1d-inj}.

  \proofparskip{Double induction:
    $\PropPP(d-1,k)\Conj\PropPP(d,k-1)\Implies\PropPP(d,k)$,
    for all $d,k\geq2$}\\
  Let~$d\geq2$.
  Let~$k\geq2$.
  Assume that~$\PropPP(d-1,k)$ and~$\PropPP(d,k-1)$ hold.\\
  Let~$\famvertd{\vv}$ be $d+1$ affinely independent points in~$\matRd$.
  Let~$p\in\matPkd$ such that for all $\aalpha\in\calAkd$, $p(\aa_{\aalpha})=0$
  where the Lagrange nodes~$\famnodekd{\aa}$ are defined in
  Definition~\thref{d:lag-nodes-Pkd}.

  \proofparskip{Step 1: factorization using $\PropPP(d-1,k)$}
  Let~$p_0\eqdef p\circ\phindv{\thetinjdmi{0}}$.\\
  Let~$\aalphap\in\calAkdmi$.
  Then, from
  Lemma~\threfc{l:hyperface-geo-mapping-of-Pkd-is-Pkdmi}{%
    with $i\eqdef0$, thus $p_0\in\matPkdmi$},
  Lemma~\thref{l:Ckd-layers-Akd},
  Lemma~\threfc{l:card-Ckd-Akdm1}{%
    thus $\fkdo(\aalphap)\in\calCkd\subset\calAkd$},
  Lemma~\threfc{l:im-nodes-by-geo-hyperface-mapping}{with $i\eqdef0$}, and
  the assumption on~$p$,
  we have $p_0(\haa^{d-1}_{\aalphap})=p(\aa_{\fkdo(\aalphap)})=0$.\\
  Then, from
  Lemma~\threfc{l:ref-affine-vert-is-affinely-indep}{with $d-1\geq1$}, and
  the property~$\PropPP(d-1,k)$ applied to the reference
  vertices~$\famvertdmi{\hvv^{d-1}}$
    (thus using the reference nodes~$\famnodekdmi{\haa}$),
  $\phi_{\Sigma^{\famnodekdmi{\haa}}}$ is injective. Hence, from
  Lemma~\thref{LM-l:injective-linear-map-has-zero-kernel},
  Definition~\thref{LM-d:kernel},
  Definition~\thref{d:fe-triple}, and
  Definition~\thref{d:lag-lin-forms-Pkd},
  we have $\phi_{\Sigma^{\famnodekdmi{\haa}}}(p_0)
  =(p_0(\haa^{d-1}_{\aalphap}))_{\aalphap\in\calAkdmi}=0$ implies $p_0=0$.\\
  Finally, from
  Lemma~\threfc{l:geo-hyperface-mapping}{%
    with $i\eqdef0$, bijection from $\matRdmi$ onto $\calHvd_0$},
  \assume{the rules of composition with a bijective function}, and
  Lemma~\threfc{l:factor-zero-pol-hyperpl-Pkd}{with $i\eqdef0$},
  we have
  \[
    \forall \xx \in \calHvd_0,\qquad
    p (\xx) = p_0 \circ \left( \phindv{\thetinjdmi{0}} \right)^{-1} (\xx) = 0,
  \]
  {\ie} $p_{|\calHvd_0}=0$, and there exists $q\in\matPkmid$ such that
  $p=\barcvd_0\,q$.

  \proofparskip{Step 2: cancellation using $\PropPP(d,k-1)$}
  Let~$\aalpha\in\calAkmid$.\\
  Then, from
  Lemma~\threfc{l:Ckd-layers-Akd}{%
    thus $\aalpha\in\calAkd$ and $\calAkd\setminus\calCkd=\calAkmid$},
  Lemma~\threfc{l:face-hyperpl-lag-nodes-Pkd}{%
    \eqref{e:face-hyperpl-lag-nodes-Pkd-0}, contrapositive,
    thus $\aa_\aalpha\nin\calHvd_0$},
  Definition~\thref{LM-d:kernel},
  Lemma~\threfc{l:equiv-def-face-hyperpl}{%
    \eqref{e:equiv-def-face-hyperpl-2} with $i\eqdef0$, contrapositive,
    thus $\barcvd_0(\aa_\aalpha)\neq0$},
  the assumption on~$p$, and
  \assume{the zero-product property in~$\matR$
    (applied to $\barcvd_0(\aa_\aalpha)\,q(\aa_\aalpha)=0$)},
  we have $q(\aa_\aalpha)=0$.
  Thus, from
  Lemma~\thref{l:Pkm1d-sub-nodes-sub-vert-are-some-nodes-Pkd},
  we have $q(\uaa_\aalpha)=0$.\\
  Finally, from
  Lemma~\thref{l:sub-vert-aff-indep}, and
  the property~$\PropPP(d,k-1)$ applied to the sub-vertices~$\famvertd{\uvv}$
  (thus using the $\matPkmid$~sub-nodes~$\famnodekmid{\uaa}$),
  $\phi_{\Sigma^{\famnodekmid{\uaa}}}$ is injective.
  Hence, from
  Lemma~\thref{LM-l:injective-linear-map-has-zero-kernel},
  Definition~\thref{LM-d:kernel}, and
  Definition~\thref{d:fe-triple},
  we have
  \[
    \phi_{\Sigma^{\famnodekmid{\uaa}}} (q)
    = (q (\uaa_\aalpha))_{\aalpha \in \calAkmid}
    = 0
    \IMPLIES
    q = 0.
  \]
  Thus, we also have~$p=0$, and from
  Lemma~\thref{LM-l:injective-linear-map-has-zero-kernel},
  Definition~\thref{LM-d:kernel}, and
  Definition~\thref{d:fe-triple},
  $\phi_{\Sigma^{\famnodekd{\aa}}}$ is injective, {\ie} the
  property~$\PropPP(d,k)$ holds, which concludes the double induction.

  \medskip\noindent
  Therefore, from
  Lemma~\threfc{l:double-induction-by-diagonal}{%
    starting from~1 for both $d$ and $k$},
  the property~$\PropPP(d,k)$ holds for all $d,k\geq1$.
\end{proof}

\begin{theorem}[unisolvence of $\matPkd$]
  \label{t:unisolvence-Pkd}
  \mbox{}\hfill
  Let~$d\geq1$.
  Let~$k\geq1$.
  Let~$\famvertd{\vv}$ be $d+1$ affinely independent points in~$\matRd$.
  Then, $\left(\Kvd,\matPkd,\Sigma^{\famnodekd{\aa}}\right)$ satisfies the
  unisolvence property.
\end{theorem}

\begin{proof}
  Direct consequence of
  Lemma~\thref{l:dim-Pkd},
  Lemma~\thref{l:card-lag-lin-forms-Pkd},
  \assume{the definition of order (reflexivity,
    thus $\dim\matPkd\geq\card\left(\Sigma^{\famnodekd{\aa}}\right)$)},
  Definition~\thref{d:lag-lin-forms-Pkd},
  Lemma~\thref{l:lag-lin-forms-Pkd-inj}, and
  Lemma~\thref{l:inj-implies-unisolvence}.
\end{proof}

\begin{remark}
  \mbox{}\\
  The proof of next lemma (from left to right) is similar to step~1
  in the proof of Lemma~\ref{l:lag-lin-forms-Pkd-inj}.

  This face unisolvence result allows to ensure the continuity of functions
  that are piecewise~$\matPkd$ in each element (assuming that the mesh is
  conforming).
  Indeed, let~$f$ be a function that is in~$\matPkd$ in two adjacent mesh
  cells~$K$ and~$\Kp$, sharing a face $\calF=K\cap\Kp$.
  To have~$f$ continuous, it suffices to enforce that the two
  polynomials~$f_{|K}$ and~$f_{|\Kp}$ take the same values at the nodes of the
  face~$\calF$.
\end{remark}


\begin{lemma}[face unisolvence of $\matPkd$]
  \label{l:face-unisolvence-Pkd}
  \mbox{}\hfill
  Let~$d\geq2$.
  Let~$k\geq1$.
  Let~$\calAkdo\eqdef\calCkd$.\\
  Let~$\famvertd{\vv}$ be $d+1$ affinely independent points in~$\matRd$.
  Let~$p\in\matPkd$.
  Let~$i\in[0..d]$.\\
  Then, we have the {\em face unisolvence property}:
  \begin{equation}
    \label{e:face-unisolvence-Pkd}
    \left( \forall \aalpha \in \calAkdi,\quad
    p (\aa_\aalpha) = 0 \right) \EQUIV p_{|\calHvd_i} = 0.
  \end{equation}
\end{lemma}

\begin{proof}
  \proofpar{From right to left}\\
  Direct consequence of
  Lemma~\threfc{l:face-hyperpl-lag-nodes-Pkd}{thus $\aa_\aalpha\in\calHvd_i$}.

  \proofparskip{From left to right}
  Assume that for all $\aalpha\in\calAkdi$, we have $p(\aa_{\aalpha})=0$.
  Let~$p_i\eqdef p\circ\phindv{\thetinjdmi{i}}$.\\
  Let~$\aalphap\in\calAkdmi$.
  Then, from
  Lemma~\threfc{l:hyperface-geo-mapping-of-Pkd-is-Pkdmi}{%
    thus $p_i\in\matPkdmi$},
  Lemma~\threfc{l:card-Ckd-Akdm1}{%
    for $i=0$, thus $\fkdo(\aalphap)\in\calAkdo$},
  Lemma~\threfc{l:card-Akdi-Akdm1}{%
    for $i\in[1..d]$, thus $\fkdi(\aalphap)\in\calAkdi$},
  Lemma~\thref{l:im-nodes-by-geo-hyperface-mapping}, and
  the assumption on~$p$,
  we have $p_i(\haa^{d-1}_{\aalphap})=p(\aa_{\fkdi(\aalphap)})=0$.\\
  Then, from
  Lemma~\threfc{l:ref-affine-vert-is-affinely-indep}{with $d-1\geq1$},
  Lemma~\threfc{l:lag-lin-forms-Pkd-inj}{%
    with $d-1\geq1$ and $\famvertdmi{\hvv^{d-1}}$},
  Lem\-ma~\thref{LM-l:injective-linear-map-has-zero-kernel},
  Definition~\thref{LM-d:kernel},
  Definition~\thref{d:fe-triple}, and
  Definition~\thref{d:lag-lin-forms-Pkd},
  $\phi_{\Sigma^{\famnodekdmi{\haa}}}$ is injective, and thus
  $\phi_{\Sigma^{\famnodekdmi{\haa}}}(p_i)
  =(p_i(\haa^{d-1}_{\aalphap}))_{\aalphap\in\calAkdmi}=0$ implies $p_i=0$.\\
  Finally, from
  Lemma~\threfc{l:geo-hyperface-mapping}{%
    bijection from $\matRdmi$ onto $\calHvd_i$}, and
  \assume{the rules of composition with a bijective function},
  we have
  \[
    \forall \xx \in \calHvd_i,\quad
    p (\xx) = p_i \circ \left( \phindv{\thetinjdmi{i}} \right)^{-1} (\xx) = 0,
  \]
  {\ie} $p_{|\calHvd_i}=0$.
\end{proof}

\clearpage
\section{{\ToPFEkd} Lagrange finite element}
\label{s:Pkd-lag-fe}

\begin{theorem}[$\FElagP{k}{d}$ Lagrange finite element]
  \label{t:Pkd-lag-fe}
  \mbox{}\hfill
  Let~$d\geq1$.
  Let~$k\in\matN$.
  Let $\famvertd{\vv}$ be $d+1$ affinely independent points of~$\matRd$.
  Then, $\FElagP{k}{d}\eqdef\left(\Kvd,\matPkd,\Sigma^{\famnodekd{\aa}}\right)$
  is a finite element.

  It is called the
  {\em Lagrange finite element of degree~$k$ associated with
    vertices~$\famvertd{\vv}$}.
\end{theorem}

\begin{proof}
  Direct consequence of
  Lemma~\thref{l:non-trivial-simplex},
  Lemma~\threfc{l:dim-Pkd}{thus $\dim\matPkd\geq1$ and $\matPkd\neq\{0\}$},
  Lemma~\threfc{l:unisolvence-P0d}{for $k=0$},
  Theorem~\threfc{t:unisolvence-Pkd}{for $k\geq1$}, and
  Definition~\threfc{d:fe-triple}{with $q\eqdef1$}.
\end{proof}

\begin{lemma}[$\FElagP{k}{d}$ Lagrange finite element for $d=1$ is
  Lagrange finite element on a current segment]
  \label{l:Pkd-lag-fe-for-d-eq-1-is-Pk1-lag-fe}
  \mbox{}\hfill
  Let~$k\in\matN$.
  Let~$\famverti{\vv}=\famverti{v}$ be two points in~$\matR^1$.\\
  Then, $\FElagP{k}{d}$ for $d\eqdef1$ and~$\FElagP{k}{1}$ from
  Theorem~\ref{t:Pk1-lag-fe} coincide.
\end{lemma}

\begin{proof}
  Direct consequence of
  Theorem~\threfc{t:Pkd-lag-fe}{with $d\eqdef1$},
  Lem\-ma~\thref{l:lag-lin-forms-Pkd-for-d-eq-1-are-lag-lin-forms-Pk1}, and
  Theorem~\thref{t:Pk1-lag-fe}.
\end{proof}

\begin{theorem}[$\FElagPref{k}{d}$ reference Lagrange finite element]
  \label{t:Pkd-lag-fe-ref}
  \mbox{}\\
  Let~$d\geq1$.
  Let~$k\in\matN$.
  Then, $\FElagPref{k}{d}\eqdef\left(\hKd,\matPkd,\hSigmakd\right)$ is a
  finite element.

  It is called the
  {\em reference Lagrange finite element of degree~$k$ in dimension~$d$}.
\end{theorem}

\begin{proof}
  Direct consequence of
  Lemma~\thref{l:ref-affine-vert-is-affinely-indep},
  Lemma~\thref{l:simplex-of-ref-vert-is-ref-simplex},
  Lemma~\thref{l:lag-nodes-Pkd-ref},
  Definition~\threfc{d:lag-lin-forms-Pkd-ref}{%
    thus $\hSigmakd=\Sigma^{\famnodekd{\haa}}$}, and
  Theorem~\thref{t:Pkd-lag-fe}.
\end{proof}

\begin{lemma}[$\FElagPref{k}{d}$ reference Lagrange finite element for $d=1$ is
  Lagrange finite element on the reference segment]
  \label{l:Pkd-lag-fe-ref-for-d-eq-1-is-Pk1-lag-fe-ref}
  \mbox{}\\
  Let~$k\in\matN$.
  Then, $\FElagPref{k}{d}$ for $d\eqdef1$ and~$\FElagPref{k}{1}$ from
  Theorem~\ref{t:Pk1-lag-fe-ref} coincide.
\end{lemma}

\begin{proof}
  Direct consequence of
  Theorem~\threfc{t:Pkd-lag-fe-ref}{with $d\eqdef1$},
  Lemma~\thref{l:lag-lin-forms-Pkd-ref-for-d-eq-1-are-lag-lin-forms-Pk1-ref},
  and
  Theorem~\thref{t:Pk1-lag-fe-ref}.
\end{proof}

\chapter{Conclusions, perspectives}
\label{c:conclusions-perspectives}

We have presented very detailed proofs for the building of the simplicial
Lagrange finite element (FE) in any nonzero dimension~$d$, and of any
order~$k$.
This includes the general definition of FE, some results about simplicial
geometry, the simplified construction in the monodimensional case on a segment
(mainly for didactic purposes, but also to prepare the tensorization in
$d$-cuboids), the construction of multi-indices of given maximum length,
results about multi-variate polynomials, $\matPid$~Lagrange polynomials, affine
geometric mappings, Lagrange nodes and Lagrange linear forms
of~$\FElagP{k}{d}$, and finally the proofs of unisolvence of~$\FElagP{k}{d}$,
and of face unisolvence.

The short-term purpose of this work was to help the formalization in a formal
proof assistant such as {\coq} of the basic concepts of FE, and the building of
the simplest example: $\FElagP{k}{d}$, that is emblematic and widely used in
practice.
First milestones towards this will be~\cite{mou:phd:24,bol:cfs:24} where
special attention will be paid to the formalization of multi-indices.

Our mid-term purpose is now to continue up to the formalization of the
quadrangular and hexahedric Lagrange FE, then of other (non-nodal) families of
FE such as the face-flux driven Raviart-Thomas FE.

The long-term purpose of these studies is the formal proof of programs
implementing the {\fem}.
As a consequence, after having addressed the formalization of the
{\LMT}~\cite{cm:lmt:16,bol:cfp:17}, and the formalization of parts of measure
theory and Lebesgue integration of nonnegative
functions~\cite{cm:li:21,bol:cfl:22,bol:cfl:23}, we will also have to write
very detailed pen-and-paper proofs for $L^p$~Lebesgue spaces and
$W^{m,p}$~Sobolev spaces as Banach spaces, including parts of the distribution
theory, and the concepts and results of the interpolation and approximation
theory to define the {\FEM} itself.

\chapter*{Acknowledgment}

The authors thank Alexandre Ern for fruitful discussions, especially about
unisolvence issues.

\tocotherhead{chapter}

\bibliography{biblio_fem}
\bibliographystyle{plainnat}


\appendix

\listoffigures

\chapter{Lists of statements}
\label{c:lists-of-statements}

This appendix collects the references (name and number) for all statements
present in Part~\ref{p:detailed-proofs}.
These are split into definitions, lemmas, and theorems.

\tocotherhead{section}

\renewcommand{\listtheoremname}{List of Definitions}
\phantomsection \label{s:list-of-definitions}
\listoftheorems[ignoreall,show=definition]

\clearpage
\renewcommand{\listtheoremname}{List of Lemmas}
\phantomsection \label{s:list-of-lemmas}
\listoftheorems[ignoreall,show=lemma]

\clearpage
\renewcommand{\listtheoremname}{List of Theorems}
\phantomsection \label{s:list-of-theorems}
\listoftheorems[ignoreall,show=theorem]



\chapter{The proof cites explicitly\ldots}
\label{c:the-proof-cites-explicitly}

This appendix gathers the explicit citations of the statements listed in
Appendix~\ref{c:lists-of-statements} that appear in the proof of each result
(lemmas and theorems).
Statements from~\cite{cm:lmt:16,cm:li:21} are anonymized.

The dual graph is described in
Appendix~\ref{c:is-explicitly-cited-in-the-proof-of}.

Printing is not advised!

\bigskip

\begin{description}[style=unboxed]

\item[The proof of Lemma~\thref{l:double-induction-by-diagonal}] \mbox{}\\
  has no explicit citation.

\item[The proof of Lemma~\thref{l:strong-double-induction}] \mbox{}\\
  has no explicit citation.

\item[The proof of Lemma~\thref{l:prop-binom-coef}] \mbox{}\\
  cites explicitly:\\
  Definition~\thref{d:binom-coef}.

\item[The proof of Lemma~\thref{l:circ-permut}] \mbox{}\\
  has no explicit citation.

\item[The proof of Lemma~\thref{l:trsp}] \mbox{}\\
  has no explicit citation.

\item[The proof of Lemma~\thref{l:jump-enum}] \mbox{}\\
  has no explicit citation.

\item[The proof of Lemma~\thref{l:im-ker-incl-ker}] \mbox{}\\
  cites explicitly:\\
  Statement(s) from~\cite{cm:lmt:16}.

\item[The proof of Lemma~\thref{l:im-ker-eq-ker}] \mbox{}\\
  cites explicitly:\\
  Statement(s) from~\cite{cm:lmt:16},\\
  Lemma~\thref{l:im-ker-incl-ker}.

\item[The proof of Lemma~\thref{l:sub-sp-inv-transl}] \mbox{}\\
  cites explicitly:\\
  Statement(s) from~\cite{cm:lmt:16}.

\item[The proof of Lemma~\thref{l:rg-lin-map-is-sub-sp}] \mbox{}\\
  cites explicitly:\\
  Statement(s) from~\cite{cm:lmt:16}.

\item[The proof of Lemma~\thref{l:inj-or-surj-and-dim-implies-bij}] \mbox{}\\
  cites explicitly:\\
  Statement(s) from~\cite{cm:lmt:16}.

\item[The proof of Lemma~\thref{l:inverse-of-isomorphism-is-linear-map}] \mbox{}\\
  cites explicitly:\\
  Statement(s) from~\cite{cm:lmt:16}.

\item[The proof of Lemma~\thref{l:free-family-of-dim-elements-is-basis}] \mbox{}\\
  has no explicit citation.

\item[The proof of Lemma~\thref{l:orig-is-in-aff-sub-sp}] \mbox{}\\
  cites explicitly:\\
  Statement(s) from~\cite{cm:lmt:16},\\
  Definition~\thref{d:aff-sub-sp}.

\item[The proof of Lemma~\thref{l:equiv-def-aff-sub-sp}] \mbox{}\\
  cites explicitly:\\
  Statement(s) from~\cite{cm:lmt:16},\\
  Lemma~\thref{l:sub-sp-inv-transl},\\
  Definition~\thref{d:aff-sub-sp}.

\item[The proof of Lemma~\thref{l:aff-sub-sp-inv-chg-orig}] \mbox{}\\
  cites explicitly:\\
  Statement(s) from~\cite{cm:lmt:16},\\
  Lemma~\thref{l:sub-sp-inv-transl},\\
  Definition~\thref{d:aff-sub-sp},\\
  Lemma~\thref{l:equiv-def-aff-sub-sp}.

\item[The proof of Lemma~\thref{l:vect_sub-sp-is-aff-sub-sp}] \mbox{}\\
  cites explicitly:\\
  Lemma~\thref{l:sub-sp-inv-transl},\\
  Lemma~\thref{l:equiv-def-aff-sub-sp}.

\item[The proof of Lemma~\thref{l:aff-plus-vect-is-aff-sub-sp}] \mbox{}\\
  cites explicitly:\\
  Statement(s) from~\cite{cm:lmt:16},\\
  Definition~\thref{d:aff-sub-sp},\\
  Lemma~\thref{l:equiv-def-aff-sub-sp}.

\item[The proof of Lemma~\thref{l:closed-under-baryc-is-aff-sub-sp}] \mbox{}\\
  cites explicitly:\\
  Statement(s) from~\cite{cm:lmt:16},\\
  Definition~\thref{d:aff-sub-sp},\\
  Lemma~\thref{l:equiv-def-aff-sub-sp}.

\item[The proof of Lemma~\thref{l:baryc-closure-is-aff-sub-sp}] \mbox{}\\
  cites explicitly:\\
  Statement(s) from~\cite{cm:lmt:16},\\
  Definition~\thref{d:aff-sub-sp}.

\item[The proof of Lemma~\thref{l:equiv-def-aff-map}] \mbox{}\\
  cites explicitly:\\
  Definition~\thref{d:aff-sub-sp},\\
  Lemma~\thref{l:equiv-def-aff-sub-sp},\\
  Definition~\thref{d:aff-map}.

\item[The proof of Lemma~\thref{l:change-orig-aff-map}] \mbox{}\\
  cites explicitly:\\
  Statement(s) from~\cite{cm:lmt:16},\\
  Definition~\thref{d:aff-sub-sp},\\
  Lemma~\thref{l:equiv-def-aff-sub-sp},\\
  Definition~\thref{d:aff-map}.

\item[The proof of Lemma~\thref{l:rg-aff-map-is-aff-sub-sp}] \mbox{}\\
  cites explicitly:\\
  Lemma~\thref{l:rg-lin-map-is-sub-sp},\\
  Definition~\thref{d:aff-sub-sp},\\
  Definition~\thref{d:aff-map}.

\item[The proof of Lemma~\thref{l:space-aff-maps}] \mbox{}\\
  cites explicitly:\\
  Statement(s) from~\cite{cm:lmt:16},\\
  Definition~\thref{d:set-aff-maps}.

\item[The proof of Lemma~\thref{l:out-restr-aff-map}] \mbox{}\\
  cites explicitly:\\
  Lemma~\thref{l:aff-sub-sp-inv-chg-orig},\\
  Lemma~\thref{l:rg-aff-map-is-aff-sub-sp},\\
  Definition~\thref{d:set-aff-maps}.

\item[The proof of Lemma~\thref{l:aff-sub-map}] \mbox{}\\
  cites explicitly:\\
  Definition~\thref{d:aff-sub-sp},\\
  Definition~\thref{d:aff-map},\\
  Lemma~\thref{l:change-orig-aff-map},\\
  Definition~\thref{d:set-aff-maps},\\
  Lemma~\thref{l:out-restr-aff-map}.

\item[The proof of Lemma~\thref{l:equiv-def-aff-map-finite-dim}] \mbox{}\\
  cites explicitly:\\
  Definition~\thref{d:aff-map}.

\item[The proof of Lemma~\thref{l:aff-map-preserves-baryc}] \mbox{}\\
  cites explicitly:\\
  Statement(s) from~\cite{cm:lmt:16},\\
  Definition~\thref{d:aff-map}.

\item[The proof of Lemma~\thref{l:aff-map-preserves-isobaryc}] \mbox{}\\
  cites explicitly:\\
  Statement(s) from~\cite{cm:lmt:16},\\
  Lemma~\thref{l:aff-map-preserves-baryc},\\
  Definition~\thref{d:isobaryc}.

\item[The proof of Lemma~\thref{l:aff-map-are-closed-by-composition}] \mbox{}\\
  cites explicitly:\\
  Statement(s) from~\cite{cm:lmt:16},\\
  Definition~\thref{d:aff-map},\\
  Lemma~\thref{l:equiv-def-aff-map},\\
  Definition~\thref{d:set-aff-maps}.

\item[The proof of Lemma~\thref{l:cont-aff-map-is-cont-linear-map}] \mbox{}\\
  cites explicitly:\\
  Statement(s) from~\cite{cm:lmt:16},\\
  Lemma~\thref{l:equiv-def-aff-map}.

\item[The proof of Lemma~\thref{l:inj-aff-sub-map-is-zero-linear-ker}] \mbox{}\\
  cites explicitly:\\
  Statement(s) from~\cite{cm:lmt:16},\\
  Lemma~\thref{l:aff-sub-map}.

\item[The proof of Lemma~\thref{l:inj-aff-map-is-zero-linear-ker}] \mbox{}\\
  cites explicitly:\\
  Lemma~\thref{l:inj-aff-sub-map-is-zero-linear-ker}.

\item[The proof of Lemma~\thref{l:surj-aff-sub-map-is-full-linear-rg}] \mbox{}\\
  cites explicitly:\\
  Statement(s) from~\cite{cm:lmt:16},\\
  Definition~\thref{d:aff-sub-sp},\\
  Lemma~\thref{l:equiv-def-aff-sub-sp},\\
  Lemma~\thref{l:aff-sub-sp-inv-chg-orig},\\
  Definition~\thref{d:aff-map},\\
  Lemma~\thref{l:equiv-def-aff-map}.

\item[The proof of Lemma~\thref{l:surj-aff-map-is-full-linear-rg}] \mbox{}\\
  cites explicitly:\\
  Lemma~\thref{l:surj-aff-sub-map-is-full-linear-rg}.

\item[The proof of Lemma~\thref{l:inv-of-aff-sub-map-is-aff-sub-map}] \mbox{}\\
  cites explicitly:\\
  Statement(s) from~\cite{cm:lmt:16},\\
  Lemma~\thref{l:inverse-of-isomorphism-is-linear-map},\\
  Definition~\thref{d:aff-sub-sp},\\
  Lemma~\thref{l:equiv-def-aff-sub-sp},\\
  Definition~\thref{d:aff-map},\\
  Lemma~\thref{l:equiv-def-aff-map},\\
  Lemma~\thref{l:inj-aff-sub-map-is-zero-linear-ker},\\
  Lemma~\thref{l:surj-aff-sub-map-is-full-linear-rg}.

\item[The proof of Lemma~\thref{l:inv-of-aff-map-is-aff-map}] \mbox{}\\
  cites explicitly:\\
  Lemma~\thref{l:inv-of-aff-sub-map-is-aff-sub-map}.

\item[The proof of Lemma~\thref{l:equiv-def-aff-indep-family}] \mbox{}\\
  cites explicitly:\\
  Statement(s) from~\cite{cm:lmt:16},\\
  Definition~\thref{d:aff-indep-family}.

\item[The proof of Lemma~\thref{l:aff-indep-family-of-two-elements}] \mbox{}\\
  cites explicitly:\\
  Definition~\thref{d:aff-indep-family}.

\item[The proof of Lemma~\thref{l:aff-indep-closed-by-sub-family}] \mbox{}\\
  cites explicitly:\\
  Definition~\thref{d:aff-indep-family},\\
  Lemma~\thref{l:equiv-def-aff-indep-family}.

\item[The proof of Lemma~\thref{l:deg-Pi-leq-k}] \mbox{}\\
  cites explicitly:\\
  Lemma~\thref{l:space-aff-maps},\\
  Lemma~\thref{l:equiv-def-aff-map-finite-dim},\\
  Definition~\thref{d:monom-k1},\\
  Definition~\thref{d:pol-space-Pk1}.

\item[The proof of Lemma~\thref{l:monom-free-in-Pk1}] \mbox{}\\
  has no explicit citation.

\item[The proof of Lemma~\thref{l:dim-Pk1}] \mbox{}\\
  cites explicitly:\\
  Lemma~\thref{l:free-family-of-dim-elements-is-basis},\\
  Definition~\thref{d:pol-space-Pk1},\\
  Lemma~\thref{l:monom-free-in-Pk1}.

\item[The proof of Lemma~\thref{l:prod-2-polynom-univ}] \mbox{}\\
  has no explicit citation.

\item[The proof of Lemma~\thref{l:inj-implies-unisolvence}] \mbox{}\\
  cites explicitly:\\
  Lemma~\thref{l:inj-or-surj-and-dim-implies-bij},\\
  Definition~\thref{d:fe-triple}.

\item[The proof of Lemma~\thref{l:dim-of-P}] \mbox{}\\
  has no explicit citation.

\item[The proof of Lemma~\thref{l:dof-basis}] \mbox{}\\
  cites explicitly:\\
  Lemma~\thref{l:free-family-of-dim-elements-is-basis},\\
  Lemma~\thref{l:dim-of-P}.

\item[The proof of Lemma~\thref{l:ref-isobaryc}] \mbox{}\\
  cites explicitly:\\
  Definition~\thref{d:isobaryc},\\
  Definition~\thref{d:fam-ref-aff-pts}.

\item[The proof of Lemma~\thref{l:ref-affine-vert-is-affinely-indep}] \mbox{}\\
  cites explicitly:\\
  Definition~\thref{d:canon-fam},\\
  Definition~\thref{d:aff-indep-family},\\
  Definition~\thref{d:fam-ref-aff-pts}.

\item[The proof of Lemma~\thref{l:coord-in-ref-simplex-smaller-than-1}] \mbox{}\\
  cites explicitly:\\
  Definition~\thref{d:ref-simplex}.

\item[The proof of Lemma~\thref{l:non-trivial-ref-simplex}] \mbox{}\\
  cites explicitly:\\
  Definition~\thref{d:ref-simplex}.

\item[The proof of Lemma~\thref{l:coord-in-simplex-smaller-than-1}] \mbox{}\\
  cites explicitly:\\
  Definition~\thref{d:simplex}.

\item[The proof of Lemma~\thref{l:simplex-of-ref-vert-is-ref-simplex}] \mbox{}\\
  cites explicitly:\\
  Definition~\thref{d:fam-ref-aff-pts},\\
  Definition~\thref{d:ref-simplex},\\
  Definition~\thref{d:simplex}.

\item[The proof of Lemma~\thref{l:card-Ak1}] \mbox{}\\
  cites explicitly:\\
  Lemma~\thref{l:prop-binom-coef},\\
  Definition~\thref{d:multi-ind-Ak1}.

\item[The proof of Lemma~\thref{l:lag-pol-is-basis-Pk1}] \mbox{}\\
  cites explicitly:\\
  Lemma~\thref{l:free-family-of-dim-elements-is-basis},\\
  Definition~\thref{d:pol-space-Pk1},\\
  Lemma~\thref{l:deg-Pi-leq-k},\\
  Lemma~\thref{l:dim-Pk1},\\
  Lemma~\thref{l:prod-2-polynom-univ},\\
  Definition~\thref{d:lag-pol-Pk1}.

\item[The proof of Lemma~\thref{l:decomp-Pk1-pol-in-lag-basis}] \mbox{}\\
  cites explicitly:\\
  Lemma~\thref{l:lag-pol-is-basis-Pk1}.

\item[The proof of Lemma~\thref{l:ref-simplex-non-trivial-in-R}] \mbox{}\\
  cites explicitly:\\
  Definition~\thref{d:ref-simplex},\\
  Lemma~\thref{l:non-trivial-ref-simplex}.

\item[The proof of Lemma~\thref{l:lag-nodes-distinct-Pk1-ref}] \mbox{}\\
  cites explicitly:\\
  Definition~\thref{d:lag-nodes-Pk1-ref}.

\item[The proof of Lemma~\thref{l:lag-basis-Pk1-ref}] \mbox{}\\
  cites explicitly:\\
  Lemma~\thref{l:lag-pol-is-basis-Pk1},\\
  Lemma~\thref{l:lag-nodes-distinct-Pk1-ref}.

\item[The proof of Lemma~\thref{l:lag-lin-forms-are-linear-Pk1-ref}] \mbox{}\\
  cites explicitly:\\
  Definition~\thref{d:lag-lin-forms-Pk1-ref}.

\item[The proof of Lemma~\thref{l:lag-lin-forms-Pk1-ref-inj}] \mbox{}\\
  cites explicitly:\\
  Statement(s) from~\cite{cm:lmt:16},\\
  Definition~\thref{d:fe-triple},\\
  Lemma~\thref{l:decomp-Pk1-pol-in-lag-basis},\\
  Lemma~\thref{l:lag-nodes-distinct-Pk1-ref},\\
  Definition~\thref{d:lag-lin-forms-Pk1-ref}.

\item[The proof of Lemma~\thref{l:unisolvence-Pk1-ref}] \mbox{}\\
  cites explicitly:\\
  Lemma~\thref{l:dim-Pk1},\\
  Lemma~\thref{l:inj-implies-unisolvence},\\
  Lemma~\thref{l:lag-nodes-distinct-Pk1-ref},\\
  Lemma~\thref{l:lag-lin-forms-Pk1-ref-inj}.

\item[The proof of Theorem~\thref{t:Pk1-lag-fe-ref}] \mbox{}\\
  cites explicitly:\\
  Lemma~\thref{l:dim-Pk1},\\
  Definition~\thref{d:fe-triple},\\
  Lemma~\thref{l:ref-simplex-non-trivial-in-R},\\
  Lemma~\thref{l:unisolvence-Pk1-ref}.

\item[The proof of Lemma~\thref{l:prop-geo-mapping-1d}] \mbox{}\\
  cites explicitly:\\
  Lemma~\thref{l:equiv-def-aff-map-finite-dim},\\
  Lemma~\thref{l:inv-of-aff-map-is-aff-map},\\
  Lemma~\thref{l:deg-Pi-leq-k},\\
  Definition~\thref{d:lag-nodes-Pk1-ref},\\
  Lemma~\thref{l:lag-basis-Pk1-ref},\\
  Definition~\thref{d:geo-mapping-1d}.

\item[The proof of Lemma~\thref{l:curr-simplex-non-trivial-in-R}] \mbox{}\\
  cites explicitly:\\
  Definition~\thref{d:simplex}.

\item[The proof of Lemma~\thref{l:curr-simplex-im-ref-in-R}] \mbox{}\\
  cites explicitly:\\
  Lemma~\thref{l:ref-simplex-non-trivial-in-R},\\
  Definition~\thref{d:geo-mapping-1d},\\
  Lemma~\thref{l:curr-simplex-non-trivial-in-R}.

\item[The proof of Lemma~\thref{l:lag-nodes-distinct-Pk1}] \mbox{}\\
  cites explicitly:\\
  Definition~\thref{d:lag-nodes-Pk1}.

\item[The proof of Lemma~\thref{l:lag-nodes-Pk1-im-ref}] \mbox{}\\
  cites explicitly:\\
  Definition~\thref{d:lag-nodes-Pk1-ref},\\
  Definition~\thref{d:geo-mapping-1d},\\
  Definition~\thref{d:lag-nodes-Pk1}.

\item[The proof of Lemma~\thref{l:lag-basis-Pk1}] \mbox{}\\
  cites explicitly:\\
  Lemma~\thref{l:lag-pol-is-basis-Pk1},\\
  Lemma~\thref{l:lag-nodes-distinct-Pk1}.

\item[The proof of Lemma~\thref{l:lag-pol-Pk1-im-ref}] \mbox{}\\
  cites explicitly:\\
  Definition~\thref{d:lag-pol-Pk1},\\
  Lemma~\thref{l:lag-basis-Pk1-ref},\\
  Definition~\thref{d:geo-mapping-1d},\\
  Lemma~\thref{l:lag-nodes-Pk1-im-ref},\\
  Lemma~\thref{l:lag-basis-Pk1}.

\item[The proof of Lemma~\thref{l:geo-mapping-of-Pk1-is-Pk1}] \mbox{}\\
  cites explicitly:\\
  Definition~\thref{d:pol-space-Pk1},\\
  Lemma~\thref{l:lag-basis-Pk1-ref},\\
  Definition~\thref{d:geo-mapping-1d},\\
  Lemma~\thref{l:lag-basis-Pk1},\\
  Lemma~\thref{l:lag-pol-Pk1-im-ref}.

\item[The proof of Lemma~\thref{l:lag-lin-forms-are-linear-Pk1}] \mbox{}\\
  cites explicitly:\\
  Definition~\thref{d:lag-lin-forms-Pk1}.

\item[The proof of Lemma~\thref{l:lag-lin-forms-Pk1-im-ref}] \mbox{}\\
  cites explicitly:\\
  Definition~\thref{d:lag-lin-forms-Pk1-ref},\\
  Lemma~\thref{l:lag-nodes-Pk1-im-ref},\\
  Definition~\thref{d:lag-lin-forms-Pk1}.

\item[The proof of Lemma~\thref{l:lag-lin-forms-Pk1-inj}] \mbox{}\\
  cites explicitly:\\
  Statement(s) from~\cite{cm:lmt:16},\\
  Definition~\thref{d:fe-triple},\\
  Lemma~\thref{l:decomp-Pk1-pol-in-lag-basis},\\
  Lemma~\thref{l:lag-nodes-distinct-Pk1},\\
  Definition~\thref{d:lag-lin-forms-Pk1}.

\item[The proof of Lemma~\thref{l:unisolvence-Pk1}] \mbox{}\\
  cites explicitly:\\
  Lemma~\thref{l:dim-Pk1},\\
  Lemma~\thref{l:inj-implies-unisolvence},\\
  Lemma~\thref{l:lag-nodes-distinct-Pk1},\\
  Lemma~\thref{l:lag-lin-forms-Pk1-inj}.

\item[The proof of Theorem~\thref{t:Pk1-lag-fe}] \mbox{}\\
  cites explicitly:\\
  Lemma~\thref{l:dim-Pk1},\\
  Definition~\thref{d:fe-triple},\\
  Lemma~\thref{l:curr-simplex-non-trivial-in-R},\\
  Lemma~\thref{l:unisolvence-Pk1}.

\item[The proof of Lemma~\thref{l:len-multi-ind-is-add}] \mbox{}\\
  cites explicitly:\\
  Definition~\thref{d:len-multi-ind}.

\item[The proof of Lemma~\thref{l:fact-multi-ind-pos}] \mbox{}\\
  cites explicitly:\\
  Definition~\thref{d:fact-multi-ind}.

\item[The proof of Lemma~\thref{l:kron-multi-ind-val}] \mbox{}\\
  cites explicitly:\\
  Definition~\thref{d:kron-multi-ind}.

\item[The proof of Lemma~\thref{l:multi-ind-Akd-for-d-eq-1-is-Ak1}] \mbox{}\\
  cites explicitly:\\
  Definition~\thref{d:multi-ind-Ak1},\\
  Definition~\thref{d:len-multi-ind},\\
  Definition~\thref{d:multi-ind-Akd-Ckd}.

\item[The proof of Lemma~\thref{l:ind-smaller-than-max-len}] \mbox{}\\
  cites explicitly:\\
  Definition~\thref{d:len-multi-ind},\\
  Definition~\thref{d:multi-ind-Akd-Ckd}.

\item[The proof of Lemma~\thref{l:first-Ckd}] \mbox{}\\
  cites explicitly:\\
  Definition~\thref{d:canon-fam},\\
  Definition~\thref{d:multi-ind-Akd-Ckd},\\
  Lemma~\thref{l:ind-smaller-than-max-len}.

\item[The proof of Lemma~\thref{l:slices-of-multi-ind-Ckd}] \mbox{}\\
  cites explicitly:\\
  Definition~\thref{d:multi-ind-Akd-Ckd},\\
  Lemma~\thref{l:ind-smaller-than-max-len},\\
  Definition~\thref{d:slices-Sckdi-and-Stkdi}.

\item[The proof of Lemma~\thref{l:card-slices-of-Ckd}] \mbox{}\\
  cites explicitly:\\
  Definition~\thref{d:slices-Sckdi-and-Stkdi}.

\item[The proof of Lemma~\thref{l:card-Ckd}] \mbox{}\\
  cites explicitly:\\
  Definition~\thref{d:binom-coef},\\
  Lemma~\thref{l:prop-binom-coef},\\
  Lemma~\thref{l:first-Ckd},\\
  Lemma~\thref{l:slices-of-multi-ind-Ckd},\\
  Lemma~\thref{l:card-slices-of-Ckd}.

\item[The proof of Lemma~\thref{l:Ckd-layers-Akd}] \mbox{}\\
  cites explicitly:\\
  Definition~\thref{d:multi-ind-Akd-Ckd}.

\item[The proof of Lemma~\thref{l:first-multi-ind-Akd}] \mbox{}\\
  cites explicitly:\\
  Lemma~\thref{l:first-Ckd},\\
  Lemma~\thref{l:Ckd-layers-Akd}.

\item[The proof of Lemma~\thref{l:card-Akd}] \mbox{}\\
  cites explicitly:\\
  Lemma~\thref{l:prop-binom-coef},\\
  Lemma~\thref{l:card-Ckd},\\
  Lemma~\thref{l:Ckd-layers-Akd}.

\item[The proof of Lemma~\thref{l:card-Ckd-Akdm1}] \mbox{}\\
  cites explicitly:\\
  Definition~\thref{d:len-multi-ind},\\
  Definition~\thref{d:multi-ind-Akd-Ckd},\\
  Lemma~\thref{l:card-Ckd},\\
  Lemma~\thref{l:card-Akd}.

\item[The proof of Lemma~\thref{l:card-Akdi-Akdm1}] \mbox{}\\
  cites explicitly:\\
  Definition~\thref{d:len-multi-ind},\\
  Definition~\thref{d:multi-ind-Akd-Ckd}.

\item[The proof of Lemma~\thref{l:monom-kd-for-d-eq-1-is-monom-k1}] \mbox{}\\
  cites explicitly:\\
  Definition~\thref{d:monom-k1},\\
  Lemma~\thref{l:first-Ckd},\\
  Definition~\thref{d:monom-kd}.

\item[The proof of Lemma~\thref{l:pol-space-Pkd-for-d-eq-1-is-Pk1}] \mbox{}\\
  cites explicitly:\\
  Definition~\thref{d:pol-space-Pk1},\\
  Lemma~\thref{l:first-multi-ind-Akd},\\
  Definition~\thref{d:pol-space-Pkd}.

\item[The proof of Lemma~\thref{l:Pkd-sp}] \mbox{}\\
  cites explicitly:\\
  Definition~\thref{d:pol-space-Pkd}.

\item[The proof of Lemma~\thref{l:Pkd-nondecr-k}] \mbox{}\\
  cites explicitly:\\
  Lemma~\thref{l:Ckd-layers-Akd},\\
  Definition~\thref{d:pol-space-Pkd}.

\item[The proof of Lemma~\thref{l:pol-space-P0d-P1d}] \mbox{}\\
  cites explicitly:\\
  Definition~\thref{d:canon-fam},\\
  Lemma~\thref{l:equiv-def-aff-map-finite-dim},\\
  Lemma~\thref{l:first-multi-ind-Akd},\\
  Definition~\thref{d:pol-space-Pkd}.

\item[The proof of Lemma~\thref{l:val-deg-pol}] \mbox{}\\
  cites explicitly:\\
  Definition~\thref{d:deg-pol}.

\item[The proof of Lemma~\thref{l:deg-monom-Ckd-is-k}] \mbox{}\\
  cites explicitly:\\
  Definition~\thref{d:multi-ind-Akd-Ckd},\\
  Definition~\thref{d:monom-kd},\\
  Definition~\thref{d:deg-pol}.

\item[The proof of Lemma~\thref{l:deg-Pkd-leq-k}] \mbox{}\\
  cites explicitly:\\
  Lemma~\thref{l:Ckd-layers-Akd},\\
  Definition~\thref{d:pol-space-Pkd},\\
  Definition~\thref{d:deg-pol}.

\item[The proof of Lemma~\thref{l:prod-monom}] \mbox{}\\
  cites explicitly:\\
  Lemma~\thref{l:len-multi-ind-is-add},\\
  Definition~\thref{d:monom-kd},\\
  Lemma~\thref{l:deg-monom-Ckd-is-k}.

\item[The proof of Lemma~\thref{l:prod-monom-polynom}] \mbox{}\\
  cites explicitly:\\
  Definition~\thref{d:multi-ind-Akd-Ckd},\\
  Definition~\thref{d:pol-space-Pkd},\\
  Lemma~\thref{l:prod-monom}.

\item[The proof of Lemma~\thref{l:prod-2-polynom}] \mbox{}\\
  cites explicitly:\\
  Definition~\thref{d:pol-space-Pkd},\\
  Lemma~\thref{l:Pkd-sp},\\
  Lemma~\thref{l:prod-monom-polynom}.

\item[The proof of Lemma~\thref{l:pder-monom}] \mbox{}\\
  cites explicitly:\\
  Definition~\thref{d:monom-kd},\\
  Lemma~\thref{l:deg-monom-Ckd-is-k}.

\item[The proof of Lemma~\thref{l:pder-is-lin}] \mbox{}\\
  cites explicitly:\\
  Definition~\thref{d:pol-space-Pkd},\\
  Lemma~\thref{l:pder-monom}.

\item[The proof of Lemma~\thref{l:pder-0}] \mbox{}\\
  cites explicitly:\\
  Lemma~\thref{l:pder-is-lin}.

\item[The proof of Lemma~\thref{l:pdef-more-than-deg-is-0}] \mbox{}\\
  cites explicitly:\\
  Definition~\thref{d:multi-ind-Akd-Ckd},\\
  Lemma~\thref{l:pder-monom}.

\item[The proof of Lemma~\thref{l:pder-monom-at-0}] \mbox{}\\
  cites explicitly:\\
  Definition~\thref{d:fact-multi-ind},\\
  Definition~\thref{d:kron-multi-ind},\\
  Lemma~\thref{l:pder-monom}.

\item[The proof of Lemma~\thref{l:monom-free-in-Pkd}] \mbox{}\\
  cites explicitly:\\
  Lemma~\thref{l:fact-multi-ind-pos},\\
  Definition~\thref{d:pol-space-Pkd},\\
  Lemma~\thref{l:pder-is-lin},\\
  Lemma~\thref{l:pder-0},\\
  Lemma~\thref{l:pder-monom-at-0}.

\item[The proof of Lemma~\thref{l:monom-basis-Pkd}] \mbox{}\\
  cites explicitly:\\
  Definition~\thref{d:pol-space-Pkd},\\
  Lemma~\thref{l:monom-free-in-Pkd}.

\item[The proof of Lemma~\thref{l:dim-Pkd}] \mbox{}\\
  cites explicitly:\\
  Lemma~\thref{l:card-Akd},\\
  Lemma~\thref{l:monom-basis-Pkd}.

\item[The proof of Lemma~\thref{l:isom-P0d-P0dm1}] \mbox{}\\
  cites explicitly:\\
  Statement(s) from~\cite{cm:lmt:16},\\
  Lemma~\thref{l:inj-or-surj-and-dim-implies-bij},\\
  Lemma~\thref{l:pol-space-P0d-P1d},\\
  Lemma~\thref{l:dim-Pkd}.

\item[The proof of Lemma~\thref{l:decomp-Pkd}] \mbox{}\\
  cites explicitly:\\
  Definition~\thref{d:slices-Sckdi-and-Stkdi},\\
  Lemma~\thref{l:slices-of-multi-ind-Ckd},\\
  Lemma~\thref{l:card-slices-of-Ckd},\\
  Lemma~\thref{l:Ckd-layers-Akd},\\
  Definition~\thref{d:pol-space-Pkd},\\
  Lemma~\thref{l:Pkd-sp},\\
  Lemma~\thref{l:Pkd-nondecr-k},\\
  Lemma~\thref{l:pol-space-P0d-P1d},\\
  Lemma~\thref{l:deg-monom-Ckd-is-k},\\
  Lemma~\thref{l:monom-free-in-Pkd}.

\item[The proof of Lemma~\thref{l:isom-Pkd-Pkdm1xPkm1d}] \mbox{}\\
  cites explicitly:\\
  Statement(s) from~\cite{cm:lmt:16},\\
  Lemma~\thref{l:prop-binom-coef},\\
  Lemma~\thref{l:inj-or-surj-and-dim-implies-bij},\\
  Lemma~\thref{l:Pkd-sp},\\
  Lemma~\thref{l:dim-Pkd},\\
  Lemma~\thref{l:decomp-Pkd}.

\item[The proof of Lemma~\thref{l:Pkd-nondecr-d}] \mbox{}\\
  cites explicitly:\\
  Definition~\thref{d:pol-space-Pkd},\\
  Lemma~\thref{l:Pkd-sp},\\
  Lemma~\thref{l:pol-space-P0d-P1d},\\
  Lemma~\thref{l:isom-Pkd-Pkdm1xPkm1d}.

\item[The proof of Lemma~\thref{l:Pkd-as-pol-xd}] \mbox{}\\
  cites explicitly:\\
  Lemma~\thref{l:monom-free-in-Pk1},\\
  Lemma~\thref{l:isom-P0d-P0dm1},\\
  Lemma~\thref{l:decomp-Pkd}.

\item[The proof of Lemma~\thref{l:prod-2-polynom-alt-proof}] \mbox{}\\
  cites explicitly:\\
  Statement(s) from~\cite{cm:lmt:16},\\
  Lemma~\thref{l:strong-double-induction},\\
  Lemma~\thref{l:prod-2-polynom-univ},\\
  Lemma~\thref{l:Pkd-sp},\\
  Lemma~\thref{l:Pkd-nondecr-k},\\
  Lemma~\thref{l:pol-space-P0d-P1d},\\
  Lemma~\thref{l:prod-monom-polynom},\\
  Lemma~\thref{l:decomp-Pkd},\\
  Lemma~\thref{l:Pkd-nondecr-d}.

\item[The proof of Lemma~\thref{l:prod-polynom}] \mbox{}\\
  cites explicitly:\\
  Lemma~\thref{l:prod-2-polynom},\\
  Lemma~\thref{l:prod-2-polynom-alt-proof}.

\item[The proof of Lemma~\thref{l:aff-mapping-of-monom-is-Pkl}] \mbox{}\\
  cites explicitly:\\
  Lemma~\thref{l:equiv-def-aff-map-finite-dim},\\
  Definition~\thref{d:len-multi-ind},\\
  Definition~\thref{d:multi-ind-Akd-Ckd},\\
  Definition~\thref{d:monom-kd},\\
  Lemma~\thref{l:pol-space-P0d-P1d},\\
  Lemma~\thref{l:prod-polynom}.

\item[The proof of Lemma~\thref{l:aff-mapping-of-Pkd-is-Pkl}] \mbox{}\\
  cites explicitly:\\
  Lemma~\thref{l:Ckd-layers-Akd},\\
  Definition~\thref{d:pol-space-Pkd},\\
  Lemma~\thref{l:Pkd-sp},\\
  Lemma~\thref{l:Pkd-nondecr-k},\\
  Lemma~\thref{l:aff-mapping-of-monom-is-Pkl}.

\item[The proof of Lemma~\thref{l:ref-lag-pol-P1d-for-d-eq-1-are-ref-lag-pol-Pk1-for-k}] \mbox{}\\
  cites explicitly:\\
  Definition~\thref{d:lag-pol-Pk1},\\
  Definition~\thref{d:lag-nodes-Pk1-ref},\\
  Lemma~\thref{l:lag-basis-Pk1-ref},\\
  Definition~\thref{d:lag-pol-P1d-ref}.

\item[The proof of Lemma~\thref{l:lag-pol-is-basis-P1d-ref}] \mbox{}\\
  cites explicitly:\\
  Lemma~\thref{l:free-family-of-dim-elements-is-basis},\\
  Definition~\thref{d:fam-ref-aff-pts},\\
  Lemma~\thref{l:pol-space-P0d-P1d},\\
  Definition~\thref{d:deg-pol},\\
  Lemma~\thref{l:dim-Pkd},\\
  Definition~\thref{d:lag-pol-P1d-ref}.

\item[The proof of Lemma~\thref{l:diff-lag-pol-ref}] \mbox{}\\
  cites explicitly:\\
  Lemma~\thref{l:pol-space-P0d-P1d},\\
  Lemma~\thref{l:lag-pol-is-basis-P1d-ref}.

\item[The proof of Lemma~\thref{l:geo-mapping-for-d-eq-1-is-geo-mapping-1d}] \mbox{}\\
  cites explicitly:\\
  Definition~\thref{d:geo-mapping-1d},\\
  Definition~\thref{d:lag-pol-P1d-ref},\\
  Definition~\thref{d:geo-mapping}.

\item[The proof of Lemma~\thref{l:ref-geo-mapping-is-id}] \mbox{}\\
  cites explicitly:\\
  Definition~\thref{d:fam-ref-aff-pts},\\
  Definition~\thref{d:lag-pol-P1d-ref},\\
  Definition~\thref{d:geo-mapping}.

\item[The proof of Lemma~\thref{l:prop-geo-mapping}] \mbox{}\\
  cites explicitly:\\
  Statement(s) from~\cite{cm:lmt:16},\\
  Lemma~\thref{l:inj-or-surj-and-dim-implies-bij},\\
  Definition~\thref{d:aff-map},\\
  Lemma~\thref{l:inj-aff-map-is-zero-linear-ker},\\
  Lemma~\thref{l:inv-of-aff-map-is-aff-map},\\
  Definition~\thref{d:aff-indep-family},\\
  Definition~\thref{d:ref-simplex},\\
  Definition~\thref{d:simplex},\\
  Definition~\thref{d:lag-pol-P1d-ref},\\
  Lemma~\thref{l:lag-pol-is-basis-P1d-ref},\\
  Definition~\thref{d:geo-mapping}.

\item[The proof of Lemma~\thref{l:diff-geo-mapping}] \mbox{}\\
  cites explicitly:\\
  Lemma~\thref{l:prop-geo-mapping}.

\item[The proof of Lemma~\thref{l:lag-pol-P1d}] \mbox{}\\
  cites explicitly:\\
  Lemma~\thref{l:prop-geo-mapping}.

\item[The proof of Lemma~\thref{l:lag-pol-of-ref-vert-are-ref-lag-pol-P1d}] \mbox{}\\
  cites explicitly:\\
  Lemma~\thref{l:ref-geo-mapping-is-id},\\
  Lemma~\thref{l:lag-pol-P1d}.

\item[The proof of Lemma~\thref{l:lag-pol-is-basis-P1d}] \mbox{}\\
  cites explicitly:\\
  Lemma~\thref{l:free-family-of-dim-elements-is-basis},\\
  Lemma~\thref{l:aff-map-are-closed-by-composition},\\
  Lemma~\thref{l:pol-space-P0d-P1d},\\
  Lemma~\thref{l:dim-Pkd},\\
  Lemma~\thref{l:lag-pol-is-basis-P1d-ref},\\
  Lemma~\thref{l:prop-geo-mapping},\\
  Lemma~\thref{l:lag-pol-P1d}.

\item[The proof of Lemma~\thref{l:decomp-P1d-pol-in-lag-basis}] \mbox{}\\
  cites explicitly:\\
  Lemma~\thref{l:lag-pol-is-basis-P1d}.

\item[The proof of Lemma~\thref{l:diff-lag-pol-P1d}] \mbox{}\\
  cites explicitly:\\
  Lemma~\thref{l:pol-space-P0d-P1d},\\
  Lemma~\thref{l:diff-lag-pol-ref},\\
  Lemma~\thref{l:diff-geo-mapping},\\
  Lemma~\thref{l:lag-pol-is-basis-P1d}.

\item[The proof of Lemma~\thref{l:non-trivial-simplex}] \mbox{}\\
  cites explicitly:\\
  Lemma~\thref{l:non-trivial-ref-simplex},\\
  Lemma~\thref{l:prop-geo-mapping},\\
  Lemma~\thref{l:diff-geo-mapping}.

\item[The proof of Lemma~\thref{l:baryc-coor}] \mbox{}\\
  cites explicitly:\\
  Statement(s) from~\cite{cm:lmt:16},\\
  Lemma~\thref{l:free-family-of-dim-elements-is-basis},\\
  Definition~\thref{d:aff-indep-family},\\
  Lemma~\thref{l:equiv-def-aff-indep-family}.

\item[The proof of Lemma~\thref{l:lag-pol-P1d-is-baryc-coor}] \mbox{}\\
  cites explicitly:\\
  Definition~\thref{d:geo-mapping},\\
  Lemma~\thref{l:prop-geo-mapping},\\
  Lemma~\thref{l:lag-pol-P1d},\\
  Lemma~\thref{l:baryc-coor}.

\item[The proof of Lemma~\thref{l:decomp-aff-map-with-baryc-coor}] \mbox{}\\
  cites explicitly:\\
  Lemma~\thref{l:decomp-P1d-pol-in-lag-basis},\\
  Lemma~\thref{l:lag-pol-P1d-is-baryc-coor}.

\item[The proof of Lemma~\thref{l:equiv-def-face-hyperpl}] \mbox{}\\
  cites explicitly:\\
  Statement(s) from~\cite{cm:lmt:16},\\
  Lemma~\thref{l:baryc-closure-is-aff-sub-sp},\\
  Lemma~\thref{l:baryc-coor},\\
  Definition~\thref{d:face-hyperpl}.

\item[The proof of Lemma~\thref{l:ref-face-hyperpl}] \mbox{}\\
  cites explicitly:\\
  Statement(s) from~\cite{cm:lmt:16},\\
  Lemma~\thref{l:ref-affine-vert-is-affinely-indep},\\
  Definition~\thref{d:lag-pol-P1d-ref},\\
  Lemma~\thref{l:lag-pol-of-ref-vert-are-ref-lag-pol-P1d},\\
  Lemma~\thref{l:lag-pol-P1d-is-baryc-coor},\\
  Lemma~\thref{l:equiv-def-face-hyperpl}.

\item[The proof of Lemma~\thref{l:face-hyperpl-im-ref-face-hyperpl}] \mbox{}\\
  cites explicitly:\\
  Definition~\thref{d:fam-ref-aff-pts},\\
  Definition~\thref{d:lag-pol-P1d-ref},\\
  Definition~\thref{d:geo-mapping},\\
  Definition~\thref{d:face-hyperpl}.

\item[The proof of Lemma~\thref{l:hyperface-is-incl-face-hyperplane}] \mbox{}\\
  cites explicitly:\\
  Lemma~\thref{l:equiv-def-face-hyperpl},\\
  Definition~\thref{d:hyperface}.

\item[The proof of Lemma~\thref{l:equiv-def-l-face-aff-space}] \mbox{}\\
  cites explicitly:\\
  Lemma~\thref{l:baryc-closure-is-aff-sub-sp},\\
  Definition~\thref{d:l-face-aff-space}.

\item[The proof of Lemma~\thref{l:d-face-aff-space-is-full-space}] \mbox{}\\
  cites explicitly:\\
  Lemma~\thref{l:free-family-of-dim-elements-is-basis},\\
  Definition~\thref{d:aff-indep-family},\\
  Lemma~\thref{l:aff-indep-closed-by-sub-family},\\
  Lemma~\thref{l:equiv-def-l-face-aff-space}.

\item[The proof of Lemma~\thref{l:0-face-aff-space-is-vertex}] \mbox{}\\
  cites explicitly:\\
  Lemma~\thref{l:equiv-def-l-face-aff-space}.

\item[The proof of Lemma~\thref{l:l-face-is-incl-l-face-aff-space}] \mbox{}\\
  cites explicitly:\\
  Definition~\thref{d:l-face-aff-space},\\
  Definition~\thref{d:l-face}.

\item[The proof of Lemma~\thref{l:l-face-is-simplex}] \mbox{}\\
  cites explicitly:\\
  Definition~\thref{d:simplex},\\
  Definition~\thref{d:l-face}.

\item[The proof of Lemma~\thref{l:d-1-face-aff-sp-is-hyperface-aff-sp}] \mbox{}\\
  cites explicitly:\\
  Lemma~\thref{l:jump-enum},\\
  Lemma~\thref{l:equiv-def-face-hyperpl},\\
  Definition~\thref{d:hyperface},\\
  Definition~\thref{d:l-face-aff-space},\\
  Definition~\thref{d:l-face}.

\item[The proof of Lemma~\thref{l:geo-d-face-mapping-is-geo-mapping}] \mbox{}\\
  cites explicitly:\\
  Definition~\thref{d:geo-mapping},\\
  Definition~\thref{d:geo-l-face-mapping}.

\item[The proof of Lemma~\thref{l:prop-geo-l-face-mapping}] \mbox{}\\
  cites explicitly:\\
  Statement(s) from~\cite{cm:lmt:16},\\
  Definition~\thref{d:aff-map},\\
  Lemma~\thref{l:inj-aff-map-is-zero-linear-ker},\\
  Lemma~\thref{l:inv-of-aff-sub-map-is-aff-sub-map},\\
  Lemma~\thref{l:aff-indep-closed-by-sub-family},\\
  Definition~\thref{d:ref-simplex},\\
  Definition~\thref{d:lag-pol-P1d-ref},\\
  Lemma~\thref{l:lag-pol-is-basis-P1d-ref},\\
  Definition~\thref{d:l-face-aff-space},\\
  Lemma~\thref{l:equiv-def-l-face-aff-space},\\
  Definition~\thref{d:l-face},\\
  Definition~\thref{d:geo-l-face-mapping}.

\item[The proof of Lemma~\thref{l:geo-l-face-mapping-of-Pkd-is-Pkl}] \mbox{}\\
  cites explicitly:\\
  Lemma~\thref{l:aff-mapping-of-Pkd-is-Pkl},\\
  Lemma~\thref{l:prop-geo-l-face-mapping}.

\item[The proof of Lemma~\thref{l:geo-mapping-of-Pkd-is-Pkd}] \mbox{}\\
  cites explicitly:\\
  Lemma~\thref{l:aff-mapping-of-Pkd-is-Pkl},\\
  Lemma~\thref{l:prop-geo-mapping},\\
  Lemma~\thref{l:geo-d-face-mapping-is-geo-mapping},\\
  Lemma~\thref{l:geo-l-face-mapping-of-Pkd-is-Pkl}.

\item[The proof of Lemma~\thref{l:geo-hyperface-mapping}] \mbox{}\\
  cites explicitly:\\
  Lemma~\thref{l:d-1-face-aff-sp-is-hyperface-aff-sp},\\
  Lemma~\thref{l:prop-geo-l-face-mapping}.

\item[The proof of Lemma~\thref{l:hyperface-geo-mapping-of-Pkd-is-Pkdmi}] \mbox{}\\
  cites explicitly:\\
  Lemma~\thref{l:d-1-face-aff-sp-is-hyperface-aff-sp},\\
  Lemma~\thref{l:geo-l-face-mapping-of-Pkd-is-Pkl}.

\item[The proof of Lemma~\thref{l:geo-mapping-permut}] \mbox{}\\
  cites explicitly:\\
  Lemma~\thref{l:im-ker-eq-ker},\\
  Lemma~\thref{l:lag-pol-is-basis-P1d-ref},\\
  Lemma~\thref{l:baryc-coor},\\
  Lemma~\thref{l:lag-pol-P1d-is-baryc-coor},\\
  Lemma~\thref{l:equiv-def-face-hyperpl},\\
  Lemma~\thref{l:ref-face-hyperpl},\\
  Lemma~\thref{l:d-face-aff-space-is-full-space},\\
  Lemma~\thref{l:l-face-is-simplex},\\
  Definition~\thref{d:geo-l-face-mapping},\\
  Lemma~\thref{l:prop-geo-l-face-mapping}.

\item[The proof of Lemma~\thref{l:lag-nodes-Pkd-for-d-eq-1-are-lag-nodes-Pk1}] \mbox{}\\
  cites explicitly:\\
  Definition~\thref{d:lag-nodes-Pk1},\\
  Lemma~\thref{l:first-multi-ind-Akd},\\
  Definition~\thref{d:lag-nodes-Pkd}.

\item[The proof of Lemma~\thref{l:num-lag-nodes-Pkd}] \mbox{}\\
  cites explicitly:\\
  Lemma~\thref{l:card-Akd},\\
  Lemma~\thref{l:baryc-coor},\\
  Definition~\thref{d:lag-nodes-Pkd}.

\item[The proof of Lemma~\thref{l:baryc-coor-lag-nodes-Pkd}] \mbox{}\\
  cites explicitly:\\
  Definition~\thref{d:isobaryc},\\
  Definition~\thref{d:len-multi-ind},\\
  Definition~\thref{d:multi-ind-Akd-Ckd},\\
  Lemma~\thref{l:ind-smaller-than-max-len},\\
  Lemma~\thref{l:baryc-coor},\\
  Definition~\thref{d:lag-nodes-Pkd}.

\item[The proof of Lemma~\thref{l:vert-lag-nodes-Pkd}] \mbox{}\\
  cites explicitly:\\
  Definition~\thref{d:canon-fam},\\
  Definition~\thref{d:len-multi-ind},\\
  Definition~\thref{d:multi-ind-Akd-Ckd},\\
  Definition~\thref{d:lag-nodes-Pkd}.

\item[The proof of Lemma~\thref{l:lag-nodes-Pid-are-vert}] \mbox{}\\
  cites explicitly:\\
  Lemma~\thref{l:first-multi-ind-Akd},\\
  Lemma~\thref{l:vert-lag-nodes-Pkd}.

\item[The proof of Lemma~\thref{l:equiv-def-sub-vert-lag-nodes-Pkd}] \mbox{}\\
  cites explicitly:\\
  Definition~\thref{d:canon-fam},\\
  Definition~\thref{d:len-multi-ind},\\
  Definition~\thref{d:multi-ind-Akd-Ckd},\\
  Lemma~\thref{l:baryc-coor-lag-nodes-Pkd},\\
  Definition~\thref{d:sub-vert-lag-nodes-Pkd}.

\item[The proof of Lemma~\thref{l:sub-vert-aff-indep}] \mbox{}\\
  cites explicitly:\\
  Definition~\thref{d:aff-indep-family},\\
  Lemma~\thref{l:equiv-def-sub-vert-lag-nodes-Pkd}.

\item[The proof of Lemma~\thref{l:Pkm1d-sub-nodes-sub-vert-are-some-nodes-Pkd}] \mbox{}\\
  cites explicitly:\\
  Definition~\thref{d:len-multi-ind},\\
  Lemma~\thref{l:Ckd-layers-Akd},\\
  Lemma~\thref{l:baryc-coor-lag-nodes-Pkd},\\
  Lemma~\thref{l:equiv-def-sub-vert-lag-nodes-Pkd},\\
  Lemma~\thref{l:sub-vert-aff-indep}.

\item[The proof of Lemma~\thref{l:lag-nodes-Pkd-ref}] \mbox{}\\
  cites explicitly:\\
  Definition~\thref{d:fam-ref-aff-pts},\\
  Lemma~\thref{l:ref-isobaryc},\\
  Definition~\thref{d:lag-nodes-Pkd},\\
  Lemma~\thref{l:baryc-coor-lag-nodes-Pkd}.

\item[The proof of Lemma~\thref{l:lag-nodes-Pkd-ref-for-d-eq-1-are-lag-nodes-Pk1-ref}] \mbox{}\\
  cites explicitly:\\
  Definition~\thref{d:canon-fam},\\
  Definition~\thref{d:lag-nodes-Pk1-ref},\\
  Lemma~\thref{l:first-multi-ind-Akd},\\
  Lemma~\thref{l:lag-nodes-Pkd-ref}.

\item[The proof of Lemma~\thref{l:equiv-def-lag-nodes-Pkd-ref}] \mbox{}\\
  cites explicitly:\\
  Definition~\thref{d:canon-fam},\\
  Lemma~\thref{l:ref-isobaryc},\\
  Lemma~\thref{l:lag-nodes-Pkd-ref}.

\item[The proof of Lemma~\thref{l:num-lag-nodes-Pkd-ref}] \mbox{}\\
  cites explicitly:\\
  Lemma~\thref{l:ref-affine-vert-is-affinely-indep},\\
  Lemma~\thref{l:card-Akd},\\
  Lemma~\thref{l:baryc-coor},\\
  Lemma~\thref{l:lag-nodes-Pkd-ref}.

\item[The proof of Lemma~\thref{l:lag-nodes-Pkd-im-ref}] \mbox{}\\
  cites explicitly:\\
  Lemma~\thref{l:prop-geo-mapping},\\
  Definition~\thref{d:lag-nodes-Pkd},\\
  Lemma~\thref{l:equiv-def-lag-nodes-Pkd-ref}.

\item[The proof of Lemma~\thref{l:face-hyperpl-lag-nodes-Pkd}] \mbox{}\\
  cites explicitly:\\
  Definition~\thref{d:multi-ind-Akd-Ckd},\\
  Lemma~\thref{l:card-Ckd},\\
  Lemma~\thref{l:card-Akd},\\
  Lemma~\thref{l:card-Ckd-Akdm1},\\
  Lemma~\thref{l:card-Akdi-Akdm1},\\
  Lemma~\thref{l:equiv-def-face-hyperpl},\\
  Lemma~\thref{l:num-lag-nodes-Pkd},\\
  Lemma~\thref{l:baryc-coor-lag-nodes-Pkd}.

\item[The proof of Lemma~\thref{l:im-nodes-by-geo-hyperface-mapping}] \mbox{}\\
  cites explicitly:\\
  Lemma~\thref{l:aff-map-preserves-baryc},\\
  Definition~\thref{d:multi-ind-Akd-Ckd},\\
  Lemma~\thref{l:card-Ckd-Akdm1},\\
  Lemma~\thref{l:card-Akdi-Akdm1},\\
  Lemma~\thref{l:geo-hyperface-mapping},\\
  Lemma~\thref{l:baryc-coor-lag-nodes-Pkd},\\
  Lemma~\thref{l:lag-nodes-Pkd-ref}.

\item[The proof of Lemma~\thref{l:lag-lin-forms-Pkd-for-d-eq-1-are-lag-lin-forms-Pk1}] \mbox{}\\
  cites explicitly:\\
  Definition~\thref{d:lag-lin-forms-Pk1},\\
  Lemma~\thref{l:first-multi-ind-Akd},\\
  Lemma~\thref{l:lag-nodes-Pkd-for-d-eq-1-are-lag-nodes-Pk1},\\
  Definition~\thref{d:lag-lin-forms-Pkd}.

\item[The proof of Lemma~\thref{l:lag-lin-forms-are-linear-Pkd}] \mbox{}\\
  cites explicitly:\\
  Definition~\thref{d:lag-lin-forms-Pkd}.

\item[The proof of Lemma~\thref{l:card-lag-lin-forms-Pkd}] \mbox{}\\
  cites explicitly:\\
  Lemma~\thref{l:num-lag-nodes-Pkd},\\
  Definition~\thref{d:lag-lin-forms-Pkd}.

\item[The proof of Lemma~\thref{l:lag-lin-forms-Pkd-ref-for-d-eq-1-are-lag-lin-forms-Pk1-ref}] \mbox{}\\
  cites explicitly:\\
  Definition~\thref{d:lag-lin-forms-Pk1-ref},\\
  Lemma~\thref{l:first-multi-ind-Akd},\\
  Lemma~\thref{l:lag-nodes-Pkd-ref-for-d-eq-1-are-lag-nodes-Pk1-ref},\\
  Definition~\thref{d:lag-lin-forms-Pkd-ref}.

\item[The proof of Lemma~\thref{l:lag-lin-forms-Pkd-im-ref}] \mbox{}\\
  cites explicitly:\\
  Lemma~\thref{l:lag-nodes-Pkd-im-ref},\\
  Definition~\thref{d:lag-lin-forms-Pkd},\\
  Definition~\thref{d:lag-lin-forms-Pkd-ref}.

\item[The proof of Lemma~\thref{l:lag-lin-forms-P0d-inj}] \mbox{}\\
  cites explicitly:\\
  Statement(s) from~\cite{cm:lmt:16},\\
  Definition~\thref{d:fe-triple},\\
  Lemma~\thref{l:pol-space-P0d-P1d},\\
  Definition~\thref{d:lag-lin-forms-Pkd}.

\item[The proof of Lemma~\thref{l:unisolvence-P0d}] \mbox{}\\
  cites explicitly:\\
  Lemma~\thref{l:prop-binom-coef},\\
  Lemma~\thref{l:inj-implies-unisolvence},\\
  Lemma~\thref{l:dim-Pkd},\\
  Lemma~\thref{l:card-lag-lin-forms-Pkd},\\
  Lemma~\thref{l:lag-lin-forms-P0d-inj}.

\item[The proof of Lemma~\thref{l:decomp-P1d-pol-with-sigma}] \mbox{}\\
  cites explicitly:\\
  Lemma~\thref{l:decomp-P1d-pol-in-lag-basis},\\
  Lemma~\thref{l:decomp-aff-map-with-baryc-coor},\\
  Lemma~\thref{l:lag-nodes-Pid-are-vert},\\
  Definition~\thref{d:lag-lin-forms-Pkd}.

\item[The proof of Lemma~\thref{l:lag-lin-forms-P1d-inj}] \mbox{}\\
  cites explicitly:\\
  Statement(s) from~\cite{cm:lmt:16},\\
  Definition~\thref{d:fe-triple},\\
  Lemma~\thref{l:first-multi-ind-Akd},\\
  Lemma~\thref{l:decomp-P1d-pol-with-sigma}.

\item[The proof of Lemma~\thref{l:unisolvence-P1d}] \mbox{}\\
  cites explicitly:\\
  Lemma~\thref{l:prop-binom-coef},\\
  Lemma~\thref{l:inj-implies-unisolvence},\\
  Lemma~\thref{l:dim-Pkd},\\
  Lemma~\thref{l:card-lag-lin-forms-Pkd},\\
  Lemma~\thref{l:lag-lin-forms-P1d-inj}.

\item[The proof of Lemma~\thref{l:factor-zero-pol-last-ref-hyperpl}] \mbox{}\\
  cites explicitly:\\
  Statement(s) from~\cite{cm:lmt:16},\\
  Lemma~\thref{l:decomp-Pkd},\\
  Definition~\thref{d:lag-pol-P1d-ref},\\
  Lemma~\thref{l:ref-face-hyperpl}.

\item[The proof of Lemma~\thref{l:factor-zero-pol-hyperpl-Pkd}] \mbox{}\\
  cites explicitly:\\
  Statement(s) from~\cite{cm:lmt:16},\\
  Lemma~\thref{l:circ-permut},\\
  Lemma~\thref{l:inv-of-aff-map-is-aff-map},\\
  Lemma~\thref{l:aff-mapping-of-Pkd-is-Pkl},\\
  Lemma~\thref{l:equiv-def-face-hyperpl},\\
  Lemma~\thref{l:geo-mapping-permut},\\
  Lemma~\thref{l:factor-zero-pol-last-ref-hyperpl}.

\item[The proof of Lemma~\thref{l:lag-lin-forms-Pkd-inj}] \mbox{}\\
  cites explicitly:\\
  Statement(s) from~\cite{cm:lmt:16},\\
  Lemma~\thref{l:double-induction-by-diagonal},\\
  Lemma~\thref{l:aff-indep-family-of-two-elements},\\
  Definition~\thref{d:fe-triple},\\
  Lemma~\thref{l:ref-affine-vert-is-affinely-indep},\\
  Lemma~\thref{l:lag-lin-forms-Pk1-inj},\\
  Lemma~\thref{l:Ckd-layers-Akd},\\
  Lemma~\thref{l:card-Ckd-Akdm1},\\
  Lemma~\thref{l:equiv-def-face-hyperpl},\\
  Lemma~\thref{l:geo-hyperface-mapping},\\
  Lemma~\thref{l:hyperface-geo-mapping-of-Pkd-is-Pkdmi},\\
  Definition~\thref{d:lag-nodes-Pkd},\\
  Lemma~\thref{l:sub-vert-aff-indep},\\
  Lemma~\thref{l:Pkm1d-sub-nodes-sub-vert-are-some-nodes-Pkd},\\
  Lemma~\thref{l:face-hyperpl-lag-nodes-Pkd},\\
  Lemma~\thref{l:im-nodes-by-geo-hyperface-mapping},\\
  Definition~\thref{d:lag-lin-forms-Pkd},\\
  Lemma~\thref{l:lag-lin-forms-Pkd-for-d-eq-1-are-lag-lin-forms-Pk1},\\
  Lemma~\thref{l:lag-lin-forms-P1d-inj},\\
  Lemma~\thref{l:factor-zero-pol-hyperpl-Pkd}.

\item[The proof of Theorem~\thref{t:unisolvence-Pkd}] \mbox{}\\
  cites explicitly:\\
  Lemma~\thref{l:inj-implies-unisolvence},\\
  Lemma~\thref{l:dim-Pkd},\\
  Definition~\thref{d:lag-lin-forms-Pkd},\\
  Lemma~\thref{l:card-lag-lin-forms-Pkd},\\
  Lemma~\thref{l:lag-lin-forms-Pkd-inj}.

\item[The proof of Lemma~\thref{l:face-unisolvence-Pkd}] \mbox{}\\
  cites explicitly:\\
  Statement(s) from~\cite{cm:lmt:16},\\
  Definition~\thref{d:fe-triple},\\
  Lemma~\thref{l:ref-affine-vert-is-affinely-indep},\\
  Lemma~\thref{l:card-Ckd-Akdm1},\\
  Lemma~\thref{l:card-Akdi-Akdm1},\\
  Lemma~\thref{l:geo-hyperface-mapping},\\
  Lemma~\thref{l:hyperface-geo-mapping-of-Pkd-is-Pkdmi},\\
  Lemma~\thref{l:face-hyperpl-lag-nodes-Pkd},\\
  Lemma~\thref{l:im-nodes-by-geo-hyperface-mapping},\\
  Definition~\thref{d:lag-lin-forms-Pkd},\\
  Lemma~\thref{l:lag-lin-forms-Pkd-inj}.

\item[The proof of Theorem~\thref{t:Pkd-lag-fe}] \mbox{}\\
  cites explicitly:\\
  Definition~\thref{d:fe-triple},\\
  Lemma~\thref{l:dim-Pkd},\\
  Lemma~\thref{l:non-trivial-simplex},\\
  Lemma~\thref{l:unisolvence-P0d},\\
  Theorem~\thref{t:unisolvence-Pkd}.

\item[The proof of Lemma~\thref{l:Pkd-lag-fe-for-d-eq-1-is-Pk1-lag-fe}] \mbox{}\\
  cites explicitly:\\
  Theorem~\thref{t:Pk1-lag-fe},\\
  Lemma~\thref{l:lag-lin-forms-Pkd-for-d-eq-1-are-lag-lin-forms-Pk1},\\
  Theorem~\thref{t:Pkd-lag-fe}.

\item[The proof of Theorem~\thref{t:Pkd-lag-fe-ref}] \mbox{}\\
  cites explicitly:\\
  Lemma~\thref{l:ref-affine-vert-is-affinely-indep},\\
  Lemma~\thref{l:simplex-of-ref-vert-is-ref-simplex},\\
  Lemma~\thref{l:lag-nodes-Pkd-ref},\\
  Definition~\thref{d:lag-lin-forms-Pkd-ref},\\
  Theorem~\thref{t:Pkd-lag-fe}.

\item[The proof of Lemma~\thref{l:Pkd-lag-fe-ref-for-d-eq-1-is-Pk1-lag-fe-ref}] \mbox{}\\
  cites explicitly:\\
  Theorem~\thref{t:Pk1-lag-fe-ref},\\
  Lemma~\thref{l:lag-lin-forms-Pkd-ref-for-d-eq-1-are-lag-lin-forms-Pk1-ref},\\
  Theorem~\thref{t:Pkd-lag-fe-ref}.

\end{description}

\chapter{Is explicitly cited in the proof of\ldots}
\label{c:is-explicitly-cited-in-the-proof-of}

This appendix gathers the explicit citations that appear in the proof of
results (lemmas and theorems) for each statement listed in
Appendix~\ref{c:lists-of-statements}.
Statements from~\cite{cm:lmt:16,cm:li:21} are anonymized.

The dual graph is described in Appendix~\ref{c:the-proof-cites-explicitly}.

Printing is not advised!

\bigskip

\begin{description}[style=unboxed]

\item[Statement(s) from~\cite{cm:lmt:16}] \mbox{}\\
  are explicitly cited in the proof of:\\
  Lemma~\thref{l:im-ker-incl-ker},\\
  Lemma~\thref{l:im-ker-eq-ker},\\
  Lemma~\thref{l:sub-sp-inv-transl},\\
  Lemma~\thref{l:rg-lin-map-is-sub-sp},\\
  Lemma~\thref{l:inj-or-surj-and-dim-implies-bij},\\
  Lemma~\thref{l:inverse-of-isomorphism-is-linear-map},\\
  Lemma~\thref{l:orig-is-in-aff-sub-sp},\\
  Lemma~\thref{l:equiv-def-aff-sub-sp},\\
  Lemma~\thref{l:aff-sub-sp-inv-chg-orig},\\
  Lemma~\thref{l:aff-plus-vect-is-aff-sub-sp},\\
  Lemma~\thref{l:closed-under-baryc-is-aff-sub-sp},\\
  Lemma~\thref{l:baryc-closure-is-aff-sub-sp},\\
  Lemma~\thref{l:change-orig-aff-map},\\
  Lemma~\thref{l:space-aff-maps},\\
  Lemma~\thref{l:aff-map-preserves-baryc},\\
  Lemma~\thref{l:aff-map-preserves-isobaryc},\\
  Lemma~\thref{l:aff-map-are-closed-by-composition},\\
  Lemma~\thref{l:cont-aff-map-is-cont-linear-map},\\
  Lemma~\thref{l:inj-aff-sub-map-is-zero-linear-ker},\\
  Lemma~\thref{l:surj-aff-sub-map-is-full-linear-rg},\\
  Lemma~\thref{l:inv-of-aff-sub-map-is-aff-sub-map},\\
  Lemma~\thref{l:equiv-def-aff-indep-family},\\
  Lemma~\thref{l:lag-lin-forms-Pk1-ref-inj},\\
  Lemma~\thref{l:lag-lin-forms-Pk1-inj},\\
  Lemma~\thref{l:isom-P0d-P0dm1},\\
  Lemma~\thref{l:isom-Pkd-Pkdm1xPkm1d},\\
  Lemma~\thref{l:prod-2-polynom-alt-proof},\\
  Lemma~\thref{l:prop-geo-mapping},\\
  Lemma~\thref{l:baryc-coor},\\
  Lemma~\thref{l:equiv-def-face-hyperpl},\\
  Lemma~\thref{l:ref-face-hyperpl},\\
  Lemma~\thref{l:prop-geo-l-face-mapping},\\
  Lemma~\thref{l:lag-lin-forms-P0d-inj},\\
  Lemma~\thref{l:lag-lin-forms-P1d-inj},\\
  Lemma~\thref{l:factor-zero-pol-last-ref-hyperpl},\\
  Lemma~\thref{l:factor-zero-pol-hyperpl-Pkd},\\
  Lemma~\thref{l:lag-lin-forms-Pkd-inj},\\
  Lemma~\thref{l:face-unisolvence-Pkd}.

\item[Statement(s) from~\cite{cm:li:21}] \mbox{}\\
  is not yet used.

\item[Lemma~\thref{l:double-induction-by-diagonal}] \mbox{}\\
  is explicitly cited in the proof of:\\
  Lemma~\thref{l:lag-lin-forms-Pkd-inj}.

\item[Lemma~\thref{l:strong-double-induction}] \mbox{}\\
  is explicitly cited in the proof of:\\
  Lemma~\thref{l:prod-2-polynom-alt-proof}.

\item[Definition~\thref{d:binom-coef}] \mbox{}\\
  is explicitly cited in the proof of:\\
  Lemma~\thref{l:prop-binom-coef},\\
  Lemma~\thref{l:card-Ckd}.

\item[Lemma~\thref{l:prop-binom-coef}] \mbox{}\\
  is explicitly cited in the proof of:\\
  Lemma~\thref{l:card-Ak1},\\
  Lemma~\thref{l:card-Ckd},\\
  Lemma~\thref{l:card-Akd},\\
  Lemma~\thref{l:isom-Pkd-Pkdm1xPkm1d},\\
  Lemma~\thref{l:unisolvence-P0d},\\
  Lemma~\thref{l:unisolvence-P1d}.

\item[Definition~\thref{d:canon-fam}] \mbox{}\\
  is explicitly cited in the proof of:\\
  Lemma~\thref{l:ref-affine-vert-is-affinely-indep},\\
  Lemma~\thref{l:first-Ckd},\\
  Lemma~\thref{l:pol-space-P0d-P1d},\\
  Lemma~\thref{l:vert-lag-nodes-Pkd},\\
  Lemma~\thref{l:equiv-def-sub-vert-lag-nodes-Pkd},\\
  Lemma~\thref{l:lag-nodes-Pkd-ref-for-d-eq-1-are-lag-nodes-Pk1-ref},\\
  Lemma~\thref{l:equiv-def-lag-nodes-Pkd-ref}.

\item[Lemma~\thref{l:circ-permut}] \mbox{}\\
  is explicitly cited in the proof of:\\
  Lemma~\thref{l:factor-zero-pol-hyperpl-Pkd}.

\item[Lemma~\thref{l:trsp}] \mbox{}\\
  is not yet used.

\item[Lemma~\thref{l:jump-enum}] \mbox{}\\
  is explicitly cited in the proof of:\\
  Lemma~\thref{l:d-1-face-aff-sp-is-hyperface-aff-sp}.

\item[Lemma~\thref{l:im-ker-incl-ker}] \mbox{}\\
  is explicitly cited in the proof of:\\
  Lemma~\thref{l:im-ker-eq-ker}.

\item[Lemma~\thref{l:im-ker-eq-ker}] \mbox{}\\
  is explicitly cited in the proof of:\\
  Lemma~\thref{l:geo-mapping-permut}.

\item[Lemma~\thref{l:sub-sp-inv-transl}] \mbox{}\\
  is explicitly cited in the proof of:\\
  Lemma~\thref{l:equiv-def-aff-sub-sp},\\
  Lemma~\thref{l:aff-sub-sp-inv-chg-orig},\\
  Lemma~\thref{l:vect_sub-sp-is-aff-sub-sp}.

\item[Lemma~\thref{l:rg-lin-map-is-sub-sp}] \mbox{}\\
  is explicitly cited in the proof of:\\
  Lemma~\thref{l:rg-aff-map-is-aff-sub-sp}.

\item[Lemma~\thref{l:inj-or-surj-and-dim-implies-bij}] \mbox{}\\
  is explicitly cited in the proof of:\\
  Lemma~\thref{l:inj-implies-unisolvence},\\
  Lemma~\thref{l:isom-P0d-P0dm1},\\
  Lemma~\thref{l:isom-Pkd-Pkdm1xPkm1d},\\
  Lemma~\thref{l:prop-geo-mapping}.

\item[Lemma~\thref{l:inverse-of-isomorphism-is-linear-map}] \mbox{}\\
  is explicitly cited in the proof of:\\
  Lemma~\thref{l:inv-of-aff-sub-map-is-aff-sub-map}.

\item[Lemma~\thref{l:free-family-of-dim-elements-is-basis}] \mbox{}\\
  is explicitly cited in the proof of:\\
  Lemma~\thref{l:dim-Pk1},\\
  Lemma~\thref{l:dof-basis},\\
  Lemma~\thref{l:lag-pol-is-basis-Pk1},\\
  Lemma~\thref{l:lag-pol-is-basis-P1d-ref},\\
  Lemma~\thref{l:lag-pol-is-basis-P1d},\\
  Lemma~\thref{l:baryc-coor},\\
  Lemma~\thref{l:d-face-aff-space-is-full-space}.

\item[Definition~\thref{d:aff-sub-sp}] \mbox{}\\
  is explicitly cited in the proof of:\\
  Lemma~\thref{l:orig-is-in-aff-sub-sp},\\
  Lemma~\thref{l:equiv-def-aff-sub-sp},\\
  Lemma~\thref{l:aff-sub-sp-inv-chg-orig},\\
  Lemma~\thref{l:aff-plus-vect-is-aff-sub-sp},\\
  Lemma~\thref{l:closed-under-baryc-is-aff-sub-sp},\\
  Lemma~\thref{l:baryc-closure-is-aff-sub-sp},\\
  Lemma~\thref{l:equiv-def-aff-map},\\
  Lemma~\thref{l:change-orig-aff-map},\\
  Lemma~\thref{l:rg-aff-map-is-aff-sub-sp},\\
  Lemma~\thref{l:aff-sub-map},\\
  Lemma~\thref{l:surj-aff-sub-map-is-full-linear-rg},\\
  Lemma~\thref{l:inv-of-aff-sub-map-is-aff-sub-map}.

\item[Lemma~\thref{l:orig-is-in-aff-sub-sp}] \mbox{}\\
  is not yet used.

\item[Lemma~\thref{l:equiv-def-aff-sub-sp}] \mbox{}\\
  is explicitly cited in the proof of:\\
  Lemma~\thref{l:aff-sub-sp-inv-chg-orig},\\
  Lemma~\thref{l:vect_sub-sp-is-aff-sub-sp},\\
  Lemma~\thref{l:aff-plus-vect-is-aff-sub-sp},\\
  Lemma~\thref{l:closed-under-baryc-is-aff-sub-sp},\\
  Lemma~\thref{l:equiv-def-aff-map},\\
  Lemma~\thref{l:change-orig-aff-map},\\
  Lemma~\thref{l:surj-aff-sub-map-is-full-linear-rg},\\
  Lemma~\thref{l:inv-of-aff-sub-map-is-aff-sub-map}.

\item[Lemma~\thref{l:aff-sub-sp-inv-chg-orig}] \mbox{}\\
  is explicitly cited in the proof of:\\
  Lemma~\thref{l:out-restr-aff-map},\\
  Lemma~\thref{l:surj-aff-sub-map-is-full-linear-rg}.

\item[Lemma~\thref{l:vect_sub-sp-is-aff-sub-sp}] \mbox{}\\
  is not yet used.

\item[Lemma~\thref{l:aff-plus-vect-is-aff-sub-sp}] \mbox{}\\
  is not yet used.

\item[Lemma~\thref{l:closed-under-baryc-is-aff-sub-sp}] \mbox{}\\
  is not yet used.

\item[Lemma~\thref{l:baryc-closure-is-aff-sub-sp}] \mbox{}\\
  is explicitly cited in the proof of:\\
  Lemma~\thref{l:equiv-def-face-hyperpl},\\
  Lemma~\thref{l:equiv-def-l-face-aff-space}.

\item[Definition~\thref{d:aff-map}] \mbox{}\\
  is explicitly cited in the proof of:\\
  Lemma~\thref{l:equiv-def-aff-map},\\
  Lemma~\thref{l:change-orig-aff-map},\\
  Lemma~\thref{l:rg-aff-map-is-aff-sub-sp},\\
  Lemma~\thref{l:aff-sub-map},\\
  Lemma~\thref{l:equiv-def-aff-map-finite-dim},\\
  Lemma~\thref{l:aff-map-preserves-baryc},\\
  Lemma~\thref{l:aff-map-are-closed-by-composition},\\
  Lemma~\thref{l:surj-aff-sub-map-is-full-linear-rg},\\
  Lemma~\thref{l:inv-of-aff-sub-map-is-aff-sub-map},\\
  Lemma~\thref{l:prop-geo-mapping},\\
  Lemma~\thref{l:prop-geo-l-face-mapping}.

\item[Lemma~\thref{l:equiv-def-aff-map}] \mbox{}\\
  is explicitly cited in the proof of:\\
  Lemma~\thref{l:aff-map-are-closed-by-composition},\\
  Lemma~\thref{l:cont-aff-map-is-cont-linear-map},\\
  Lemma~\thref{l:surj-aff-sub-map-is-full-linear-rg},\\
  Lemma~\thref{l:inv-of-aff-sub-map-is-aff-sub-map}.

\item[Lemma~\thref{l:change-orig-aff-map}] \mbox{}\\
  is explicitly cited in the proof of:\\
  Lemma~\thref{l:aff-sub-map}.

\item[Lemma~\thref{l:rg-aff-map-is-aff-sub-sp}] \mbox{}\\
  is explicitly cited in the proof of:\\
  Lemma~\thref{l:out-restr-aff-map}.

\item[Definition~\thref{d:set-aff-maps}] \mbox{}\\
  is explicitly cited in the proof of:\\
  Lemma~\thref{l:space-aff-maps},\\
  Lemma~\thref{l:out-restr-aff-map},\\
  Lemma~\thref{l:aff-sub-map},\\
  Lemma~\thref{l:aff-map-are-closed-by-composition}.

\item[Lemma~\thref{l:space-aff-maps}] \mbox{}\\
  is explicitly cited in the proof of:\\
  Lemma~\thref{l:deg-Pi-leq-k}.

\item[Lemma~\thref{l:out-restr-aff-map}] \mbox{}\\
  is explicitly cited in the proof of:\\
  Lemma~\thref{l:aff-sub-map}.

\item[Lemma~\thref{l:aff-sub-map}] \mbox{}\\
  is explicitly cited in the proof of:\\
  Lemma~\thref{l:inj-aff-sub-map-is-zero-linear-ker}.

\item[Lemma~\thref{l:equiv-def-aff-map-finite-dim}] \mbox{}\\
  is explicitly cited in the proof of:\\
  Lemma~\thref{l:deg-Pi-leq-k},\\
  Lemma~\thref{l:prop-geo-mapping-1d},\\
  Lemma~\thref{l:pol-space-P0d-P1d},\\
  Lemma~\thref{l:aff-mapping-of-monom-is-Pkl}.

\item[Lemma~\thref{l:aff-map-preserves-baryc}] \mbox{}\\
  is explicitly cited in the proof of:\\
  Lemma~\thref{l:aff-map-preserves-isobaryc},\\
  Lemma~\thref{l:im-nodes-by-geo-hyperface-mapping}.

\item[Definition~\thref{d:isobaryc}] \mbox{}\\
  is explicitly cited in the proof of:\\
  Lemma~\thref{l:aff-map-preserves-isobaryc},\\
  Lemma~\thref{l:ref-isobaryc},\\
  Lemma~\thref{l:baryc-coor-lag-nodes-Pkd}.

\item[Lemma~\thref{l:aff-map-preserves-isobaryc}] \mbox{}\\
  is not yet used.

\item[Lemma~\thref{l:aff-map-are-closed-by-composition}] \mbox{}\\
  is explicitly cited in the proof of:\\
  Lemma~\thref{l:lag-pol-is-basis-P1d}.

\item[Lemma~\thref{l:cont-aff-map-is-cont-linear-map}] \mbox{}\\
  is not yet used.

\item[Lemma~\thref{l:inj-aff-sub-map-is-zero-linear-ker}] \mbox{}\\
  is explicitly cited in the proof of:\\
  Lemma~\thref{l:inj-aff-map-is-zero-linear-ker},\\
  Lemma~\thref{l:inv-of-aff-sub-map-is-aff-sub-map}.

\item[Lemma~\thref{l:inj-aff-map-is-zero-linear-ker}] \mbox{}\\
  is explicitly cited in the proof of:\\
  Lemma~\thref{l:prop-geo-mapping},\\
  Lemma~\thref{l:prop-geo-l-face-mapping}.

\item[Lemma~\thref{l:surj-aff-sub-map-is-full-linear-rg}] \mbox{}\\
  is explicitly cited in the proof of:\\
  Lemma~\thref{l:surj-aff-map-is-full-linear-rg},\\
  Lemma~\thref{l:inv-of-aff-sub-map-is-aff-sub-map}.

\item[Lemma~\thref{l:surj-aff-map-is-full-linear-rg}] \mbox{}\\
  is not yet used.

\item[Lemma~\thref{l:inv-of-aff-sub-map-is-aff-sub-map}] \mbox{}\\
  is explicitly cited in the proof of:\\
  Lemma~\thref{l:inv-of-aff-map-is-aff-map},\\
  Lemma~\thref{l:prop-geo-l-face-mapping}.

\item[Lemma~\thref{l:inv-of-aff-map-is-aff-map}] \mbox{}\\
  is explicitly cited in the proof of:\\
  Lemma~\thref{l:prop-geo-mapping-1d},\\
  Lemma~\thref{l:prop-geo-mapping},\\
  Lemma~\thref{l:factor-zero-pol-hyperpl-Pkd}.

\item[Definition~\thref{d:aff-indep-family}] \mbox{}\\
  is explicitly cited in the proof of:\\
  Lemma~\thref{l:equiv-def-aff-indep-family},\\
  Lemma~\thref{l:aff-indep-family-of-two-elements},\\
  Lemma~\thref{l:aff-indep-closed-by-sub-family},\\
  Lemma~\thref{l:ref-affine-vert-is-affinely-indep},\\
  Lemma~\thref{l:prop-geo-mapping},\\
  Lemma~\thref{l:baryc-coor},\\
  Lemma~\thref{l:d-face-aff-space-is-full-space},\\
  Lemma~\thref{l:sub-vert-aff-indep}.

\item[Lemma~\thref{l:equiv-def-aff-indep-family}] \mbox{}\\
  is explicitly cited in the proof of:\\
  Lemma~\thref{l:aff-indep-closed-by-sub-family},\\
  Lemma~\thref{l:baryc-coor}.

\item[Lemma~\thref{l:aff-indep-family-of-two-elements}] \mbox{}\\
  is explicitly cited in the proof of:\\
  Lemma~\thref{l:lag-lin-forms-Pkd-inj}.

\item[Lemma~\thref{l:aff-indep-closed-by-sub-family}] \mbox{}\\
  is explicitly cited in the proof of:\\
  Lemma~\thref{l:d-face-aff-space-is-full-space},\\
  Lemma~\thref{l:prop-geo-l-face-mapping}.

\item[Definition~\thref{d:monom-k1}] \mbox{}\\
  is explicitly cited in the proof of:\\
  Lemma~\thref{l:deg-Pi-leq-k},\\
  Lemma~\thref{l:monom-kd-for-d-eq-1-is-monom-k1}.

\item[Definition~\thref{d:pol-space-Pk1}] \mbox{}\\
  is explicitly cited in the proof of:\\
  Lemma~\thref{l:deg-Pi-leq-k},\\
  Lemma~\thref{l:dim-Pk1},\\
  Lemma~\thref{l:lag-pol-is-basis-Pk1},\\
  Lemma~\thref{l:geo-mapping-of-Pk1-is-Pk1},\\
  Lemma~\thref{l:pol-space-Pkd-for-d-eq-1-is-Pk1}.

\item[Lemma~\thref{l:deg-Pi-leq-k}] \mbox{}\\
  is explicitly cited in the proof of:\\
  Lemma~\thref{l:lag-pol-is-basis-Pk1},\\
  Lemma~\thref{l:prop-geo-mapping-1d}.

\item[Lemma~\thref{l:monom-free-in-Pk1}] \mbox{}\\
  is explicitly cited in the proof of:\\
  Lemma~\thref{l:dim-Pk1},\\
  Lemma~\thref{l:Pkd-as-pol-xd}.

\item[Lemma~\thref{l:dim-Pk1}] \mbox{}\\
  is explicitly cited in the proof of:\\
  Lemma~\thref{l:lag-pol-is-basis-Pk1},\\
  Lemma~\thref{l:unisolvence-Pk1-ref},\\
  Theorem~\thref{t:Pk1-lag-fe-ref},\\
  Lemma~\thref{l:unisolvence-Pk1},\\
  Theorem~\thref{t:Pk1-lag-fe}.

\item[Lemma~\thref{l:prod-2-polynom-univ}] \mbox{}\\
  is explicitly cited in the proof of:\\
  Lemma~\thref{l:lag-pol-is-basis-Pk1},\\
  Lemma~\thref{l:prod-2-polynom-alt-proof}.

\item[Definition~\thref{d:fe-triple}] \mbox{}\\
  is explicitly cited in the proof of:\\
  Lemma~\thref{l:inj-implies-unisolvence},\\
  Lemma~\thref{l:lag-lin-forms-Pk1-ref-inj},\\
  Theorem~\thref{t:Pk1-lag-fe-ref},\\
  Lemma~\thref{l:lag-lin-forms-Pk1-inj},\\
  Theorem~\thref{t:Pk1-lag-fe},\\
  Lemma~\thref{l:lag-lin-forms-P0d-inj},\\
  Lemma~\thref{l:lag-lin-forms-P1d-inj},\\
  Lemma~\thref{l:lag-lin-forms-Pkd-inj},\\
  Lemma~\thref{l:face-unisolvence-Pkd},\\
  Theorem~\thref{t:Pkd-lag-fe}.

\item[Lemma~\thref{l:inj-implies-unisolvence}] \mbox{}\\
  is explicitly cited in the proof of:\\
  Lemma~\thref{l:unisolvence-Pk1-ref},\\
  Lemma~\thref{l:unisolvence-Pk1},\\
  Lemma~\thref{l:unisolvence-P0d},\\
  Lemma~\thref{l:unisolvence-P1d},\\
  Theorem~\thref{t:unisolvence-Pkd}.

\item[Lemma~\thref{l:dim-of-P}] \mbox{}\\
  is explicitly cited in the proof of:\\
  Lemma~\thref{l:dof-basis}.

\item[Lemma~\thref{l:dof-basis}] \mbox{}\\
  is not yet used.

\item[Definition~\thref{d:shape-fun}] \mbox{}\\
  is not yet used.

\item[Definition~\thref{d:fam-aff-pts}] \mbox{}\\
  is not yet used.

\item[Definition~\thref{d:fam-ref-aff-pts}] \mbox{}\\
  is explicitly cited in the proof of:\\
  Lemma~\thref{l:ref-isobaryc},\\
  Lemma~\thref{l:ref-affine-vert-is-affinely-indep},\\
  Lemma~\thref{l:simplex-of-ref-vert-is-ref-simplex},\\
  Lemma~\thref{l:lag-pol-is-basis-P1d-ref},\\
  Lemma~\thref{l:ref-geo-mapping-is-id},\\
  Lemma~\thref{l:face-hyperpl-im-ref-face-hyperpl},\\
  Lemma~\thref{l:lag-nodes-Pkd-ref}.

\item[Lemma~\thref{l:ref-isobaryc}] \mbox{}\\
  is explicitly cited in the proof of:\\
  Lemma~\thref{l:lag-nodes-Pkd-ref},\\
  Lemma~\thref{l:equiv-def-lag-nodes-Pkd-ref}.

\item[Lemma~\thref{l:ref-affine-vert-is-affinely-indep}] \mbox{}\\
  is explicitly cited in the proof of:\\
  Lemma~\thref{l:ref-face-hyperpl},\\
  Lemma~\thref{l:num-lag-nodes-Pkd-ref},\\
  Lemma~\thref{l:lag-lin-forms-Pkd-inj},\\
  Lemma~\thref{l:face-unisolvence-Pkd},\\
  Theorem~\thref{t:Pkd-lag-fe-ref}.

\item[Definition~\thref{d:ref-simplex}] \mbox{}\\
  is explicitly cited in the proof of:\\
  Lemma~\thref{l:coord-in-ref-simplex-smaller-than-1},\\
  Lemma~\thref{l:non-trivial-ref-simplex},\\
  Lemma~\thref{l:simplex-of-ref-vert-is-ref-simplex},\\
  Lemma~\thref{l:ref-simplex-non-trivial-in-R},\\
  Lemma~\thref{l:prop-geo-mapping},\\
  Lemma~\thref{l:prop-geo-l-face-mapping}.

\item[Lemma~\thref{l:coord-in-ref-simplex-smaller-than-1}] \mbox{}\\
  is not yet used.

\item[Lemma~\thref{l:non-trivial-ref-simplex}] \mbox{}\\
  is explicitly cited in the proof of:\\
  Lemma~\thref{l:ref-simplex-non-trivial-in-R},\\
  Lemma~\thref{l:non-trivial-simplex}.

\item[Definition~\thref{d:simplex}] \mbox{}\\
  is explicitly cited in the proof of:\\
  Lemma~\thref{l:coord-in-simplex-smaller-than-1},\\
  Lemma~\thref{l:simplex-of-ref-vert-is-ref-simplex},\\
  Lemma~\thref{l:curr-simplex-non-trivial-in-R},\\
  Lemma~\thref{l:prop-geo-mapping},\\
  Lemma~\thref{l:l-face-is-simplex}.

\item[Lemma~\thref{l:coord-in-simplex-smaller-than-1}] \mbox{}\\
  is not yet used.

\item[Lemma~\thref{l:simplex-of-ref-vert-is-ref-simplex}] \mbox{}\\
  is explicitly cited in the proof of:\\
  Theorem~\thref{t:Pkd-lag-fe-ref}.

\item[Definition~\thref{d:multi-ind-Ak1}] \mbox{}\\
  is explicitly cited in the proof of:\\
  Lemma~\thref{l:card-Ak1},\\
  Lemma~\thref{l:multi-ind-Akd-for-d-eq-1-is-Ak1}.

\item[Lemma~\thref{l:card-Ak1}] \mbox{}\\
  is not yet used.

\item[Definition~\thref{d:lag-pol-Pk1}] \mbox{}\\
  is explicitly cited in the proof of:\\
  Lemma~\thref{l:lag-pol-is-basis-Pk1},\\
  Lemma~\thref{l:lag-pol-Pk1-im-ref},\\
  Lemma~\thref{l:ref-lag-pol-P1d-for-d-eq-1-are-ref-lag-pol-Pk1-for-k}.

\item[Lemma~\thref{l:lag-pol-is-basis-Pk1}] \mbox{}\\
  is explicitly cited in the proof of:\\
  Lemma~\thref{l:decomp-Pk1-pol-in-lag-basis},\\
  Lemma~\thref{l:lag-basis-Pk1-ref},\\
  Lemma~\thref{l:lag-basis-Pk1}.

\item[Lemma~\thref{l:decomp-Pk1-pol-in-lag-basis}] \mbox{}\\
  is explicitly cited in the proof of:\\
  Lemma~\thref{l:lag-lin-forms-Pk1-ref-inj},\\
  Lemma~\thref{l:lag-lin-forms-Pk1-inj}.

\item[Lemma~\thref{l:ref-simplex-non-trivial-in-R}] \mbox{}\\
  is explicitly cited in the proof of:\\
  Theorem~\thref{t:Pk1-lag-fe-ref},\\
  Lemma~\thref{l:curr-simplex-im-ref-in-R}.

\item[Definition~\thref{d:lag-nodes-Pk1-ref}] \mbox{}\\
  is explicitly cited in the proof of:\\
  Lemma~\thref{l:lag-nodes-distinct-Pk1-ref},\\
  Lemma~\thref{l:prop-geo-mapping-1d},\\
  Lemma~\thref{l:lag-nodes-Pk1-im-ref},\\
  Lemma~\thref{l:ref-lag-pol-P1d-for-d-eq-1-are-ref-lag-pol-Pk1-for-k},\\
  Lemma~\thref{l:lag-nodes-Pkd-ref-for-d-eq-1-are-lag-nodes-Pk1-ref}.

\item[Lemma~\thref{l:lag-nodes-distinct-Pk1-ref}] \mbox{}\\
  is explicitly cited in the proof of:\\
  Lemma~\thref{l:lag-basis-Pk1-ref},\\
  Lemma~\thref{l:lag-lin-forms-Pk1-ref-inj},\\
  Lemma~\thref{l:unisolvence-Pk1-ref}.

\item[Lemma~\thref{l:lag-basis-Pk1-ref}] \mbox{}\\
  is explicitly cited in the proof of:\\
  Lemma~\thref{l:prop-geo-mapping-1d},\\
  Lemma~\thref{l:lag-pol-Pk1-im-ref},\\
  Lemma~\thref{l:geo-mapping-of-Pk1-is-Pk1},\\
  Lemma~\thref{l:ref-lag-pol-P1d-for-d-eq-1-are-ref-lag-pol-Pk1-for-k}.

\item[Definition~\thref{d:lag-lin-forms-Pk1-ref}] \mbox{}\\
  is explicitly cited in the proof of:\\
  Lemma~\thref{l:lag-lin-forms-are-linear-Pk1-ref},\\
  Lemma~\thref{l:lag-lin-forms-Pk1-ref-inj},\\
  Lemma~\thref{l:lag-lin-forms-Pk1-im-ref},\\
  Lemma~\thref{l:lag-lin-forms-Pkd-ref-for-d-eq-1-are-lag-lin-forms-Pk1-ref}.

\item[Lemma~\thref{l:lag-lin-forms-are-linear-Pk1-ref}] \mbox{}\\
  is not yet used.

\item[Lemma~\thref{l:lag-lin-forms-Pk1-ref-inj}] \mbox{}\\
  is explicitly cited in the proof of:\\
  Lemma~\thref{l:unisolvence-Pk1-ref}.

\item[Lemma~\thref{l:unisolvence-Pk1-ref}] \mbox{}\\
  is explicitly cited in the proof of:\\
  Theorem~\thref{t:Pk1-lag-fe-ref}.

\item[Theorem~\thref{t:Pk1-lag-fe-ref}] \mbox{}\\
  is explicitly cited in the proof of:\\
  Lemma~\thref{l:Pkd-lag-fe-ref-for-d-eq-1-is-Pk1-lag-fe-ref}.

\item[Definition~\thref{d:geo-mapping-1d}] \mbox{}\\
  is explicitly cited in the proof of:\\
  Lemma~\thref{l:prop-geo-mapping-1d},\\
  Lemma~\thref{l:curr-simplex-im-ref-in-R},\\
  Lemma~\thref{l:lag-nodes-Pk1-im-ref},\\
  Lemma~\thref{l:lag-pol-Pk1-im-ref},\\
  Lemma~\thref{l:geo-mapping-of-Pk1-is-Pk1},\\
  Lemma~\thref{l:geo-mapping-for-d-eq-1-is-geo-mapping-1d}.

\item[Lemma~\thref{l:prop-geo-mapping-1d}] \mbox{}\\
  is not yet used.

\item[Lemma~\thref{l:curr-simplex-non-trivial-in-R}] \mbox{}\\
  is explicitly cited in the proof of:\\
  Lemma~\thref{l:curr-simplex-im-ref-in-R},\\
  Theorem~\thref{t:Pk1-lag-fe}.

\item[Lemma~\thref{l:curr-simplex-im-ref-in-R}] \mbox{}\\
  is not yet used.

\item[Definition~\thref{d:lag-nodes-Pk1}] \mbox{}\\
  is explicitly cited in the proof of:\\
  Lemma~\thref{l:lag-nodes-distinct-Pk1},\\
  Lemma~\thref{l:lag-nodes-Pk1-im-ref},\\
  Lemma~\thref{l:lag-nodes-Pkd-for-d-eq-1-are-lag-nodes-Pk1}.

\item[Lemma~\thref{l:lag-nodes-distinct-Pk1}] \mbox{}\\
  is explicitly cited in the proof of:\\
  Lemma~\thref{l:lag-basis-Pk1},\\
  Lemma~\thref{l:lag-lin-forms-Pk1-inj},\\
  Lemma~\thref{l:unisolvence-Pk1}.

\item[Lemma~\thref{l:lag-nodes-Pk1-im-ref}] \mbox{}\\
  is explicitly cited in the proof of:\\
  Lemma~\thref{l:lag-pol-Pk1-im-ref},\\
  Lemma~\thref{l:lag-lin-forms-Pk1-im-ref}.

\item[Lemma~\thref{l:lag-basis-Pk1}] \mbox{}\\
  is explicitly cited in the proof of:\\
  Lemma~\thref{l:lag-pol-Pk1-im-ref},\\
  Lemma~\thref{l:geo-mapping-of-Pk1-is-Pk1}.

\item[Lemma~\thref{l:lag-pol-Pk1-im-ref}] \mbox{}\\
  is explicitly cited in the proof of:\\
  Lemma~\thref{l:geo-mapping-of-Pk1-is-Pk1}.

\item[Lemma~\thref{l:geo-mapping-of-Pk1-is-Pk1}] \mbox{}\\
  is not yet used.

\item[Definition~\thref{d:lag-lin-forms-Pk1}] \mbox{}\\
  is explicitly cited in the proof of:\\
  Lemma~\thref{l:lag-lin-forms-are-linear-Pk1},\\
  Lemma~\thref{l:lag-lin-forms-Pk1-im-ref},\\
  Lemma~\thref{l:lag-lin-forms-Pk1-inj},\\
  Lemma~\thref{l:lag-lin-forms-Pkd-for-d-eq-1-are-lag-lin-forms-Pk1}.

\item[Lemma~\thref{l:lag-lin-forms-are-linear-Pk1}] \mbox{}\\
  is not yet used.

\item[Lemma~\thref{l:lag-lin-forms-Pk1-im-ref}] \mbox{}\\
  is not yet used.

\item[Lemma~\thref{l:lag-lin-forms-Pk1-inj}] \mbox{}\\
  is explicitly cited in the proof of:\\
  Lemma~\thref{l:unisolvence-Pk1},\\
  Lemma~\thref{l:lag-lin-forms-Pkd-inj}.

\item[Lemma~\thref{l:unisolvence-Pk1}] \mbox{}\\
  is explicitly cited in the proof of:\\
  Theorem~\thref{t:Pk1-lag-fe}.

\item[Theorem~\thref{t:Pk1-lag-fe}] \mbox{}\\
  is explicitly cited in the proof of:\\
  Lemma~\thref{l:Pkd-lag-fe-for-d-eq-1-is-Pk1-lag-fe}.

\item[Definition~\thref{d:len-multi-ind}] \mbox{}\\
  is explicitly cited in the proof of:\\
  Lemma~\thref{l:len-multi-ind-is-add},\\
  Lemma~\thref{l:multi-ind-Akd-for-d-eq-1-is-Ak1},\\
  Lemma~\thref{l:ind-smaller-than-max-len},\\
  Lemma~\thref{l:card-Ckd-Akdm1},\\
  Lemma~\thref{l:card-Akdi-Akdm1},\\
  Lemma~\thref{l:aff-mapping-of-monom-is-Pkl},\\
  Lemma~\thref{l:baryc-coor-lag-nodes-Pkd},\\
  Lemma~\thref{l:vert-lag-nodes-Pkd},\\
  Lemma~\thref{l:equiv-def-sub-vert-lag-nodes-Pkd},\\
  Lemma~\thref{l:Pkm1d-sub-nodes-sub-vert-are-some-nodes-Pkd}.

\item[Lemma~\thref{l:len-multi-ind-is-add}] \mbox{}\\
  is explicitly cited in the proof of:\\
  Lemma~\thref{l:prod-monom}.

\item[Definition~\thref{d:fact-multi-ind}] \mbox{}\\
  is explicitly cited in the proof of:\\
  Lemma~\thref{l:fact-multi-ind-pos},\\
  Lemma~\thref{l:pder-monom-at-0}.

\item[Lemma~\thref{l:fact-multi-ind-pos}] \mbox{}\\
  is explicitly cited in the proof of:\\
  Lemma~\thref{l:monom-free-in-Pkd}.

\item[Definition~\thref{d:kron-multi-ind}] \mbox{}\\
  is explicitly cited in the proof of:\\
  Lemma~\thref{l:kron-multi-ind-val},\\
  Lemma~\thref{l:pder-monom-at-0}.

\item[Lemma~\thref{l:kron-multi-ind-val}] \mbox{}\\
  is not yet used.

\item[Definition~\thref{d:multi-ind-Akd-Ckd}] \mbox{}\\
  is explicitly cited in the proof of:\\
  Lemma~\thref{l:multi-ind-Akd-for-d-eq-1-is-Ak1},\\
  Lemma~\thref{l:ind-smaller-than-max-len},\\
  Lemma~\thref{l:first-Ckd},\\
  Lemma~\thref{l:slices-of-multi-ind-Ckd},\\
  Lemma~\thref{l:Ckd-layers-Akd},\\
  Lemma~\thref{l:card-Ckd-Akdm1},\\
  Lemma~\thref{l:card-Akdi-Akdm1},\\
  Lemma~\thref{l:deg-monom-Ckd-is-k},\\
  Lemma~\thref{l:prod-monom-polynom},\\
  Lemma~\thref{l:pdef-more-than-deg-is-0},\\
  Lemma~\thref{l:aff-mapping-of-monom-is-Pkl},\\
  Lemma~\thref{l:baryc-coor-lag-nodes-Pkd},\\
  Lemma~\thref{l:vert-lag-nodes-Pkd},\\
  Lemma~\thref{l:equiv-def-sub-vert-lag-nodes-Pkd},\\
  Lemma~\thref{l:face-hyperpl-lag-nodes-Pkd},\\
  Lemma~\thref{l:im-nodes-by-geo-hyperface-mapping}.

\item[Lemma~\thref{l:multi-ind-Akd-for-d-eq-1-is-Ak1}] \mbox{}\\
  is not yet used.

\item[Lemma~\thref{l:ind-smaller-than-max-len}] \mbox{}\\
  is explicitly cited in the proof of:\\
  Lemma~\thref{l:first-Ckd},\\
  Lemma~\thref{l:slices-of-multi-ind-Ckd},\\
  Lemma~\thref{l:baryc-coor-lag-nodes-Pkd}.

\item[Lemma~\thref{l:first-Ckd}] \mbox{}\\
  is explicitly cited in the proof of:\\
  Lemma~\thref{l:card-Ckd},\\
  Lemma~\thref{l:first-multi-ind-Akd},\\
  Lemma~\thref{l:monom-kd-for-d-eq-1-is-monom-k1}.

\item[Definition~\thref{d:slices-Sckdi-and-Stkdi}] \mbox{}\\
  is explicitly cited in the proof of:\\
  Lemma~\thref{l:slices-of-multi-ind-Ckd},\\
  Lemma~\thref{l:card-slices-of-Ckd},\\
  Lemma~\thref{l:decomp-Pkd}.

\item[Lemma~\thref{l:slices-of-multi-ind-Ckd}] \mbox{}\\
  is explicitly cited in the proof of:\\
  Lemma~\thref{l:card-Ckd},\\
  Lemma~\thref{l:decomp-Pkd}.

\item[Lemma~\thref{l:card-slices-of-Ckd}] \mbox{}\\
  is explicitly cited in the proof of:\\
  Lemma~\thref{l:card-Ckd},\\
  Lemma~\thref{l:decomp-Pkd}.

\item[Lemma~\thref{l:card-Ckd}] \mbox{}\\
  is explicitly cited in the proof of:\\
  Lemma~\thref{l:card-Akd},\\
  Lemma~\thref{l:card-Ckd-Akdm1},\\
  Lemma~\thref{l:face-hyperpl-lag-nodes-Pkd}.

\item[Lemma~\thref{l:Ckd-layers-Akd}] \mbox{}\\
  is explicitly cited in the proof of:\\
  Lemma~\thref{l:first-multi-ind-Akd},\\
  Lemma~\thref{l:card-Akd},\\
  Lemma~\thref{l:Pkd-nondecr-k},\\
  Lemma~\thref{l:deg-Pkd-leq-k},\\
  Lemma~\thref{l:decomp-Pkd},\\
  Lemma~\thref{l:aff-mapping-of-Pkd-is-Pkl},\\
  Lemma~\thref{l:Pkm1d-sub-nodes-sub-vert-are-some-nodes-Pkd},\\
  Lemma~\thref{l:lag-lin-forms-Pkd-inj}.

\item[Lemma~\thref{l:first-multi-ind-Akd}] \mbox{}\\
  is explicitly cited in the proof of:\\
  Lemma~\thref{l:pol-space-Pkd-for-d-eq-1-is-Pk1},\\
  Lemma~\thref{l:pol-space-P0d-P1d},\\
  Lemma~\thref{l:lag-nodes-Pkd-for-d-eq-1-are-lag-nodes-Pk1},\\
  Lemma~\thref{l:lag-nodes-Pid-are-vert},\\
  Lemma~\thref{l:lag-nodes-Pkd-ref-for-d-eq-1-are-lag-nodes-Pk1-ref},\\
  Lemma~\thref{l:lag-lin-forms-Pkd-for-d-eq-1-are-lag-lin-forms-Pk1},\\
  Lemma~\thref{l:lag-lin-forms-Pkd-ref-for-d-eq-1-are-lag-lin-forms-Pk1-ref},\\
  Lemma~\thref{l:lag-lin-forms-P1d-inj}.

\item[Lemma~\thref{l:card-Akd}] \mbox{}\\
  is explicitly cited in the proof of:\\
  Lemma~\thref{l:card-Ckd-Akdm1},\\
  Lemma~\thref{l:dim-Pkd},\\
  Lemma~\thref{l:num-lag-nodes-Pkd},\\
  Lemma~\thref{l:num-lag-nodes-Pkd-ref},\\
  Lemma~\thref{l:face-hyperpl-lag-nodes-Pkd}.

\item[Lemma~\thref{l:card-Ckd-Akdm1}] \mbox{}\\
  is explicitly cited in the proof of:\\
  Lemma~\thref{l:face-hyperpl-lag-nodes-Pkd},\\
  Lemma~\thref{l:im-nodes-by-geo-hyperface-mapping},\\
  Lemma~\thref{l:lag-lin-forms-Pkd-inj},\\
  Lemma~\thref{l:face-unisolvence-Pkd}.

\item[Lemma~\thref{l:card-Akdi-Akdm1}] \mbox{}\\
  is explicitly cited in the proof of:\\
  Lemma~\thref{l:face-hyperpl-lag-nodes-Pkd},\\
  Lemma~\thref{l:im-nodes-by-geo-hyperface-mapping},\\
  Lemma~\thref{l:face-unisolvence-Pkd}.

\item[Definition~\thref{d:monom-kd}] \mbox{}\\
  is explicitly cited in the proof of:\\
  Lemma~\thref{l:monom-kd-for-d-eq-1-is-monom-k1},\\
  Lemma~\thref{l:deg-monom-Ckd-is-k},\\
  Lemma~\thref{l:prod-monom},\\
  Lemma~\thref{l:pder-monom},\\
  Lemma~\thref{l:aff-mapping-of-monom-is-Pkl}.

\item[Lemma~\thref{l:monom-kd-for-d-eq-1-is-monom-k1}] \mbox{}\\
  is not yet used.

\item[Definition~\thref{d:pol-space-Pkd}] \mbox{}\\
  is explicitly cited in the proof of:\\
  Lemma~\thref{l:pol-space-Pkd-for-d-eq-1-is-Pk1},\\
  Lemma~\thref{l:Pkd-sp},\\
  Lemma~\thref{l:Pkd-nondecr-k},\\
  Lemma~\thref{l:pol-space-P0d-P1d},\\
  Lemma~\thref{l:deg-Pkd-leq-k},\\
  Lemma~\thref{l:prod-monom-polynom},\\
  Lemma~\thref{l:prod-2-polynom},\\
  Lemma~\thref{l:pder-is-lin},\\
  Lemma~\thref{l:monom-free-in-Pkd},\\
  Lemma~\thref{l:monom-basis-Pkd},\\
  Lemma~\thref{l:decomp-Pkd},\\
  Lemma~\thref{l:Pkd-nondecr-d},\\
  Lemma~\thref{l:aff-mapping-of-Pkd-is-Pkl}.

\item[Lemma~\thref{l:pol-space-Pkd-for-d-eq-1-is-Pk1}] \mbox{}\\
  is not yet used.

\item[Lemma~\thref{l:Pkd-sp}] \mbox{}\\
  is explicitly cited in the proof of:\\
  Lemma~\thref{l:prod-2-polynom},\\
  Lemma~\thref{l:decomp-Pkd},\\
  Lemma~\thref{l:isom-Pkd-Pkdm1xPkm1d},\\
  Lemma~\thref{l:Pkd-nondecr-d},\\
  Lemma~\thref{l:prod-2-polynom-alt-proof},\\
  Lemma~\thref{l:aff-mapping-of-Pkd-is-Pkl}.

\item[Lemma~\thref{l:Pkd-nondecr-k}] \mbox{}\\
  is explicitly cited in the proof of:\\
  Lemma~\thref{l:decomp-Pkd},\\
  Lemma~\thref{l:prod-2-polynom-alt-proof},\\
  Lemma~\thref{l:aff-mapping-of-Pkd-is-Pkl}.

\item[Lemma~\thref{l:pol-space-P0d-P1d}] \mbox{}\\
  is explicitly cited in the proof of:\\
  Lemma~\thref{l:isom-P0d-P0dm1},\\
  Lemma~\thref{l:decomp-Pkd},\\
  Lemma~\thref{l:Pkd-nondecr-d},\\
  Lemma~\thref{l:prod-2-polynom-alt-proof},\\
  Lemma~\thref{l:aff-mapping-of-monom-is-Pkl},\\
  Lemma~\thref{l:lag-pol-is-basis-P1d-ref},\\
  Lemma~\thref{l:diff-lag-pol-ref},\\
  Lemma~\thref{l:lag-pol-is-basis-P1d},\\
  Lemma~\thref{l:diff-lag-pol-P1d},\\
  Lemma~\thref{l:lag-lin-forms-P0d-inj}.

\item[Definition~\thref{d:deg-pol}] \mbox{}\\
  is explicitly cited in the proof of:\\
  Lemma~\thref{l:val-deg-pol},\\
  Lemma~\thref{l:deg-monom-Ckd-is-k},\\
  Lemma~\thref{l:deg-Pkd-leq-k},\\
  Lemma~\thref{l:lag-pol-is-basis-P1d-ref}.

\item[Lemma~\thref{l:val-deg-pol}] \mbox{}\\
  is not yet used.

\item[Lemma~\thref{l:deg-monom-Ckd-is-k}] \mbox{}\\
  is explicitly cited in the proof of:\\
  Lemma~\thref{l:prod-monom},\\
  Lemma~\thref{l:pder-monom},\\
  Lemma~\thref{l:decomp-Pkd}.

\item[Lemma~\thref{l:deg-Pkd-leq-k}] \mbox{}\\
  is not yet used.

\item[Lemma~\thref{l:prod-monom}] \mbox{}\\
  is explicitly cited in the proof of:\\
  Lemma~\thref{l:prod-monom-polynom}.

\item[Lemma~\thref{l:prod-monom-polynom}] \mbox{}\\
  is explicitly cited in the proof of:\\
  Lemma~\thref{l:prod-2-polynom},\\
  Lemma~\thref{l:prod-2-polynom-alt-proof}.

\item[Lemma~\thref{l:prod-2-polynom}] \mbox{}\\
  is explicitly cited in the proof of:\\
  Lemma~\thref{l:prod-polynom}.

\item[Lemma~\thref{l:pder-monom}] \mbox{}\\
  is explicitly cited in the proof of:\\
  Lemma~\thref{l:pder-is-lin},\\
  Lemma~\thref{l:pdef-more-than-deg-is-0},\\
  Lemma~\thref{l:pder-monom-at-0}.

\item[Lemma~\thref{l:pder-is-lin}] \mbox{}\\
  is explicitly cited in the proof of:\\
  Lemma~\thref{l:pder-0},\\
  Lemma~\thref{l:monom-free-in-Pkd}.

\item[Lemma~\thref{l:pder-0}] \mbox{}\\
  is explicitly cited in the proof of:\\
  Lemma~\thref{l:monom-free-in-Pkd}.

\item[Lemma~\thref{l:pdef-more-than-deg-is-0}] \mbox{}\\
  is not yet used.

\item[Lemma~\thref{l:pder-monom-at-0}] \mbox{}\\
  is explicitly cited in the proof of:\\
  Lemma~\thref{l:monom-free-in-Pkd}.

\item[Lemma~\thref{l:monom-free-in-Pkd}] \mbox{}\\
  is explicitly cited in the proof of:\\
  Lemma~\thref{l:monom-basis-Pkd},\\
  Lemma~\thref{l:decomp-Pkd}.

\item[Lemma~\thref{l:monom-basis-Pkd}] \mbox{}\\
  is explicitly cited in the proof of:\\
  Lemma~\thref{l:dim-Pkd}.

\item[Lemma~\thref{l:dim-Pkd}] \mbox{}\\
  is explicitly cited in the proof of:\\
  Lemma~\thref{l:isom-P0d-P0dm1},\\
  Lemma~\thref{l:isom-Pkd-Pkdm1xPkm1d},\\
  Lemma~\thref{l:lag-pol-is-basis-P1d-ref},\\
  Lemma~\thref{l:lag-pol-is-basis-P1d},\\
  Lemma~\thref{l:unisolvence-P0d},\\
  Lemma~\thref{l:unisolvence-P1d},\\
  Theorem~\thref{t:unisolvence-Pkd},\\
  Theorem~\thref{t:Pkd-lag-fe}.

\item[Lemma~\thref{l:isom-P0d-P0dm1}] \mbox{}\\
  is explicitly cited in the proof of:\\
  Lemma~\thref{l:Pkd-as-pol-xd}.

\item[Lemma~\thref{l:decomp-Pkd}] \mbox{}\\
  is explicitly cited in the proof of:\\
  Lemma~\thref{l:isom-Pkd-Pkdm1xPkm1d},\\
  Lemma~\thref{l:Pkd-as-pol-xd},\\
  Lemma~\thref{l:prod-2-polynom-alt-proof},\\
  Lemma~\thref{l:factor-zero-pol-last-ref-hyperpl}.

\item[Lemma~\thref{l:isom-Pkd-Pkdm1xPkm1d}] \mbox{}\\
  is explicitly cited in the proof of:\\
  Lemma~\thref{l:Pkd-nondecr-d}.

\item[Lemma~\thref{l:Pkd-nondecr-d}] \mbox{}\\
  is explicitly cited in the proof of:\\
  Lemma~\thref{l:prod-2-polynom-alt-proof}.

\item[Lemma~\thref{l:Pkd-as-pol-xd}] \mbox{}\\
  is not yet used.

\item[Lemma~\thref{l:prod-2-polynom-alt-proof}] \mbox{}\\
  is explicitly cited in the proof of:\\
  Lemma~\thref{l:prod-polynom}.

\item[Lemma~\thref{l:prod-polynom}] \mbox{}\\
  is explicitly cited in the proof of:\\
  Lemma~\thref{l:aff-mapping-of-monom-is-Pkl}.

\item[Lemma~\thref{l:aff-mapping-of-monom-is-Pkl}] \mbox{}\\
  is explicitly cited in the proof of:\\
  Lemma~\thref{l:aff-mapping-of-Pkd-is-Pkl}.

\item[Lemma~\thref{l:aff-mapping-of-Pkd-is-Pkl}] \mbox{}\\
  is explicitly cited in the proof of:\\
  Lemma~\thref{l:geo-l-face-mapping-of-Pkd-is-Pkl},\\
  Lemma~\thref{l:geo-mapping-of-Pkd-is-Pkd},\\
  Lemma~\thref{l:factor-zero-pol-hyperpl-Pkd}.

\item[Definition~\thref{d:lag-pol-P1d-ref}] \mbox{}\\
  is explicitly cited in the proof of:\\
  Lemma~\thref{l:ref-lag-pol-P1d-for-d-eq-1-are-ref-lag-pol-Pk1-for-k},\\
  Lemma~\thref{l:lag-pol-is-basis-P1d-ref},\\
  Lemma~\thref{l:geo-mapping-for-d-eq-1-is-geo-mapping-1d},\\
  Lemma~\thref{l:ref-geo-mapping-is-id},\\
  Lemma~\thref{l:prop-geo-mapping},\\
  Lemma~\thref{l:ref-face-hyperpl},\\
  Lemma~\thref{l:face-hyperpl-im-ref-face-hyperpl},\\
  Lemma~\thref{l:prop-geo-l-face-mapping},\\
  Lemma~\thref{l:factor-zero-pol-last-ref-hyperpl}.

\item[Lemma~\thref{l:ref-lag-pol-P1d-for-d-eq-1-are-ref-lag-pol-Pk1-for-k}] \mbox{}\\
  is not yet used.

\item[Lemma~\thref{l:lag-pol-is-basis-P1d-ref}] \mbox{}\\
  is explicitly cited in the proof of:\\
  Lemma~\thref{l:diff-lag-pol-ref},\\
  Lemma~\thref{l:prop-geo-mapping},\\
  Lemma~\thref{l:lag-pol-is-basis-P1d},\\
  Lemma~\thref{l:prop-geo-l-face-mapping},\\
  Lemma~\thref{l:geo-mapping-permut}.

\item[Lemma~\thref{l:diff-lag-pol-ref}] \mbox{}\\
  is explicitly cited in the proof of:\\
  Lemma~\thref{l:diff-lag-pol-P1d}.

\item[Definition~\thref{d:geo-mapping}] \mbox{}\\
  is explicitly cited in the proof of:\\
  Lemma~\thref{l:geo-mapping-for-d-eq-1-is-geo-mapping-1d},\\
  Lemma~\thref{l:ref-geo-mapping-is-id},\\
  Lemma~\thref{l:prop-geo-mapping},\\
  Lemma~\thref{l:lag-pol-P1d-is-baryc-coor},\\
  Lemma~\thref{l:face-hyperpl-im-ref-face-hyperpl},\\
  Lemma~\thref{l:geo-d-face-mapping-is-geo-mapping}.

\item[Lemma~\thref{l:geo-mapping-for-d-eq-1-is-geo-mapping-1d}] \mbox{}\\
  is not yet used.

\item[Lemma~\thref{l:ref-geo-mapping-is-id}] \mbox{}\\
  is explicitly cited in the proof of:\\
  Lemma~\thref{l:lag-pol-of-ref-vert-are-ref-lag-pol-P1d}.

\item[Lemma~\thref{l:prop-geo-mapping}] \mbox{}\\
  is explicitly cited in the proof of:\\
  Lemma~\thref{l:diff-geo-mapping},\\
  Lemma~\thref{l:lag-pol-P1d},\\
  Lemma~\thref{l:lag-pol-is-basis-P1d},\\
  Lemma~\thref{l:non-trivial-simplex},\\
  Lemma~\thref{l:lag-pol-P1d-is-baryc-coor},\\
  Lemma~\thref{l:geo-mapping-of-Pkd-is-Pkd},\\
  Lemma~\thref{l:lag-nodes-Pkd-im-ref}.

\item[Lemma~\thref{l:diff-geo-mapping}] \mbox{}\\
  is explicitly cited in the proof of:\\
  Lemma~\thref{l:diff-lag-pol-P1d},\\
  Lemma~\thref{l:non-trivial-simplex}.

\item[Lemma~\thref{l:lag-pol-P1d}] \mbox{}\\
  is explicitly cited in the proof of:\\
  Lemma~\thref{l:lag-pol-of-ref-vert-are-ref-lag-pol-P1d},\\
  Lemma~\thref{l:lag-pol-is-basis-P1d},\\
  Lemma~\thref{l:lag-pol-P1d-is-baryc-coor}.

\item[Lemma~\thref{l:lag-pol-of-ref-vert-are-ref-lag-pol-P1d}] \mbox{}\\
  is explicitly cited in the proof of:\\
  Lemma~\thref{l:ref-face-hyperpl}.

\item[Lemma~\thref{l:lag-pol-is-basis-P1d}] \mbox{}\\
  is explicitly cited in the proof of:\\
  Lemma~\thref{l:decomp-P1d-pol-in-lag-basis},\\
  Lemma~\thref{l:diff-lag-pol-P1d}.

\item[Lemma~\thref{l:decomp-P1d-pol-in-lag-basis}] \mbox{}\\
  is explicitly cited in the proof of:\\
  Lemma~\thref{l:decomp-aff-map-with-baryc-coor},\\
  Lemma~\thref{l:decomp-P1d-pol-with-sigma}.

\item[Lemma~\thref{l:diff-lag-pol-P1d}] \mbox{}\\
  is not yet used.

\item[Lemma~\thref{l:non-trivial-simplex}] \mbox{}\\
  is explicitly cited in the proof of:\\
  Theorem~\thref{t:Pkd-lag-fe}.

\item[Lemma~\thref{l:baryc-coor}] \mbox{}\\
  is explicitly cited in the proof of:\\
  Lemma~\thref{l:lag-pol-P1d-is-baryc-coor},\\
  Lemma~\thref{l:equiv-def-face-hyperpl},\\
  Lemma~\thref{l:geo-mapping-permut},\\
  Lemma~\thref{l:num-lag-nodes-Pkd},\\
  Lemma~\thref{l:baryc-coor-lag-nodes-Pkd},\\
  Lemma~\thref{l:num-lag-nodes-Pkd-ref}.

\item[Lemma~\thref{l:lag-pol-P1d-is-baryc-coor}] \mbox{}\\
  is explicitly cited in the proof of:\\
  Lemma~\thref{l:decomp-aff-map-with-baryc-coor},\\
  Lemma~\thref{l:ref-face-hyperpl},\\
  Lemma~\thref{l:geo-mapping-permut}.

\item[Lemma~\thref{l:decomp-aff-map-with-baryc-coor}] \mbox{}\\
  is explicitly cited in the proof of:\\
  Lemma~\thref{l:decomp-P1d-pol-with-sigma}.

\item[Definition~\thref{d:face-hyperpl}] \mbox{}\\
  is explicitly cited in the proof of:\\
  Lemma~\thref{l:equiv-def-face-hyperpl},\\
  Lemma~\thref{l:face-hyperpl-im-ref-face-hyperpl}.

\item[Lemma~\thref{l:equiv-def-face-hyperpl}] \mbox{}\\
  is explicitly cited in the proof of:\\
  Lemma~\thref{l:ref-face-hyperpl},\\
  Lemma~\thref{l:hyperface-is-incl-face-hyperplane},\\
  Lemma~\thref{l:d-1-face-aff-sp-is-hyperface-aff-sp},\\
  Lemma~\thref{l:geo-mapping-permut},\\
  Lemma~\thref{l:face-hyperpl-lag-nodes-Pkd},\\
  Lemma~\thref{l:factor-zero-pol-hyperpl-Pkd},\\
  Lemma~\thref{l:lag-lin-forms-Pkd-inj}.

\item[Lemma~\thref{l:ref-face-hyperpl}] \mbox{}\\
  is explicitly cited in the proof of:\\
  Lemma~\thref{l:geo-mapping-permut},\\
  Lemma~\thref{l:factor-zero-pol-last-ref-hyperpl}.

\item[Lemma~\thref{l:face-hyperpl-im-ref-face-hyperpl}] \mbox{}\\
  is not yet used.

\item[Definition~\thref{d:hyperface}] \mbox{}\\
  is explicitly cited in the proof of:\\
  Lemma~\thref{l:hyperface-is-incl-face-hyperplane},\\
  Lemma~\thref{l:d-1-face-aff-sp-is-hyperface-aff-sp}.

\item[Lemma~\thref{l:hyperface-is-incl-face-hyperplane}] \mbox{}\\
  is not yet used.

\item[Definition~\thref{d:l-face-aff-space}] \mbox{}\\
  is explicitly cited in the proof of:\\
  Lemma~\thref{l:equiv-def-l-face-aff-space},\\
  Lemma~\thref{l:l-face-is-incl-l-face-aff-space},\\
  Lemma~\thref{l:d-1-face-aff-sp-is-hyperface-aff-sp},\\
  Lemma~\thref{l:prop-geo-l-face-mapping}.

\item[Lemma~\thref{l:equiv-def-l-face-aff-space}] \mbox{}\\
  is explicitly cited in the proof of:\\
  Lemma~\thref{l:d-face-aff-space-is-full-space},\\
  Lemma~\thref{l:0-face-aff-space-is-vertex},\\
  Lemma~\thref{l:prop-geo-l-face-mapping}.

\item[Lemma~\thref{l:d-face-aff-space-is-full-space}] \mbox{}\\
  is explicitly cited in the proof of:\\
  Lemma~\thref{l:geo-mapping-permut}.

\item[Lemma~\thref{l:0-face-aff-space-is-vertex}] \mbox{}\\
  is not yet used.

\item[Definition~\thref{d:l-face}] \mbox{}\\
  is explicitly cited in the proof of:\\
  Lemma~\thref{l:l-face-is-incl-l-face-aff-space},\\
  Lemma~\thref{l:l-face-is-simplex},\\
  Lemma~\thref{l:d-1-face-aff-sp-is-hyperface-aff-sp},\\
  Lemma~\thref{l:prop-geo-l-face-mapping}.

\item[Lemma~\thref{l:l-face-is-incl-l-face-aff-space}] \mbox{}\\
  is not yet used.

\item[Lemma~\thref{l:l-face-is-simplex}] \mbox{}\\
  is explicitly cited in the proof of:\\
  Lemma~\thref{l:geo-mapping-permut}.

\item[Lemma~\thref{l:d-1-face-aff-sp-is-hyperface-aff-sp}] \mbox{}\\
  is explicitly cited in the proof of:\\
  Lemma~\thref{l:geo-hyperface-mapping},\\
  Lemma~\thref{l:hyperface-geo-mapping-of-Pkd-is-Pkdmi}.

\item[Definition~\thref{d:geo-l-face-mapping}] \mbox{}\\
  is explicitly cited in the proof of:\\
  Lemma~\thref{l:geo-d-face-mapping-is-geo-mapping},\\
  Lemma~\thref{l:prop-geo-l-face-mapping},\\
  Lemma~\thref{l:geo-mapping-permut}.

\item[Lemma~\thref{l:geo-d-face-mapping-is-geo-mapping}] \mbox{}\\
  is explicitly cited in the proof of:\\
  Lemma~\thref{l:geo-mapping-of-Pkd-is-Pkd}.

\item[Lemma~\thref{l:prop-geo-l-face-mapping}] \mbox{}\\
  is explicitly cited in the proof of:\\
  Lemma~\thref{l:geo-l-face-mapping-of-Pkd-is-Pkl},\\
  Lemma~\thref{l:geo-hyperface-mapping},\\
  Lemma~\thref{l:geo-mapping-permut}.

\item[Lemma~\thref{l:geo-l-face-mapping-of-Pkd-is-Pkl}] \mbox{}\\
  is explicitly cited in the proof of:\\
  Lemma~\thref{l:geo-mapping-of-Pkd-is-Pkd},\\
  Lemma~\thref{l:hyperface-geo-mapping-of-Pkd-is-Pkdmi}.

\item[Lemma~\thref{l:geo-mapping-of-Pkd-is-Pkd}] \mbox{}\\
  is not yet used.

\item[Lemma~\thref{l:geo-hyperface-mapping}] \mbox{}\\
  is explicitly cited in the proof of:\\
  Lemma~\thref{l:im-nodes-by-geo-hyperface-mapping},\\
  Lemma~\thref{l:lag-lin-forms-Pkd-inj},\\
  Lemma~\thref{l:face-unisolvence-Pkd}.

\item[Lemma~\thref{l:hyperface-geo-mapping-of-Pkd-is-Pkdmi}] \mbox{}\\
  is explicitly cited in the proof of:\\
  Lemma~\thref{l:lag-lin-forms-Pkd-inj},\\
  Lemma~\thref{l:face-unisolvence-Pkd}.

\item[Lemma~\thref{l:geo-mapping-permut}] \mbox{}\\
  is explicitly cited in the proof of:\\
  Lemma~\thref{l:factor-zero-pol-hyperpl-Pkd}.

\item[Definition~\thref{d:lag-nodes-Pkd}] \mbox{}\\
  is explicitly cited in the proof of:\\
  Lemma~\thref{l:lag-nodes-Pkd-for-d-eq-1-are-lag-nodes-Pk1},\\
  Lemma~\thref{l:num-lag-nodes-Pkd},\\
  Lemma~\thref{l:baryc-coor-lag-nodes-Pkd},\\
  Lemma~\thref{l:vert-lag-nodes-Pkd},\\
  Lemma~\thref{l:lag-nodes-Pkd-ref},\\
  Lemma~\thref{l:lag-nodes-Pkd-im-ref},\\
  Lemma~\thref{l:lag-lin-forms-Pkd-inj}.

\item[Lemma~\thref{l:lag-nodes-Pkd-for-d-eq-1-are-lag-nodes-Pk1}] \mbox{}\\
  is explicitly cited in the proof of:\\
  Lemma~\thref{l:lag-lin-forms-Pkd-for-d-eq-1-are-lag-lin-forms-Pk1}.

\item[Lemma~\thref{l:num-lag-nodes-Pkd}] \mbox{}\\
  is explicitly cited in the proof of:\\
  Lemma~\thref{l:face-hyperpl-lag-nodes-Pkd},\\
  Lemma~\thref{l:card-lag-lin-forms-Pkd}.

\item[Lemma~\thref{l:baryc-coor-lag-nodes-Pkd}] \mbox{}\\
  is explicitly cited in the proof of:\\
  Lemma~\thref{l:equiv-def-sub-vert-lag-nodes-Pkd},\\
  Lemma~\thref{l:Pkm1d-sub-nodes-sub-vert-are-some-nodes-Pkd},\\
  Lemma~\thref{l:lag-nodes-Pkd-ref},\\
  Lemma~\thref{l:face-hyperpl-lag-nodes-Pkd},\\
  Lemma~\thref{l:im-nodes-by-geo-hyperface-mapping}.

\item[Lemma~\thref{l:vert-lag-nodes-Pkd}] \mbox{}\\
  is explicitly cited in the proof of:\\
  Lemma~\thref{l:lag-nodes-Pid-are-vert}.

\item[Lemma~\thref{l:lag-nodes-Pid-are-vert}] \mbox{}\\
  is explicitly cited in the proof of:\\
  Lemma~\thref{l:decomp-P1d-pol-with-sigma}.

\item[Definition~\thref{d:sub-vert-lag-nodes-Pkd}] \mbox{}\\
  is explicitly cited in the proof of:\\
  Lemma~\thref{l:equiv-def-sub-vert-lag-nodes-Pkd}.

\item[Lemma~\thref{l:equiv-def-sub-vert-lag-nodes-Pkd}] \mbox{}\\
  is explicitly cited in the proof of:\\
  Lemma~\thref{l:sub-vert-aff-indep},\\
  Lemma~\thref{l:Pkm1d-sub-nodes-sub-vert-are-some-nodes-Pkd}.

\item[Lemma~\thref{l:sub-vert-aff-indep}] \mbox{}\\
  is explicitly cited in the proof of:\\
  Lemma~\thref{l:Pkm1d-sub-nodes-sub-vert-are-some-nodes-Pkd},\\
  Lemma~\thref{l:lag-lin-forms-Pkd-inj}.

\item[Lemma~\thref{l:Pkm1d-sub-nodes-sub-vert-are-some-nodes-Pkd}] \mbox{}\\
  is explicitly cited in the proof of:\\
  Lemma~\thref{l:lag-lin-forms-Pkd-inj}.

\item[Lemma~\thref{l:lag-nodes-Pkd-ref}] \mbox{}\\
  is explicitly cited in the proof of:\\
  Lemma~\thref{l:lag-nodes-Pkd-ref-for-d-eq-1-are-lag-nodes-Pk1-ref},\\
  Lemma~\thref{l:equiv-def-lag-nodes-Pkd-ref},\\
  Lemma~\thref{l:num-lag-nodes-Pkd-ref},\\
  Lemma~\thref{l:im-nodes-by-geo-hyperface-mapping},\\
  Theorem~\thref{t:Pkd-lag-fe-ref}.

\item[Lemma~\thref{l:lag-nodes-Pkd-ref-for-d-eq-1-are-lag-nodes-Pk1-ref}] \mbox{}\\
  is explicitly cited in the proof of:\\
  Lemma~\thref{l:lag-lin-forms-Pkd-ref-for-d-eq-1-are-lag-lin-forms-Pk1-ref}.

\item[Lemma~\thref{l:equiv-def-lag-nodes-Pkd-ref}] \mbox{}\\
  is explicitly cited in the proof of:\\
  Lemma~\thref{l:lag-nodes-Pkd-im-ref}.

\item[Lemma~\thref{l:num-lag-nodes-Pkd-ref}] \mbox{}\\
  is not yet used.

\item[Lemma~\thref{l:lag-nodes-Pkd-im-ref}] \mbox{}\\
  is explicitly cited in the proof of:\\
  Lemma~\thref{l:lag-lin-forms-Pkd-im-ref}.

\item[Lemma~\thref{l:face-hyperpl-lag-nodes-Pkd}] \mbox{}\\
  is explicitly cited in the proof of:\\
  Lemma~\thref{l:lag-lin-forms-Pkd-inj},\\
  Lemma~\thref{l:face-unisolvence-Pkd}.

\item[Lemma~\thref{l:im-nodes-by-geo-hyperface-mapping}] \mbox{}\\
  is explicitly cited in the proof of:\\
  Lemma~\thref{l:lag-lin-forms-Pkd-inj},\\
  Lemma~\thref{l:face-unisolvence-Pkd}.

\item[Definition~\thref{d:lag-lin-forms-Pkd}] \mbox{}\\
  is explicitly cited in the proof of:\\
  Lemma~\thref{l:lag-lin-forms-Pkd-for-d-eq-1-are-lag-lin-forms-Pk1},\\
  Lemma~\thref{l:lag-lin-forms-are-linear-Pkd},\\
  Lemma~\thref{l:card-lag-lin-forms-Pkd},\\
  Lemma~\thref{l:lag-lin-forms-Pkd-im-ref},\\
  Lemma~\thref{l:lag-lin-forms-P0d-inj},\\
  Lemma~\thref{l:decomp-P1d-pol-with-sigma},\\
  Lemma~\thref{l:lag-lin-forms-Pkd-inj},\\
  Theorem~\thref{t:unisolvence-Pkd},\\
  Lemma~\thref{l:face-unisolvence-Pkd}.

\item[Lemma~\thref{l:lag-lin-forms-Pkd-for-d-eq-1-are-lag-lin-forms-Pk1}] \mbox{}\\
  is explicitly cited in the proof of:\\
  Lemma~\thref{l:lag-lin-forms-Pkd-inj},\\
  Lemma~\thref{l:Pkd-lag-fe-for-d-eq-1-is-Pk1-lag-fe}.

\item[Lemma~\thref{l:lag-lin-forms-are-linear-Pkd}] \mbox{}\\
  is not yet used.

\item[Lemma~\thref{l:card-lag-lin-forms-Pkd}] \mbox{}\\
  is explicitly cited in the proof of:\\
  Lemma~\thref{l:unisolvence-P0d},\\
  Lemma~\thref{l:unisolvence-P1d},\\
  Theorem~\thref{t:unisolvence-Pkd}.

\item[Definition~\thref{d:lag-lin-forms-Pkd-ref}] \mbox{}\\
  is explicitly cited in the proof of:\\
  Lemma~\thref{l:lag-lin-forms-Pkd-ref-for-d-eq-1-are-lag-lin-forms-Pk1-ref},\\
  Lemma~\thref{l:lag-lin-forms-Pkd-im-ref},\\
  Theorem~\thref{t:Pkd-lag-fe-ref}.

\item[Lemma~\thref{l:lag-lin-forms-Pkd-ref-for-d-eq-1-are-lag-lin-forms-Pk1-ref}] \mbox{}\\
  is explicitly cited in the proof of:\\
  Lemma~\thref{l:Pkd-lag-fe-ref-for-d-eq-1-is-Pk1-lag-fe-ref}.

\item[Lemma~\thref{l:lag-lin-forms-Pkd-im-ref}] \mbox{}\\
  is not yet used.

\item[Lemma~\thref{l:lag-lin-forms-P0d-inj}] \mbox{}\\
  is explicitly cited in the proof of:\\
  Lemma~\thref{l:unisolvence-P0d}.

\item[Lemma~\thref{l:unisolvence-P0d}] \mbox{}\\
  is explicitly cited in the proof of:\\
  Theorem~\thref{t:Pkd-lag-fe}.

\item[Lemma~\thref{l:decomp-P1d-pol-with-sigma}] \mbox{}\\
  is explicitly cited in the proof of:\\
  Lemma~\thref{l:lag-lin-forms-P1d-inj}.

\item[Lemma~\thref{l:lag-lin-forms-P1d-inj}] \mbox{}\\
  is explicitly cited in the proof of:\\
  Lemma~\thref{l:unisolvence-P1d},\\
  Lemma~\thref{l:lag-lin-forms-Pkd-inj}.

\item[Lemma~\thref{l:unisolvence-P1d}] \mbox{}\\
  is not yet used.

\item[Lemma~\thref{l:factor-zero-pol-last-ref-hyperpl}] \mbox{}\\
  is explicitly cited in the proof of:\\
  Lemma~\thref{l:factor-zero-pol-hyperpl-Pkd}.

\item[Lemma~\thref{l:factor-zero-pol-hyperpl-Pkd}] \mbox{}\\
  is explicitly cited in the proof of:\\
  Lemma~\thref{l:lag-lin-forms-Pkd-inj}.

\item[Lemma~\thref{l:lag-lin-forms-Pkd-inj}] \mbox{}\\
  is explicitly cited in the proof of:\\
  Theorem~\thref{t:unisolvence-Pkd},\\
  Lemma~\thref{l:face-unisolvence-Pkd}.

\item[Theorem~\thref{t:unisolvence-Pkd}] \mbox{}\\
  is explicitly cited in the proof of:\\
  Theorem~\thref{t:Pkd-lag-fe}.

\item[Lemma~\thref{l:face-unisolvence-Pkd}] \mbox{}\\
  is not yet used.

\item[Theorem~\thref{t:Pkd-lag-fe}] \mbox{}\\
  is explicitly cited in the proof of:\\
  Lemma~\thref{l:Pkd-lag-fe-for-d-eq-1-is-Pk1-lag-fe},\\
  Theorem~\thref{t:Pkd-lag-fe-ref}.

\item[Lemma~\thref{l:Pkd-lag-fe-for-d-eq-1-is-Pk1-lag-fe}] \mbox{}\\
  is not yet used.

\item[Theorem~\thref{t:Pkd-lag-fe-ref}] \mbox{}\\
  is explicitly cited in the proof of:\\
  Lemma~\thref{l:Pkd-lag-fe-ref-for-d-eq-1-is-Pk1-lag-fe-ref}.

\item[Lemma~\thref{l:Pkd-lag-fe-ref-for-d-eq-1-is-Pk1-lag-fe-ref}] \mbox{}\\
  is not yet used.

\end{description}

\end{document}